\documentclass[11pt,letterpaper]{report}

\usepackage[T1]{fontenc}
\usepackage[utf8]{inputenc}
\usepackage{lmodern}
\usepackage[english]{babel}
\usepackage{microtype}             

\usepackage{geometry}
\usepackage{amsmath}
\usepackage{amssymb}
\usepackage{amsfonts}
\usepackage{amsthm}
\newtheorem{theorem}{Theorem}       
\usepackage{mathtools}
\usepackage{dsfont}

\usepackage{enumitem}
\usepackage[table]{xcolor}          
\usepackage{longtable}              
\usepackage{multicol}               
\usepackage{makeidx}
\makeindex
\definecolor{azulobs}{HTML}{037FFC}
\definecolor{rojocita}{HTML}{E8491D}
\definecolor{rosasobrio}{HTML}{C2185B}

\usepackage{graphicx}
\graphicspath{{figures/}}
\usepackage[most]{tcolorbox}

\tcbset{
    enhanced,
    colback=cyan!20,
    colbacktitle=black,
    coltitle=white,
    frame hidden,
    fonttitle=\bfseries,
    boxrule=0pt,
    opacityback=0.6,
    opacitybacktitle=0.6,
    sharp corners=downhill
}

\newtcolorbox{proofbox}{
    enhanced,
    blanker,
    borderline west={1.2pt}{0pt}{azulobs!70},
    left=10pt,
    breakable,
    sharp corners
}
\renewenvironment{proof}[1][\proofname]{%
    \begin{proofbox}%
    \noindent\textit{#1.}\space
}{%
    \hfill$\square$%
    \end{proofbox}%
}
\renewcommand{\proofname}{Proof}

\usepackage{subcaption}
\usepackage{booktabs}
\usepackage{array}
\usepackage{float}

\usepackage{tikz}
\usetikzlibrary{decorations.pathreplacing}

\usepackage{fancyhdr}
\fancypagestyle{plain}{%
    \fancyhf{}
    \fancyfoot[R]{\thepage}

}

\usepackage{titlesec}

\titleformat{\chapter}[display]
  {\normalfont\Large\bfseries\raggedleft}
  {\titlerule[1pt]\vspace{0.5em}Chapter \thechapter}
  {0.4em}
  {\hfill}
  [{\vspace{0.5em}\titlerule[1pt]}]
  
\titlespacing*{\chapter}{0pt}{1.5em}{1.2em}

\titleformat{\section}
  {\normalfont\Large\bfseries}{\thesection}{0.8em}{}
\titlespacing*{\section}{0pt}{1em}{0.5em}

\titleformat{\subsection}
  {\normalfont\large\bfseries}{\thesubsection}{0.8em}{}
\titlespacing*{\subsection}{0pt}{0.8em}{0.4em}

\titleformat{\subsubsection}
  {\normalfont\normalsize\bfseries}{\thesubsubsection}{0.8em}{}
\titlespacing*{\subsubsection}{0pt}{0.7em}{0.3em}

\newcommand{\formatoapendices}{%
  \titleformat{\chapter}[display]
    {\normalfont\Large\bfseries\raggedleft}
    {\titlerule[1pt]\vspace{0.5em}Appendix \thechapter}
    {0.4em}
    {\hfill}
    [{\vspace{0.5em}\titlerule[1pt]}]%
}



\usepackage{csquotes}              

\usepackage{hyperref}
\hypersetup{
    colorlinks=true,
    linktoc=page,
    linkcolor=blue,
    citecolor=rojocita,
    urlcolor=rosasobrio,
    pdfauthor={Javier Gomez Morales},
    pdftitle={Analysis of the Principal Components of Correlation Matrices of S\&P 500 Financial Data from an Econophysics Perspective}
}

\usepackage[backend=biber,
            style=numeric,
            sorting=none,
            doi=true,
            url=true,
            isbn=true,
            giveninits=true]{biblatex}
\DeclareFieldFormat{doi}{%
  \mkbibacro{DOI}\addcolon\space
  \href{https://doi.org/#1}{\nolinkurl{#1}}}

\title{Analysis of the Principal Components of Correlation Matrices of S\&P 500 Financial Data from an Econophysics Perspective}
\author{
Javier G\'omez Morales\\
Faculty of Sciences, National Autonomous University of Mexico
}
\date{2026}

\begin{document}

\begin{titlepage}
\centering
\vspace*{1.0cm}
{\Huge\bfseries
Analysis of the Principal Components\\[0.15cm]
of Correlation Matrices\\[0.15cm]
of S\&P 500 Financial Data\\[0.15cm]
from an Econophysics Perspective
\par}
\vspace{0.8cm}
{\par\noindent\centering\rule{0.82\textwidth}{0.5pt}\par}
\vspace{1.0cm}
{\Large
Javier G\'omez Morales
\par}
\vspace{0.25cm}
{\large
\href{mailto:japogm@ciencias.unam.mx}{\texttt{japogm@ciencias.unam.mx}}
\par}
\vspace{0.8cm}
{\large
Faculty of Sciences\\
Universidad Nacional Aut\'onoma de M\'exico
\par}
\vspace{1.4cm}
{\large
Undergraduate thesis submitted in fulfillment of the requirements for the degree of\\[0.15cm]
\textbf{Bachelor of Physics}
\par}
\vspace{1.0cm}
{\large
Advisor:\\[0.15cm]
Dr. Manuel Mija\'il Mart\'inez Ramos
\par}
\vfill
\begin{minipage}{0.82\textwidth}
\centering
{\normalsize
English version prepared for arXiv\\[0.15cm]
Original Spanish version: \href{https://tesiunamdocumentos.dgb.unam.mx/ptd2026/ene_mar/0881908/Index.html}{UNAM digital thesis repository}\\[0.15cm]
2026
\par}
\end{minipage}
\vspace*{0.8cm}
\end{titlepage}

\clearpage
\pagenumbering{roman}

\chapter*{Abstract}
\phantomsection
\addcontentsline{toc}{chapter}{Abstract}

This work studies the collective behavior of the financial market using the daily adjusted prices of a group of S\&P 500 companies, spanning from \textbf{January 3, 2012} to \textbf{December 29, 2023}.

The purpose of this work is to describe how correlation matrices evolve by means of spectral methods. To this end, the eigenvalues were analyzed, the eigenvectors were examined through their squared entries, and matrices defined as linear combinations of eigenvectors, weighted by their eigenvalues, were constructed. With this approach, the existence and evolution of the Discrete Market States (\emph{MS}) were identified, and the distinctiveness of the \emph{COVID} State as a differentiated regime was assessed.

As noted by Mart\'inez Ramos in \emph{Series de tiempo financieras en el marco de la F\'isica Estad\'istica}~\cite{martinezramosSeriesTiempoFinancieras2024}, the \emph{COVID} State manifests itself as an atypical \emph{State of the Market} in the correlation matrices. In the present work, the analysis of eigenvectors through their squared entries, together with \emph{k}-Means \emph{clustering}, made it possible to identify the sectors and companies with the greatest relative participation in the collective dynamics of the market. In particular, the configuration with windows of $q=40$ and $k=5$ allowed the \emph{COVID} State to be recognized as a singular regime. Nevertheless, none of the approaches considered accurately reproduced the Discrete Market States (\emph{MS}) obtained through the direct \emph{clustering} of the correlation matrices.
\vspace{2.5cm}

\noindent\textbf{Keywords:} Econophysics, Discrete Market States, \emph{COVID} State, correlation matrices, eigenvalue, eigenvector, \emph{k}-Means.

\clearpage
\tableofcontents

\clearpage
\pagenumbering{arabic}

\chapter{Introduction}
\label{chapter:intro}


\section{Background}

Prior to the consolidation of Econophysics, some physicists applied their knowledge to financial problems. A prominent example is the work of Isaac Newton at the \emph{Royal Mint} during the period following the \emph{Great Recoinage} of 1696. As \emph{Warden}\footnote{The office of \emph{Warden of the Mint} entailed administrative and judicial responsibilities for protecting the currency of the United Kingdom.}, and later as \emph{Master of the Mint}\footnote{The \emph{Master} was the highest-ranking official at the \emph{Royal Mint}, responsible for the production of coinage, overseeing its quality and preventing crimes such as counterfeiting.}, Newton reorganized the minting processes, applied experimental methods to assess metals, and led a legal campaign against the most sophisticated counterfeiters of the era, such as William Chaloner\cite{levensonNewtonCounterfeiterUnknown2009}. His method, which combined observation, numerical analysis, and legal strategy, had a positive impact on the stability of the English monetary system.
\index{Background!Newton, Isaac}

The doctoral thesis of Louis Bachelier, \emph{Th\'eorie de la Sp\'eculation}, presented in 1900, is a precursor of the union between Physics and Economics. Bachelier, under the supervision of Henri Poincar\'e, represented the prices of a set of stocks listed on the Paris stock exchange as continuous-time random walks~\cite{bachelierTheorieSpeculation1900}. Today this behavior is understood as a stochastic process with diffusion~\cite{voitStatisticalMechanicsFinancial2005}. This work was pioneering, anticipating the mathematics of random motion and laying the foundations for quantitative financial theory and Econophysics.

\index{Background!Bachelier, Louis}

\section{Motivation}

From the perspective of Econophysics, price fluctuations in stock markets can be modeled as a complex system that transitions between \textbf{discrete market states} \cite{munnixIdentifyingStatesFinancial2012}. M\"unnix et al. proposed identifying such states from correlation patterns that emerge and persist over well-defined time intervals. Under this approach, the market is represented as a sequence of regimes that alternate between periods of high and low synchronization, which makes it possible to analyze its dynamics through transitions between states.

Although \emph{PCA} has been widely used in finance, the detailed application of the leading eigenvector of the market correlation matrix has been comparatively less explored relative to other areas. A recent example is found in \cite{ochoa-gonzalezSubnarrowBandSleep2024}, where the joint analysis of eigenvalues and eigenvectors in multiband correlation matrices made it possible to distinguish sleep stages from electroencephalography recordings. This parallel is illustrative: even in the presence of high variability across individuals, or, in the financial context, across assets and sectors, spectral decomposition makes it possible to extract stable features that facilitate the distinction between states.

\index{S\&P 500}

The \textbf{S\&P 500} is one of the principal stock-market indices of the United States, developed by the financial firm \emph{Standard \& Poor's}\footnote{\emph{Standard \& Poor's} is a financial services company founded in 1860, specializing in the construction of stock-market indices and in market analysis.}. The index brings together approximately 500 companies\footnote{At the time of writing of this thesis, the S\&P is composed of \textbf{503 companies}.} selected on the basis of criteria of size, liquidity, and representativeness of the U.S. economy, whose shares are traded on the \emph{New York Stock Exchange (NYSE)}\footnote{The \emph{New York Stock Exchange} (NYSE) is the stock exchange with the largest market capitalization, located on Wall Street, New York.}\index{NYSE} and on the \emph{NASDAQ}\footnote{The \emph{NASDAQ} is an electronic stock exchange headquartered in New York, with a high concentration of technology companies.}.\index{NASDAQ}

In research in econophysics and quantitative finance, the S\&P~500 is used as a representative indicator (\emph{proxy})\index{Proxy} of the U.S. stock market. Owing to its sectoral breadth and its depth of liquidity, it provides a reliable quantitative reference for describing the joint behavior of a broad spectrum of stocks.

\index{Matrix!correlation}
\index{Matrix!random}

A tool that was developed to study \emph{complex systems} is random matrix theory (\textbf{RMT})\index{Theory!random matrix}. Historically, \textbf{RMT} originated in the context of nuclear physics with the pioneering works of Eugene P. Wigner, motivated by the need to understand the energy spectrum of complex atomic nuclei (for example, heavy nuclei such as uranium) \cite{wignerDistributionRootsCertain1958, wignerCharacteristicVectorsBordered1955}. Decades later, this approach was transferred to finance: instead of energy levels, price fluctuations and their correlations are studied. Correlation matrices of returns serve to describe how financial markets are related. From a spectral perspective, Wirtz et al. demonstrated that the extreme eigenvalues, both the largest and the smallest, of correlated random matrices concentrate statistically relevant information, whereas the central region of the spectrum is dominated by largely noisy contributions \cite{wirtzLimitingStatisticsLargest2015}. It has likewise been found that a large part of the spectrum is noise (the Mar\v{c}enko--Pastur distribution)\index{Mar\v{c}enko-Pastur!distribution} and that only a few eigenvalues reflect genuine economic information \cite{mantegnaIntroductionEconophysicsCorrelations2004, lalouxNoiseDressingFinancial1999, plerouRandomMatrixTheory2000}.

The way in which markets move can be explained by recurrent patterns, called \emph{State of the Market} (\emph{SOTM})\footnote{\emph{State of the Market} (SOTM) refers to the value of the largest eigenvalue of the correlation matrix \cite{mantegnaIntroductionEconophysicsCorrelations2004}.} and \emph{Market States} (\emph{MS})\footnote{\emph{Market States} (MS): the correlation matrices can be grouped into different \emph{clusters} known as \textit{discrete market states} \cite{munnixIdentifyingStatesFinancial2012}. From this point onward, SOTM and MS will be used as their English-language acronyms.}.\index{State of the Market} \index{Market!States} These patterns help us to identify trends in the change of correlations between financial assets over time, which makes it possible to describe and compare events such as crises or significant shifts in the market.

The concept of \emph{stylized facts}\index{Facts!stylized} was originally introduced by Kaldor (1961) in the field of macroeconomics, and was later adopted in the field of finance to describe empirical regularities that hold over time and across different contexts. These facts serve as a starting point in the construction of theoretical models that seek to reproduce the properties observed in financial data\cite{chakrabortiEconophysicsReviewEmpirical2011}.

Among the most relevant \emph{stylized facts} in financial time series, the following stand out:

\begin{enumerate}[label=\roman*)]
    \item \textbf{Heavy-tailed return distributions.}
    The distributions of returns exhibit heavier tails than those of a normal distribution, which indicates a higher probability of occurrence of extreme events in financial prices \cite{chakrabortiEconophysicsReviewEmpirical2011}.

    \item \textbf{Absence of linear autocorrelation in returns.}
    On short time scales, successive returns show no significant linear correlation, which suggests an essentially unpredictable evolution in the direction of prices \cite{mantegnaIntroductionEconophysicsCorrelations2004}.

    \item \textbf{Temporal persistence in volatility (\emph{volatility clustering}).}
    Periods of high volatility tend to cluster temporally, being followed by phases of lower variability. This pattern, known as \emph{volatility clustering}, has been widely documented in the literature \cite{voitStatisticalMechanicsFinancial2005,chakrabortiEconophysicsReviewEmpirical2011}.

    \item \textbf{Correlation between the dominant eigenvalue and the average correlation.}
    The largest eigenvalue $\lambda_{N}$ of the correlation matrices exhibits a marked correlation with the average of the correlation coefficients between assets. This average, denoted as $\langle C_{ij} \rangle$, is:

    \begin{equation}
        \langle C_{ij} \rangle = \frac{2}{N(N-1)} \sum_{i<j} C_{ij},
        \label{eq:promedio_correlaciones}
    \end{equation}

    where $C_{ij}$ represents the Pearson correlation coefficient between the returns of assets $i$ and $j$, and $N$ is the total number of assets considered. This behavior motivates the study of the collective interactions of the financial market through dominant spectral modes. \cite{martinezramosSeriesTiempoFinancieras2024}.
\end{enumerate}

The \emph{ansatz}\footnote{A German term meaning \textbf{initial assumption}. In physics and mathematics it designates the form proposed for the solution of a problem.}\index{Ansatz} of this thesis consists in treating financial markets as a \emph{complex system}\index{System!complex}. In this way, the general aim is to understand the structure and dynamics of the correlations of the returns of a subset of S\&P~500 stocks by means of tools from statistical mechanics and complex systems. \cite{munnixIdentifyingStatesFinancial2012,pharasiIdentifyingLongtermPrecursors2018,pharasiDynamicsMarketStates2024}.

In the Latin American context, Mart\'inez Ramos et al.~\cite{martinezramosSeriesTiempoFinancieras2024} analyzed the dynamics of the leading eigenvectors $\{\mathbf{v}_j(t)\}_{j=1}^{3}$. In particular, the eigenvector $\mathbf{v}_N(t)$, associated with the largest eigenvalue $\lambda_N(t)$, indicates the direction of maximum variance and is usually interpreted as the predominant collective component of the system.

The main objective of this thesis is to use the correlation structure and the information contained in the eigenvalues and their associated eigenvectors, in particular the dominant one, to characterize discrete market states and their transitions, as well as to identify a representative regime associated with the \emph{COVID} State.

\chapter{Theoretical Framework}
\label{chapter:marco-teorico}
\section{Logarithmic Returns}
\index{Return!logarithmic}
Let $S_i(t)$ be the daily closing prices for $i\in\{1,\dots,N\}$ assets and $t\in\{0, 1,\dots,T\}$ dates.
Following the notation of \cite{heckensNewCollectivityMeasures2022}, the logarithmic returns are
\begin{equation}
  G_i(t)\;=\;\ln\!\left(\frac{S_i(t)}{S_i(t-\Delta t)}\right),\qquad \Delta t=1,\quad t\in\{1,\dots,T\},\quad i=1,\dots,N.
  \label{eq:def-retornos-HG}
\end{equation}

The returns were organized into matrices as shown:
\begin{equation}
  \mathds{X} \;=\;
  \begin{pmatrix}
    G_{1}(1)   & \cdots & G_{N}(1)\\
    \vdots     & \ddots & \vdots  \\
    G_{1}(T) & \cdots & G_{N}(T)
  \end{pmatrix}
  \in \mathbb{R}^{T\times N},
  \qquad t\in\{1,\ldots,T\},\quad i\in\{1,\ldots,N\}.
  \label{eq:def-Gm-ventanas}
\end{equation}

\section{Short-Epoch Window and Full-Horizon Matrices}\label{section:epocas-horizonte-completo}
\index{Window!short-epoch}\index{Horizon!full}\index{Matrix!correlation}

In the analysis of financial series, an \emph{epoch} is a contiguous time window of length \(q\) over which returns and associated statistics (mean, covariance, and correlation) are computed. For each time \(t\), the return matrix \(G(t,q)\in\mathbb{R}^{q\times N}\) is constructed, with \(N\) assets as columns and \(q\) consecutive observations as rows; this arrangement preserves \(N\) columns and \emph{short-epoch} windows of length \(q\). From \(G(t,q)\) the correlation matrix \(C(t,q)\) is estimated, usually with overlapping windows in order to track the temporal evolution.

Given the non-stationary nature of the stock market, the choice of \(q\) must balance two effects: with \emph{short} \(q\), the estimate \(C(t,q)\) becomes noisier; with \emph{long} \(q\), it loses sensitivity to the prevailing state of the market~\cite{munnixIdentifyingStatesFinancial2012}. By contrast, the \emph{full-horizon} matrix \(C^{(\mathrm{full})}\) is computed with the entire available sample and summarizes the long-term average structure of the market; it is used as a reference point for comparing the local estimates \(C(t,q)\), which capture the short-term dynamics.

For daily data, it is reasonable to employ \textbf{one-calendar-month windows}, which operationally correspond to approximately \(20\) trading days\footnote{The exact figure depends on the holiday calendar and on any market suspensions.}. This horizon offers a temporal resolution sufficiently fine to track short-term variations in the correlation structure \cite{martinezramosSeriesTiempoFinancieras2024, pharasiDynamicsMarketStates2024}. Another common choice consists in using \textbf{two-calendar-month windows} (approx.\ \(40\) trading days) when computing sliding correlation matrices, which helps to attenuate statistical noise \cite{munnixIdentifyingStatesFinancial2012}.

\subsection{Overlapping Moving Windows}

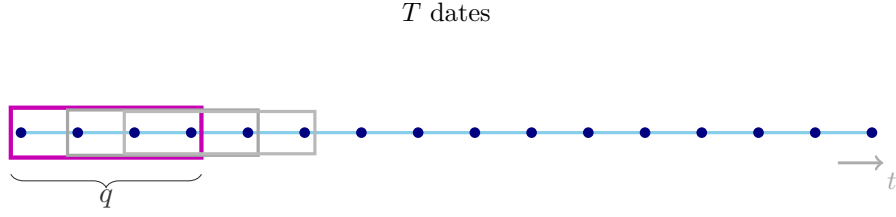
\begin{figure}[ht]
\centering
\begin{tikzpicture}[x=0.75cm,y=1.0cm]
  \definecolor{skyblue}{RGB}{135,206,235} 
  \definecolor{navy}{RGB}{0,0,128}        
  \definecolor{magentaedge}{RGB}{200,0,180}

  \def\T{16}    
  \def\q{4}     
  \def\eps{0.18}

  \draw[line width=1.2pt, skyblue] (1,0) -- (\T,0);

  \foreach \i in {1,...,\T} \fill[navy] (\i,0) circle (2.0pt);

  \draw[draw=magentaedge, line width=1.6pt]
       ({1-\eps},-0.33) rectangle ({\q+\eps},0.33);

  \draw[draw=gray!70, line width=1.2pt]
       (2-\eps,-0.30) rectangle (\q+1+\eps,0.30);

  \draw[draw=gray!55, line width=1.2pt]
       (3-\eps,-0.28) rectangle (\q+2+\eps,0.28);

  \draw[decorate,decoration={brace,mirror,amplitude=5pt},black!80]
       ({1-\eps},-0.55) -- ({\q+\eps},-0.55);
  \node[below=2pt,black!80] at ({(1+\q)/2},-0.55) {$q$};

  \node[font=\small] at ({(1+\T)/2},1.6) {$T$ dates};

  \node[below, gray!70] at (\T+0.35,-0.40) {\small $t$};
  \draw[->, gray!65, line width=1.1pt] (\T-0.6,-0.40) -- (\T+0.2,-0.40);
\end{tikzpicture}

\caption[Short-epoch windows] {For a total of $T$ dates and a window length $q$ with $1<q\le T$, the moving windows ending at $t=q,\dots,T$ generate exactly $T-q+1$ contiguous subsamples. Consecutive windows share $q-1$ dates.}
\label{fig:ventanas-moviles}
\end{figure}

Let \(\mathds{X}\in\mathbb{R}^{T\times N}\) be the \emph{logarithmic return matrix} defined in Eq.~\ref{eq:def-Gm-ventanas}; fixing an \emph{epoch length} \(q\in\mathbb{N}\) with \(1<q\le T\), the moving window that \emph{ends} at \(s\) is defined, for each \emph{closing index} \(s\in\{q,\dots,T\}\), as the set of consecutive rows
\begin{equation}
\mathcal{E}(s,q)\;=\;\{\,s-q+1,\;s-q+2,\;\dots,\;s \,\},\qquad s\in\{q,\dots,T\}.
\label{eq:def-epoca}
\end{equation}

Each window induces the submatrix
\begin{equation}
    \mathds{X}_{\mathcal{E}(s,q)}\in\mathbb{R}^{q\times N}
\label{eq:submatriz-q}
\end{equation}

formed by the \(q\) rows of \(\mathds{X}\) with indices in \(\mathcal{E}(s,q)\). When the closing index is incremented from \(s\) to \(s+1\), the window is displaced one day forward: the new window retains exactly \(q-1\) rows of the previous one and adds one date at the right end. This overlap is shown in Fig.~\ref{fig:ventanas-moviles}. The total number of windows generated is
\begin{equation}
\left|\left\{\mathcal{E}(s,q):\;s\in\{q,\dots,T\}\right\}\right|
\;=\;T-q+1.
\label{eq:conteo-ventanas}
\end{equation}

\section{Covariance Matrix}
\index{Covariance!matrix}

Fixing a moving window $\mathcal{E}(s,q)$ and the submatrix $\mathds{X}_{\mathcal{E}(s,q)}$ defined in~\eqref{eq:submatriz-q}, the mean for asset $i$ is defined as:
\begin{equation}
\mu_i \;=\; \frac{1}{q}\sum_{u\in\mathcal{E}(s,q)} G_i(u)\,\qquad i\in\{1,\dots,N\},
\label{eq:promedios-locales}
\end{equation}

The covariance between asset $i$ and $j$ is:
\begin{equation}
    \Sigma_{ij}[\mathcal{E}(s,q)]=\frac{1}{q-1}\sum_{u\in\mathcal{E}(s,q)}[G_i(u)-\mu_i][G_j(u)-\mu_j]
\label{eq:covarianza-ij}
\end{equation}

This construction induces the full covariance matrix 
$\boldsymbol{\Sigma}_{\mathcal{E}(s,q)}$, 
which, by definition, is symmetric and positive semidefinite.

\section{Correlation Matrix}
\index{Correlation!matrix}

The standard deviation $\sigma_i$ is defined as\footnote{Strictly speaking, it should be denoted as $\sigma_i[\mathcal{E}(s,q)]$ to emphasize its explicit dependence on the time window $\mathcal{E}(s,q)$. Nevertheless, throughout the text the simplified notation $\sigma_i$ is employed in order to avoid overloading the notation.}:
\begin{equation}
\sigma_i = \sqrt{\frac{1}{q-1}\sum_{u\in\mathcal{E}(s,q)} [G_i(u)-\mu_i]^2}
\label{eq:desv-est}
\end{equation}

The correlation between assets $i$ and $j$ in the window $\mathcal{E}(s,q)$ is expressed as:
\begin{equation}
C_{ij}\big[\mathcal{E}(s,q)\big] \;=\; 
\frac{\Sigma_{ij}\big[\mathcal{E}(s,q)\big]}{\sigma_i\,\sigma_j}\,, 
\label{eq:def-corr}
\end{equation}
where $\sigma_i$ and $\sigma_j$ represent the standard deviations of the returns of assets $i$ and $j$ in the window considered.

This definition induces the full correlation matrix 
$\mathbf{C}\big[\mathcal{E}(s,q)\big]$, 
which, like the covariance matrix, is symmetric and positive semidefinite. Its diagonal elements satisfy
\begin{equation}
C_{ii}\big[\mathcal{E}(s,q)\big]\;=\; 1.
\label{eq:correlacion-diagonal}
\end{equation}

By symmetry, it suffices to consider the strict upper triangular part (or, equivalently, the strict lower triangular part). Let
\begin{equation}
    M=\left\{0,1,...,N-1\right\},
\end{equation}
the total number of distinct correlations corresponds to the number of combinations of $N$ elements taken two at a time:

\begin{equation}
\left|\left\{(i,j)\in{M\times M}:\ i<j\right\}\right|
=\binom{N}{2}
=\frac{N!}{2!\,(N-2)!}
=\frac{N(N-1)}{2}.
\label{eq:corr-pares-distintos}
\end{equation}

Within this framework, Pharasi et al.~\cite{pharasiIdentifyingLongtermPrecursors2018} applied \emph{clustering} techniques to low-dimensional representations derived from correlation matrices and proposed the use of the \emph{power map}, a nonlinear transformation that suppresses the value of correlations close to zero\footnote{The exponent \(\varepsilon\) of the \emph{power map} is used to measure the closeness of the correlation matrix to zero.}. The aim of this transformation is to mitigate noise and improve the separability between groups.

\section{The \textit{k}-Means Method}\label{section:kmeans}
\index{K-Means}\index{Clusters}\index{Norm!Euclidean}\index{Standardization!z-score}

The \textit{k}-Means algorithm partitions a set of observations into \(K\) disjoint \emph{clusters}, seeking that each group forms a compact agglomeration under the distance induced by a norm. In this thesis the Euclidean norm (\(\ell_2\)) is used; for a discussion of alternative norms, see Appendix~\ref{appendix: Normas}. To carry out the \textit{k}-Means algorithm, the library 
\href{https://scikit-learn.org/stable/}{Scikit-learn} (version 1.3.1) was used \cite{blairPythonDataScience2019}. Each observation corresponds to an \emph{epoch} \(s\): a \emph{simple flattening} is applied to the correlation matrix \(\mathbf{C}\!\big[\mathcal{E}(s,q)\big]\) to obtain a vector of correlations \(\mathbf{x_k}^{(s)}\). To make the dimensions comparable over time, component-wise \textit{z-score} standardization is applied:

\[
z^{(s)}_{k}=\frac{x^{(s)}_{k}-\mu_{k}}{\sigma_{k}},\qquad
\mu_{k}=\big\langle x^{(s)}_{k}\big\rangle_{s},\quad
\sigma_{k}=\sqrt{\big\langle\!\big(x^{(s)}_{k}-\mu_{k}\big)^{2}\big\rangle_{s}},
\]
yielding the standardized vector \(\mathbf{z}^{(s)}\). Here, \(\langle\cdot\rangle_{s}\) denotes the average over epochs.

Fixing the total number of \emph{clusters} \(K\), the criterion to be optimized is written as
\begin{equation}
\min \;\sum_{k=1}^{K}\;\sum_{\mathbf{C}\!\big[\mathcal{E}(s,q)\big]\in\mathcal{C}_k}
\big\lVert \mathbf{z}^{(s)}-\mathbf{M}_k\big\rVert_{2}^{2},
\label{eq:kmeans-objetivo-l2-final}
\end{equation}
where \(\mathcal{C}_k\) is the cluster of matrices assigned to \emph{cluster} \(k\) and \(\mathbf{M}_k\) is its \textbf{centroid} in the standardized space. This centroid is obtained as the average of the standardized vectors associated with the matrices of the \emph{cluster}:
\begin{equation}
\mathbf{M}_k=\frac{1}{|\mathcal{C}_k|}\sum_{\mathbf{C}\!\big[\mathcal{E}(s,q)\big]\in\mathcal{C}_k}\mathbf{z}^{(s)}.
\label{eq:centroide-promedio-z}
\end{equation}

\paragraph{Characteristics}
\begin{enumerate}[label=\roman*)]
  \item \textbf{Initialization (\textit{k}-Means++).} 
  In the present work, the $K$ initial centroids were selected with the \textit{k}-Means++ procedure\footnote{\textit{k}-Means++ is also known as \emph{D$^{2}$-seeding}.}, preferred over a naive random selection because it produces initial centroids that are sufficiently separated and reduces the probability of converging to low-quality solutions \cite{blairPythonDataScience2019}. 
  The process is replicated for \emph{each} of the $K$ centroids:
  \begin{enumerate}[label=\alph*)]
    \item The first centroids are chosen at random from among the observations.
    \item For each $\bar{\mathbf{x}}$, $D(\bar{\mathbf{x}})$ is computed as the distance to the nearest centroid already chosen.
    \item The next centroid is selected by sampling over the observations with the probability law
    \begin{equation}
    P\!\left(\bar{\mathbf{x}}\right)=
    \frac{D\!\left(\bar{\mathbf{x}}\right)^{2}}{\sum_{\bar{\mathbf{x}}} D\!\left(\bar{\mathbf{x}}\right)^{2}},
    \label{eq:kmeanspp-prob}
    \end{equation}
    \item Steps (b)--(c) are repeated $K-1$ times.
  \end{enumerate}
 After initialization, the procedure alternates two steps.
 \begin{itemize}
     \item [i)] \textbf{Assignment:} each observation is sent to the \emph{cluster} whose centroid is nearest according to the chosen metric. 
     \item [ii)] \textbf{Update:} each centroid is recomputed as the average of the observations in its \emph{cluster}. 
These two steps are repeated until convergence. In this thesis a stopping criterion of \(10^{-4}\) was used for vector representations \(\bar{\mathbf{x}}\in\mathbb{R}^{d}\) and of \(10^{-5}\) for correlation matrices \(\mathbf{C}\!\big[\mathcal{E}(s,q)\big]\in\mathbb{R}^{N\times N}\).
 \end{itemize}

  \item \textbf{Iterations and repetitions (ensemble).} In each iteration, 10 independent runs with different initial conditions were carried out. From those 10 initializations, the partition that attained the \textbf{lowest value} of the \textit{k}-Means objective function was chosen. This scheme (an ensemble of 10 runs and selection of the best) was repeated 15 times in total.

  \item \textbf{Reproducibility control.} 
  By fixing the random seed and keeping the preprocessing, the scaling, and the order of the observations constant, the same partition was always obtained.

 \item \textbf{Cluster ordering criterion.}
For the presentation of results, the \emph{clusters} were numbered according to the \emph{average maximum eigenvalue} computed over the matrices of each cluster. In particular, for \emph{cluster} \(k\) it is defined as:
\[
\overline{\lambda}_{\max}(k)
=\frac{1}{|\mathcal{C}_k|}
\sum_{\mathbf{C}\!\big[\mathcal{E}(s,q)\big]\in\mathcal{C}_k}
\lambda_{\max}\!\big(\mathbf{C}\!\big[\mathcal{E}(s,q)\big]\big).
\]
The ordering is established \textbf{in ascending order} of \(\overline{\lambda}_{\max}(k)\); therefore, the \emph{cluster} with the lowest index corresponds to the lowest average of \(\lambda_{\max}\), and the subsequent ones to increasing averages. 
\index{Clusters!ordering}

\end{enumerate}

\subsection{Ordering and Stability (IEM)}
\index{Stability!IEM}
After ordering the \emph{clusters} by the average maximum eigenvalue criterion, 15 independent executions of \textit{k}-Means were considered. For each instant $t$, the \emph{ensemble mode} was defined as the most frequent label among the $15$ assignments:
\begin{equation}
\widehat{L}_t \;=\; \operatorname{mode}\big\{L_t^{(r)}\big\}_{r=1}^{15}
\;=\; \arg\max_{k\in\{1,\dots,K\}} \sum_{r=1}^{15} \mathbf{1}\!\left\{\,L_t^{(r)}=k\,\right\}.
\label{eq:iem-moda}
\end{equation}

\paragraph*{IEM stability.}
From the already ordered labels of the \textit{ensemble}, stability was quantified as the fraction of runs that coincided with the mode at each $t$:
\begin{equation}
\mathrm{IEM}(t) \;=\; \frac{1}{15}\;\max_{k\in\{1,\dots,K\}}\; \sum_{r=1}^{15} \mathbf{1}\!\left\{\,L_t^{(r)}=k\,\right\}.
\label{eq:iem-estabilidad}
\end{equation}

\paragraph*{Example (5 executions):}
For $K=3$ and a time $t_0$, consider the set of labels from five runs
\[
\mathcal{L}(t_0)=\{2,\,2,\,3,\,2,\,1\}.
\]
The counts were $f_1=1$, $f_2=3$, $f_3=1$; therefore,
\[
\widehat{L}_{t_0}=2,\qquad
\mathrm{IEM}(t_0)=\frac{\max\{f_1,f_2,f_3\}}{5}=\frac{3}{5}=0.6.
\]

In general terms, \(\mathrm{IEM}(t)\in[0,1]\): values close to \(0\) indicate \emph{low consensus} among executions (high dispersion in the assigned labels), whereas values close to \(1\) indicate \emph{high consensus} or maximum stability (all executions coincide on the same label).

\section{Spectral Decomposition, Dyadic Decomposition, and Change of Basis}
\index{Decomposition!spectral}
\index{Decomposition!dyadic}
\index{Change!of basis}
\index{Diagonalization!orthogonal}

The spectral theorem is formulated in general for Hermitian matrices 
(in complex spaces). In this thesis, complex numbers are not used. For this reason, the focus will be on the real case, 
where the matrices are \textbf{symmetric} and their diagonalization is \textbf{orthogonal}. In addition, the action of a matrix is studied in terms of its sum of scaled projections \cite{hassaniMathematicalPhysicsModern2013}.

\subsection{Diagonalization and Dyadic Decomposition}
\index{Diagonalization!equation}

Let $M\in\mathbb{R}^{N\times N}$ be a square matrix that admits diagonalization; there exist an invertible matrix $P$ and a diagonal matrix $D$ such that
\begin{equation}
M \;=\; P\,D\,P^{-1}.
\label{eq:diag-general}
\end{equation}
where $P=\bigl[\mathbf{v}_{1}\,\mathbf{v}_{2}\,\cdots\,\mathbf{v}_{N}\bigr]$ is the matrix of eigenvectors and
\[
P^{-1}=
\begin{bmatrix}
\mathbf{w}_{1}^{\top}\\
\mathbf{w}_{2}^{\top}\\
\vdots\\
\mathbf{w}_{N}^{\top}
\end{bmatrix},
\]
with $\mathbf{w}_{i}^{\top}\mathbf{v}_{j}=\delta_{ij}$\footnote{\(\delta_{ij}\) is the \textbf{Kronecker delta}.}.

In this way, the \emph{dyadic decomposition} is obtained. Moreover, the diagonal matrix $D$ contains the eigenvalues of $M$ as shown below:
\begin{equation}
D \;=\; \begin{bmatrix}
\lambda_{1} & 0 & \cdots & 0 \\
0 & \lambda_{2} & \cdots & 0 \\
\vdots & \vdots & \ddots & \vdots \\
0 & 0 & \cdots & \lambda_{N}
\end{bmatrix}.
\label{matrix:hermitianD}
\end{equation}

\begin{equation}
M \;=\; \sum_{i=1}^{N} \lambda_{i}\,\mathbf{v}_{i}\,\mathbf{w}_{i}^{\top}.
\label{eq:dyadic}
\end{equation}

\subsection{Spectral Theorem for Real and Symmetric Matrices}
\index{Theorem!spectral}
\index{Matrix!symmetric}
\index{Matrix!Hermitian}

The spectral theorem states that every Hermitian matrix is diagonalizable by means of a unitary transformation. 
In the real case of interest here, if $A\in\mathbb{R}^{N\times N}$ is symmetric ($A=A^{\top}$), there exists an orthogonal matrix $Q$ such that
\begin{equation}
Q^{\top}Q \;=\; QQ^{\top} \;=\; I \qquad\Rightarrow\qquad Q^{\top}=Q^{-1},
\label{eq:ortogonalidad}
\end{equation}
and the orthogonal diagonalization holds
\begin{equation}
A \;=\; Q\,D\,Q^{\top},
\label{eq:diag-ortogonal}
\end{equation}
where $D$ is diagonal with the real eigenvalues of $A$ as in matrix~\eqref{matrix:hermitianD} and the columns of $Q$ are orthonormal eigenvectors. 

\subsection{Spectral Decomposition and Functions of Operators}\label{subsection:funciones-operadores}
\index{Operator!normal}\index{Decomposition!spectral}\index{Projectors}\index{Functions!of operators}

Let \(\mathcal{V}\) be a finite-dimensional real vector space with an inner product. 
A linear operator \(A:\mathcal{V}\to\mathcal{V}\) is said to be \textbf{normal} if \(AA^{\top}=A^{\top}A\).
By fixing an orthonormal basis of \(\mathcal{V}\) we can identify \(A\) with its matrix in that basis and work interchangeably with the operator and its matrix representation.

If \(A\) is \textbf{symmetric} \((A=A^{\top})\), there exists an orthonormal basis of eigenvectors \(\{v_{i}\}_{i=1}^{N}\) with real eigenvalues \(\{\lambda_i\}_{i=1}^{N}\).
By defining the orthogonal projectors \(P_i=v_i\,v_i^{\top}\), the following decomposition is obtained \cite{hassaniMathematicalPhysicsModern2013}:
\begin{equation}\label{eq:descomp-espectral-real}
A \;=\; \sum_{i=1}^{N}\lambda_i\,P_i,
\qquad
\sum_{i=1}^{N}P_i \;=\; I,
\qquad
P_i^{\top}=P_i .
\end{equation}
Similarly to \eqref{eq:diag-general}, the multiplication of projectors is:
\begin{equation}\label{eq:proyectores-kronecker}
P_iP_j
=\big(v_i v_i^{\top}\big)\big(v_j v_j^{\top}\big)
= v_i\,(v_i^{\top}v_j)\,v_j^{\top}
= \delta_{ij}P_i.
\end{equation}

The preceding orthogonality makes it possible to compute powers, polynomials, and functions defined by power series on an interval containing the spectrum of \(A\) \cite{hassaniMathematicalPhysicsModern2013}:
\begin{equation}\label{eq:func-operator}
A^{n}=\sum_{i=1}^{N}\lambda_i^{\,n}P_i\quad(n\in\mathbb{N}),
\qquad
p(A)=\sum_{i=1}^{N}p(\lambda_i)P_i,
\qquad
f(A)=\sum_{i=1}^{N}f(\lambda_i)P_i.
\end{equation}

The correlation matrices \(\mathbf{C}\!\big[\mathcal{E}(s,q)\big]\) that were analyzed are real and symmetric; therefore, for each epoch \(\mathcal{E}(s,q)\) it is possible to consider them in the following manner:
\[
\mathbf{C}\!\big[\mathcal{E}(s,q)\big]\;=\;\sum_{i=1}^{N}\lambda_i(s,q)\,P_i(s,q),
\]
with exactly \(N\) \textbf{eigenvalues} and \(N\) \textbf{eigenvectors}. The \textbf{Eckart--Young--Mirsky Theorem}, found in Appendix~\ref{sec:cov_correl_espectral}, establishes that, for any matrix \(A\in\mathbb{R}^{m\times n}\), the best rank-\(r\) approximation (in the Frobenius norm or in the spectral norm \ref{appendix: Normas}) is obtained by truncating its singular value decomposition. Stated differently, by retaining the \(r\) largest singular values and setting the rest to zero.

\section{\textit{MS} from Reduced-Rank Matrices}
\index{Rank!reduced}\index{Subspace!dominant}
\label{sec:subespacio-dominante-rango-reducido}
To analyze the collective and sectoral structure, it is useful to spectrally separate the covariance matrix. Following \cite{heckensNewCollectivityMeasures2022}, the covariance $\Sigma\in\mathbb{R}^{N\times N}$ admits the decomposition:
\begin{equation}
\Sigma = \sum_{i=1}^{N}\kappa_i\,\mathbf{u}_i \mathbf{u}_i^{\top}
= \underbrace{\kappa_{N}\,\mathbf{u}_{N}\mathbf{u}_{N}^{\top}}_{\Sigma_{\mathrm{SOTM}}}
+ \underbrace{\sum_{i=1}^{N-1}\kappa_i\,\mathbf{u}_i \mathbf{u}_i^{\top}}_{\Sigma_{\mathrm{res}}},
\label{eq:decomp_cov}
\end{equation}
where $\kappa_i$ and $\mathbf{u}_i$ are the eigenvalues and eigenvectors ordered from smallest to largest. The term $\Sigma_{\mathrm{SOTM}}$ contains the largest eigenvalue \emph{SOTM}, whereas $\Sigma_{\mathrm{res}}$ groups the residual modes, typically associated with economic sectors, which are found in Appendix~\ref{appendix:sectores}.

Analogously, for the correlation matrix $\mathbf{C}$ one has:
\begin{equation}
\mathbf{C} = \sum_{i=1}^{N}\lambda_i\,\mathbf{v}_i \mathbf{v}_i^{\top}
= \underbrace{\lambda_{N}\,\mathbf{v}_{N}\mathbf{v}_{N}^{\top}}_{\tilde{\Sigma}_{\mathrm{SOTM}}}
+ \underbrace{\sum_{i=1}^{N-1}\lambda_i\,\mathbf{v}_i \mathbf{v}_i^{\top}}_{\tilde{\Sigma}_{\mathrm{res}}},
\label{eq:decomp_corr}
\end{equation}
where $\lambda_i$ and $\mathbf{v}_i$ are the eigenvalues and eigenvectors of the correlation matrix. It is important to note that, after this separation, $\tilde{\Sigma}_{\mathrm{SOTM}}$ and $\tilde{\Sigma}_{\mathrm{res}}$ do \emph{not} retain the structure of correlation matrices (they lose the normalization with unit diagonal), but they do represent well-defined \emph{covariance matrices} of \emph{reduced rank} \cite{heckensNewCollectivityMeasures2022}. In this work, this property is exploited to identify relevant parameters that characterize the collective patterns of the market and the sectoral behaviors. It is worth noting that the correlation matrix has rank at most $q-1$. The justification of this property is found in Appendix~\ref{subsec:Rango efectivo corr}.
\section{\emph{MS} and the Principal Axes Theorem}
\label{section:estados-mercado-ejes}
\index{Market!States}\index{Theorem!Principal Axes}

The principal axes theorem states that every quadratic form in $\mathbb{R}^n$ can be written as:
\begin{equation}
    Q(\mathbf{x}) = \mathbf{x}^{T} A \mathbf{x},
\label{eq:quadr-ejes-principales}
\end{equation}
where $A$ is a symmetric matrix. Every symmetric matrix admits a spectral decomposition: its eigenvalues are real and it can be diagonalized by means of an orthonormal basis of eigenvectors \cite{strangLinearAlgebraEveryone2020}.  

Let $\{\mathbf{u}_1,\dots,\mathbf{u}_n\}$ be an orthonormal basis of eigenvectors associated with the eigenvalues $\lambda_1,\dots,\lambda_N$.  
Any vector $\mathbf{x}\in\mathbb{R}^N$ can be decomposed as:
\begin{equation}
\mathbf{x}=\sum_{i=1}^N c_i\,\mathbf{u}_i,
\label{eq:desc-espectral}
\end{equation}
where $c_i=\mathbf{u}_i^{\top}\mathbf{x}$ represents the projection of $\mathbf{x}$ onto the direction of the eigenvector $\mathbf{u}_i$. Substituting into $Q(\mathbf{x})$, the quadratic form is rewritten as:
\begin{equation}
    Q(\mathbf{x})=\lambda_1 c_1^2 + \lambda_2 c_2^2 + \cdots + \lambda_n c_n^2.
\label{eq:desc-quad}
\end{equation}

Each term describes the independent contribution of a principal axis: the direction is determined by the eigenvector $\mathbf{u}_i$ and the magnitude of the variation by the eigenvalue $\lambda_i$. In particular, the largest eigenvalue, $\lambda_N$, dominates the total contribution and reflects the direction of most significant variation of the market as a whole.  

Using this theorem, it is possible to interpret the correlation matrices of stock returns in terms of principal axes. The analysis of financial correlation matrices reveals that the majority of the eigenvalues are concentrated in what is statistically considered noise, whereas only a few contain relevant information \cite{lalouxNoiseDressingFinancial1999}. This structure can be described as shown below:

\begin{enumerate}[label=\roman*)]
    \item A leading eigenvalue $\lambda_N$ that typically exceeds $\lambda_+$ and represents the global market factor: the \emph{SOTM}~\cite{mantegnaIntroductionEconophysicsCorrelations2004}~\cite{lalouxNoiseDressingFinancial1999}.
    \item A small set of intermediate eigenvalues that reflect sectoral dynamics.
    \item Many smaller eigenvalues that constitute statistical noise.
\end{enumerate}

For a detailed description of the eigenvalue distribution, see Appendix \ref{sec: Marcenko-Pastur}.

\section{Stationary Vector from the Transition Matrix}
\label{sec:vector-estacionario}
\index{Vector!stationary}

Consider the context of homogeneous \emph{Markov Chains}, in a discrete state space $\mathcal{S}$. Let $P$ be a row-normalized transition matrix, that is,
$P_{xy}\ge 0$ and $\sum_{y\in\mathcal{S}}P_{xy}=1$ $\forall$ $x\in\mathcal{S}$,
where $\mathcal{S}$ denotes the set of states and $P_{xy}$ is the probability of passing from state $x$ to state $y$ in one step.

The \textbf{stationary vector} $\boldsymbol{\pi}$ is defined as a vector of probabilities over $\mathcal{S}$ that satisfies:
\begin{equation}
\boldsymbol{\pi}=\boldsymbol{\pi}P.
\label{eq:vector-estacionario}
\end{equation}
Its components are denoted by $\pi(y)$ for $y\in\mathcal{S}$, so that $\pi(y)\ge 0$ and $\sum_{y\in\mathcal{S}}\pi(y)=1$. In particular, $\pi(y)$ describes the asymptotic behavior: the probability of being in state $y$ when the system has evolved over a sufficiently long time.

Under the hypotheses of irreducibility,\footnote{$\forall$ $x,y\in\mathcal{S}$ there exists $n\ge 1$ such that $(P^{n})_{xy}>0$; that is, $y$ can be reached from $x$ in a finite number of steps with positive probability.}
positive recurrence,\footnote{$\forall$ $x\in\mathcal{S}$, the return time to state $x$ has finite expectation. Stated differently, the process returns to $x$ in a finite average time.}
and aperiodicity,\footnote{The process is not restricted to visiting a state only at times following a fixed alternation. In other words, it does not occur that returns to a state $x\in\mathcal{S}$ take place only at instants that are multiples of a single integer $d>1$.}
the $n$-step probabilities lose memory of the initial state. In particular, for $x,y\in\mathcal{S}$ it holds that
\begin{equation}
\lim_{n\to\infty}(P^{n})_{xy}=\pi(y),
\qquad x,y\in\mathcal{S}.
\label{eq:lim-estacionario}
\end{equation}
This establishes that the probability of being in $y$ at long times is approximated by $\pi(y)$, independently of the initial state $x$ \cite{hoelIntroductionStochasticProcesses1972}.

\section{Interpretation of Eigenvectors through Squared Entries}
\label{sec:probabilistic}

\index{Eigenvector!squared entries}

Let \(\mathbf{C}[\mathcal{E}(s,q)] \in \mathbb{R}^{N \times N}\) be the correlation matrix estimated in the window \(\mathcal{E}(s,q)\). Its spectral decomposition is:
\begin{equation}
\mathbf{C}[\mathcal{E}(s,q)] = \sum_{k=1}^{N} \lambda_{k} \, \mathbf{v}_{k} \, \mathbf{v}_{k}^T, 
\quad \left\| \mathbf{v}_{k} \right\|_2 = 1,
\label{eq:descomp_correlacion_ventana_ajustada}
\end{equation}
where the eigenvectors \(\mathbf{v}_{k}\) and eigenvalues \(\lambda_{k}\) are defined for the time window \(\mathcal{E}(s,q)\), and \emph{\(k\)} denotes the \(k\)-th eigenvalue--eigenvector pair \([\lambda_{k}, \mathbf{v}_{k}]\), where $\lambda_k$ is ordered from smallest to largest.
The largest eigenvalue corresponding to \(k=N\) represents the \emph{SOTM} \cite{mantegnaIntroductionEconophysicsCorrelations2004, lalouxNoiseDressingFinancial1999}.

For each time window $\mathcal{E}(s,q)$, consider the components of the normalized eigenvector $\mathbf{v}_{k}$ and define, for each asset $i$, the quantity:
\begin{equation}
w_i \;=\; \left[v_i\right]^2.
\end{equation}
Since $\mathbf{v}_{k}$ is taken with unit norm, it holds that $w_i \ge 0$ and that $\sum_{i=1}^{N} w_i = 1$. Consequently, the sequence $\{w_i\}_{i=1}^{N}$ possesses the \emph{structure} of a discrete distribution in the axiomatic sense, that is, a non-negative and normalized measure defined over a finite set, in accordance with the framework introduced by Kolmogorov \cite{a.n.kolmogorovFoundationsTheoryProbability1950}. 

Nevertheless, in this work $\left\{w_i\right\}$ is not interpreted as the probability of occurrence of a random event. Instead, these quantities are understood as a \textbf{discrete measure of participation}, which quantifies the relative contribution of each asset to the collective pattern described by the eigenvector $\mathbf{v}_{k}$.\footnote{To avoid overloading the notation, given an eigenvector $\mathbf{v}_k$ its entries were denoted as $\left[v_i\right]_{i=1}^N$. In this thesis, work will be carried out, in particular, with the eigenvector associated with the largest eigenvalue $\lambda_N$, that is, $\mathbf{v}_N$.} Within this framework, the indices $i\in\{1,\dots,N\}$ represent the set of market assets and each $w_i$ measures the relevance of asset $i$ within the collective pattern associated with the eigenvector considered, in the time window analyzed.

For each time window $\mathcal{E}(s,q)$, let $\mathbf{v}_{k}\in\mathbb{R}^N$ be the \emph{normalized eigenvector} associated with index $k$, with entries $\left[v_i\right]_{i=1}^N$ for $i=1,\dots,N$, and Euclidean normalization
\begin{equation}
\sum_{i=1}^{N}\left[v_i\right]^2=1.
\end{equation}

Following \cite{sandovalCorrelationFinancialMarkets2012}, the \emph{inverse participation ratio}\footnote{From this point onward, IPR will be used as its English-language acronym.}\index{Inverse!Participation Ratio} of the eigenvector $\mathbf{v}_{k}$ is defined as
\begin{equation}
\boldsymbol{\mathrm{IPR}}_{k}
=\sum_{i=1}^{N}\left[v_i\right]^4,
\end{equation}
and its inverse as the \emph{participation ratio}\footnote{From this point onward, PR will be used as its English-language acronym.} \index{Participation!Ratio},
\begin{equation}
\boldsymbol{\mathrm{PR}}_{k}
=\frac{1}{\boldsymbol{\mathrm{IPR}}_{k}}.
\end{equation}

$\boldsymbol{\mathrm{PR}}_{k}$ is interpreted as the \emph{effective number} of assets that contribute significantly to the collective behavior of the market for a given time window.

Below, some relevant cases are mentioned:
\begin{enumerate}
\item \textbf{Constant vector.}
Consider the idealized case in which the eigenvector distributes its weight uniformly among the $N$ assets,
\begin{equation}
\mathbf{v}_{k}=\frac{1}{\sqrt{N}}(1,1,\dots,1).
\end{equation}

\begin{equation}
\boldsymbol{\mathrm{IPR}}_k=\frac{1}{N},
\qquad
\boldsymbol{\mathrm{PR}}_k=\frac{1}{\boldsymbol{\mathrm{IPR}}_k}=N.
\end{equation}

\item \textbf{Fully concentrated vector.}
Consider an eigenvector whose norm is concentrated in a single asset. This can be represented, without loss of generality, by assuming that there exists an index $j\in\{1,\dots,N\}$ such that
\begin{equation}
\left[v_j\right]=1,
\qquad
\left[v_i\right]=0\quad \text{for}\quad i\neq j.
\end{equation}
Such a choice satisfies the Euclidean normalization. In this case,
\begin{equation}
\boldsymbol{\mathrm{IPR}}_k=1,\
\qquad
\boldsymbol{\mathrm{PR}}_k=\frac{1}{\boldsymbol{\mathrm{IPR}}_k}=1.
\end{equation}
\end{enumerate}

\section{Principal Component Analysis}
\label{section:PCA}
\index{PCA}

\emph{Principal Component Analysis} (\emph{PCA}) is a statistical technique used
for dimensionality reduction that seeks to represent a set of
highly correlated variables in a lower-dimensional space, preserving
most of the information contained in the data
\cite{janicijevicPrincipalComponentAnalysis2022}. 
In the financial context, \emph{PCA} is applied to covariance or correlation matrices
constructed from time series of returns, 
with the aim of identifying collective patterns and isolating the principal sources of variability.

Hirsa et al.~\cite{hirsaRobustRollingPCA2023} introduce the methodology of
\emph{Robust Rolling PCA}\index{PCA!Robust Rolling} (R2-PCA), designed for financial data where
\emph{non-stationarity} and temporal dependence hinder the application of standard
PCA. The central idea consists in applying \emph{PCA} in sliding time windows
of length $W$, generating sequences of eigenvalues and
eigenvectors that describe the dynamic evolution of the system. These windows
\textbf{overlap}, so that when advancing one step in time, \emph{PCA} is recomputed with
the $W$ most recent observations, reducing the variance of the estimators and
capturing short-term variations. A characteristic problem of this approach
is the \emph{sign indeterminacy} (\emph{sign flipping}), since each
eigenvector $\mathbf{v}_j$ and its opposite $-\mathbf{v}_j$ represent the same
direction in space. R2-PCA addresses this drawback by means of criteria
of temporal continuity that ensure consistent trajectories for the
eigenvectors.

For their part, Jani\'cijevi\'c et al.~\cite{janicijevicPrincipalComponentAnalysis2022} emphasize that
\emph{PCA} constitutes an essential tool for processing high-dimensional financial
information, such as balance sheets, income statements, or price series.
Its application makes it possible to reduce hundreds of variables to a limited set of
components, which facilitates tasks of classification, risk prediction,
credit allocation, and solvency analysis. Furthermore, by
combining \emph{PCA} with \emph{clustering} algorithms, grouping structures consistent
with the economic nature of the analyzed entities are revealed. Finally, an important point noted by Jani\'cijevi\'c et al. is that \emph{PCA} functions as a noise filter. 

From the covariance matrix $\mathbf{\Sigma}$ expressed in Eq. \eqref{eq:covarianza-ij}, the central problem of \emph{PCA} consists in solving the spectral decomposition:
\begin{equation}
\mathbf{\Sigma} \, \mathbf{v}_j = \lambda_j \, \mathbf{v}_j,
\label{eq:eig}
\end{equation}
where $\lambda_j$ are the \textbf{eigenvalues} and $\mathbf{v}_j$ the
\textbf{eigenvectors}. The eigenvectors form an orthonormal basis and the
eigenvalues quantify the variance in each of these directions.

\chapter{Results}
\label{chapter:results}

This section presents the main findings of the analysis of the correlation matrices constructed from financial time-series data.
First, the leading eigenvalues of these matrices are examined.
Subsequently, the leading eigenvectors with squared entries are studied, with the aim of analyzing the \emph{MS} they induce and the sectoral distribution. Finally, the dynamics are addressed using matrices of the form:
\begin{equation}
    \mathbf{C}^{l} \;=\; \mathcal{N}\!\left(\sum_{i=N-l+1}^{N} \lambda_i\, \mathbf{v}_i\, \mathbf{v}_i^{\top}\right), 
    \quad l \in \{1,2,3\},
    \label{eq:tipo_guhr}
\end{equation}
which make it possible to characterize the contribution of the eigenvectors weighted by their corresponding eigenvalues. In the remainder of the text, the notation $k$ will be used to denote the number of \emph{clusters}.

\section{Data Treatment}
\label{sec:datos}

The historical series used in this work were obtained from \href{https://finance.yahoo.com/}{Yahoo Finance}, a publicly accessible source with broad coverage of financial markets. It is appropriate to distinguish between the \textbf{raw closing price}\index{Price!raw closing}\ and the \textbf{adjusted closing price}\index{Price!adjusted closing}. The former is the last price traded at the close of the session, useful as an immediate market reference \cite{kratterBeginnersGuideStock2019}, but it does not correct for the discontinuities introduced by dividends and stock splits\footnote{A \emph{stock split} increases the number of shares and proportionally reduces the price per share, without altering the capitalization or the economic value for the shareholder. Conversely, a \emph{reverse split} reduces the number of shares and raises the price.}. The latter incorporates these transformations and provides a homogeneous measure that is more faithful to the economic value over time, which is why it is preferable for historical comparisons and for the computation of cumulative returns \cite{WhatAdjustedClosing}.

In this thesis, work is carried out exclusively with \textbf{adjusted closing prices}. In general terms, the quotations reflect the equilibrium between \textbf{supply and demand}, modulated by the performance of the companies, the expectations of investors, and the macroeconomic environment. The economic sectors of the S\&P~500 are shown in Appendix~\ref{tab:sectores-gics-sp500}, and the set of companies analyzed in Appendix~\ref{tab:companies-sp500}. The study period spans from \textbf{January 3, 2012} to \textbf{December 29, 2023}.

\section{Correlation Matrices}
\index{Matrix!correlation}
\label{sec:matrices_corr}

As a starting point, the correlation matrices estimated with short-epoch windows of size $q\in\{20,40\}$ days are analyzed. Their visual representations are displayed: heatmaps, temporal evolution of the \emph{SOTM}, transition matrices between states, and correlation histograms, with the aim of describing the \emph{morphology} of the system.

\subsection{Temporal Evolution of \textit{MS}}
\index{Market!States}
\index{Clusters}

In the discussion of the \emph{COVID} State\index{State!COVID} by Mart\'inez-Ramos \emph{et al.} \cite{martinezramosSeriesTiempoFinancieras2024}, the appearance of this state as a particular phase of the market was studied. In the results with $q=40$ in Fig.~\ref{fig:evolucion_ms_corr}, it is clearly observed that this state is activated in \emph{cluster} 2. Here, the \emph{average correlation} $\langle C_{ij} \rangle$ is attenuated by the \emph{cancellation} between positive and negative contributions. Likewise, the pattern appears in an \emph{isolated} manner at $q=40$, without any antecedent that replicates this pattern at earlier dates. Additionally, in 2017--2018 a continuous presence of \emph{cluster} 1 is recorded.

\index{Heatmaps}
\subsection{Heatmaps}

Fig.~\ref{fig:avg_corr_mats} presents the average correlation matrices obtained by \emph{cluster} for $q=40$ and $k=5$. Each matrix corresponds to the average of all the matrices belonging to the same cluster\footnote{Also referred to as the centroid, Eq.~\ref{eq:centroide-promedio-z}.}. In \emph{cluster} 1, the \textit{UT} and \textit{FN} sectors exhibit a greater presence of anticorrelations, whereas the intensity of the average correlation reflects the intersectoral correlation arising from the organization of the S\&P 500 companies. The sectoral nomenclature used is detailed in Appendix~\ref{tab:sectores-gics-sp500}. Additionally, Appendix~\ref{fig: Mapas de Calor apendices} shows complementary heatmaps that reproduce an analogous pattern, characterized by an increase in correlation as one advances toward the cluster of higher index.
\subsection{Transition Matrices}
\label{subsec:matrices_transicion_corr}
\index{Matrix!transition}

The transition matrices shown in Fig.~\ref{fig:matrices_transicion_corr} exhibit an almost tridiagonal structure, which reflects that state transitions occur predominantly between adjacent states. This feature supports the hypothesis of modeling the dynamics of the system as a \textbf{Markov chain}\index{Chain!Markov}.

In a Markov chain, the random variables \(X_n,\, n\geq 0\) have \emph{no memory}: the state at \(n+1\) depends solely on the state at \(n\), and not on the previous history \cite{hoelIntroductionStochasticProcesses1972}. Although financial markets present a \emph{non-Markovian} nature\footnote{Historical examples of dynamics with long-term memory are the \emph{Fukushima}\index{Fukushima} nuclear crisis in 2011 and the \emph{Great Depression}\index{Great Depression} of 1929, whose effects extended in a prolonged manner across global financial markets.}, the pattern of transitions between market states proves compatible with the hypothesis of \emph{Markovianity} \cite{martinezramosSeriesTiempoFinancieras2024, pharasiIdentifyingLongtermPrecursors2018}. Finally, it is worth noting that in the transition matrices of Fig.~\ref{fig:matrices_transicion_corr}, the state with the highest frequency of permanence is not preserved in a single \emph{cluster}, but rather varies between \emph{cluster} 2 and \emph{cluster} 3 for \(k=4\). For \(k=5\), the distribution of the values of the transition matrix is similar for windows of 20 and 40 days.

\subsection{\texorpdfstring{Relationship between the Largest Eigenvalue $\lambda_N$ and the Average Correlation $\langle C_{ij} \rangle$}{Relaci\'on entre el eigenvalor mayor y la correlaci\'on promedio <Cij>}}
\label{subsec:lambda_vs_corr}

Fig.~\ref{fig:lambda_vs_corr_pearson} shows the nearly linear relationship between the dominant eigenvalue $\lambda_N$ and the average correlation $\langle C_{ij}\rangle$, with a \textbf{Pearson coefficient} very close to $1$.
\index{Correlation!average}
\index{Correlation!Pearson}

In the work of Mart\'inez Ramos, the analysis was extended comparatively to different indices, such as the \textbf{Nikkei 225}\index{Nikkei 225}. There it was shown that, while in the \textbf{S\&P 500} the correlation between $\lambda_N$ and $\langle C_{ij}\rangle$ is very high, in the other markets the correlation remains high but with differences~\cite{martinezramosSeriesTiempoFinancieras2024}. Thus, both the dominant eigenvalue $\lambda_N$ and the average correlation $\langle C_{ij} \rangle$ function as indicators of the collective dynamics in stock markets. In particular, it has been observed that the average correlation tends to increase during periods of crisis \cite{mantegnaIntroductionEconophysicsCorrelations2004}.
\begin{figure}[p]
    \centering
    \vspace*{\fill}   

    \includegraphics[width=0.75\textwidth]{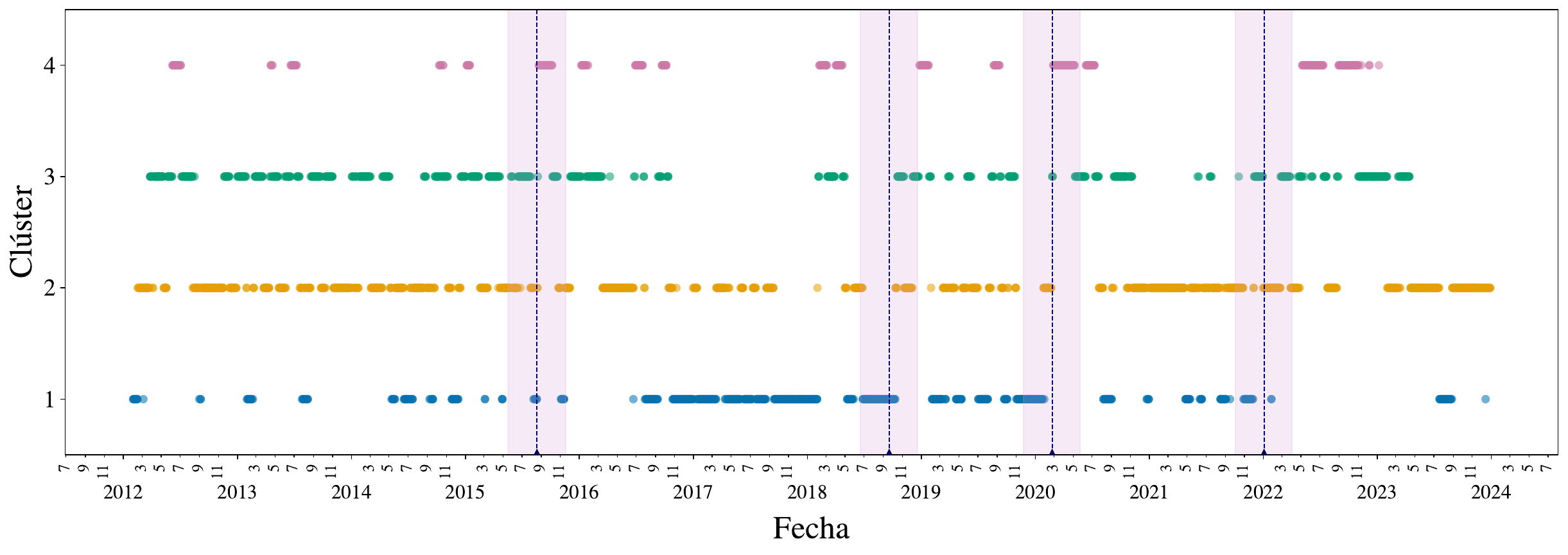}
    \caption*{a) $q=20$, $k=4$}

    \vfill   

    \includegraphics[width=0.75\textwidth]{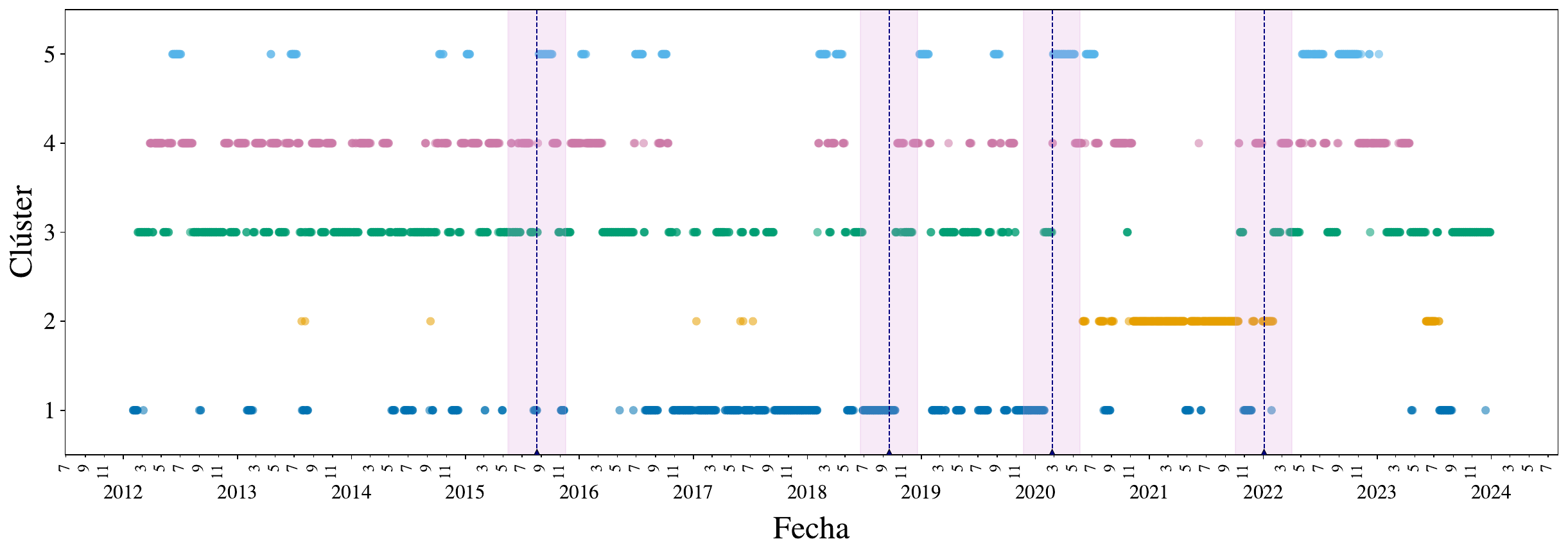}
    \caption*{b) $q=20$, $k=5$}

    \vfill

    \includegraphics[width=0.75\textwidth]{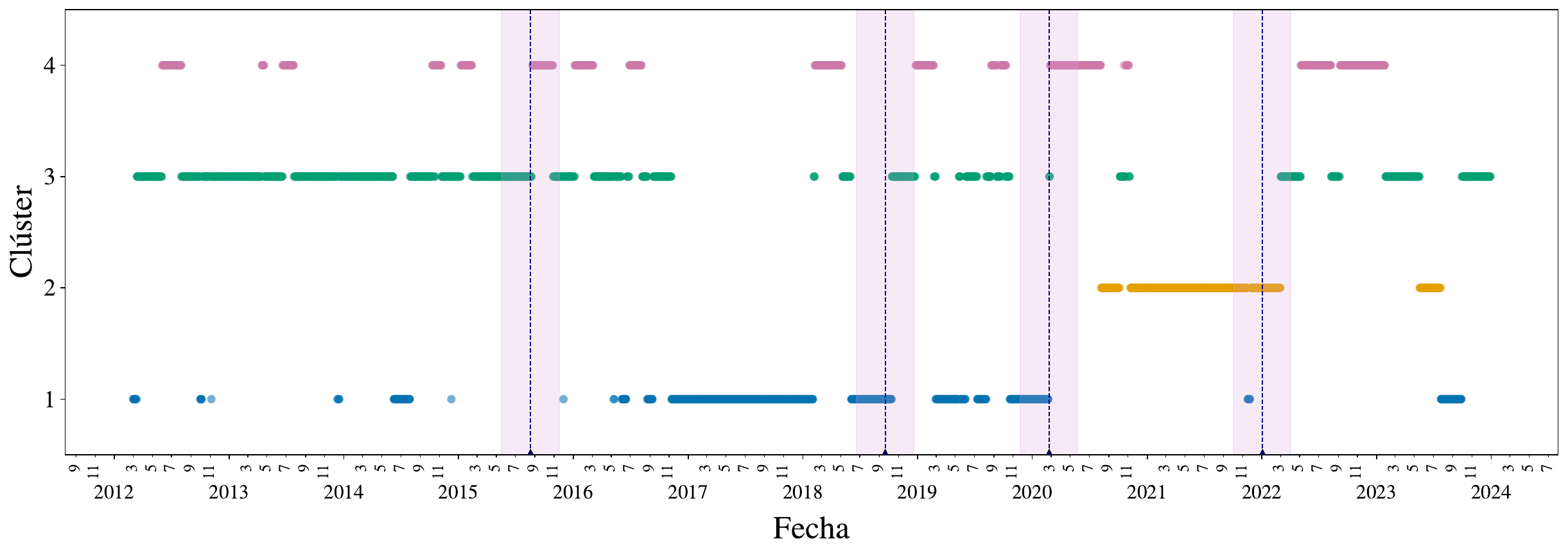}
    \caption*{c) $q=40$, $k=4$}

    \vfill

    \includegraphics[width=0.75\textwidth]{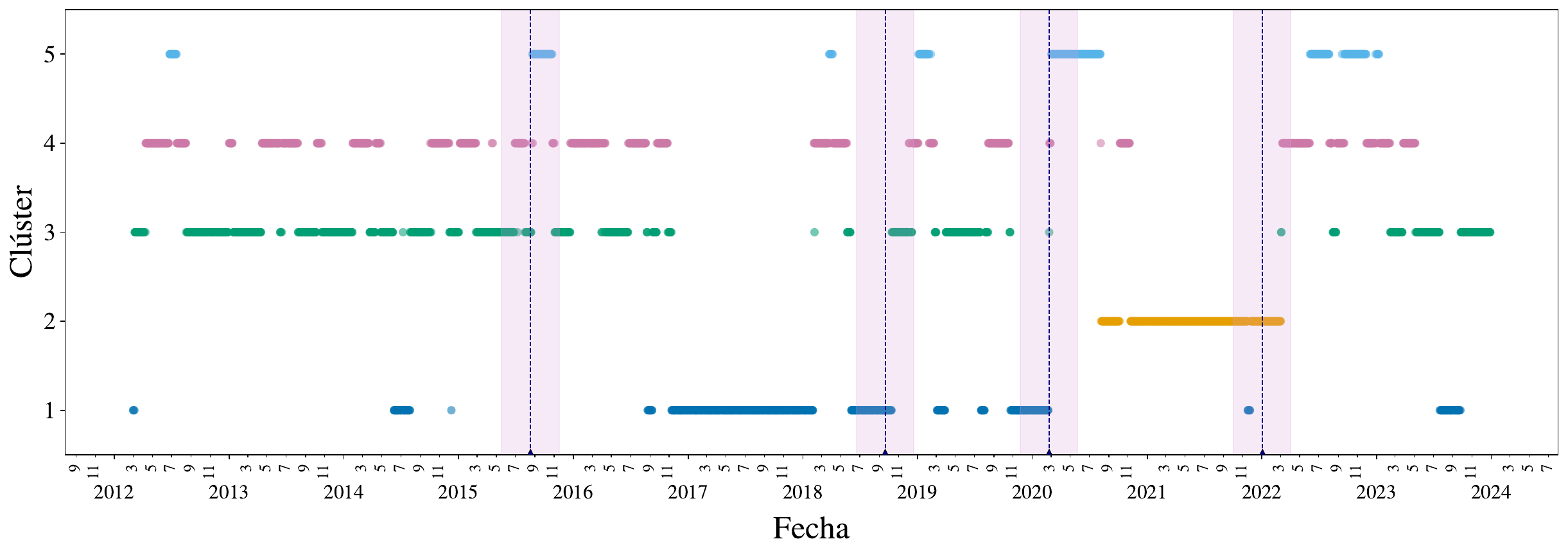}
    \caption*{d) $q=40$, $k=5$}

    \vspace*{\fill}   

    \caption[Temporal evolution of \textit{MS} from the correlation matrices]%
    {Temporal evolution of the \textit{MS} obtained by means of short-epoch-window correlation matrices.
    The atypical \emph{COVID} State is more clearly defined in \emph{cluster} 2 with $q=40$. There is also a constant appearance of \emph{cluster} 1 in 2017--2018.}
    \label{fig:evolucion_ms_corr}
\end{figure}

\begin{figure}[H]
    \centering
    \includegraphics[width=\textwidth,keepaspectratio]{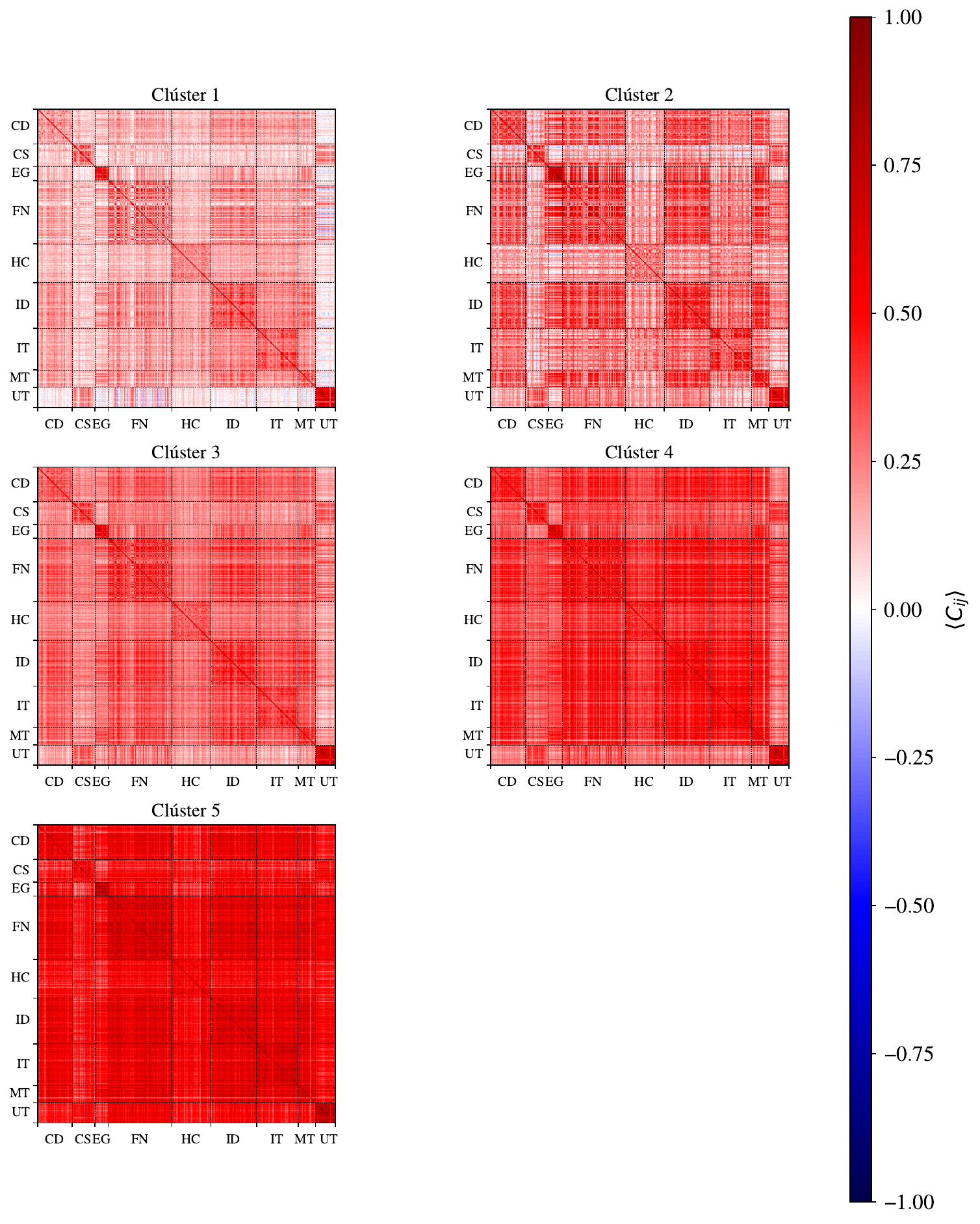}
    \caption[Heatmaps of the average correlation matrix by \emph{cluster}]%
    {Heatmaps of the average correlation matrix by \emph{cluster}. $q=40$ and $k=5$. \emph{Cluster} 1 corresponds to the lowest correlation values in the system. The sector nomenclature is detailed in Appendix~\ref{tab:sectores-gics-sp500}.}
    \label{fig:avg_corr_mats}
\end{figure}
\clearpage
\begin{figure}[ht]
\centering

\begin{subfigure}{0.49\linewidth}
  \centering
  \includegraphics[width=\linewidth]{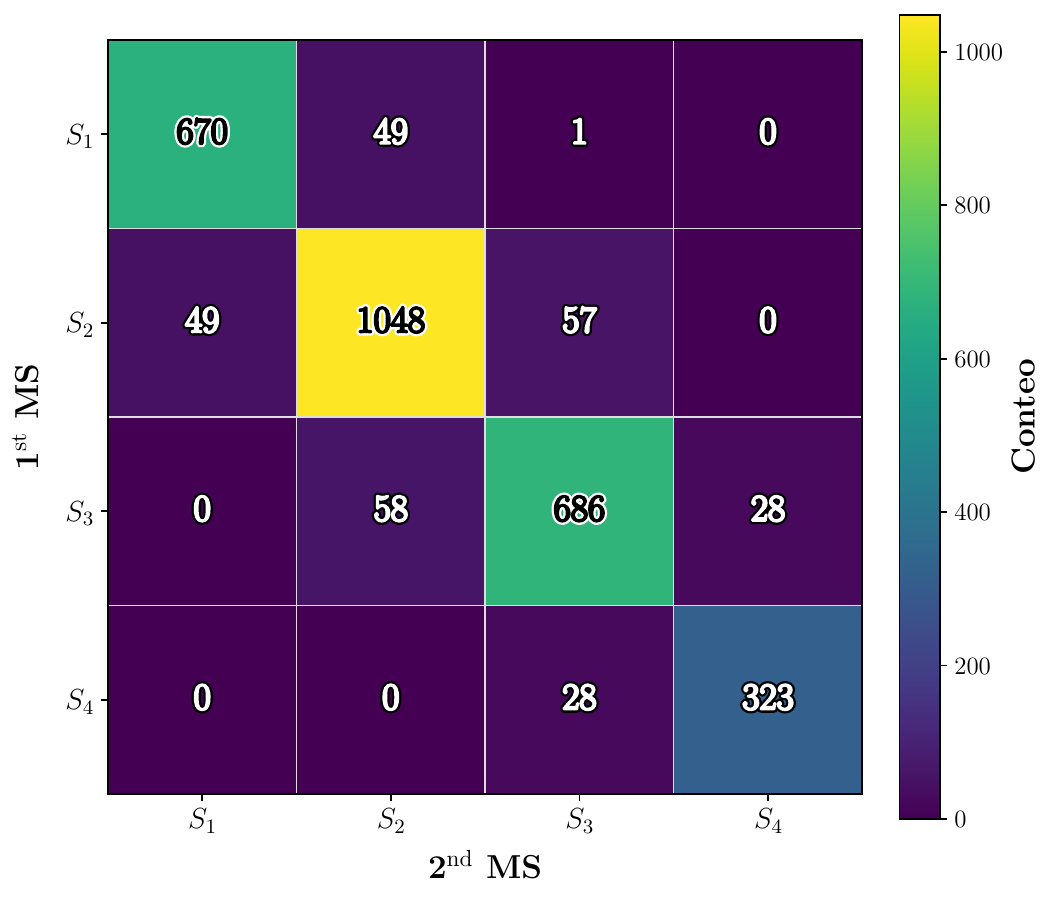}
  \caption*{a) $q=20$, $k=4$}
\end{subfigure}\hfill
\begin{subfigure}{0.49\linewidth}
  \centering
  \includegraphics[width=\linewidth]{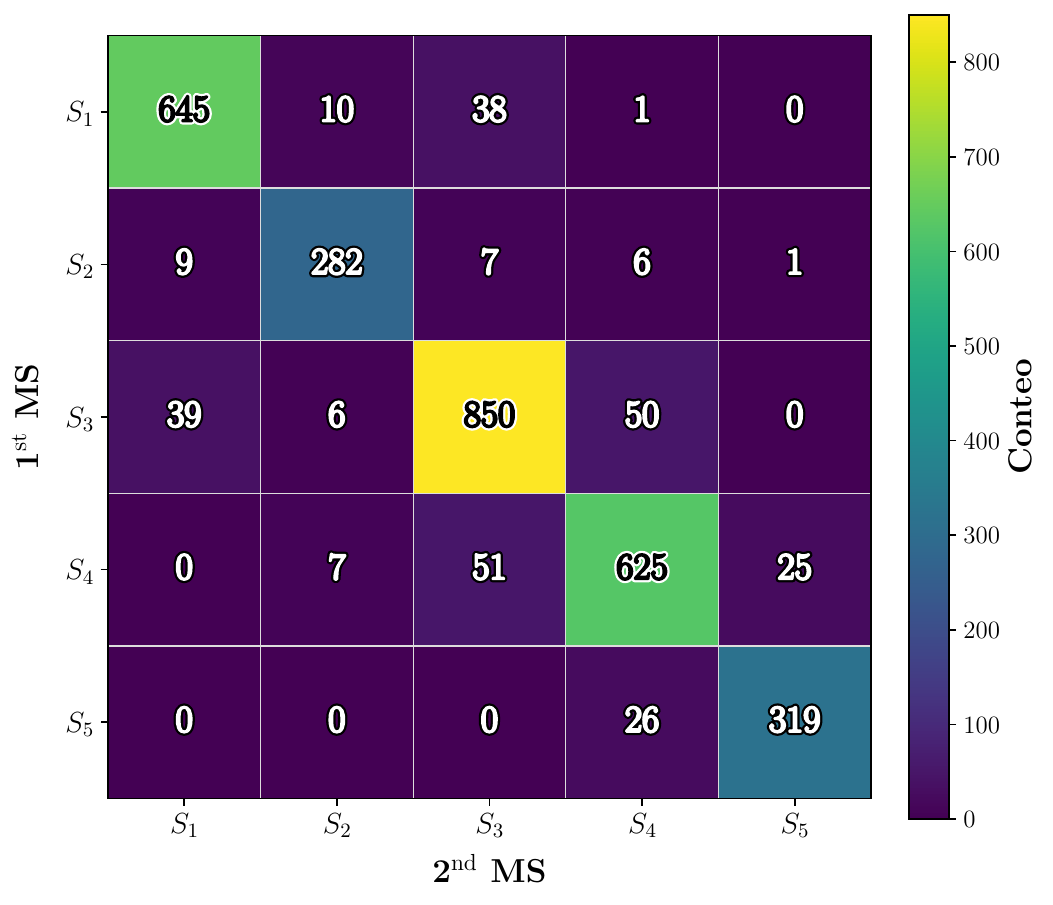}
  \caption*{b) $q=20$, $k=5$}
\end{subfigure}

\vspace{0.6em}

\begin{subfigure}{0.49\linewidth}
  \centering
  \includegraphics[width=\linewidth]{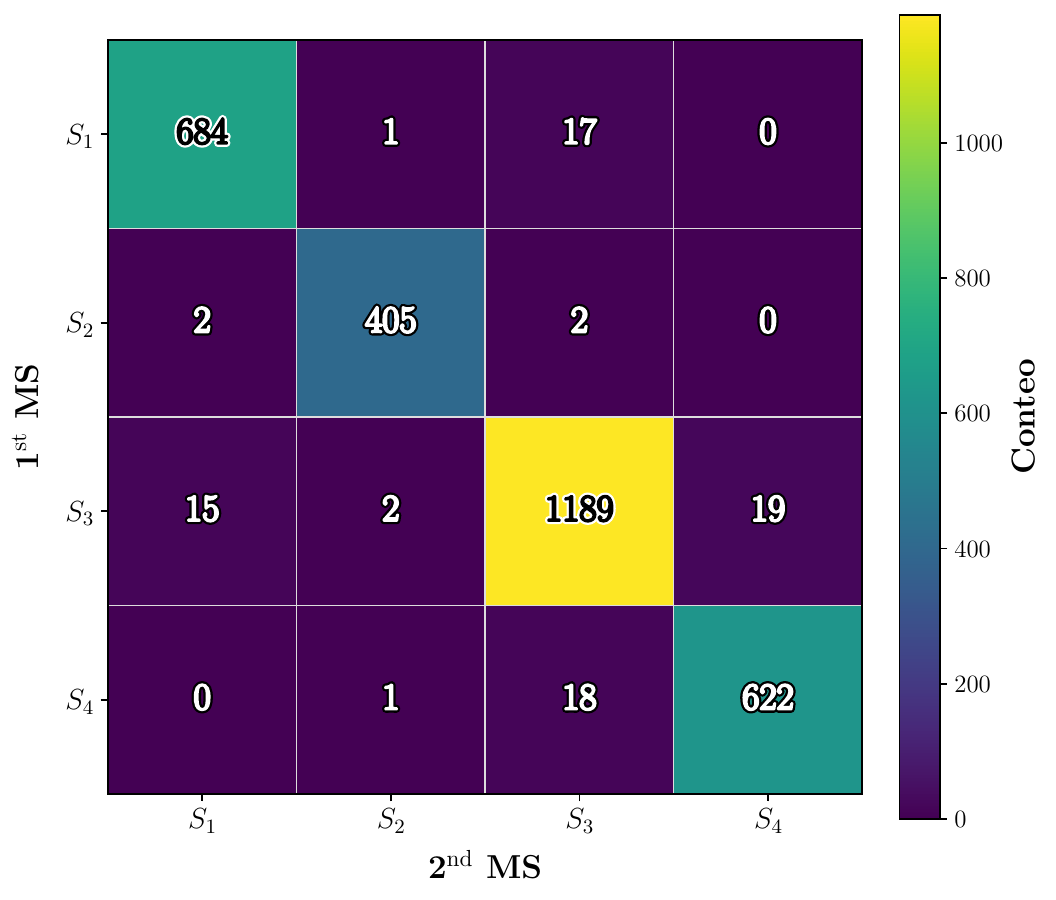}
  \caption*{c) $q=40$, $k=4$}
\end{subfigure}\hfill
\begin{subfigure}{0.49\linewidth}
  \centering
  \includegraphics[width=\linewidth]{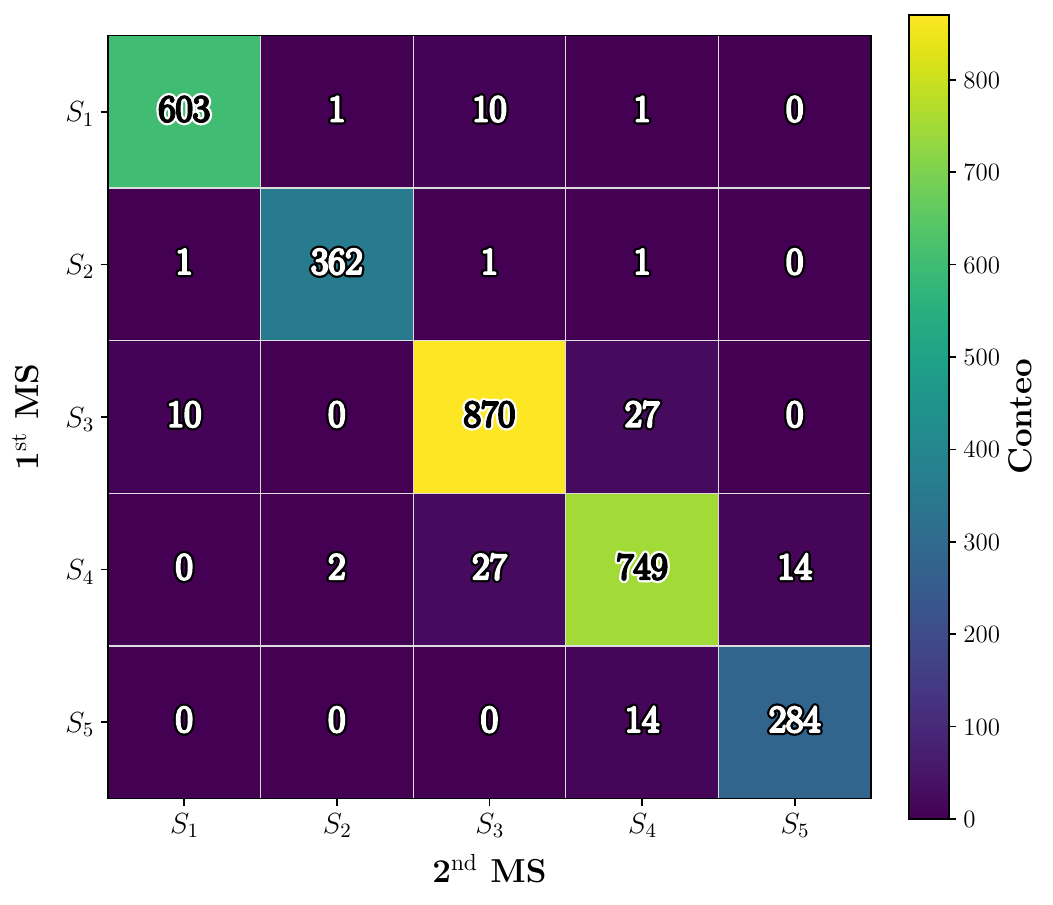}
  \caption*{d) $q=40$, $k=5$}
\end{subfigure}

\caption[Transition matrices indicated by the correlation matrices]%
{Transition matrices that quantify the jump of states from $i$ to $i\!+\!1$ of the correlation matrices. On the diagonal, the matrix with the highest occupation alternates between \emph{cluster} 2 and \emph{cluster} 3. The matrices are almost tridiagonal, since the differences between symmetric positions are of only a few units.}
\label{fig:matrices_transicion_corr}
\end{figure}
\clearpage
\begin{figure}[ht]
\centering

\begin{subfigure}{0.48\linewidth}
  \centering
  \includegraphics[width=\linewidth]{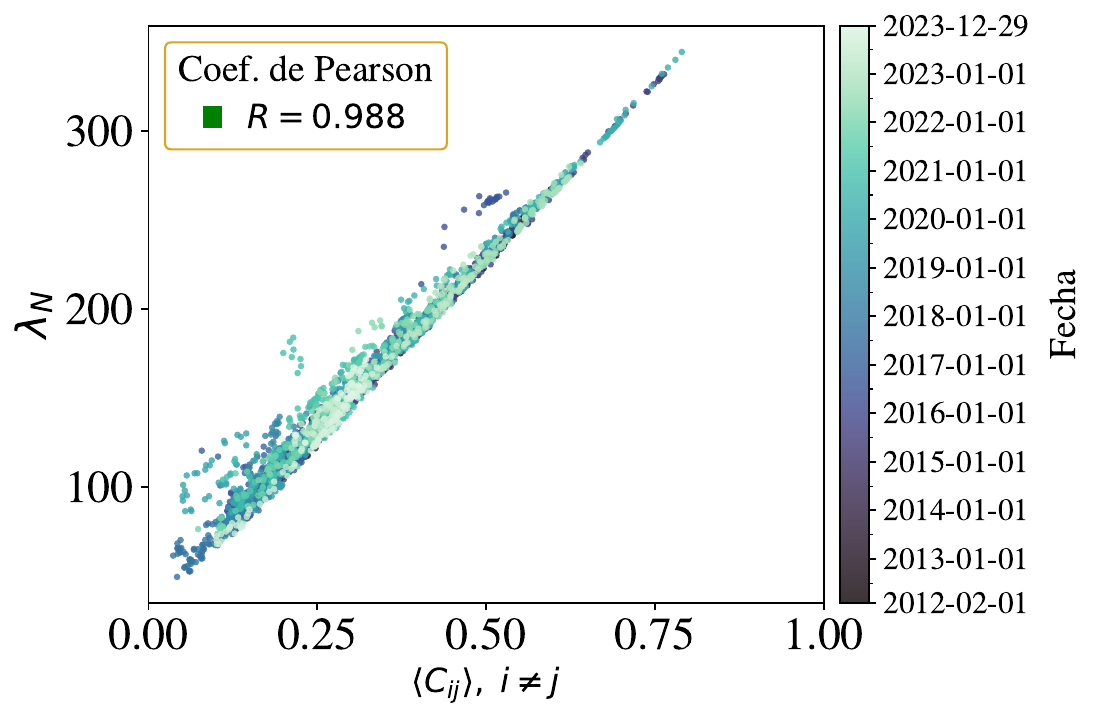}
  \caption*{a) $q=20$}
\end{subfigure}\hfill
\begin{subfigure}{0.48\linewidth}
  \centering
  \includegraphics[width=\linewidth]{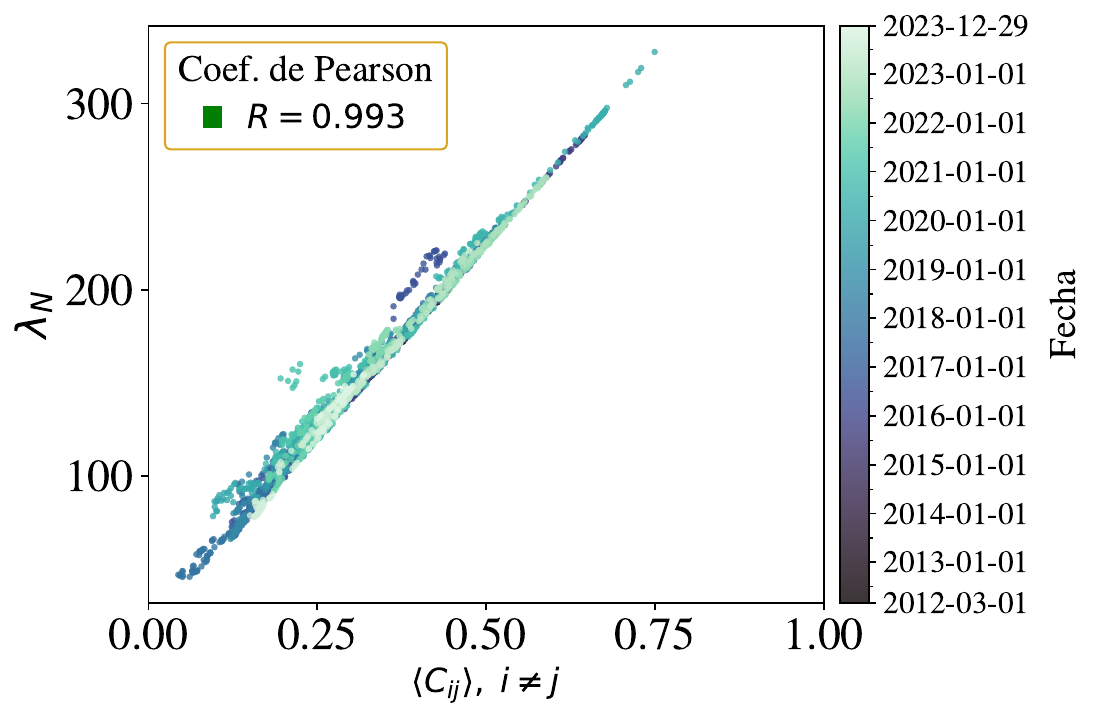}
  \caption*{b) $q=40$}
\end{subfigure}

\caption[\texorpdfstring{Relationship between $\lambda_N$ and the average correlation $\langle C_{ij}\rangle$ ($i\neq j$)}{Relaci\'on entre lambdaN y la correlaci\'on promedio}]%
{Relationship between the dominant eigenvalue $\lambda_N$ and $\langle C_{ij}\rangle$, with $i \neq j$. In both cases a strong linear relationship is observed, with a \textbf{Pearson coefficient} very close to 1.}
\label{fig:lambda_vs_corr_pearson}
\end{figure}

\subsection{Temporal Density Surfaces of Correlations}
\label{subsec:waves3d}
\index{Density!of correlations}

Fig.~\ref{fig:waves3d_corr} shows temporal surfaces of the density \(\rho(C_{ij})\) (normalized by date\footnote{\(\rho(C_{ij})\) is normalized so that \(\int_{-1}^{1}\rho(x)\,\mathrm{d}x=1\) for each day.}) for \(q=20\) and \(q=40\). In the case of \(q=20\), two relevant episodes were identified. Between \textbf{September 9 and 18, 2015}, the modal center\footnote{The modal center \(c_{\text{mode}}\) is defined as the correlation value at which the density \(\rho(C_{ij})\) attains its maximum.} was \(c_{\text{mode}}\approx 0.847\), whereas between \textbf{March 12 and 18, 2020} a value of \(c_{\text{mode}}\approx 0.870\) was obtained. In turn, for \(q=40\) a single episode was detected between \textbf{March 12 and 17, 2020}, in which the modal center reached a value of \(c_{\text{mode}}\approx 0.824\). For further detail regarding the crisis events, see Appendix~\ref{tab:fechas_episodios}.

\begin{figure}[ht]
  \centering

  \begin{subfigure}{0.67\linewidth}
    \centering
    \includegraphics[width=\linewidth]{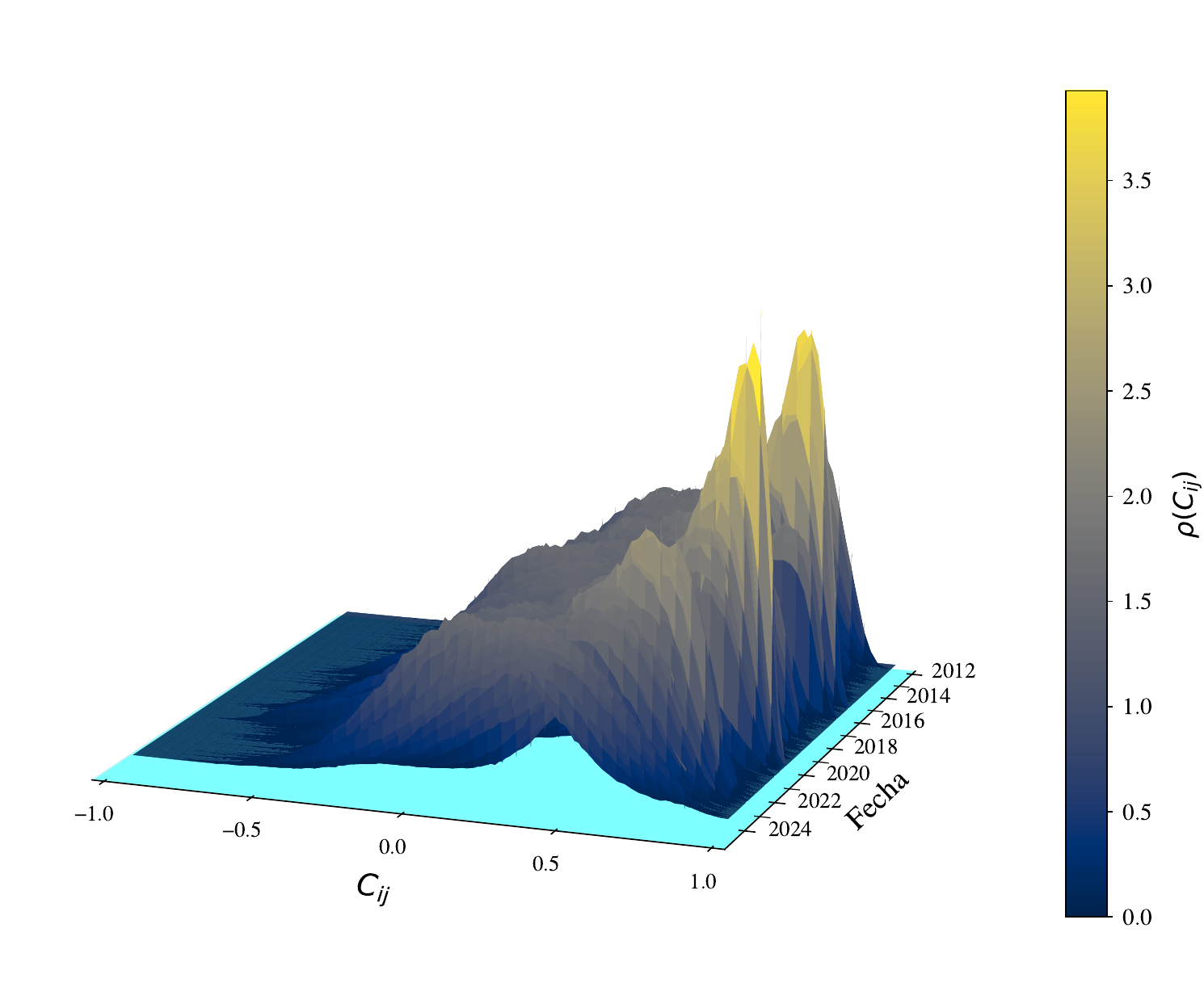}
    \caption*{a) $q=20$}
  \end{subfigure}

  \vspace{0.4cm}

  \begin{subfigure}{0.67\linewidth}
    \centering
    \includegraphics[width=\linewidth]{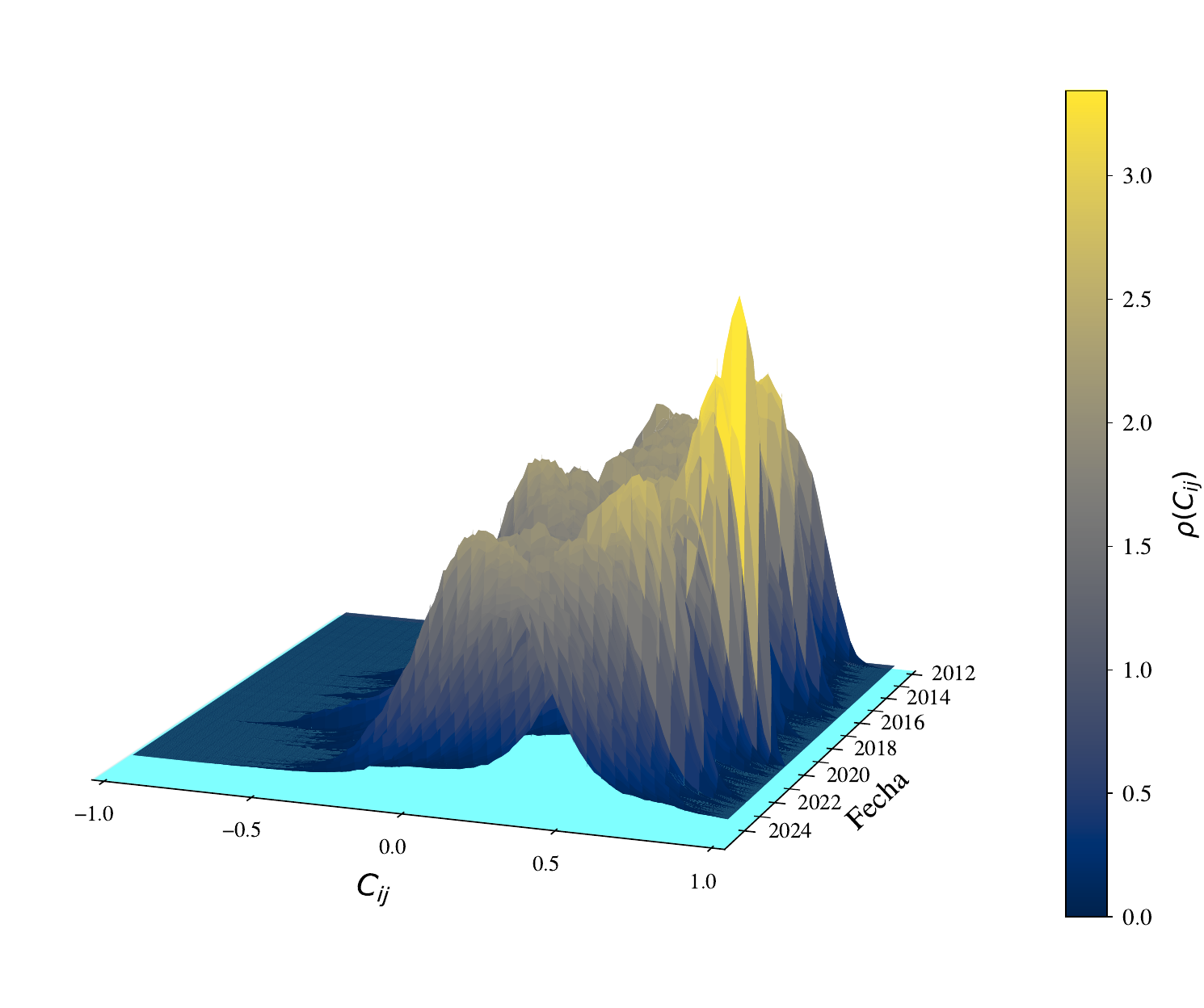}
    \caption*{b) $q=40$}
  \end{subfigure}

  \vspace{0.5cm}
  \caption[3D density surfaces of correlations]%
  {Each figure shows a \textbf{3D surface} of the temporal evolution of the distribution of correlations \(C_{ij}\), with \(i\neq j\): the color map is the \textbf{density} \(\rho(C_{ij})\), computed as a histogram \emph{normalized by date} for each day. Prominent peaks: for \(q=20\), from September 9 to 18, 2015 and from March 12 to 18, 2020 (\(\rho>3.5\)); for \(q=40\), from March 12 to 17, 2020 (\(\rho>3.5\)).}
  \label{fig:waves3d_corr}
\end{figure}

\subsection{Histograms by Cluster}
\label{subsec:histogramas_cluster}
\index{Clusters}
\index{Histograms}

Fig.~\ref{fig:histCorrByCluster_q20q40_k4k5} shows the \textbf{correlation histograms} by \emph{cluster} (on a \emph{semi-logarithmic scale}) for different windows $q$. \emph{Clusters} 1 and 2 exhibit \emph{positive skewness}. To define the \emph{bin} width $h$, the \textbf{Freedman--Diaconis rule} was employed, which establishes:
\begin{equation}
h = 2\,\frac{\mathrm{IQR}}{n^{1/3}}.
\end{equation}
In this way, the choice of \emph{bins} explicitly incorporates the variability of the sample and the number of available observations~\cite{DocumentationNumPyV112}. \emph{Multimodality} is observed in \emph{cluster} 1. The figures report the statistics: standard deviation $\sigma$, skewness $\gamma_{1}$, and kurtosis $\kappa$.
\begin{figure}[!htbp]
\centering

\begin{subfigure}[t]{0.465\linewidth}
  \centering
  \includegraphics[width=\linewidth]{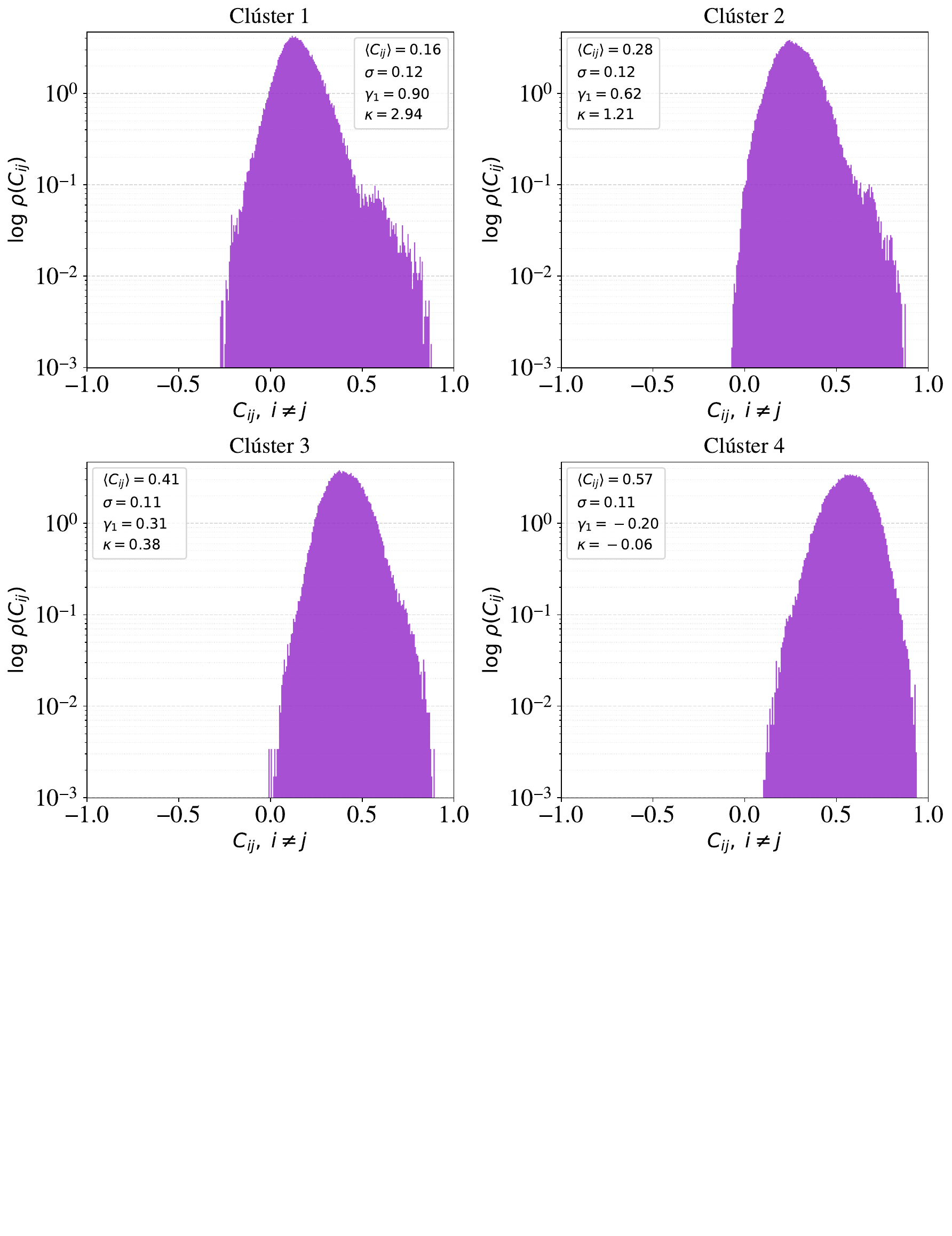}
  \caption*{\footnotesize a) $q=20$, $k=4$}
\end{subfigure}\hfill
\begin{subfigure}[t]{0.465\linewidth}
  \centering
  \includegraphics[width=\linewidth]{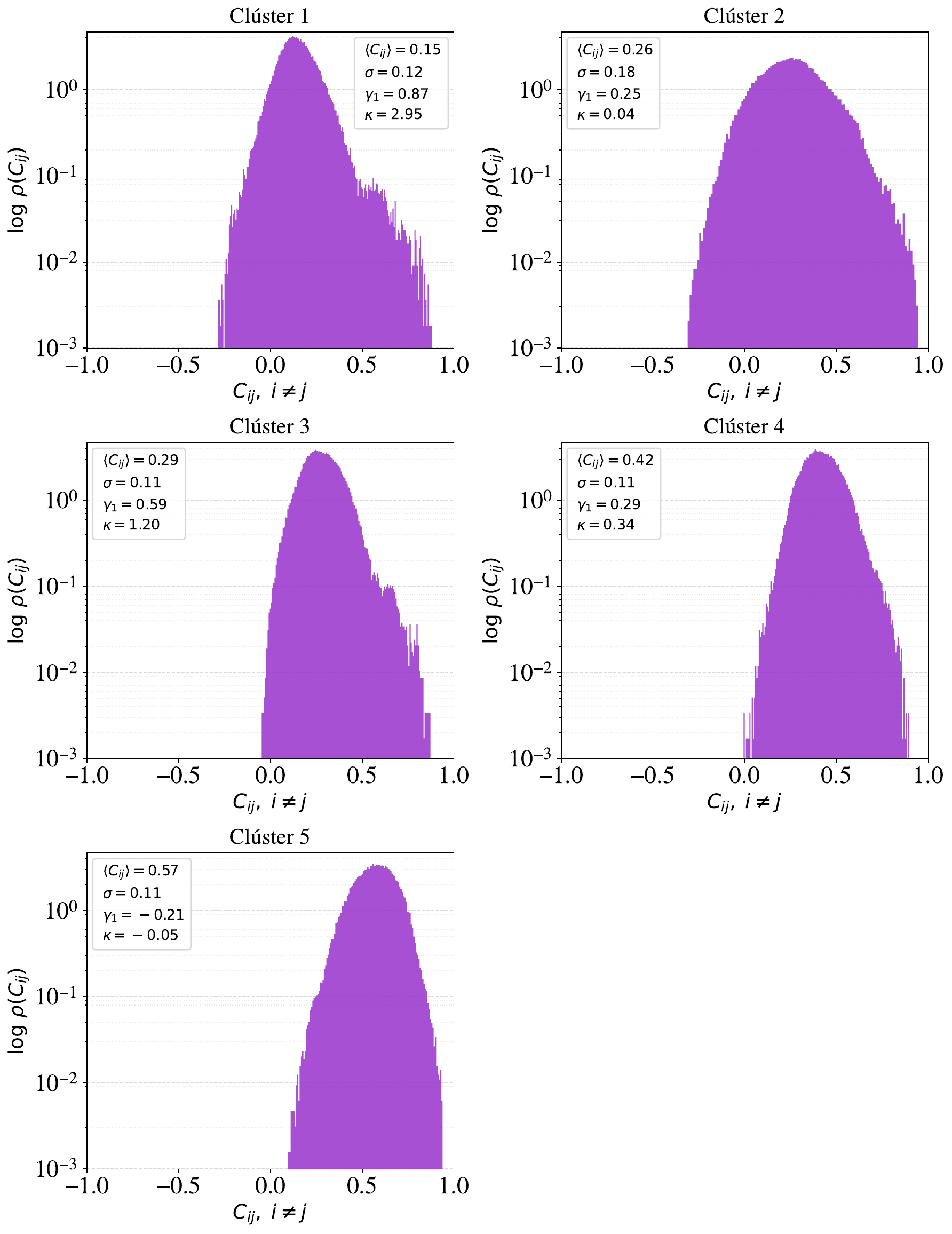}
  \caption*{\footnotesize b) $q=20$, $k=5$}
\end{subfigure}

\begin{subfigure}[t]{0.465\linewidth}
  \centering
  \includegraphics[width=\linewidth]{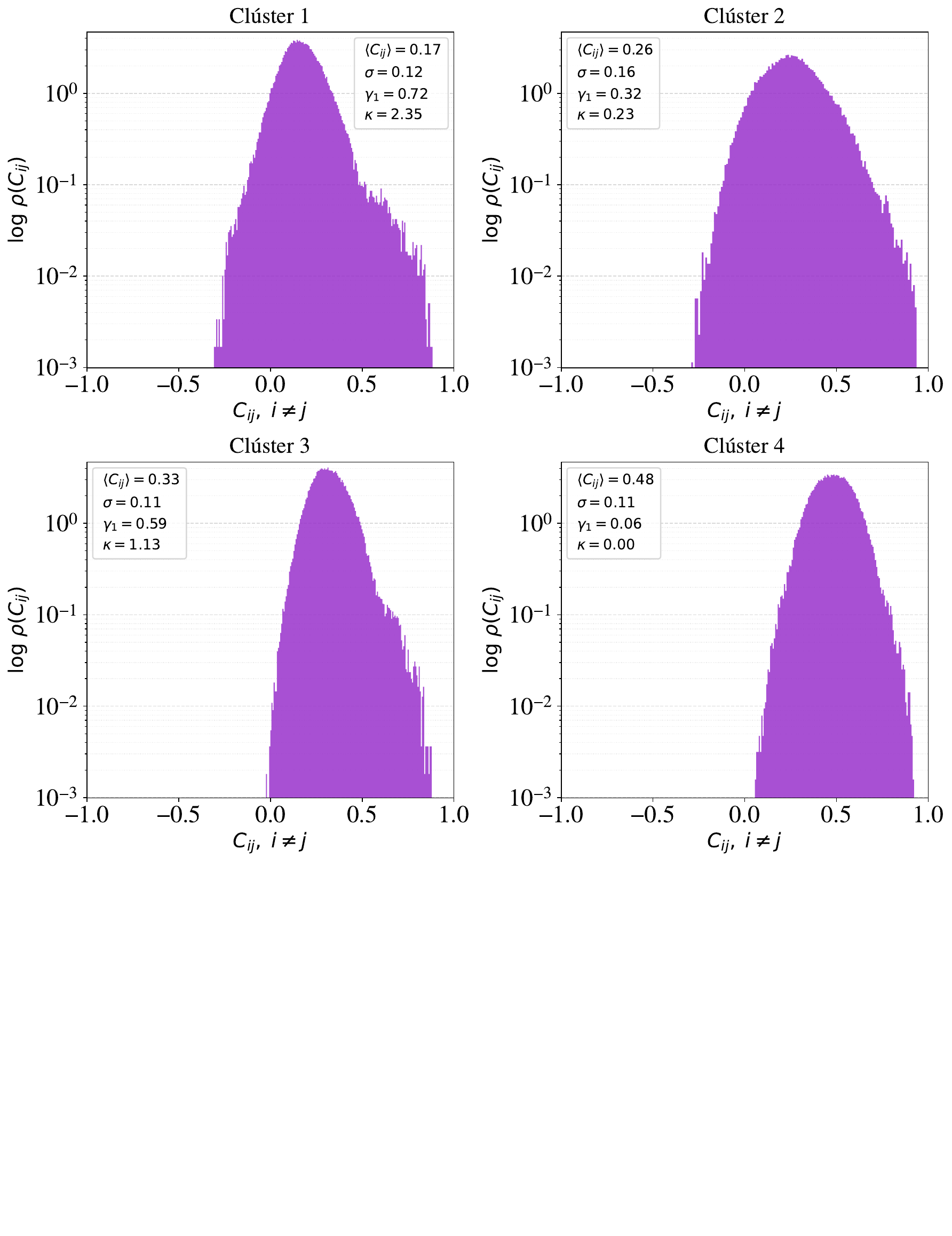}
  \caption*{\footnotesize c) $q=40$, $k=4$}
\end{subfigure}\hfill
\begin{subfigure}[t]{0.465\linewidth}
  \centering
  \includegraphics[width=\linewidth]{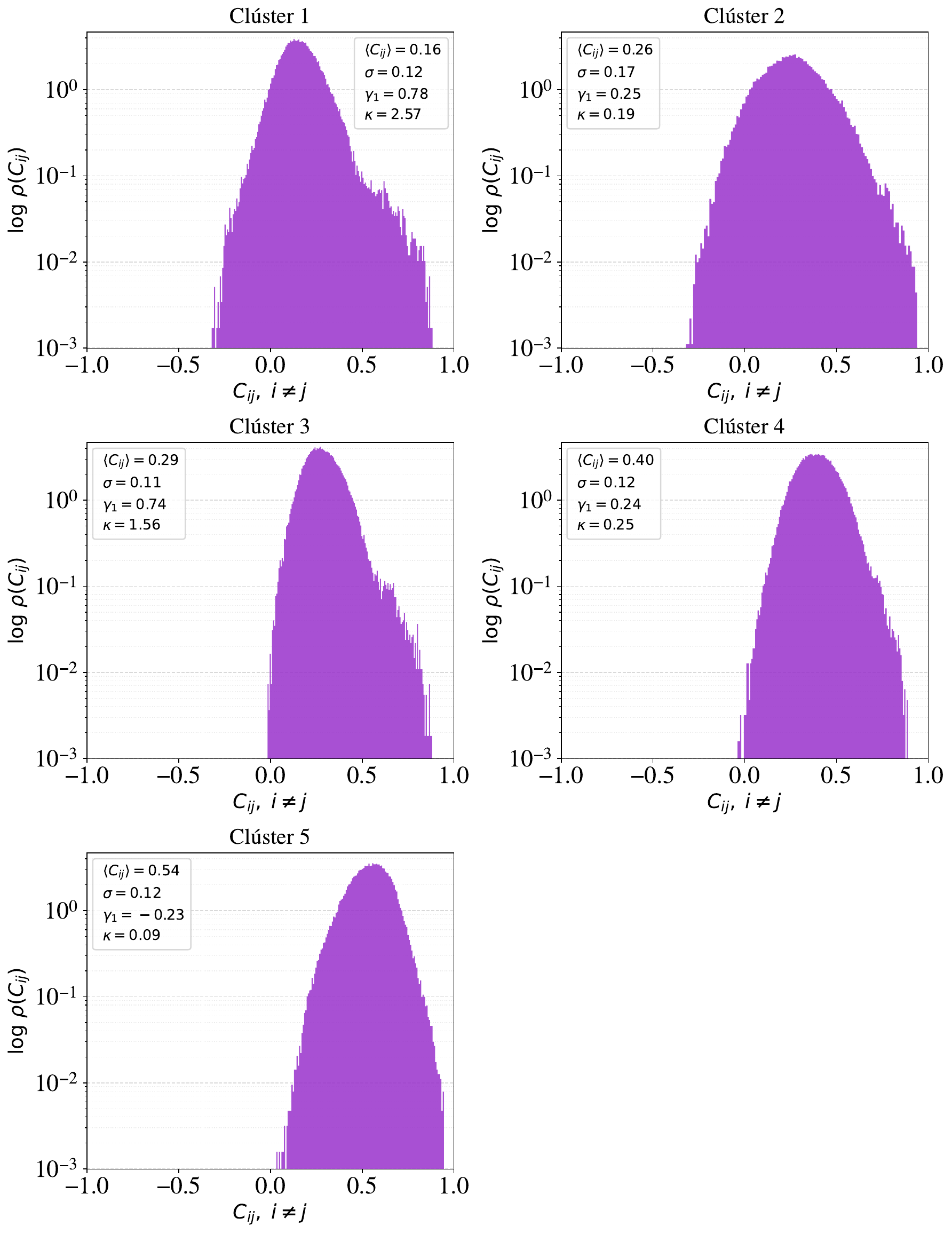}
  \caption*{\footnotesize d) $q=40$, $k=5$}
\end{subfigure}
\caption[Correlation histograms by \textit{cluster} on a semi-logarithmic scale]%
{Correlation histograms by \textit{cluster} on a \emph{semi-logarithmic scale}.
As the \textit{cluster} index increases, the average correlation $\langle C_{ij}\rangle$ grows, and the \emph{multimodality} becomes more evident in \emph{cluster} 1.
The following symbols are used: standard deviation $\sigma$, skewness $\gamma_{1}$, and kurtosis $\kappa$.}
\label{fig:histCorrByCluster_q20q40_k4k5}
\end{figure}

\clearpage

\section{\texorpdfstring{Dominant Eigenvalues $\lambda_N$}{Eigenvalores dominantes lambda N}}

\index{Eigenvalue!maximum}
\label{section: Dinamica Eigenvalores}

In this section, the dynamics indicated by the leading eigenvalues of the correlation matrix are analyzed, with emphasis on the maximum eigenvalue $\lambda_{N}$, interpreted as the \emph{SOTM} \index{State of the Market} \cite{mantegnaIntroductionEconophysicsCorrelations2004}. The temporal accumulation of $\lambda_{N}$ is examined and, in a complementary manner, the behavior of $\lambda_{N-1}$ and $\lambda_{N-2}$ is examined in Appendix \ref{sec: segundos y terceros eigen}. In turn, states were defined and ordered by means of the \emph{k}-Means algorithm \index{K-Means}, and transition matrices are constructed. Finally, the triplets $(\lambda_{N-2},\lambda_{N-1},\lambda_{N})$ are represented in three-dimensional space in order to visualize the spectral organization associated with the \emph{SOTM}.

Fig. \ref{fig:LambdaMax} shows the variation of $\lambda_N$ for short-epoch windows $q\in{20,40}$. The two series presented qualitatively similar patterns; nevertheless, for $q=20$ oscillations of greater amplitude and more abrupt local variations were observed, reflecting a greater sensitivity to short-term fluctuations. By contrast, for $q=40$ a smoothed profile was obtained.

\begin{figure}[ht]
\centering

\begin{subfigure}{0.96\linewidth}
    \centering
    \includegraphics[width=\linewidth]{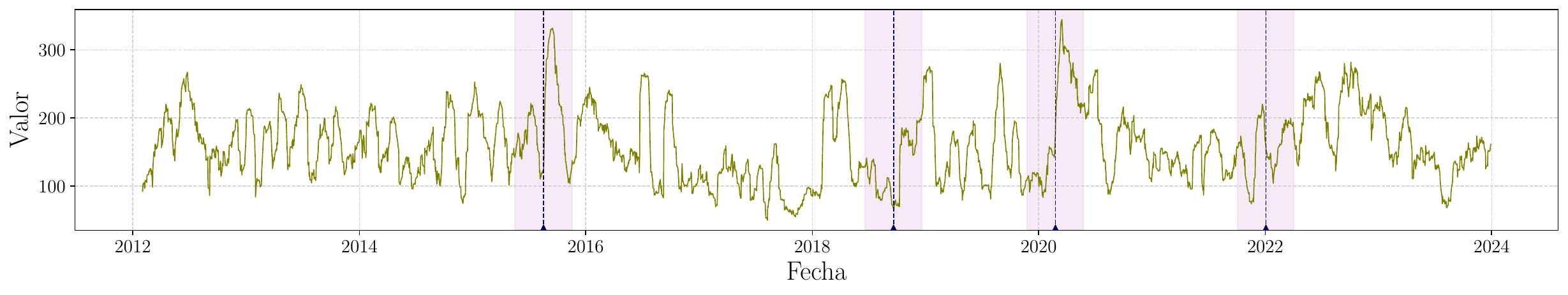}
    \caption*{a) $q=20$}
\end{subfigure}

\medskip

\begin{subfigure}{0.96\linewidth}
    \centering
    \includegraphics[width=\linewidth]{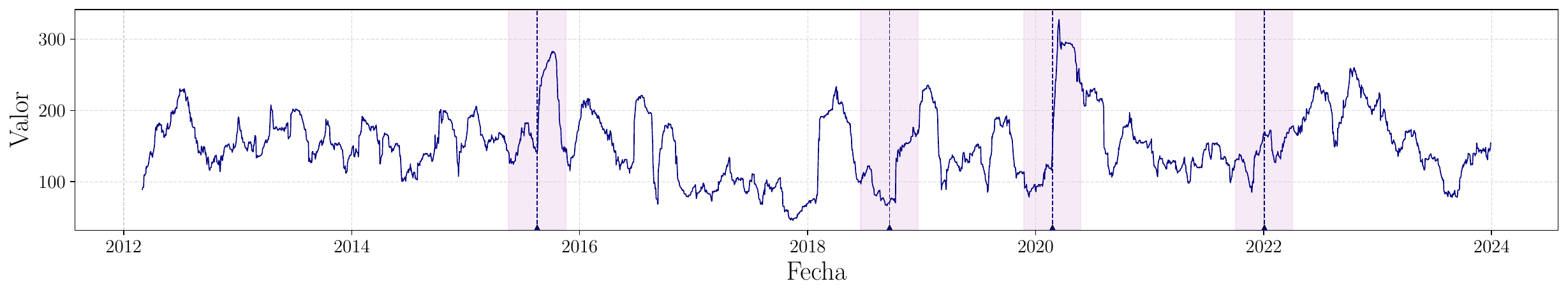}
    \caption*{b) $q=40$}
\end{subfigure}

\caption[Trajectories of the leading eigenvalue (\emph{SOTM})]%
{Trajectories of the dominant eigenvalue $\lambda_{N}(t)$, an indicator of the \emph{SOTM}. 
With $q=20$, oscillations of greater amplitude and more abrupt local variations were observed, whereas with $q=40$ the profile was more smoothed. The vertical dotted lines indicate the episodes listed in Appendix \ref{tab:fechas_episodios}, and the semi-transparent violet rectangles represent the window of three months before and three months after each event.}
\label{fig:LambdaMax}

\end{figure}

\index{Sturges!rule}
\subsection{Dispersion of the \texorpdfstring{$\lambda_N$}{eigenvalor m\'aximo} Values}
\index{Clusters}
\label{subsection: Dispersion eigenvalor maximo}
Fig.~\ref{fig:lambdaN_scatter_png} shows the values of $\lambda_{N}$, color-coded according to the assigned \emph{cluster}. Following the \textit{crashes} of 2015 and 2020 (Appendix \ref{tab:fechas_episodios}), a greater permanence in the \emph{cluster} associated with the highest eigenvalues was observed, in agreement with the increase in correlations during episodes of stress reported in the literature \cite{martinezramosSeriesTiempoFinancieras2024,pharasiDynamicsMarketStates2024,sandovalCorrelationFinancialMarkets2012}. From the end of 2016 until March 2018, a trend toward lower values of $\lambda_{N}$ is appreciated, followed by a subsequent rebound. From March 2018 onward, a transition from \emph{cluster} 1 toward higher \emph{clusters} was recorded.
\index{Crash!dates}
\begin{figure}[ht]
\centering
\begin{subfigure}{0.49\linewidth}
  \centering
  \includegraphics[width=\linewidth]{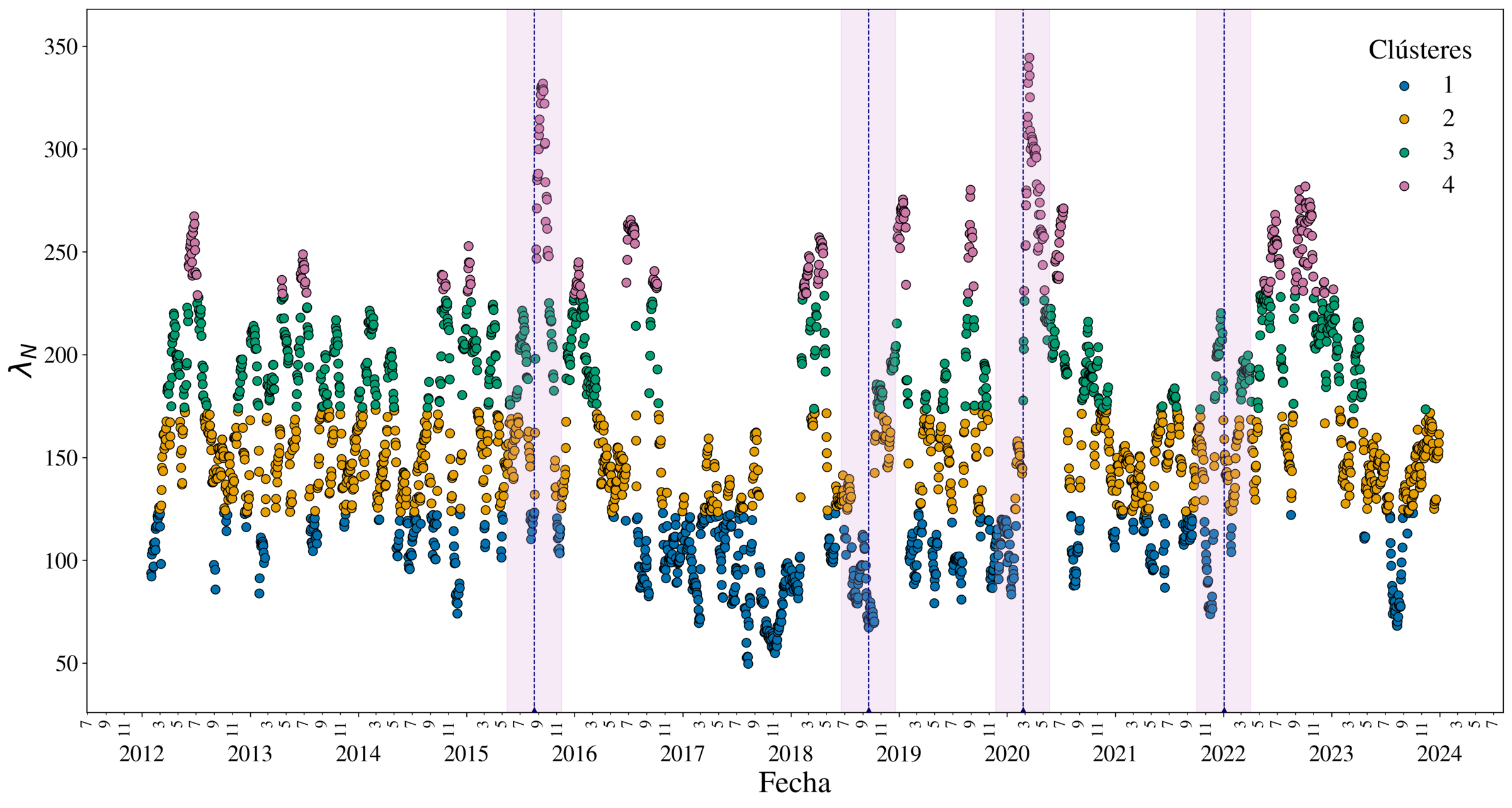}
  \caption*{a) $q=20$, $k=4$}
\end{subfigure}\hfill
\begin{subfigure}{0.49\linewidth}
  \centering
  \includegraphics[width=\linewidth]{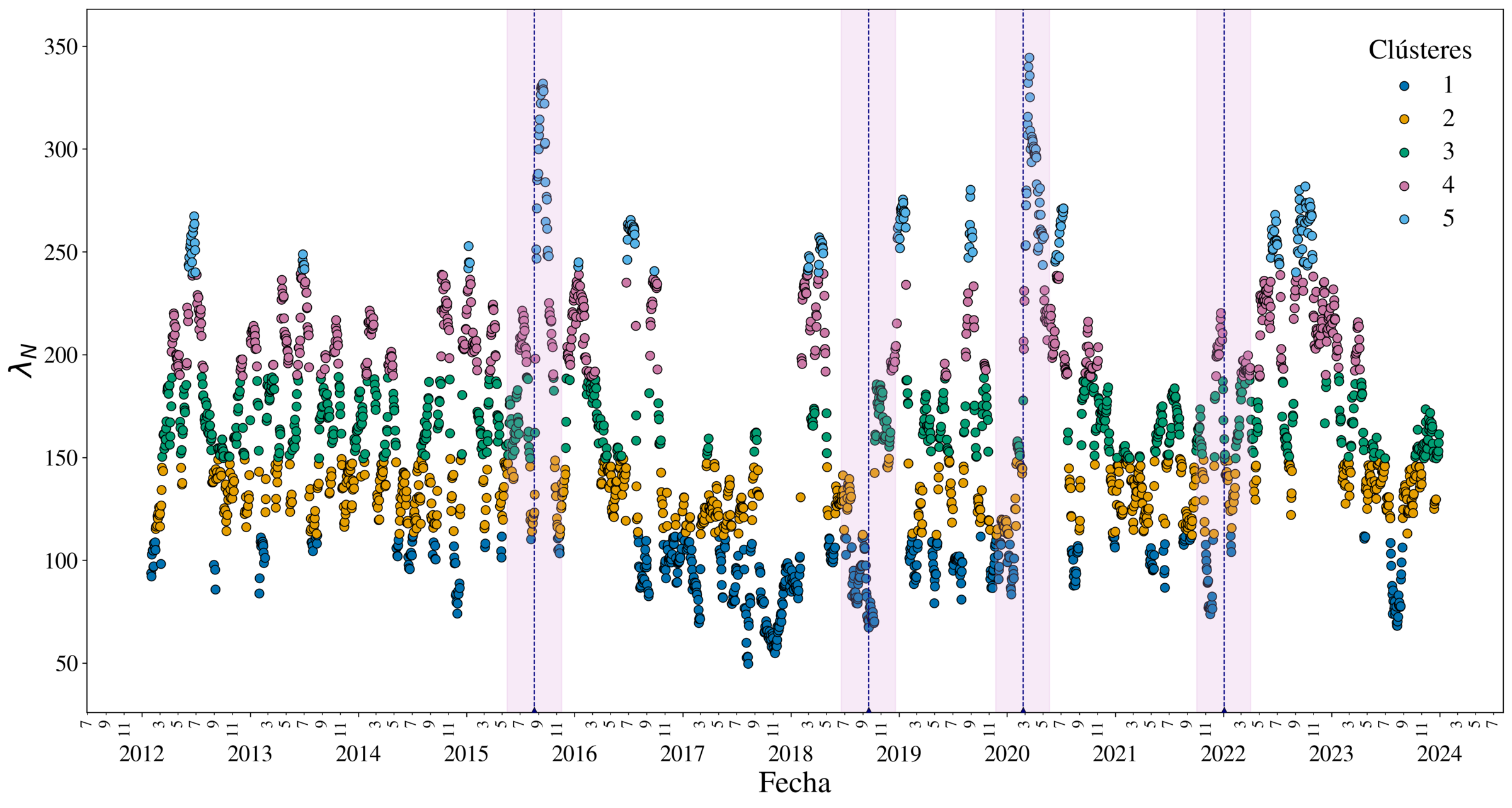}
  \caption*{b) $q=20$, $k=5$}
\end{subfigure}
\vspace{0.7em}
\begin{subfigure}{0.49\linewidth}
  \centering
  \includegraphics[width=\linewidth]{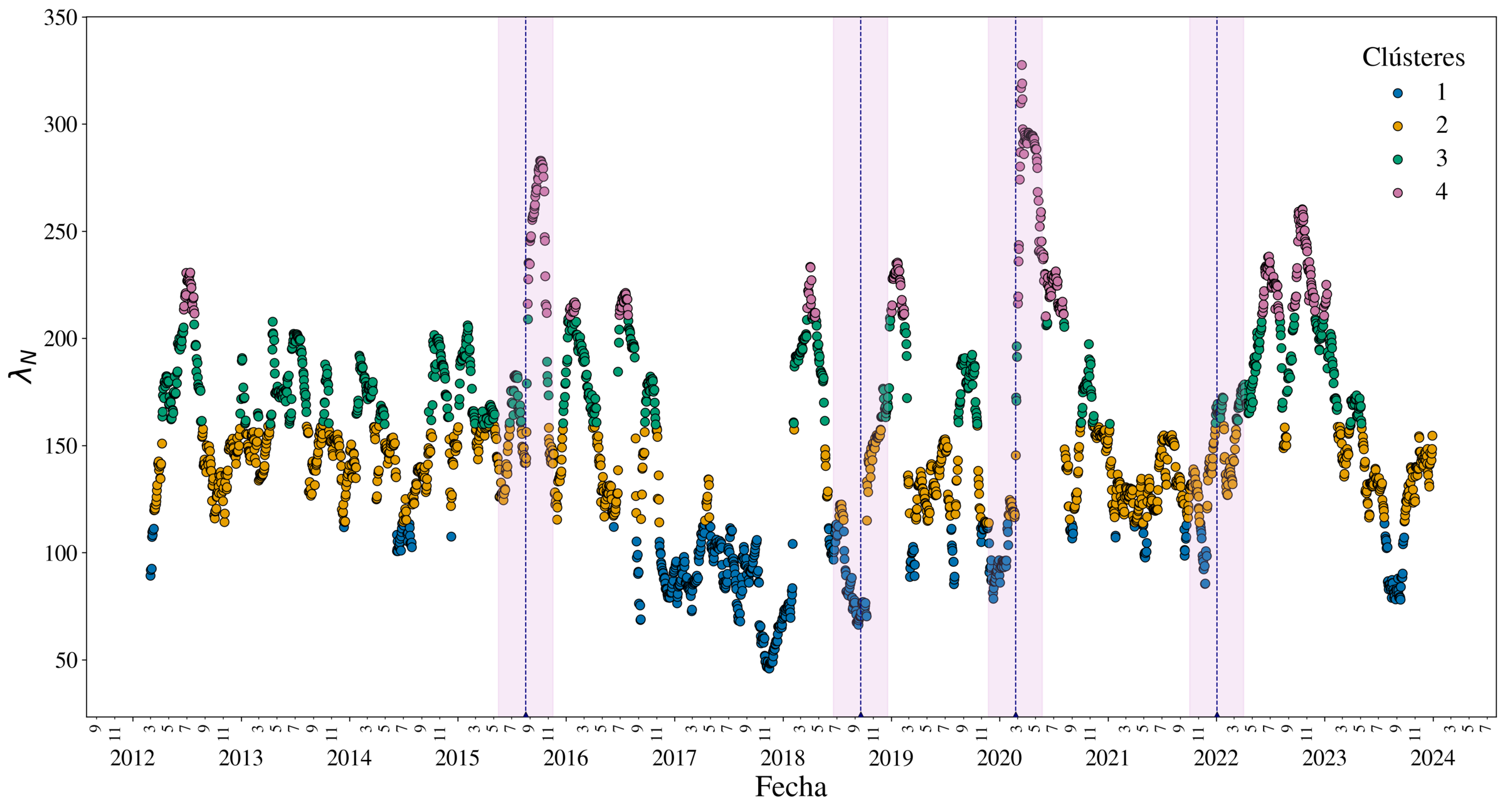}
  \caption*{c) $q=40$, $k=4$}
\end{subfigure}\hfill
\begin{subfigure}{0.49\linewidth}
  \centering
  \includegraphics[width=\linewidth]{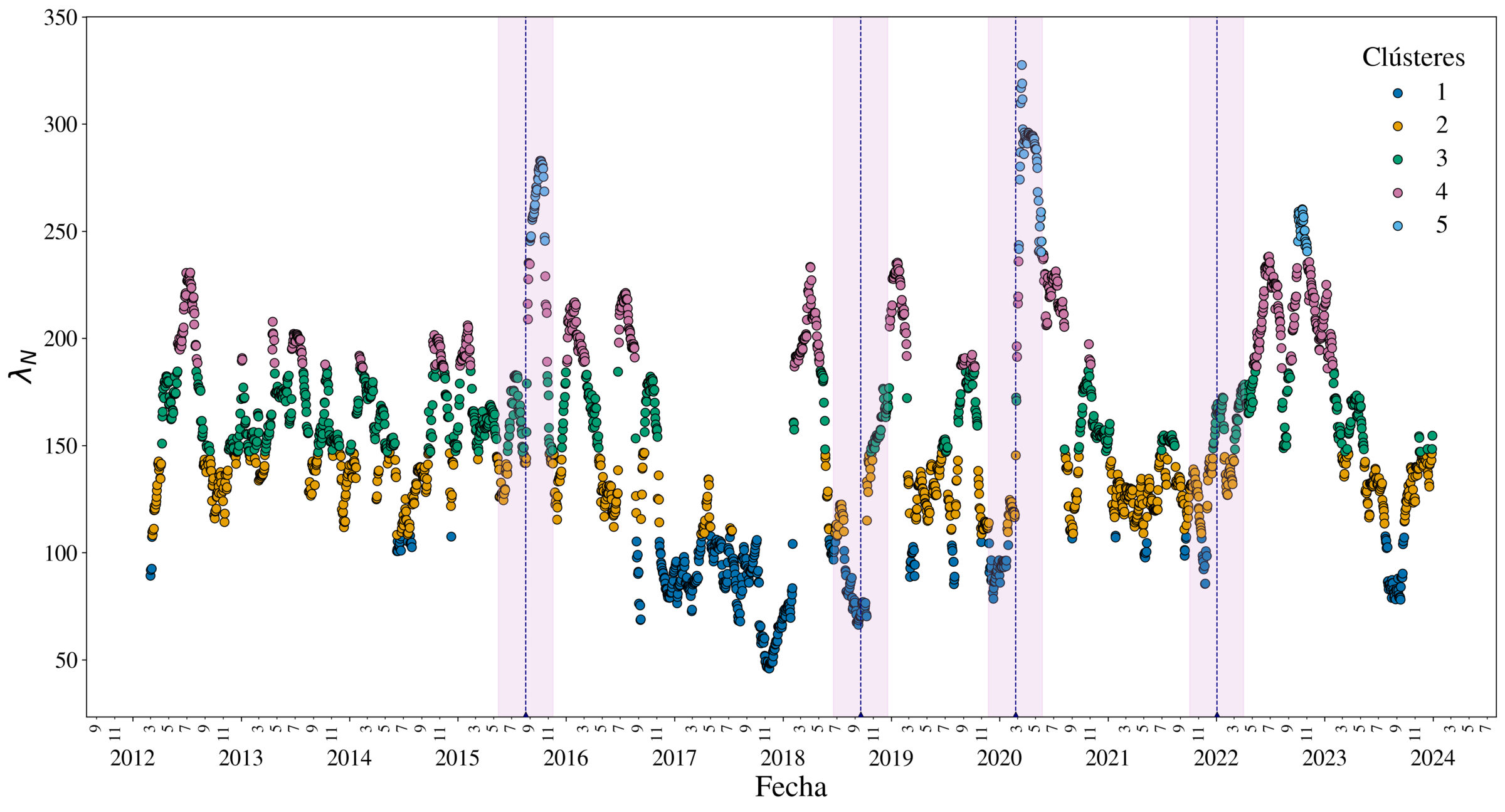}
  \caption*{d) $q=40$, $k=5$}
\end{subfigure}
\caption[Dispersion of the $\lambda_N$ values]%
{Values of $\lambda_{N}$ color-coded according to the \emph{cluster}. 
Following the \emph{crashes} of 2015 and 2020, a greater permanence in the \emph{cluster} associated with the highest eigenvalues is observed, in agreement with the increase in correlations during episodes of stress reported in the literature.}
\label{fig:lambdaN_scatter_png}
\end{figure}
In the figures of Appendix~\ref{fig:lambdaNmenos1_grid} and Appendix~\ref{fig:lambdaNmenos2_grid}, the plots of the following two leading eigenvalues $\lambda_{N-1}$ and $\lambda_{N-2}$ are found.

\subsection{Temporal Evolution of \textit{MS}}
\index{Market!States}
\index{Clusters}
\label{subsection: Evolucion temporal State of the Market}

In Fig.~\ref{fig: evolucionEstados eigenvalores lambda max} it is observed that, in the four state evolutions, \emph{cluster} 2 is activated in time intervals that do not necessarily coincide with the \emph{COVID} State\index{State!COVID}. This result suggests that a classification based solely on the value of $\lambda_N$ is not sufficient to unambiguously isolate that state. In agreement with this idea, the \emph{coarse graining} analysis of correlation matrices reported by Mart\'inez Ramos \emph{et al.} \cite{martinez-ramosCoarseGrainingCorrelation2024} indicates that the robust identification of states requires incorporating additional information from the spectral structure, such as the contribution of the eigenvectors and, in general, the complete form of the correlation matrix. It is worth noting that, although a reduction of parameters was considered in that work, the \emph{COVID} State was not included in the analysis.

\begin{figure}[p]   
\centering

\captionsetup[subfigure]{skip=2pt}
\setlength{\abovecaptionskip}{2pt}
\setlength{\belowcaptionskip}{2pt}

\begin{subfigure}{\textwidth}
  \centering
  \includegraphics[height=0.19\textheight,keepaspectratio]{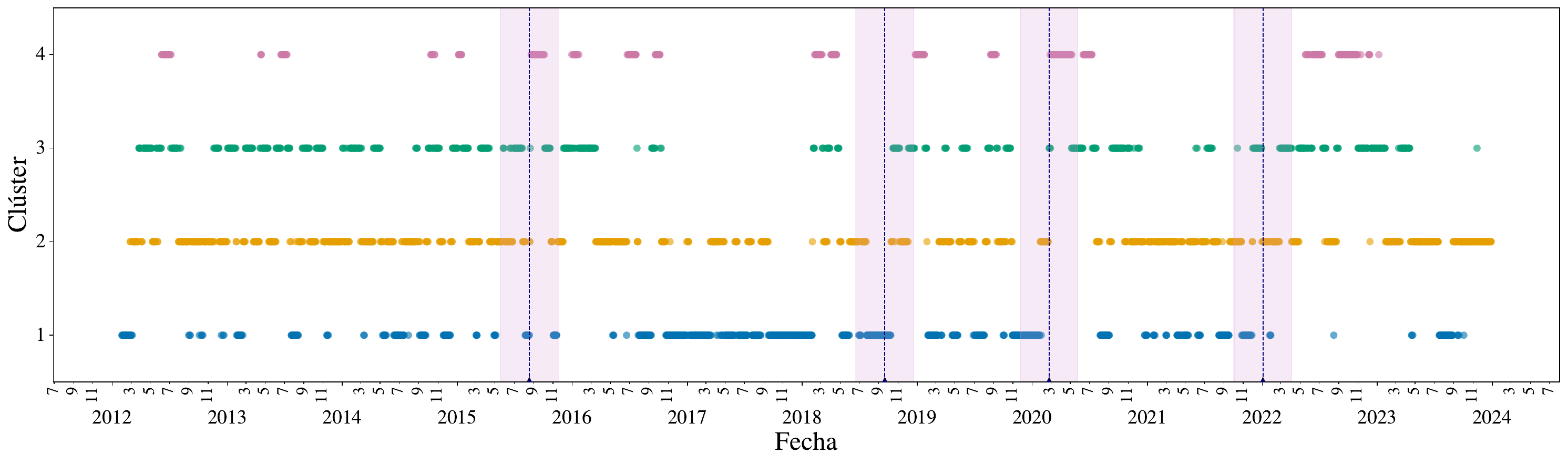}
  \caption*{a) $q=20$, $k=4$}
\end{subfigure}

\par\vspace{-2pt}

\begin{subfigure}{\textwidth}
  \centering
  \includegraphics[height=0.19\textheight,keepaspectratio]{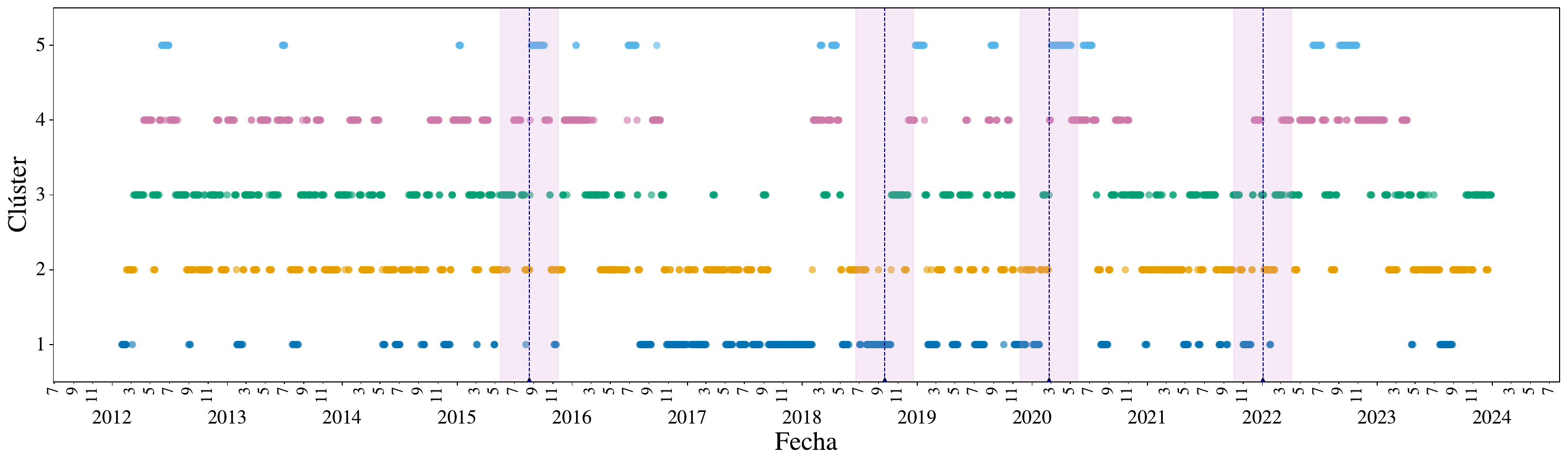}
  \caption*{b) $q=20$, $k=5$}
\end{subfigure}

\par\vspace{-2pt}

\begin{subfigure}{\textwidth}
  \centering
  \includegraphics[height=0.19\textheight,keepaspectratio]{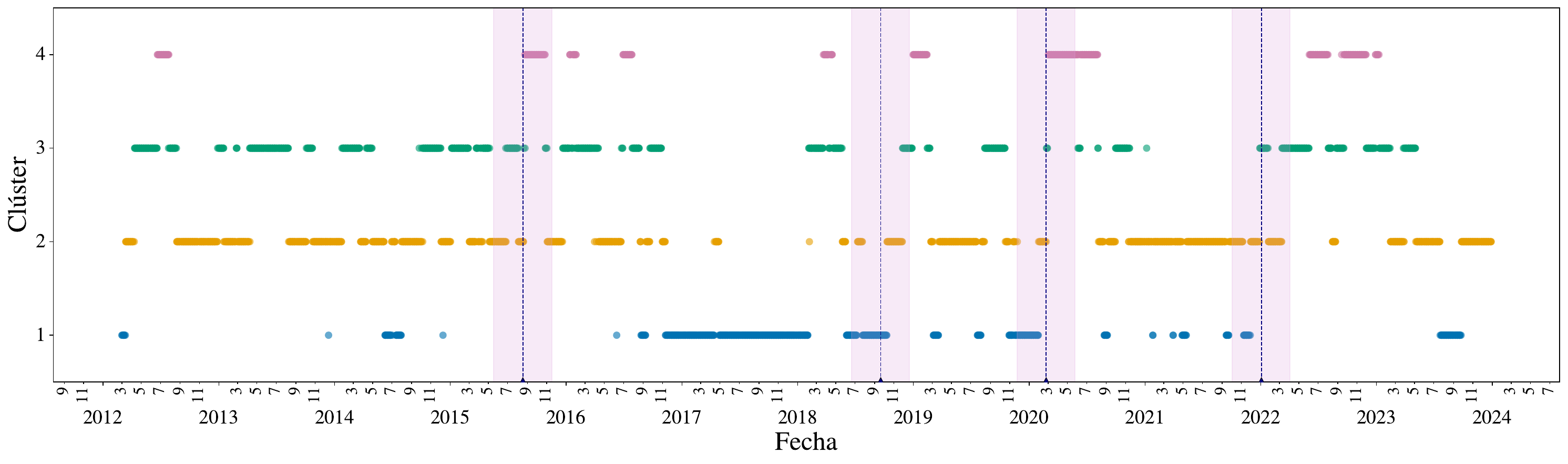}
  \caption*{c) $q=40$, $k=4$}
\end{subfigure}

\par\vspace{-2pt}

\begin{subfigure}{\textwidth}
  \centering
  \includegraphics[height=0.19\textheight,keepaspectratio]{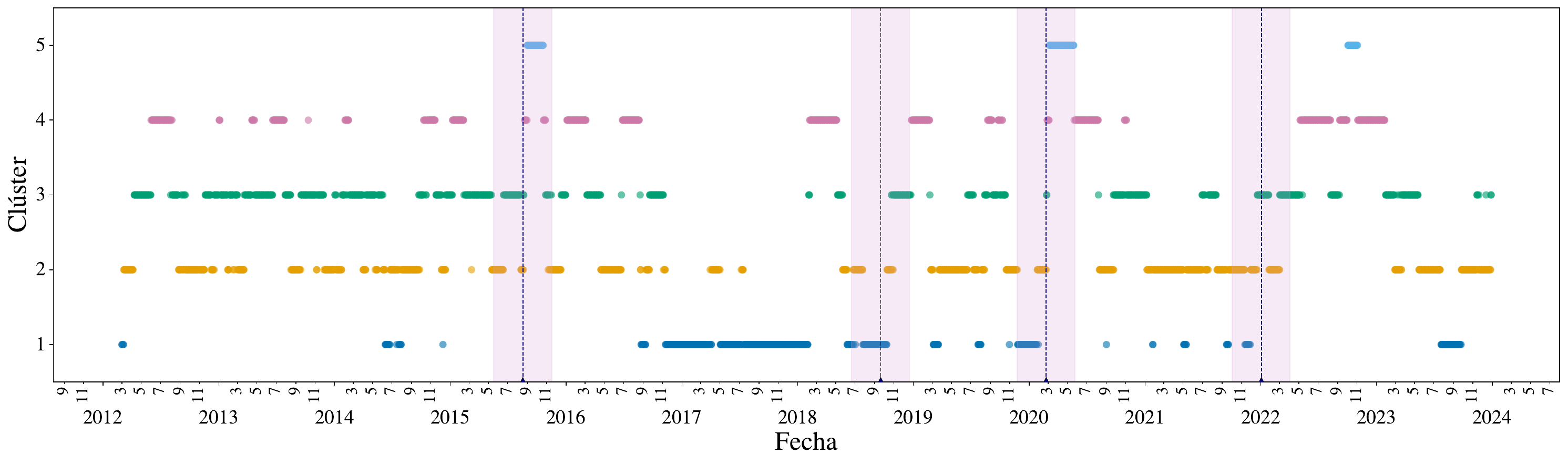}
  \caption*{d) $q=40$, $k=5$}
\end{subfigure}

\caption[Temporal evolution of \textit{MS} indicated by \emph{clustering} of $\lambda_N$]%
{Temporal evolution of the \textit{MS} induced by means of \emph{clustering} of $\lambda_N$. 
It is observed that the \emph{cluster} associated with the \emph{COVID} State appears in time intervals distinct from the episode of that event.}
\label{fig: evolucionEstados eigenvalores lambda max}
\end{figure}

\clearpage
\subsection{Transition Matrices}
\label{subsection: Matriz de transici\'on eigenvalores}
\index{Matrix!transition}
The transition matrices induced by the \emph{clustering} of the leading eigenvalues are shown in Fig.~\ref{fig:trans-matrices-eigenvalores}: the greatest persistence on the main diagonal corresponds to \emph{cluster} 2, with self-transition counts between \(701\) and \(1142\). In addition, the highest-index state (corresponding to \emph{cluster} 4 or \emph{cluster} 5, depending on the configuration) always exhibits the lowest frequency of permanence, with values in the range of 111 to 329.

It is worth noting that, unlike what occurs in the correlation matrices (Fig. \ref{fig:matrices_transicion_corr}), the largest value on the diagonal always occurs in \emph{cluster} 2. 

\begin{figure}[H]
\centering

\begin{minipage}{0.49\linewidth}
  \centering
  \includegraphics[width=\linewidth]{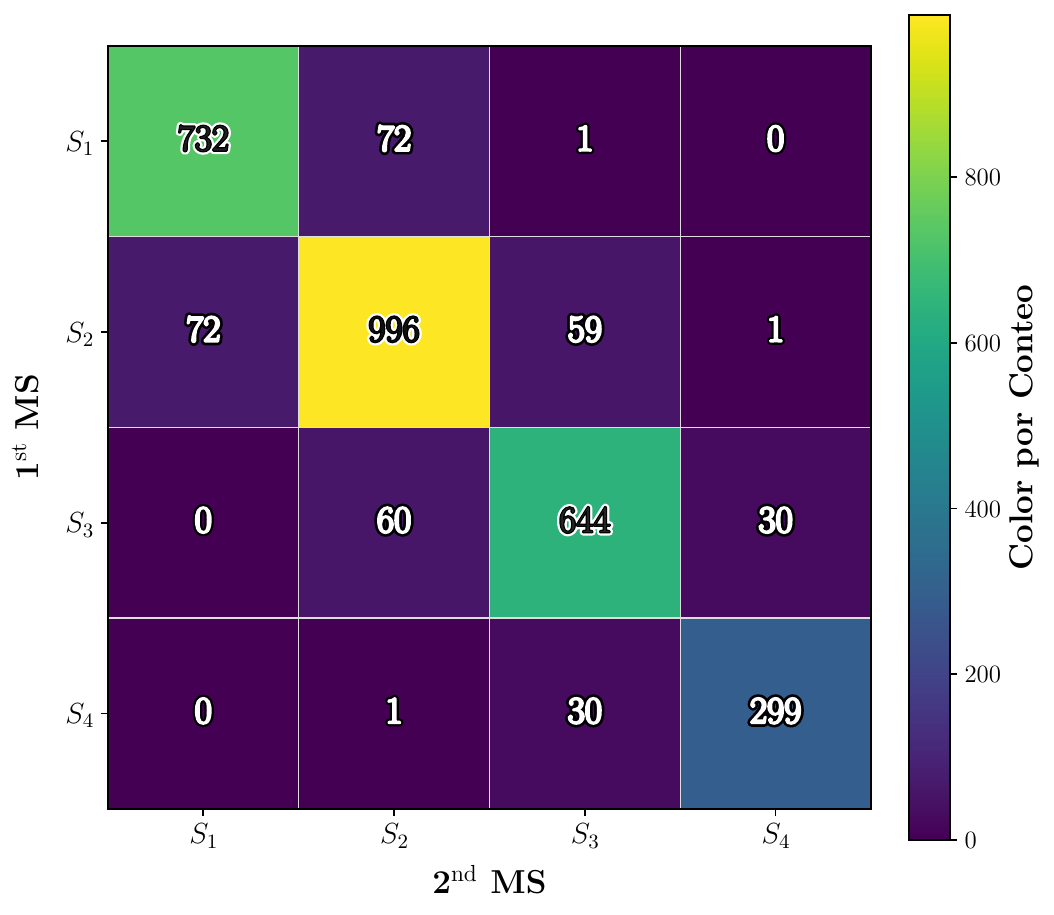}
  {\footnotesize a) $q=20$, $k=4$}
\end{minipage}\hfill
\begin{minipage}{0.49\linewidth}
  \centering
  \includegraphics[width=\linewidth]{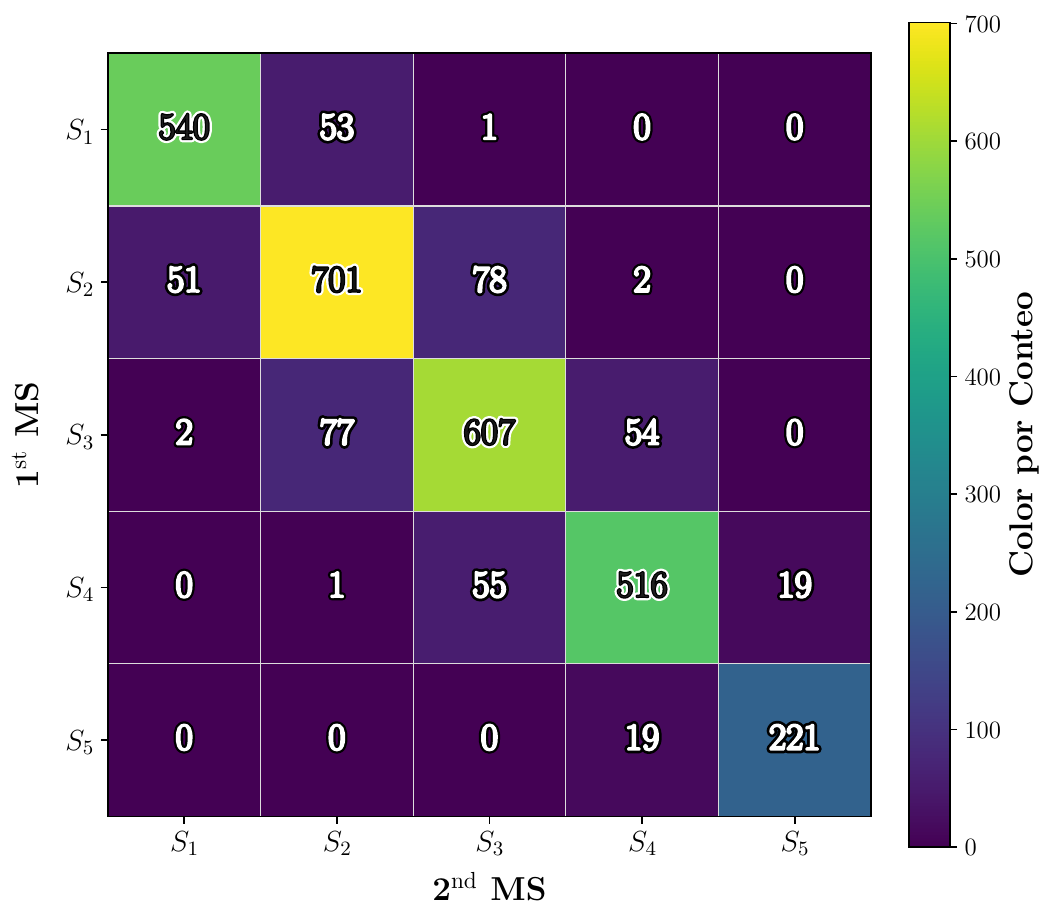}
  {\footnotesize b) $q=20$, $k=5$}
\end{minipage}

\vspace{0.8em}

\begin{minipage}{0.49\linewidth}
  \centering
  \includegraphics[width=\linewidth]{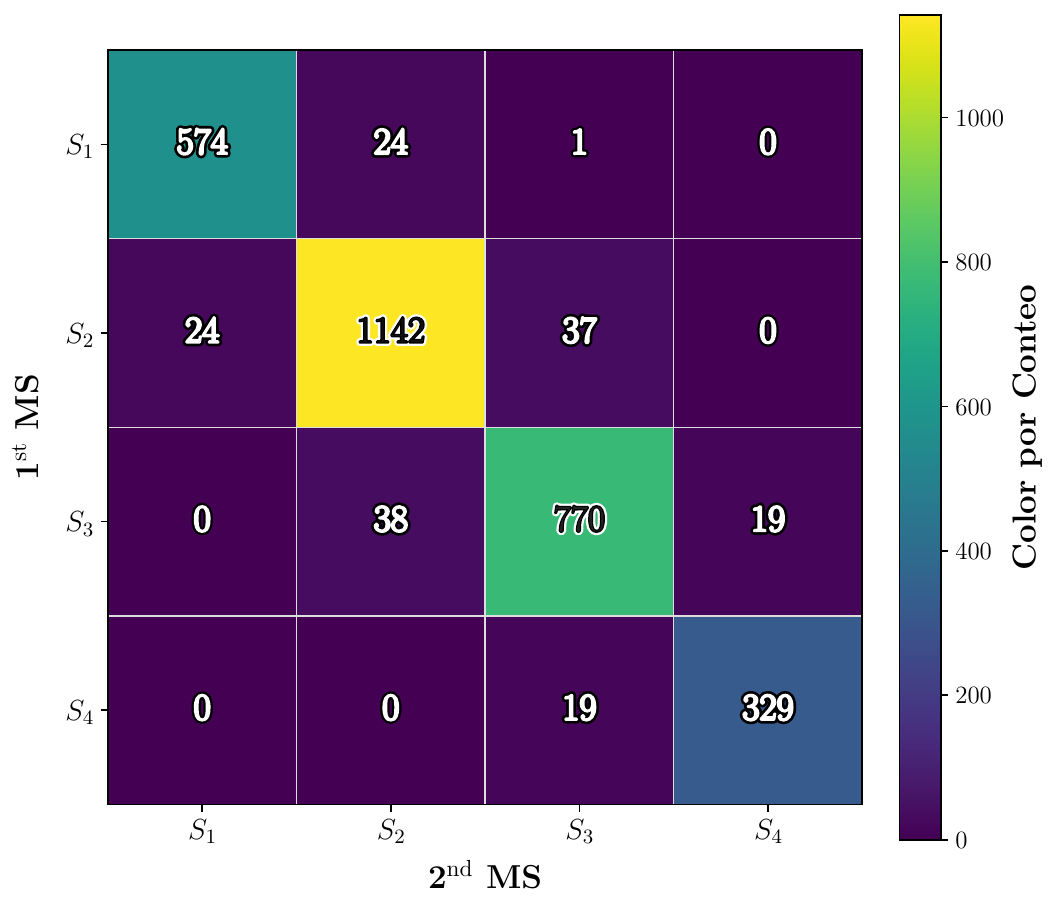}
  {\footnotesize c) $q=40$, $k=4$}
\end{minipage}\hfill
\begin{minipage}{0.49\linewidth}
  \centering
  \includegraphics[width=\linewidth]{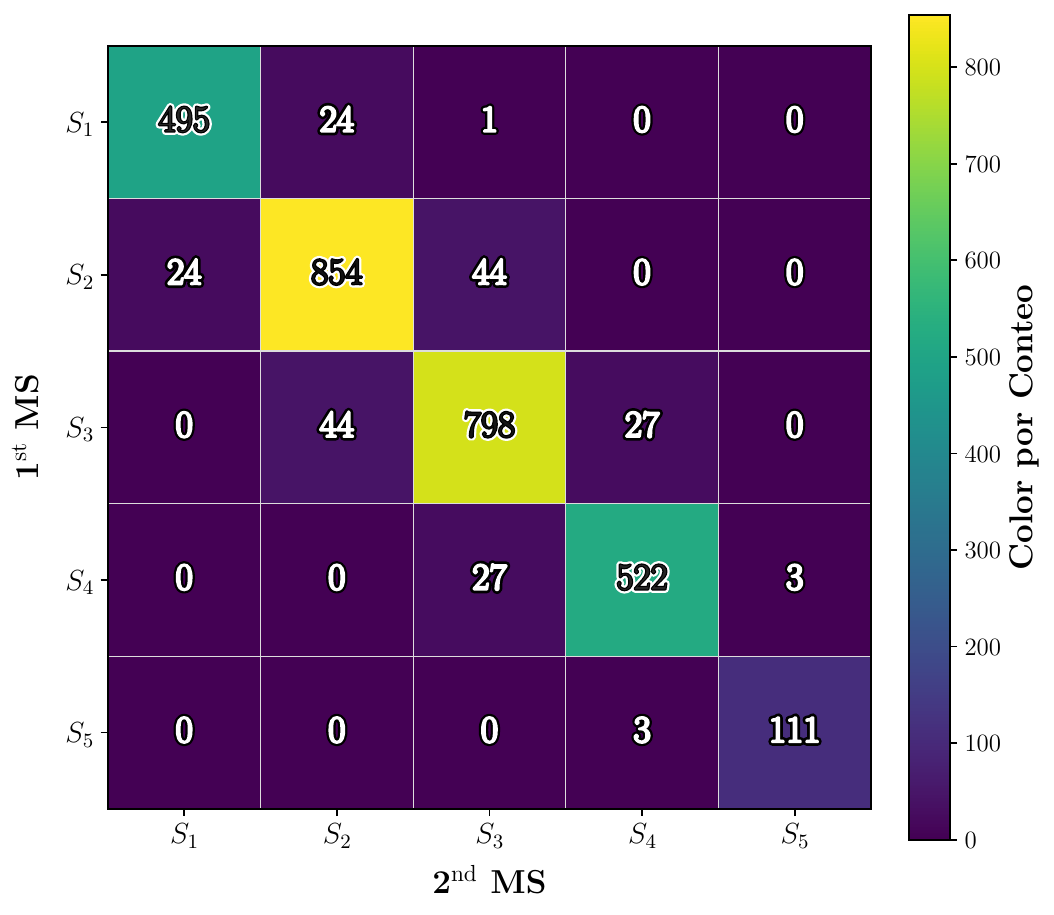}
  {\footnotesize d) $q=40$, $k=5$}
\end{minipage}

\caption[Transition matrices indicated by $\lambda_N$]{Representation of the transition matrices indicated by $\lambda_N$. In all cases an approximately tridiagonal and nearly symmetric structure is observed (with differences of the order of units).}
\label{fig:trans-matrices-eigenvalores}
\end{figure}

\clearpage

\subsection{Three-Dimensional Vector Formed by the Three Leading Eigenvalues}
\index{Triplet!leading eigenvalues}
\label{subsection: Triplete eigenvalores}
The spectral vector of each window is defined as:
\[
\boldsymbol{\bar{\lambda}(t)}\;=\;\bigl[\lambda_{N-2}(t),\,\lambda_{N-1}(t),\,\lambda_{N}(t)\bigr]\in\mathbb{R}^{3},
\]
that is, the triplet formed by the three leading eigenvalues of the correlation matrix in the window that closes at \(t\). The \textit{k}-Means algorithm\index{K-Means} was applied to the cloud of triplets \(\boldsymbol{\{{\bar{\lambda}(t)}\}_{t}}\) and, after relabeling in ascending order of the average of \(\lambda_{N}\) per cluster, a high consensus was obtained: the IEM stability turned out to be greater than or equal to \(98\%\) for each configuration of \(q\) and \(k\).

Fig.~\ref{fig:tripleteLambdaMax} shows the three-dimensional distribution of $\boldsymbol{\bar{\lambda}(t)}$ ordered in ascending order by \(\lambda_{N}\). For \(q=20\), a sharp pattern was observed: in \emph{cluster} 1, the cloud exhibits a \emph{broad base}, and as the \emph{cluster} index increases the mass \emph{contracts and sharpens}, tending toward a \emph{tip} pointing to the positive part of the vertical axis. By contrast, with \(q=40\), the dispersion remains more extended across the whole range of \emph{clusters} and the concentration toward a vertex does not follow a sharpening, so that the geometry of the higher \emph{clusters} proves less pointed and looser in comparison with the \(q=20\) case. This pattern is consistent with expression \eqref{eq:traza-espectro} of Appendix~\ref{subsec:Rango efectivo corr}: since $\mathrm{tr}\,C(t,q)=N$ is constant, any increase in $\lambda_{N}$ must be compensated by a decrease in the rest of the spectrum.

\begin{figure}[ht]
    \centering
    \begin{subfigure}{0.48\linewidth}
        \centering
        \includegraphics[width=\linewidth]{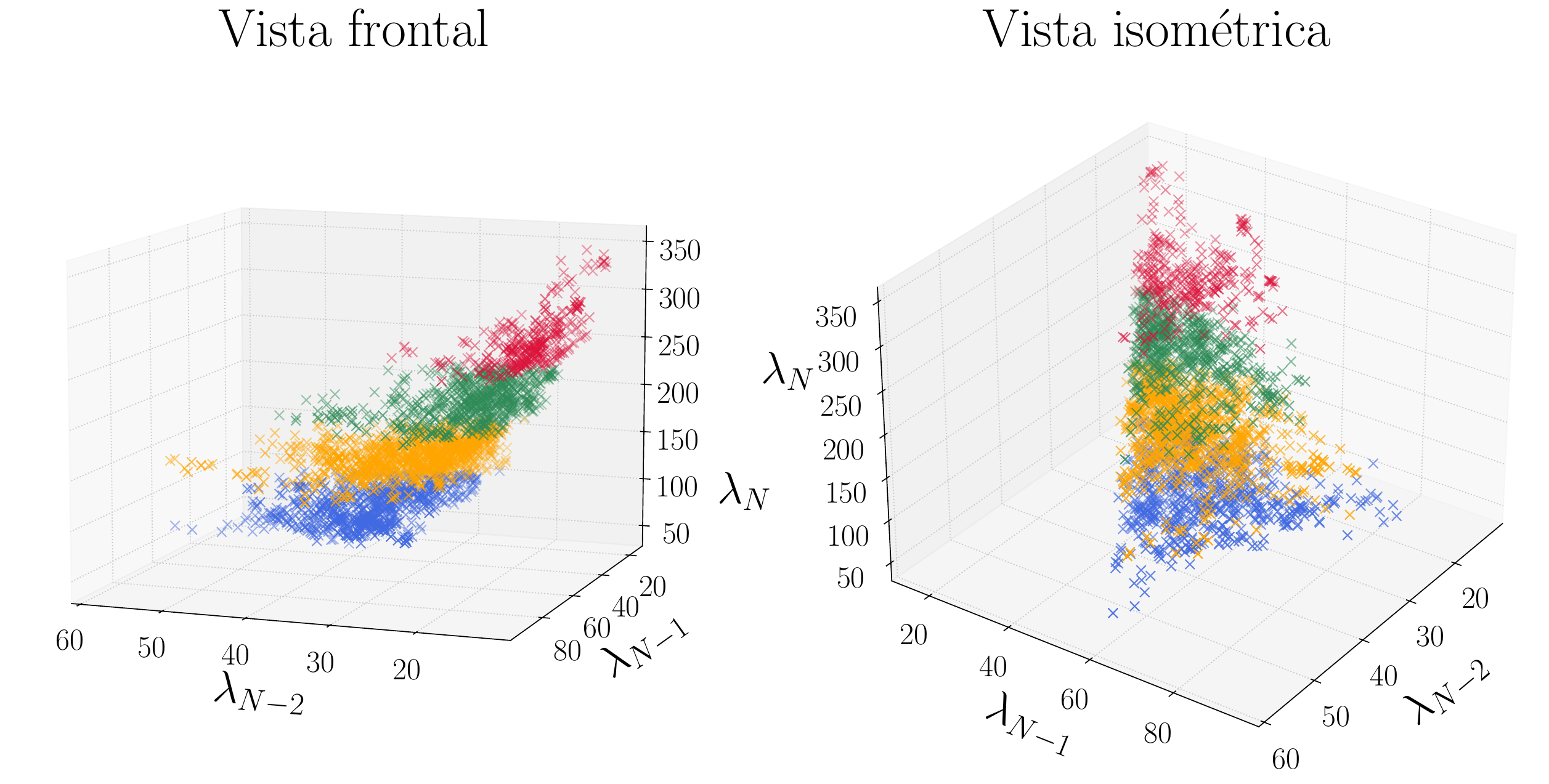}
        \caption*{a) $q=20$, $k=4$}
    \end{subfigure}
    \hfill
    \begin{subfigure}{0.48\linewidth}
        \centering
        \includegraphics[width=\linewidth]{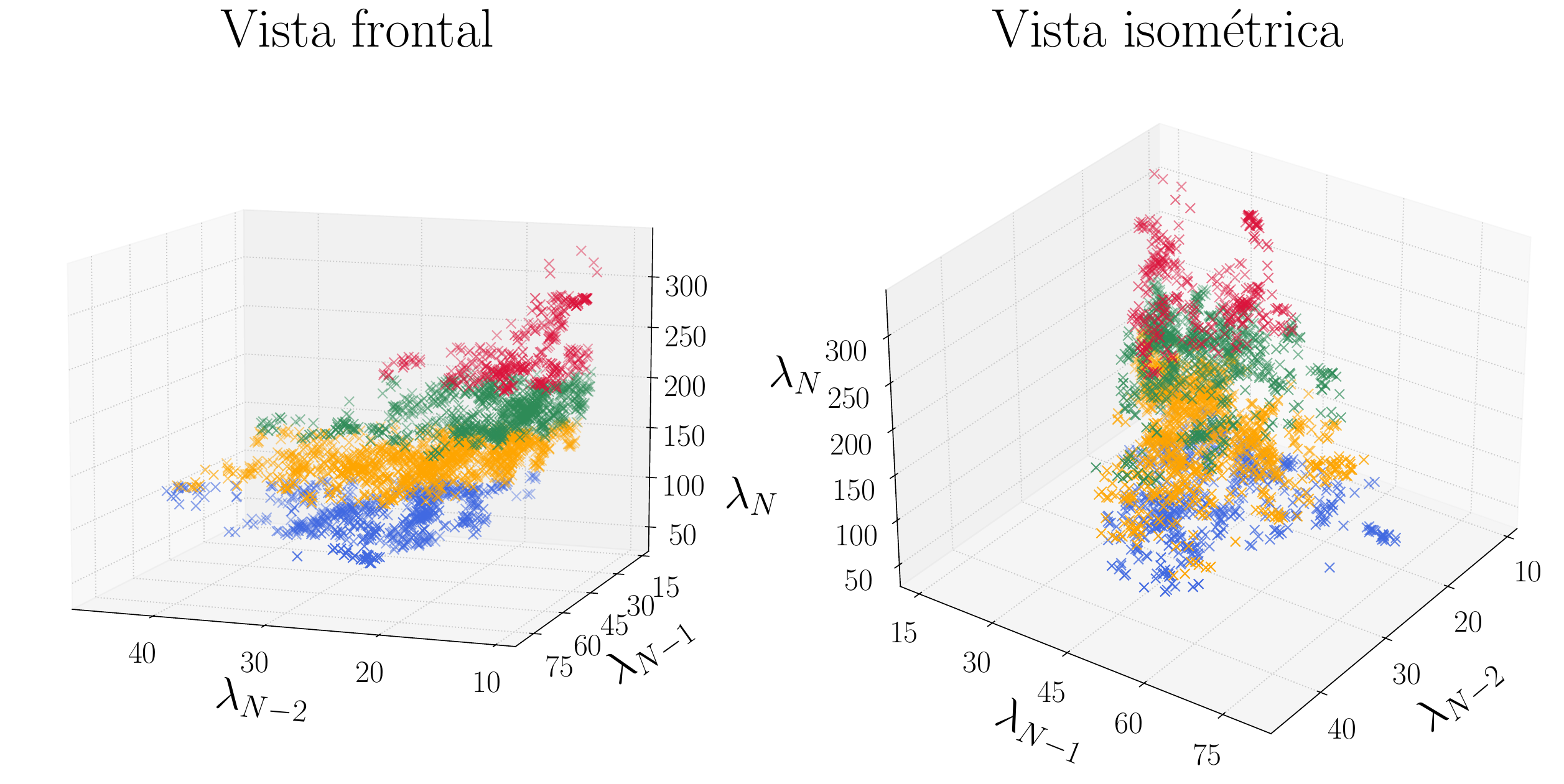}
        \caption*{b) $q=40$, $k=4$}
    \end{subfigure}

 \caption[Triplets of leading eigenvalues]%
{Triplets $\bigl[\lambda_{N-2}(t),\,\lambda_{N-1}(t),\,\lambda_{N}(t)\bigr]$ for configurations with $k=4$. 
At $q=20$, it is observed that the points tend to concentrate in a more compact structure with a broad base and a tip at the top. At $q=40$, the dispersion of the points remains more extended and the concentration toward a vertex is not as marked, 
generating a less sharpened geometry.}
\label{fig:tripleteLambdaMax}
\end{figure}

\section{Dominant Eigenvectors \texorpdfstring{$[v_i^2]_{i=1}^N$}{[vi²]}}
\label{section:Eigenvectores al cuadrado}
\index{Eigenvector!squared entries}

In this section, the eigenvectors associated with the leading eigenvalue $\lambda_N$ of the correlation matrices are analyzed through their squared entries, 
\([v_i^2]_{i=1}^N\). 
The study focuses on the \emph{MS} induced by means of the \emph{clustering} of these vectors. 
The distribution of the squared entries is presented in bar charts, where each bar is colored according to the corresponding economic sector, with the aim of estimating the \textbf{relative participation} of each stock in the collective dynamics. 

\subsection{Temporal Evolution of \textit{MS}}
\index{Market!States}
\index{Clusters}

To identify the \emph{MS}, the \textit{k}-Means algorithm \index{K-Means} was applied with configurations of \(k=4\) and \(k=5\). Fig.~\ref{fig:evolucionMS_squared_q20q40} shows the temporal evolution of the states from the eigenvectors with squared entries. It is observed that the \textit{COVID} episode emerges in \emph{cluster} 2 for \(k=4\) (both at \(q=20\) and at \(q=40\)), whereas for \(k=5\) this state is identified in \emph{cluster} 3. However, this state does not present itself as a strictly isolated regime, since the corresponding \emph{clusters} contain observations both prior to and subsequent to the \textit{COVID} State.

In addition, the 2017-2018 period appears dominated by \emph{cluster} 1, although, unlike what was observed in the analysis of the full correlation matrix, this period shows a significant participation of higher-index clusters. 
\index{State!COVID}

\subsection{Sectoral Distribution of the Dominant Eigenvector \texorpdfstring{$[v_i^2]_{i=1}^N$}{[vi²]}}
\index{Eigenvector!squared entries}

Fig.~\ref{fig: EigVect cuadrado q40} shows the temporal averages of the entries $[v_i^2]_{i=1}^N$ of the eigenvector $\mathbf{v}_{N}$ associated with the dominant eigenvalue $\lambda_N$, computed with windows $q=40$. Since the leading eigenvectors are normalized, it holds that
\begin{equation}
\sum_{i=1}^{N}\left[v_i^{2}\right] = 1.
\end{equation}
Consequently, for each asset $i$ the \textbf{relative participation weight} is defined as
\begin{equation}
w_i \;\equiv\; \left[v_i^{2}\right],
\label{eq:def-weights-participacion}
\end{equation}
so that $\sum_{i=1}^{N} w_i = 1$ as discussed in Section~\ref{sec:probabilistic}.

The dotted line indicates the reference value associated with a uniform assignment of weights over $N$ assets (also known as the \emph{discrete rectangular distribution}), for which $w_i=1/N$ \cite{balakrishnanPrimerStatisticalDistributions2003}. In this case, $N=430$, so that $w_i=\frac{1}{430}\approx 0.0023$. The coloring by sectors uses the GICS abbreviations of Table~\ref{tab:sectores-gics-sp500}. Table~\ref{tab:top10-vn2-q20} shows the 10 companies with the greatest average participation in the eigenvector associated with $\lambda_N$, computed across all the time windows. Finally, for the full time horizon a value of $\boldsymbol{\mathrm{IPR}_{N}}=0.00260$
\index{Inverse!Participation Ratio} was obtained and, consequently, $\boldsymbol{\mathrm{PR}_{N}}=385.0026$
\index{Participation!Ratio}, which suggests an effective participation of $385$ assets, equivalent to $89.53\%$ of the companies considered.

\clearpage
\begin{figure}[htbp]
  \centering
  \includegraphics[width=\textwidth]{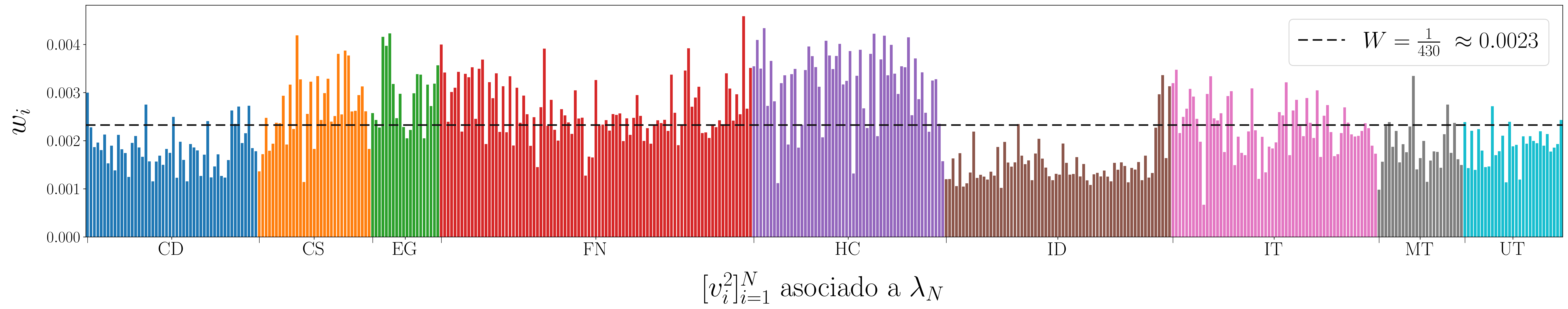}
  \caption[Averages of $v_i^2$ entries for the full time horizon, \texorpdfstring{$q=40$}{q=40}]%
  {Temporal averages of $\left[v_i^2\right]_{i=1}^{N}$ (eigenvector associated with $\lambda_{N}$) for $q=40$.
  The \emph{relative participation weight} $w_i\equiv v_i^2$ is defined and, by the normalization of the eigenvector,
  it holds that $\sum_{i=1}^{N} w_i = 1$.
  The dotted line indicates the reference value corresponding to a uniform assignment of weights over $N=430$ assets, that is, $w_i=1/N\approx 0.0023$.
  See also Fig.~\ref{fig: EigVect cuadrado q20}.
  For the full time horizon a value of $\boldsymbol{\mathrm{IPR}_{N}}=0.00260$ was obtained and, consequently, $\boldsymbol{\mathrm{PR}_{N}}=385.0026$, which suggests an effective participation of approximately $385$ assets, equivalent to $89.53\%$ of the companies considered.}
  \label{fig: EigVect cuadrado q40}
\end{figure}

\vspace{15em}
\begin{table}[ht]
  \centering
  \scriptsize
  \begin{tabular}{|c|c|c|r|}
    \hline
    \rowcolor{teal!80}
    \color{white}\textbf{Position} &
    \color{white}\textbf{Ticker} &
    \color{white}\textbf{Sector} &
    \color{white}\textbf{Value} \\
    \hline
    1  & WTW  & FN & 0.004587 \\
    2  & AMGN & HC & 0.004342 \\
    3  & DVN  & EG & 0.004233 \\
    4  & MOH  & HC & 0.004224 \\
    5  & GIS  & CS & 0.004190 \\
    6  & PFE  & HC & 0.004186 \\
    7  & CTRA & EG & 0.004162 \\
    8  & TECH & HC & 0.004149 \\
    9  & ABT  & HC & 0.004095 \\
    10 & GILD & HC & 0.004079 \\
    \hline
  \end{tabular}
  \vspace{0.35cm}
   \caption[Leading companies of $v_i^2$, $q=20$]{Leading companies in the entries $[v_i^2]_{i=1}^N$ for a window of $q=40$ days. The companies coincide with those of Table \ref{tab:top10-vn2-q20}, although in a different order. The list of companies studied in this thesis is found in Appendix~\ref{tab:companies-sp500}.}
  \label{tab:top10-vn2-q40}
\end{table}

\index{Sectors!GICS} 
\clearpage
\begin{figure}[p]
\centering
\vspace*{\fill}

\begin{subfigure}{0.8\linewidth}
  \centering
  \includegraphics[width=\linewidth,height=0.20\textheight,keepaspectratio]{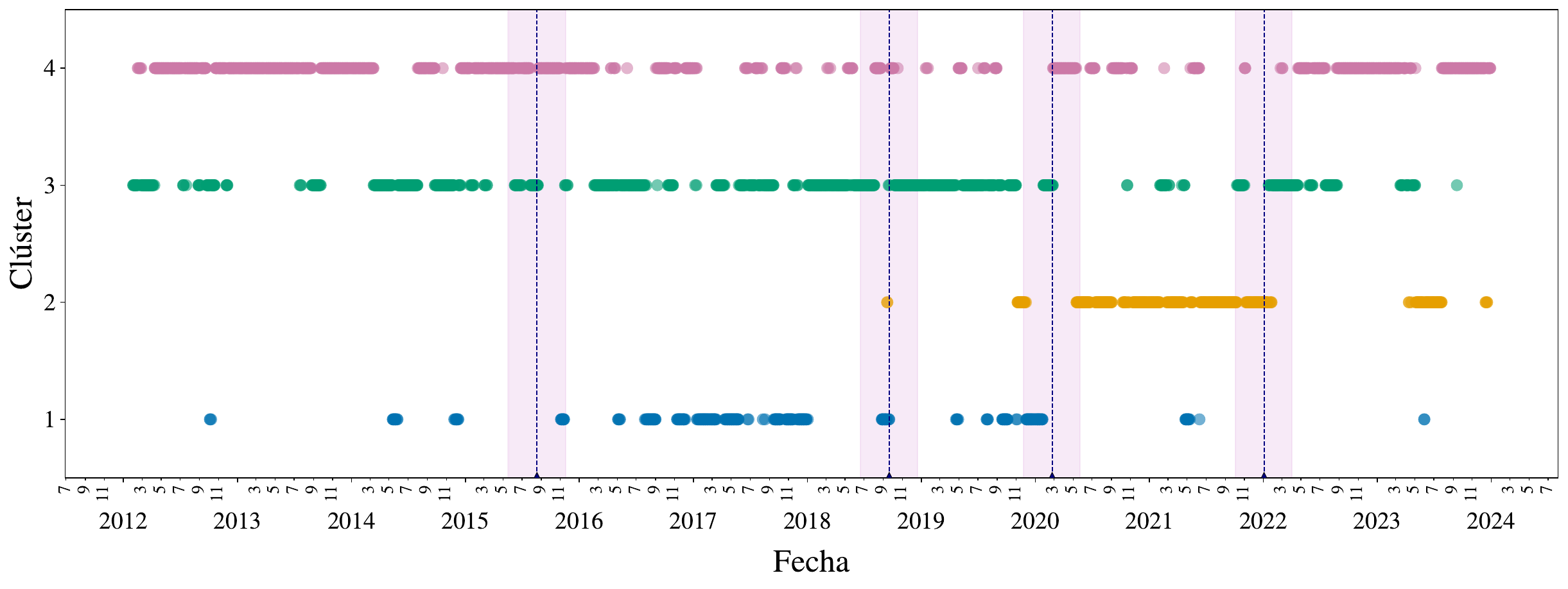}
  \caption*{a) $q=20,\; k=4$}
\end{subfigure}

\vfill

\begin{subfigure}{0.8\linewidth}
  \centering
  \includegraphics[width=\linewidth,height=0.20\textheight,keepaspectratio]{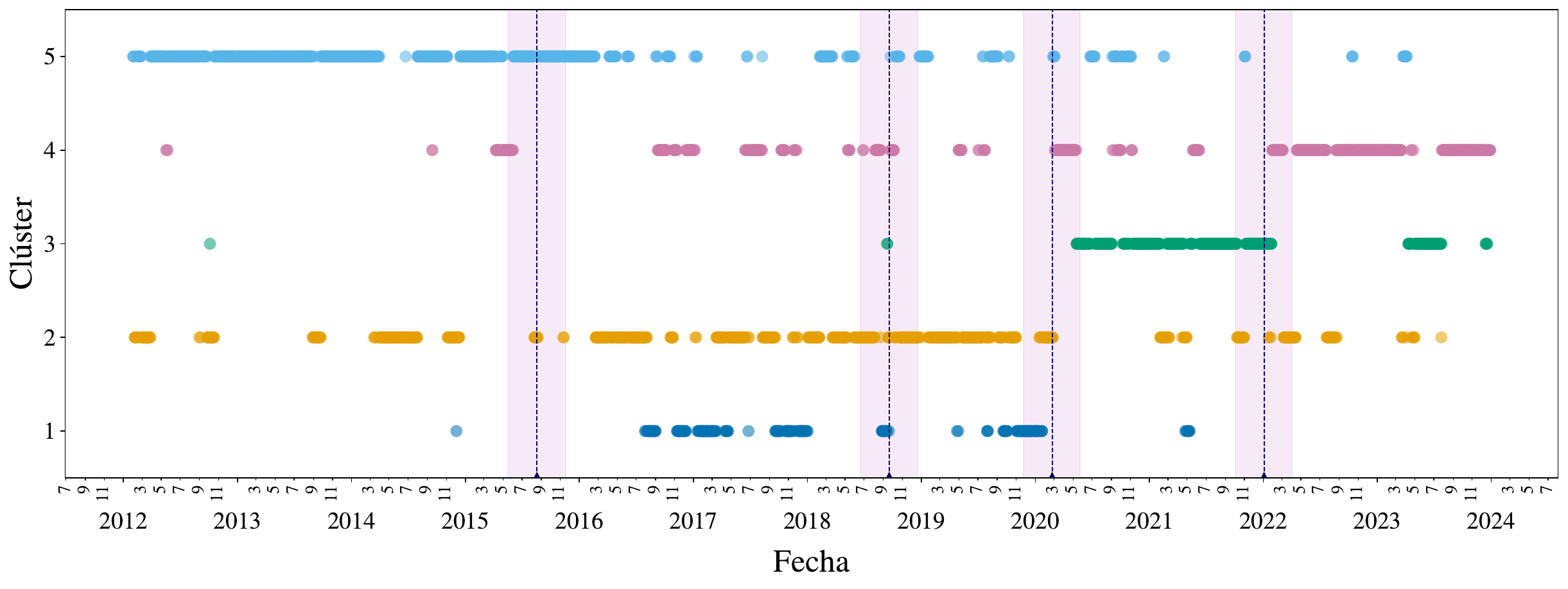}
  \caption*{b) $q=20,\; k=5$}
\end{subfigure}

\vfill

\begin{subfigure}{0.8\linewidth}
  \centering
  \includegraphics[width=\linewidth,height=0.20\textheight,keepaspectratio]{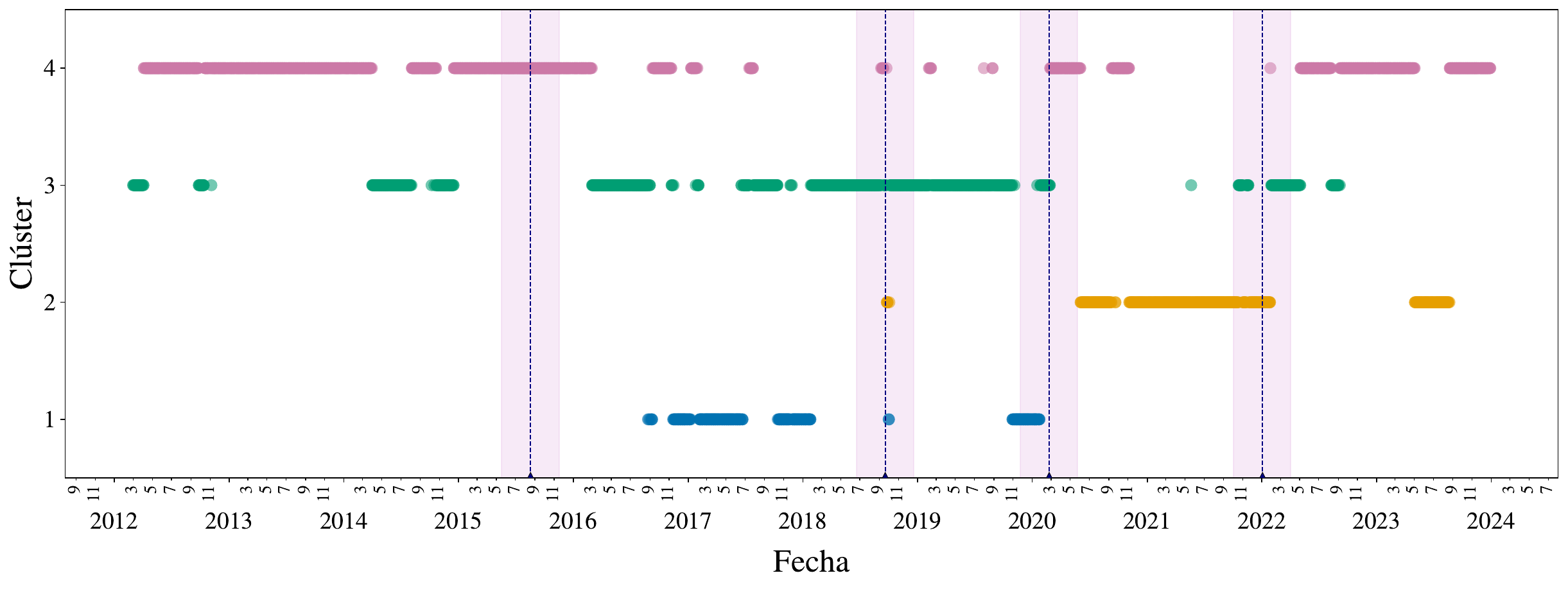}
  \caption*{c) $q=40,\; k=4$}
\end{subfigure}

\vfill

\begin{subfigure}{0.8\linewidth}
  \centering
  \includegraphics[width=\linewidth,height=0.20\textheight,keepaspectratio]{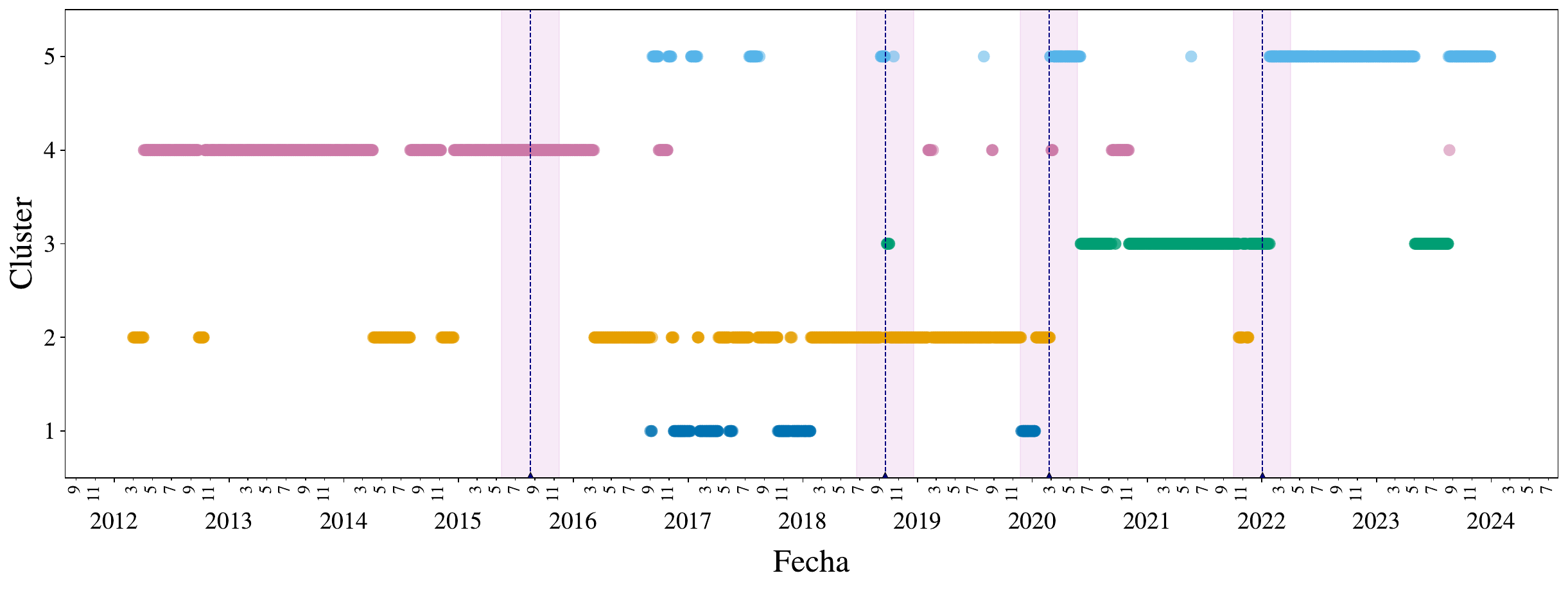}
  \caption*{d) $q=40,\; k=5$}
\end{subfigure}

\vspace*{\fill}

\caption[Temporal evolution of \textit{MS} indicated by $v_i^2$]
{Temporal evolution of the \emph{MS} identified from the entries $[v_i^2]_{i=1}^{N}$ of the eigenvector associated with $\lambda_{N}$. 
For $q=40$ and $k=4$ the episode associated with the \emph{COVID} State is reproduced, in agreement with Fig.~\ref{fig:evolucion_ms_corr}. 
Nevertheless, in this representation the presence of several intervals previously classified within \emph{cluster} 1 decreases.}
\label{fig:evolucionMS_squared_q20q40}
\end{figure}
\clearpage

\subsection{\emph{Clustering} of the Dominant Eigenvectors \texorpdfstring{$[v_i^2]_{i=1}^N$}{[vi²]}}
\index{Clusters}

The \emph{temporal averages} of each squared entry $[v_i^2]_{i=1}^{N}$ of the dominant eigenvector associated with $\lambda_{N}$, computed for $q=40$ and $k=5$, show sectoral patterns that change from one \emph{cluster} to another, as illustrated in Fig.~\ref{fig:vn2-q40-k5}. In particular, for each asset $i$ the quantities $v_i^2$ are averaged over all the dominant eigenvectors belonging to the same \emph{cluster}. It is thereby observed that, in \emph{cluster} 1, the \textbf{participation} is concentrated in a subset of companies (typically dominated by certain sectors), whereas in the higher-index \emph{clusters}, associated with periods of crisis, the participations tend to become more homogeneous across companies.

This is also seen in the values of $\boldsymbol{\mathrm{IPR}}_{N}$ and $\boldsymbol{\mathrm{PR}}_{N}$ reported in Table~\ref{tab:iprpr_q40_k5_vn2}. As can be appreciated, \emph{cluster} 1 concentrates the companies with the highest relative participations; consequently, $\boldsymbol{\mathrm{IPR}}_{N}$ takes its largest value and the effective number of companies $\boldsymbol{\mathrm{PR}}_{N}$ is the smallest, $\sim 278$. In the higher-index \emph{clusters}, $\boldsymbol{\mathrm{IPR}}_{N}$ tends to decrease and $\boldsymbol{\mathrm{PR}}_{N}$ to increase, which indicates that the \textbf{participation} is distributed among a larger set of assets.

\vspace{0.35cm}

\begin{table}[htbp]
\centering

\begin{tabular}{|c|c|c|}
\hline
\rowcolor{teal!80}
\color{white}\textbf{\emph{Cluster}} & \color{white}$\boldsymbol{\mathrm{IPR}_{N}}$ & \color{white}$\boldsymbol{\mathrm{PR}_{N}}$ \\
\hline
1 & $0.00360$ & $278.0681$ \\
\hline
2 & $0.00282$ & $354.8316$ \\
\hline
3 & $0.00291$ & $343.6030$ \\
\hline
4 & $0.00253$ & $395.9047$ \\
\hline
5 & $0.00255$ & $392.8114$ \\
\hline
\end{tabular}

\vspace{0.25cm}

\caption[Values of $\boldsymbol{\mathrm{IPR}_{N}}$ and $\boldsymbol{\mathrm{PR}_{N}}$, $q=40$ and $k=5$]{Values of $\boldsymbol{\mathrm{IPR}_{N}}$ and $\boldsymbol{\mathrm{PR}_{N}}$ for the vectors $[v_i^2]_{i=1}^{N}$, computed by \emph{cluster} for $q=40$ and $k=5$. Together with Fig.~\ref{fig:vn2-q40-k5}, it is observed that \emph{cluster} 1 presents the greatest concentration of \textbf{participation} (highest IPR) and, therefore, the smallest effective number of companies (PR $\sim 278$). In the higher-index \emph{clusters}, $\boldsymbol{\mathrm{IPR}_{N}}$ decreases and $\boldsymbol{\mathrm{PR}_{N}}$ increases, which suggests a \textbf{participation} more distributed among assets.}
\label{tab:iprpr_q40_k5_vn2}
\end{table}

\vspace{0.15cm}

Table~\ref{tab:Top10_q40_k5} shows the 10 leading companies identified for each cluster. \emph{Cluster} 1 is dominated by the \emph{HC} sector. The intersection between the 10 leading companies of each cluster is empty, although the company \emph{WTW} appears consistently in all the \emph{clusters} except the first, with high values. The companies studied are listed in Appendix~\ref{tab:companies-sp500}.

A similar analysis for $q=40$ and $k=4$ is presented in Table~\ref{tab:Top10_q40_k4} and in Appendix~\ref{sec: Clustering de los eigenvectores cuadrados}.

\begin{table}[ht]
\centering

\begin{subtable}{0.48\linewidth}
\centering
\scriptsize
\begin{tabular}{|c|c|c|r|}
\hline
\rowcolor{teal!80}
\color{white}\textbf{Position} & 
\color{white}\textbf{Ticker} & 
\color{white}\textbf{Sector} & 
\color{white}\textbf{Value} \\
\hline
1 & LH   & HC & 0.00795 \\
2 & HCA  & HC & 0.00762 \\
3 & DXCM & HC & 0.00745 \\
4 & A    & HC & 0.00741 \\
5 & BSX  & HC & 0.00737 \\
6 & PFE  & HC & 0.00727 \\
7 & WAT  & HC & 0.00718 \\
8 & HUM  & HC & 0.00710 \\
9 & GILD & HC & 0.00709 \\
10 & BDX & HC & 0.00709 \\
\hline
\end{tabular}
\caption*{a) \emph{cluster} 1}
\end{subtable}\hfill
\begin{subtable}{0.48\linewidth}
\centering
\scriptsize
\begin{tabular}{|c|c|c|r|}
\hline
\rowcolor{teal!80}
\color{white}\textbf{Position} & 
\color{white}\textbf{Ticker} & 
\color{white}\textbf{Sector} & 
\color{white}\textbf{Value} \\
\hline
1 & AMGN & HC & 0.00484 \\
2 & CTRA & EG & 0.00465 \\
3 & WTW  & FN & 0.00464 \\
4 & DVN  & EG & 0.00456 \\
5 & TECH & HC & 0.00454 \\
6 & ACGL & FN & 0.00450 \\
7 & GIS  & CS & 0.00445 \\
8 & PFE  & HC & 0.00445 \\
9 & DLR  & FN & 0.00443 \\
10 & MOH & HC & 0.00443 \\
\hline
\end{tabular}
\caption*{b) \emph{cluster} 2}
\end{subtable}

\vspace{0.6em}

\begin{subtable}{0.48\linewidth}
\centering
\scriptsize
\begin{tabular}{|c|c|c|r|}
\hline
\rowcolor{teal!80}
\color{white}\textbf{Position} & 
\color{white}\textbf{Ticker} & 
\color{white}\textbf{Sector} & 
\color{white}\textbf{Value} \\
\hline
1 & WTW  & FN & 0.00524 \\
2 & MOH  & HC & 0.00507 \\
3 & PFE  & HC & 0.00500 \\
4 & DHR  & HC & 0.00485 \\
5 & V    & FN & 0.00482 \\
6 & INCY & HC & 0.00475 \\
7 & PM   & CS & 0.00475 \\
8 & PEP  & CS & 0.00471 \\
9 & HSIC & HC & 0.00470 \\
10 & CTRA & EG & 0.00463 \\
\hline
\end{tabular}
\caption*{c) \emph{cluster} 3}
\end{subtable}\hfill
\begin{subtable}{0.48\linewidth}
\centering
\scriptsize
\begin{tabular}{|c|c|c|r|}
\hline
\rowcolor{teal!80}
\color{white}\textbf{Position} & 
\color{white}\textbf{Ticker} & 
\color{white}\textbf{Sector} & 
\color{white}\textbf{Value} \\
\hline
1 & WTW  & FN & 0.00433 \\
2 & TECH & HC & 0.00397 \\
3 & ABT  & HC & 0.00395 \\
4 & JNJ  & HC & 0.00392 \\
5 & GILD & HC & 0.00391 \\
6 & AMGN & HC & 0.00391 \\
7 & CTRA & EG & 0.00389 \\
8 & MOH  & HC & 0.00387 \\
9 & CVX  & EG & 0.00385 \\
10 & DVN & EG & 0.00385 \\
\hline
\end{tabular}
\caption*{d) \emph{cluster} 4}
\end{subtable}

\vspace{0.6em}

\begin{subtable}{0.48\linewidth}
\centering
\scriptsize
\begin{tabular}{|c|c|c|r|}
\hline
\rowcolor{teal!80}
\color{white}\textbf{Position} & 
\color{white}\textbf{Ticker} & 
\color{white}\textbf{Sector} & 
\color{white}\textbf{Value} \\
\hline
1 & DLR & FN & 0.00410 \\
2 & RJF & FN & 0.00405 \\
3 & AMGN & HC & 0.00394 \\
4 & WTW  & FN & 0.00388 \\
5 & GIS  & CS & 0.00372 \\
6 & PEP  & CS & 0.00370 \\
7 & DVN  & EG & 0.00367 \\
8 & ABT  & HC & 0.00365 \\
9 & FMC  & MT & 0.00358 \\
10 & GILD & HC & 0.00356 \\
\hline
\end{tabular}
\caption*{e) \emph{cluster} 5}
\end{subtable}

\caption[Leading companies by cluster, $q=40$ and $k=5$]%
{Leading companies indicated by the entries $[v_i^2]_{i=1}^{N}$ corresponding to Fig.~\ref{fig:vn2-q40-k5}. It stands out that, in \emph{cluster} 1, the $10$ leading contributions belong to the \emph{HC} sector. The companies studied in this thesis are found in Appendix~\ref{tab:companies-sp500}}
\label{tab:Top10_q40_k5}
\end{table}
\clearpage
\begin{figure}[p]
  \centering
  \captionsetup[subfigure]{justification=centering,singlelinecheck=false}
  \vspace*{\fill}

  \begin{subfigure}[t]{0.98\textwidth}
    \centering
    \includegraphics[width=\linewidth,height=0.15\textheight,keepaspectratio]{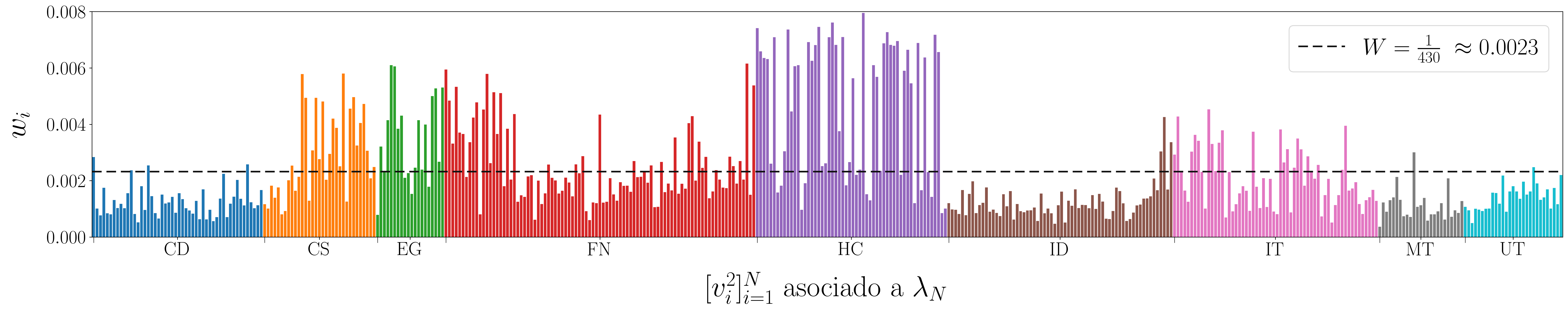}
    \caption{\emph{cluster} 1}
  \end{subfigure}

  \vfill

  \begin{subfigure}[t]{0.98\textwidth}
    \centering
    \includegraphics[width=\linewidth,height=0.15\textheight,keepaspectratio]{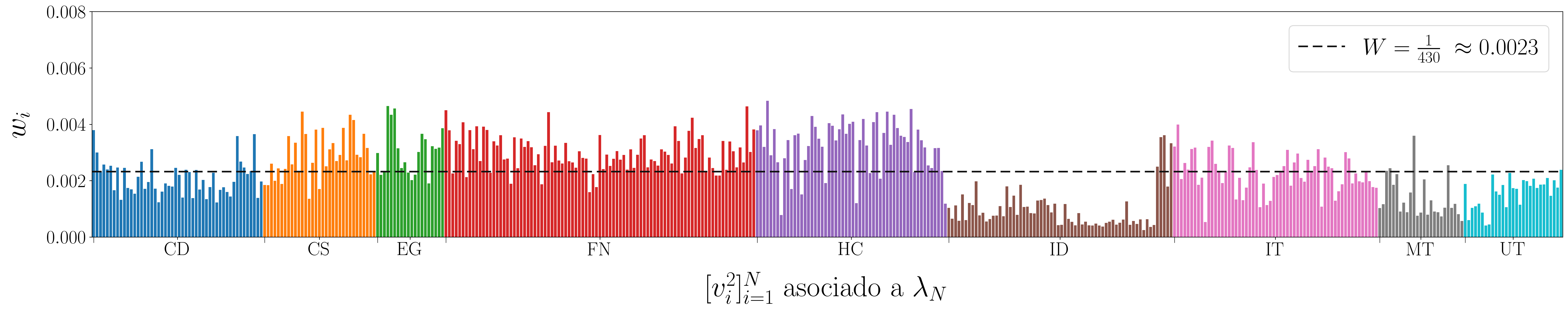}
    \caption{\emph{cluster} 2}
  \end{subfigure}

  \vfill

  \begin{subfigure}[t]{0.98\textwidth}
    \centering
    \includegraphics[width=\linewidth,height=0.15\textheight,keepaspectratio]{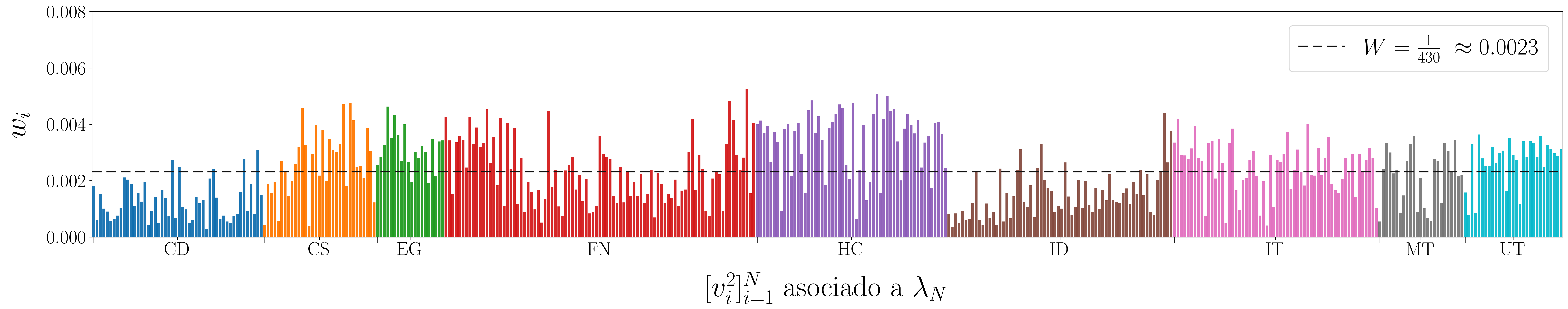}
    \caption{\emph{cluster} 3}
  \end{subfigure}

  \vfill

  \begin{subfigure}[t]{0.98\textwidth}
    \centering
    \includegraphics[width=\linewidth,height=0.15\textheight,keepaspectratio]{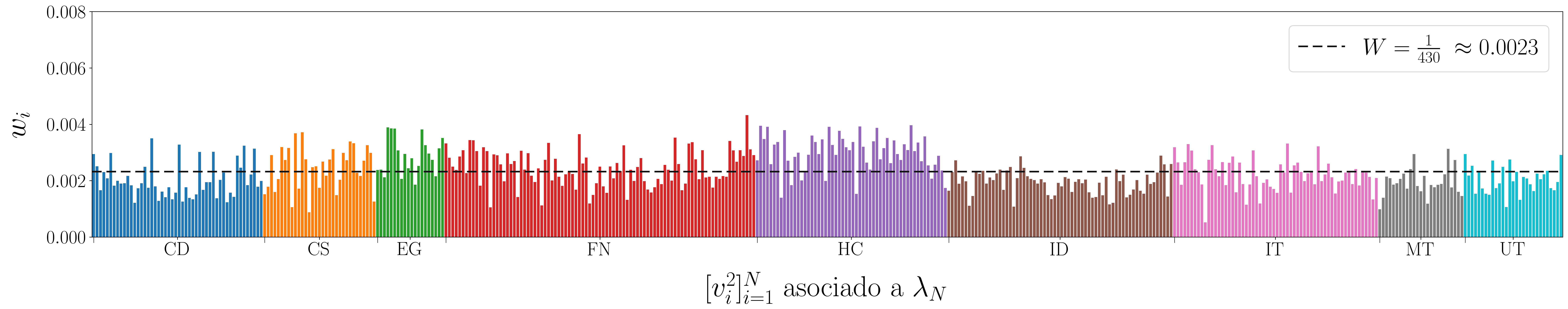}
    \caption{\emph{cluster} 4}
  \end{subfigure}

  \vfill

  \begin{subfigure}[t]{0.98\textwidth}
    \centering
    \includegraphics[width=\linewidth,height=0.15\textheight,keepaspectratio]{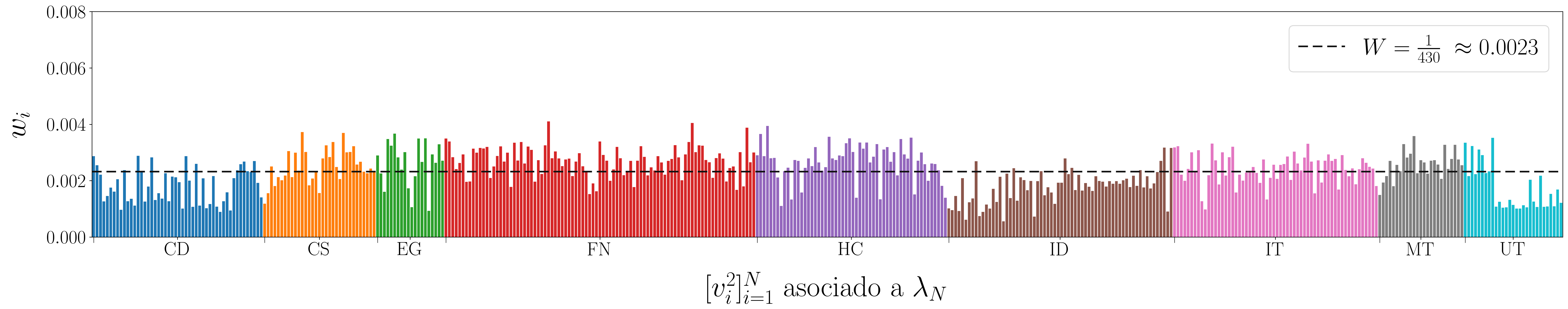}
    \caption{\emph{cluster} 5}
  \end{subfigure}

  \vspace*{\fill}

 \caption[Averages of squared entries by \emph{cluster}, \texorpdfstring{$q=40\;  \text{y}\;k=5$}{q=40 y k=5}]%
{Temporal averages of $w_i\equiv v_i^2$ of the eigenvector associated with $\lambda_N$ for $q=40$ and $k=5$.
In each \emph{cluster}, $w_i$ is computed as the temporal average of component $i$ over the windows assigned to that group, so that each bar represents the mean \textbf{participation weight} of company $i$.
Sectoral patterns are observed, with a predominance of \emph{HC} in \emph{cluster} 1 and a greater presence of \emph{CS}, \emph{EG}, \emph{FN}, and \emph{HC} in \emph{cluster} 2.}
\label{fig:vn2-q40-k5}
\end{figure}
\clearpage

\section{Matrices of Leading Eigenvectors Weighted by Eigenvalues}
\label{section: Matrices tipo Guhr}
In this section, the construction of matrices is employed
\begin{equation}
\mathbf{C}^{l} \;=\; \mathcal{N}\!\left(\sum_{i=N-l+1}^{N} \lambda_i\, \mathbf{v}_i\, \mathbf{v}_i^{\top}\right), 
\quad l \in \{1,2,3\},
\label{eq:matriz-ponderada}
\end{equation}
where $\lambda_i$ and $\mathbf{v}_i$ denote the \textbf{eigenvalues} and \textbf{eigenvectors} of the correlation matrices. 
Each $\mathbf{v}_i$ is a column vector and $\mathcal{N}$ represents the corresponding normalization to obtain matrices with properties analogous to those of correlation matrices, as will be detailed in Eqs. \eqref{eq:matriz-mk} \eqref{eq:matriz-dk} \eqref{eq:normalizacion}. 
This formulation considers only the \emph{dominant spectral components}, while the other contributions are set to zero.

The difference with the proposal of Heckens and Guhr \cite{heckensNewCollectivityMeasures2022} lies in the manner of separating the spectrum. 
They do not apply a direct truncation, but rather distinguish between the contribution of the \emph{SOTM} and the rest, as expressed in Eq.~\eqref{eq:decomp_cov} and Eq.~\eqref{eq:decomp_corr}. These authors analyze both the \emph{SOTM} part and the rest of the spectrum separately in order to isolate the collective measure.

By contrast, the present work employs the truncated construction where only the $l$ largest eigenvalues and their associated eigenvectors are retained, as expressed in Eq.~\eqref{eq:matriz-ponderada}. 

The normalization $\mathcal{N}$ is defined in the following manner: given the matrix
\begin{equation}
\mathbf{M}^{l} \;=\; \sum_{i=N-l+1}^{N} \lambda_i\, \mathbf{v}_i\, \mathbf{v}_i^{\top},
\label{eq:matriz-mk}
\end{equation}
the diagonal matrix $\mathbf{D}^{l}$ is constructed from the diagonal elements
of $\mathbf{M}^{l}$, that is:
\begin{equation}
\mathbf{D}^{l} \;=\; \operatorname{diag}\!\left(M^{l}_{11}, M^{l}_{22}, \dots, M^{l}_{NN}\right).
\label{eq:matriz-dk}
\end{equation}
The normalized matrix is obtained as
\begin{equation}
\mathbf{C}^{l} \;=\; \left(\mathbf{D}^{l}\right)^{-1/2}\, \mathbf{M}^{l}\, \left(\mathbf{D}^{l}\right)^{-1/2},
\label{eq:normalizacion}
\end{equation}
so that the diagonal of $\mathbf{C}^{l}$ is composed only of ones,
that is, $C^{l}_{ii} = 1$ for all $i\in\{1,\dots,N\}$. 
This procedure guarantees that the matrices $\mathbf{C}^{l}$ preserve the 
interpretation of correlation matrices.

\subsection{Matrices \texorpdfstring{$C^1$}{C¹}}
\label{subsection: Matrices C1}
\index{Theorem!sign matrix}
Next, the result corresponding to the matrix is presented  
\begin{equation}
\mathbf{C}^{1} \;=\; \mathcal{N}\!\left(\lambda_{N}\,\mathbf{v}_{N}\mathbf{v}_{N}^{\top}\right),
\end{equation}
\begin{tcolorbox}[title=\strut Remark (Sign matrix $\mathbf{C}^{1}$)]
Let $\mathbf{v}_N^\top=\left(x_1,\dots,x_N\right)$ be the leading eigenvector (associated with $\lambda_N$) with $x_i\neq 0$ for all $i\in\{1,\dots,N\}$.  
Then the matrix
\begin{equation}
\mathbf{C}^{1}=\textbf{s}\,\textbf{s}^{\!\top},\qquad s_i=\operatorname{sgn}\!\bigl(x_i\bigr),
\label{eq:observacion-C1}
\end{equation}
satisfies $\mathbf{C}^1_{ii}=1$ and $\mathbf{C}^1_{ij}\in\{1,-1\}$ for $i\neq j$, with the symmetry $\mathbf{C}_{ij}^{1}=\mathbf{C}_{ji}^1$.
\end{tcolorbox}
\begin{proof}
Consider
\begin{equation}
\mathbf{M}^{1} \;=\; \lambda_N\,\mathbf{v}_N\,\mathbf{v}_N^{\!\top},
\label{eq:M1-def}
\end{equation}
whose entries are
\begin{equation}
M^{1}_{ij} \;=\; \lambda_N\,x_i\,x_j, 
\qquad \forall\, i,j\in\{1,\dots,N\}.
\end{equation}
In particular, on the diagonal
\begin{equation}
M^{1}_{ii} \;=\; \lambda_N\,x_i^{2} \;>\; 0.
\end{equation}
\noindent
Define the diagonal matrix
\begin{equation}
\mathbf{D}^{1} \;=\; \operatorname{diag}\!\bigl(M^{1}_{11},\,M^{1}_{22},\,\dots,\,M^{1}_{NN}\bigr),
\end{equation}
and its inverse root
\begin{equation}
\bigl(\mathbf{D}^{1}\bigr)^{-1/2} \;=\; 
\operatorname{diag}\!\Bigl[\bigl(M^{1}_{11}\bigr)^{-1/2},\,\dots,\,\bigl(M^{1}_{NN}\bigr)^{-1/2}\Bigr]
 \;=\; \operatorname{diag}\!\Bigl(\tfrac{1}{\sqrt{\lambda_N}\,\lvert x_1\rvert},\,\dots,\,\tfrac{1}{\sqrt{\lambda_N}\,\lvert x_N\rvert}\Bigr).
\end{equation}
The normalization is performed as
\begin{equation}
\mathbf{C}^{1} \;=\; \bigl(\mathbf{D}^{1}\bigr)^{-1/2}\,\mathbf{M}^{1}\,\bigl(\mathbf{D}^{1}\bigr)^{-1/2}.
\end{equation}
Then, for the diagonal:
\begin{equation}
C^{1}_{ii}
= \frac{\lambda_N\,x_i^{2}}{\lambda_N\,x_i^{2}}
= 1,
\end{equation}
and for $i\neq j$:
\begin{equation}
C^{1}_{ij}
= \frac{\lambda_N\,x_i\,x_j}{\sqrt{\lambda_N}\,\lvert x_i\rvert \;\sqrt{\lambda_N}\,\lvert x_j\rvert}
= \frac{x_i}{\lvert x_i\rvert}\,\frac{x_j}{\lvert x_j\rvert}
= \operatorname{sgn}\!\bigl(x_i\bigr)\,\operatorname{sgn}\!\bigl(x_j\bigr).
\end{equation}
Therefore, $\mathbf{C}^{1}=\textbf{s}\,\textbf{s}^{\!\top}$ with $s_i=\operatorname{sgn}(x_i)\in\{1,-1\}$. Due to the commutativity of multiplication in $\mathbb{R}$, it is verified that $\mathbf{C}_{ij}^1=\mathbf{C}_{ji}^1$.
\end{proof}
Because the entries of the matrix $\mathbf{C}^{1}$ take only the values $1$ and $-1$, the information it provides proves too limited and does not allow the original properties of the correlation matrices to be adequately reconstructed.

\subsection{Matrices \texorpdfstring{$C^2$}{C²}}
\label{subsection: Matrices C2}
Next, the result corresponding to the matrix is presented  
\begin{equation}
\mathbf{C}^{2} \;=\; \mathcal{N}\left(\lambda_{N}\,\mathbf{v}_{N}\mathbf{v}_{N}^{\top}+\lambda_{N-1}\mathbf{v}_{N-1}\mathbf{v}_{N-1}^\top\right),
\end{equation}
\subsubsection{Temporal Evolution of \emph{MS}}
\index{Market!States}
Upon applying the \textit{k}-Means algorithm\index{K-Means} to the matrices $\mathbf{C}^{2}$, the stability of the \emph{clusters} was evaluated by means of the modal ensemble stability index (IEM). For $q=20$ and $k=4$, a consensus of $IEM=1.0$ was obtained in 96.5\% of the configurations.\index{Stability!IEM} Upon increasing to $k=5$, the results were distributed between $IEM=1.0$ with a frequency of 78.7\% and $IEM=0.9333$ with 10.7\%.  

In the case of $q=40$, an analogous pattern was observed, with differences between $k=4$ and $k=5$. For $k=4$, 92.2\% of the configurations attained $IEM=1.0$. By contrast, for $k=5$, 49.4\% of the configurations corresponded to $IEM=1.0$ and 41.7\% to $IEM=0.8$.  

Taken together, these results show that the stability of the \emph{clusters} obtained from
$\mathbf{C}^2$ depends both on the window length $q$ and on the number of \emph{clusters} $k$ considered,
the optimal values being $q=20$ and $k=4$.

Fig.~\ref{fig:C2_modaIEM_q20q40_k4k5} shows the temporal evolution of the \emph{MS} induced from the matrices $\mathbf{C}^{2}$ with $k=4$ and $k=5$. On both scales, a \textbf{large population} of the higher-index \emph{clusters} is appreciated, which contrasts with what was observed in the plots obtained from the original correlation matrices (Fig.~\ref{fig:evolucion_ms_corr}). For $q=20$, the dynamics reproduce neither the \textit{COVID} State nor the episode of permanence corresponding to the 2017--2018 period in \emph{cluster} 1. By contrast, when considering $q=40$ these episodes begin to emerge; nevertheless, the 2017--2018 interval does not appear isolated as in Fig.~\ref{fig:evolucion_ms_corr}, but rather accompanied by the population of other higher states. This result suggests that the technique based on $\mathbf{C}^{2}$ matrices tends to \textbf{overrepresent} the higher \emph{clusters}.

\subsubsection{Heatmaps}
\index{Heatmaps}
Fig.~\ref{fig:MeanMatrices_C2_q40_k5} shows the heatmaps of average matrices obtained by means of \emph{clustering} with the \textit{k}-Means algorithm applied to the matrices $\mathbf{C}^{2}$ with windows of $q=40$. 
In \emph{cluster} 1, well-delimited regions of null or negative correlation are identified. 
In \emph{cluster} 2, the anticorrelation zones (blue tones) disappear, whereas in \emph{cluster} 3 they manifest themselves again. 
\emph{Cluster} 5 concentrates the highest correlation values, which is reflected in the dark red tones that appear in Fig.~\ref{fig:evolucion_ms_corr}. 

When comparing with the matrices shown in Fig.~\ref{fig:avg_corr_mats}, it is observed that the $\mathbf{C}^{2}$ matrices do not accurately reproduce the characteristic correlations of the correlation matrices. 
The figure corresponding to the case $q=40$, $k=4$ is included in Appendix~\ref{fig:MeanMatrices_C2_q40_k4}.
\subsubsection{Transition Matrices}
\index{Matrix!transition}
Fig.~\ref{fig:MatricesTransC2} presents the \textbf{transition matrices} derived from the $\mathbf{C}^2$ constructions. Most of the occupation is concentrated on the diagonal, in particular in the highest-index \emph{cluster} (4 or 5, depending on the case). It is worth mentioning that the information represented differs from the transition matrices of the correlation matrices, Fig.~\ref{fig:matrices_transicion_corr}.

\begin{figure}[H]
\centering
\begin{minipage}{0.49\linewidth}
  \centering
  \includegraphics[width=\linewidth]{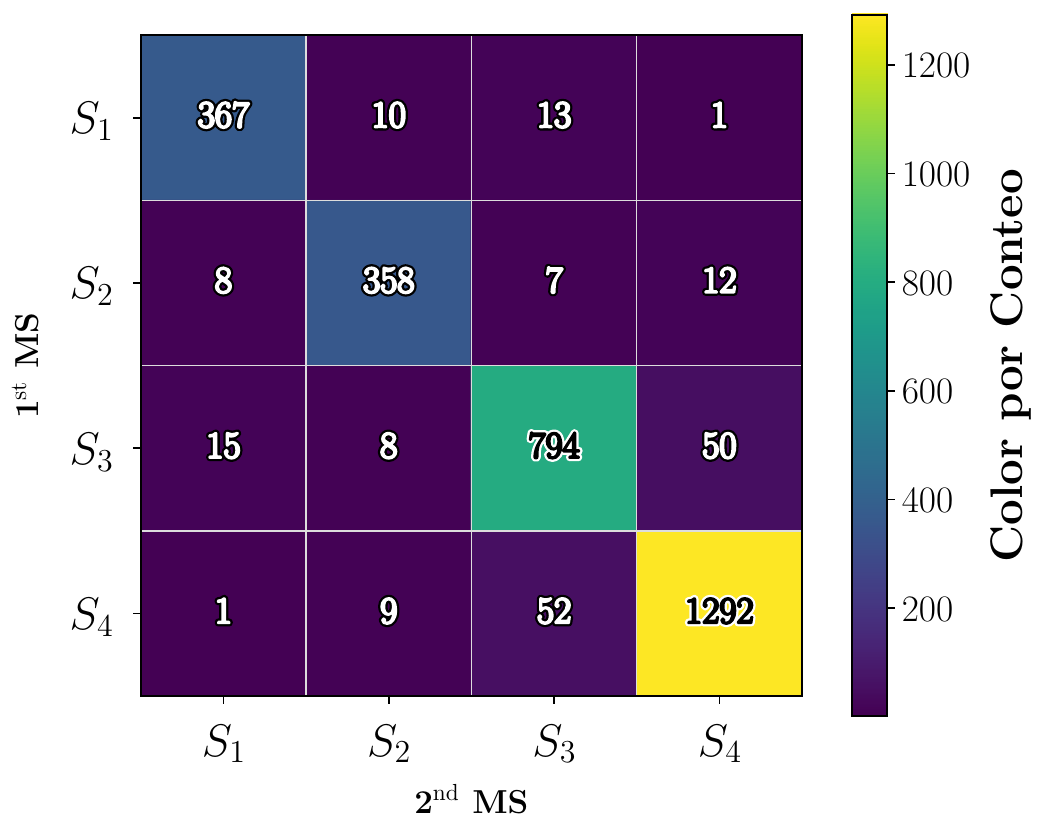}
  {\footnotesize a) $q=20$, $k=4$}
\end{minipage}\hfill
\begin{minipage}{0.49\linewidth}
  \centering
  \includegraphics[width=\linewidth]{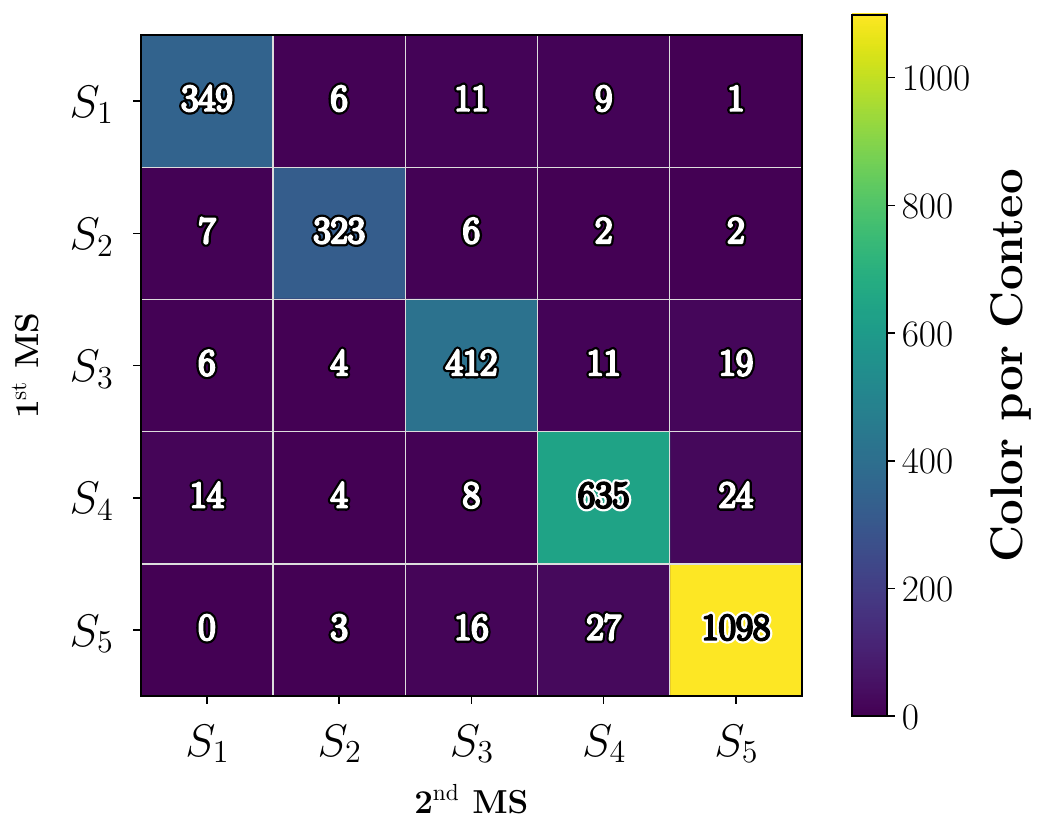}
  {\footnotesize b) $q=20$, $k=5$}
\end{minipage}
\vspace{0.8em}
\begin{minipage}{0.49\linewidth}
  \centering
  \includegraphics[width=\linewidth]{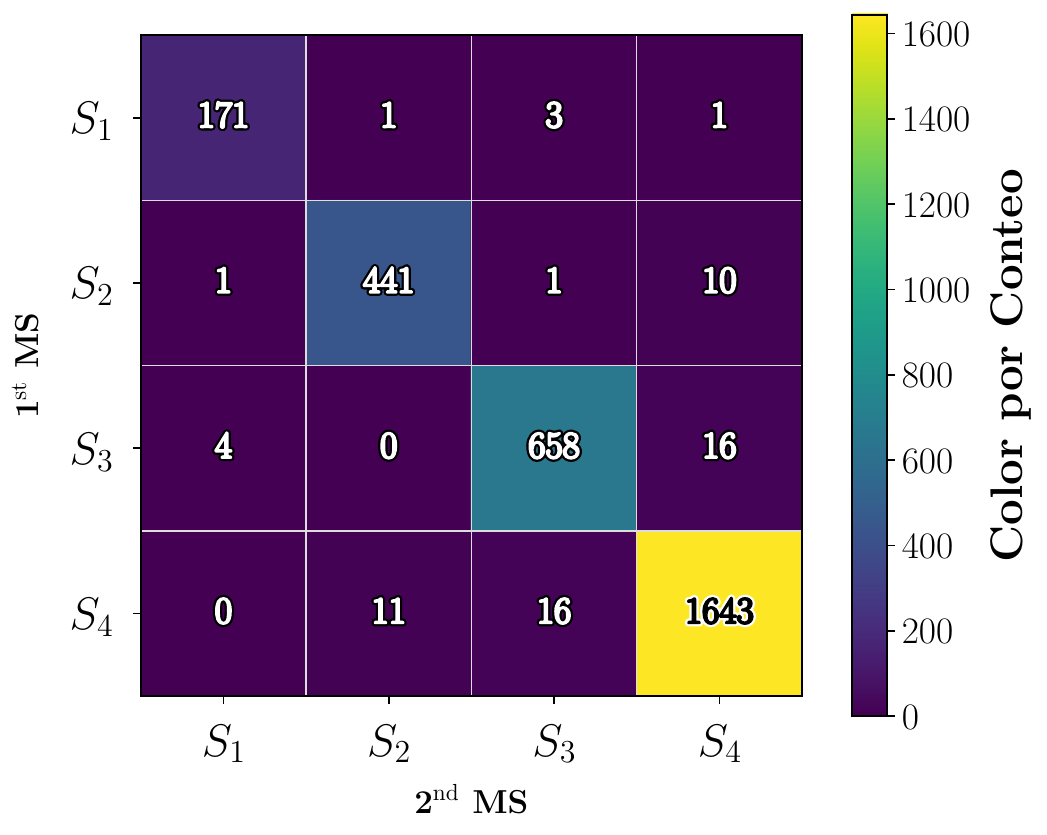}
  {\footnotesize c) $q=40$, $k=4$}
\end{minipage}\hfill
\begin{minipage}{0.49\linewidth}
  \centering
  \includegraphics[width=\linewidth]{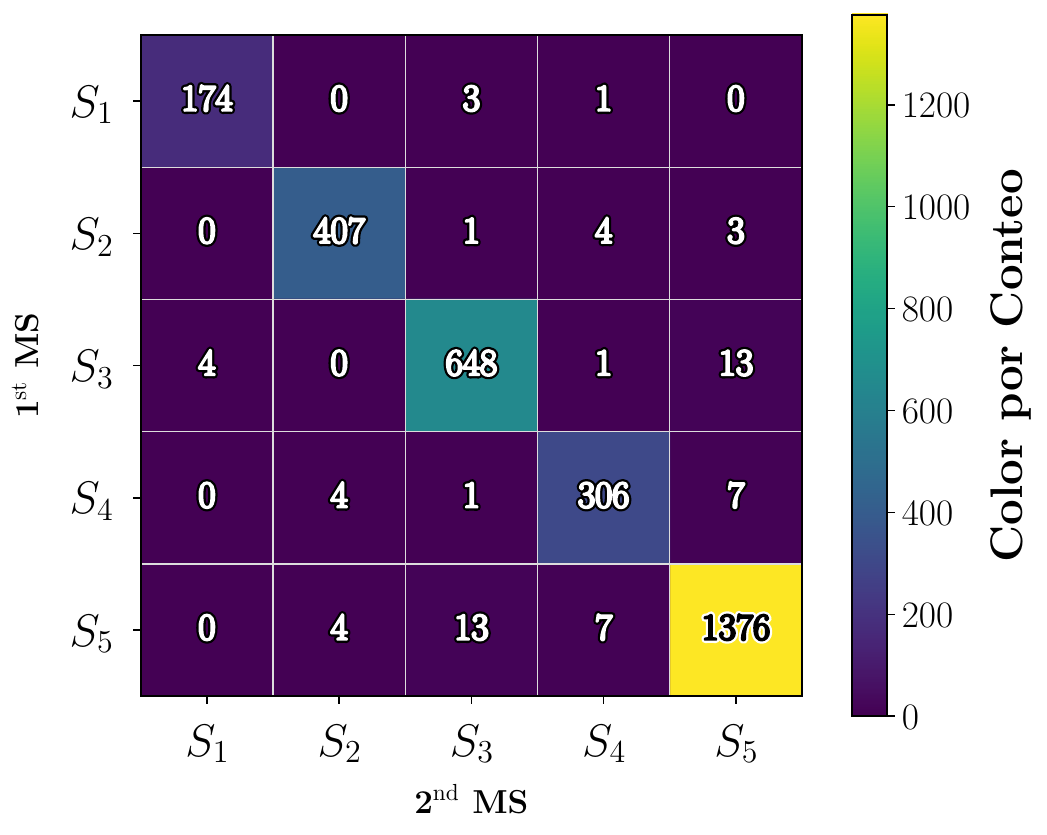}
  {\footnotesize d) $q=40$, $k=5$}
\end{minipage}
\caption[Transition matrices induced by the \texorpdfstring{$C^2$}{C²} matrices]%
{Transition matrices derived from the $\mathbf{C}^{2}$ constructions. 
It is observed that the greatest occupation is concentrated on the diagonal, 
particularly in the highest-index \emph{cluster} (4 or 5, as applicable).}
\label{fig:MatricesTransC2}
\end{figure}

\begin{figure}[p]
\centering
\vspace*{\fill}

\begin{subfigure}{0.88\linewidth}
  \centering
  \includegraphics[width=\linewidth,height=0.20\textheight,keepaspectratio]{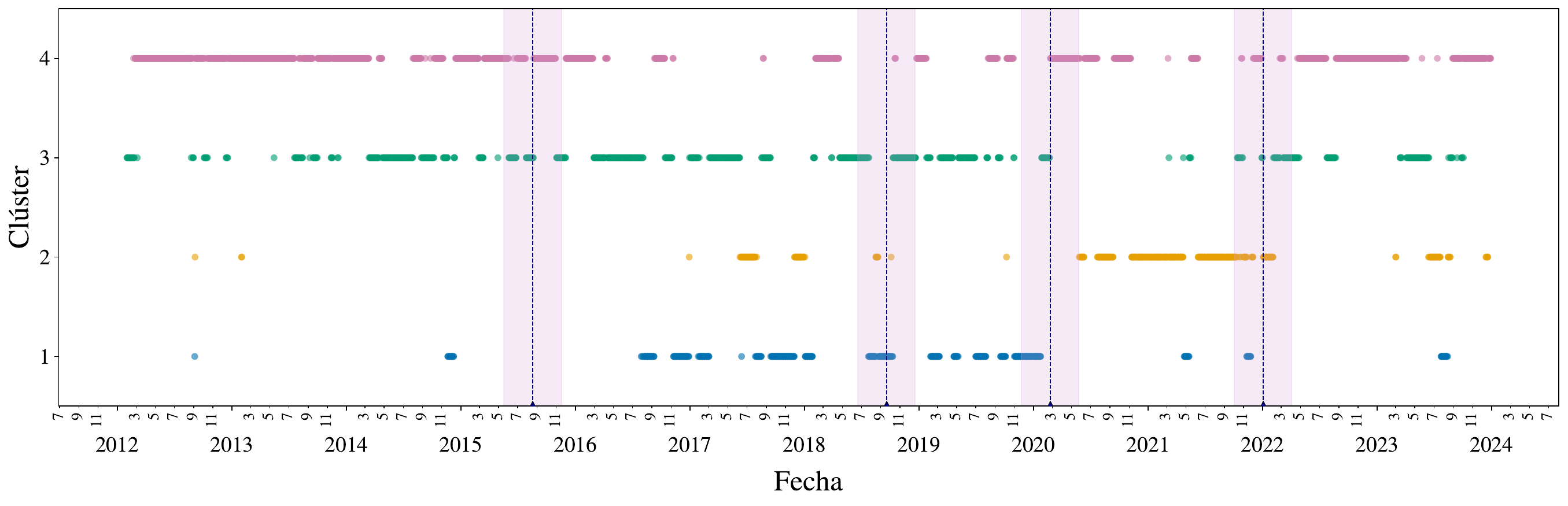}
  \caption*{\footnotesize a) $q=20$, $k=4$}
\end{subfigure}

\vfill

\begin{subfigure}{0.88\linewidth}
  \centering
  \includegraphics[width=\linewidth,height=0.20\textheight,keepaspectratio]{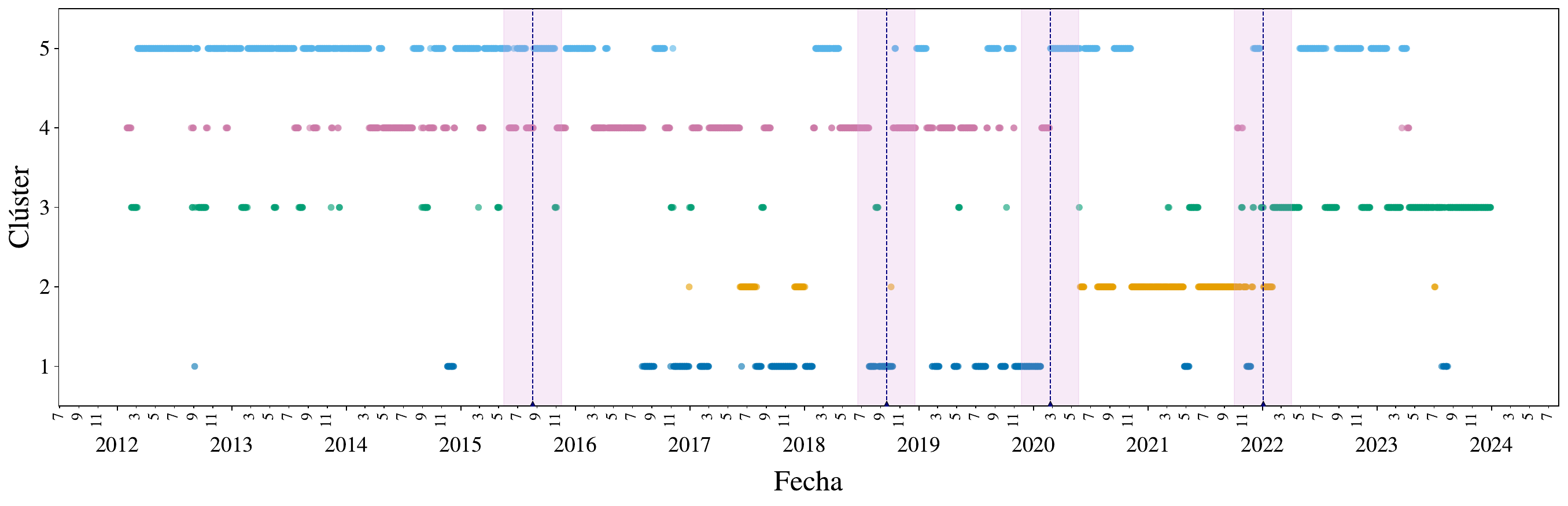}
  \caption*{\footnotesize b) $q=20$, $k=5$}
\end{subfigure}

\vfill

\begin{subfigure}{0.88\linewidth}
  \centering
  \includegraphics[width=\linewidth,height=0.20\textheight,keepaspectratio]{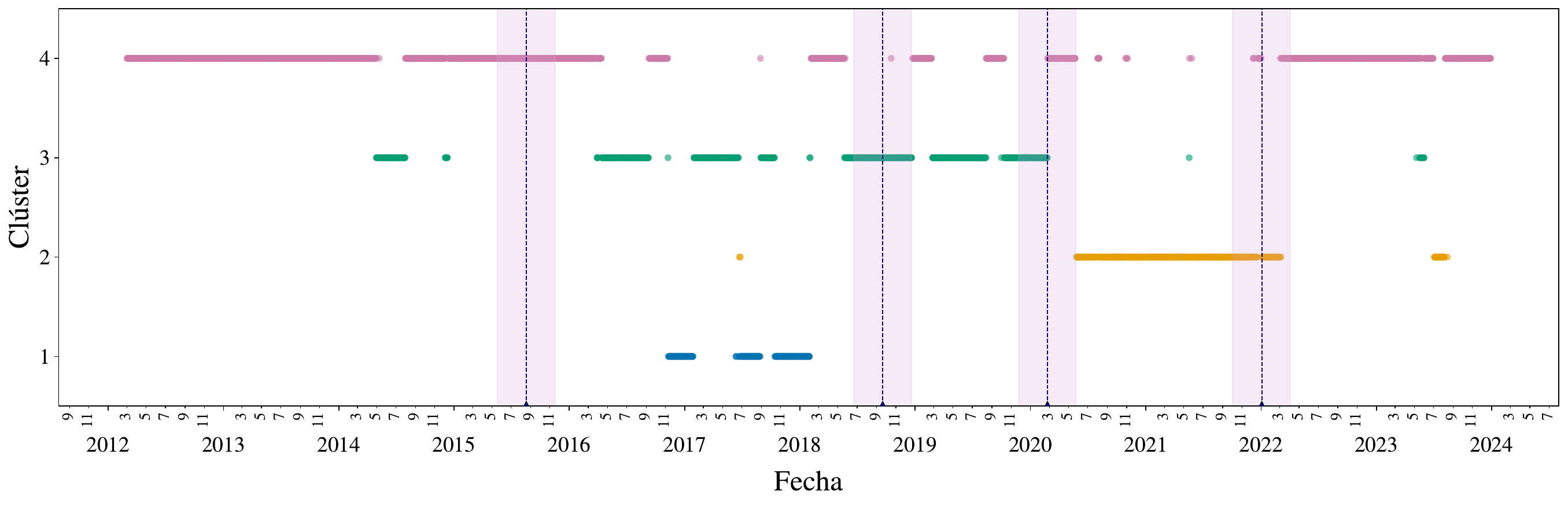}
  \caption*{\footnotesize c) $q=40$, $k=4$}
\end{subfigure}

\vfill

\begin{subfigure}{0.88\linewidth}
  \centering
  \includegraphics[width=\linewidth,height=0.20\textheight,keepaspectratio]{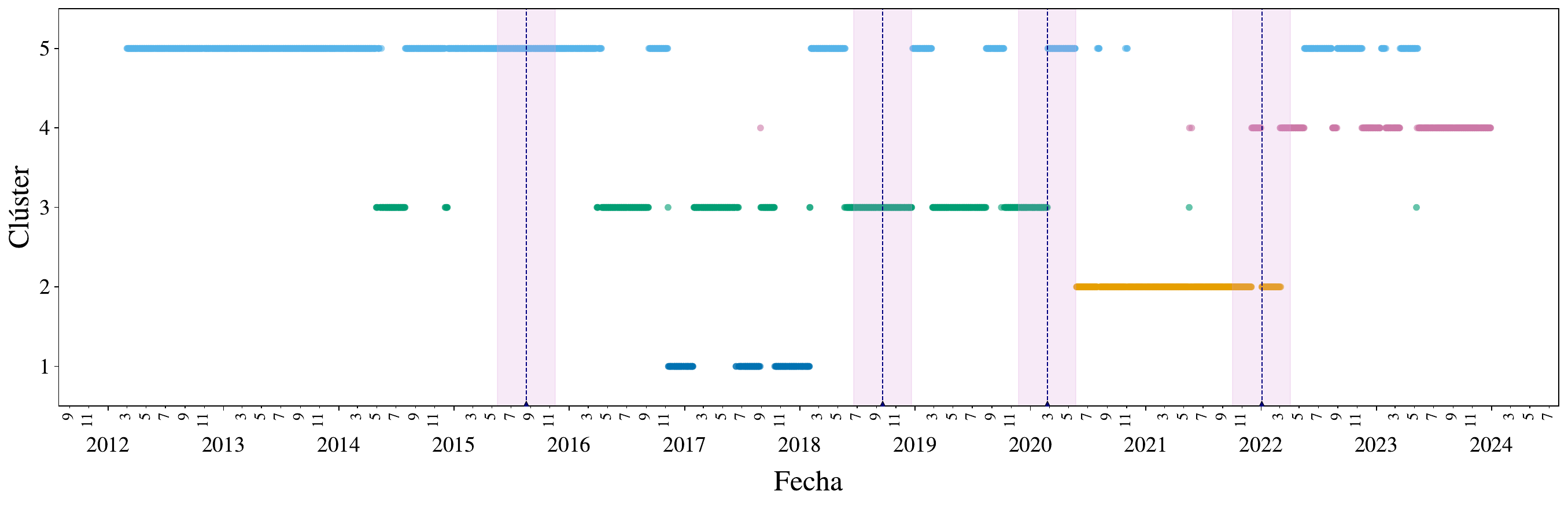}
  \caption*{\footnotesize d) $q=40$, $k=5$}
\end{subfigure}

\vspace*{\fill}

\caption[Temporal evolution of \emph{MS} from the \texorpdfstring{$C^2$}{C²} matrices]%
{Temporal evolution of \emph{MS} from the $\mathbf{C}^{2}$ matrices. On both time scales a \textbf{large population} is observed in the \emph{clusters} with higher $\langle C_{ij}\rangle$. For $q=20$, neither the \textit{COVID} State nor the 2017-2018 episode in \emph{cluster} 1 is replicated. By contrast, with $q=40$ these states begin to emerge; however, the 2017-2018 period does not appear isolated as in Fig.~\ref{fig:evolucion_ms_corr} for $q=40$.}
\label{fig:C2_modaIEM_q20q40_k4k5}
\end{figure}


\begin{figure}[t]
  \centering
  \includegraphics[width=\linewidth]{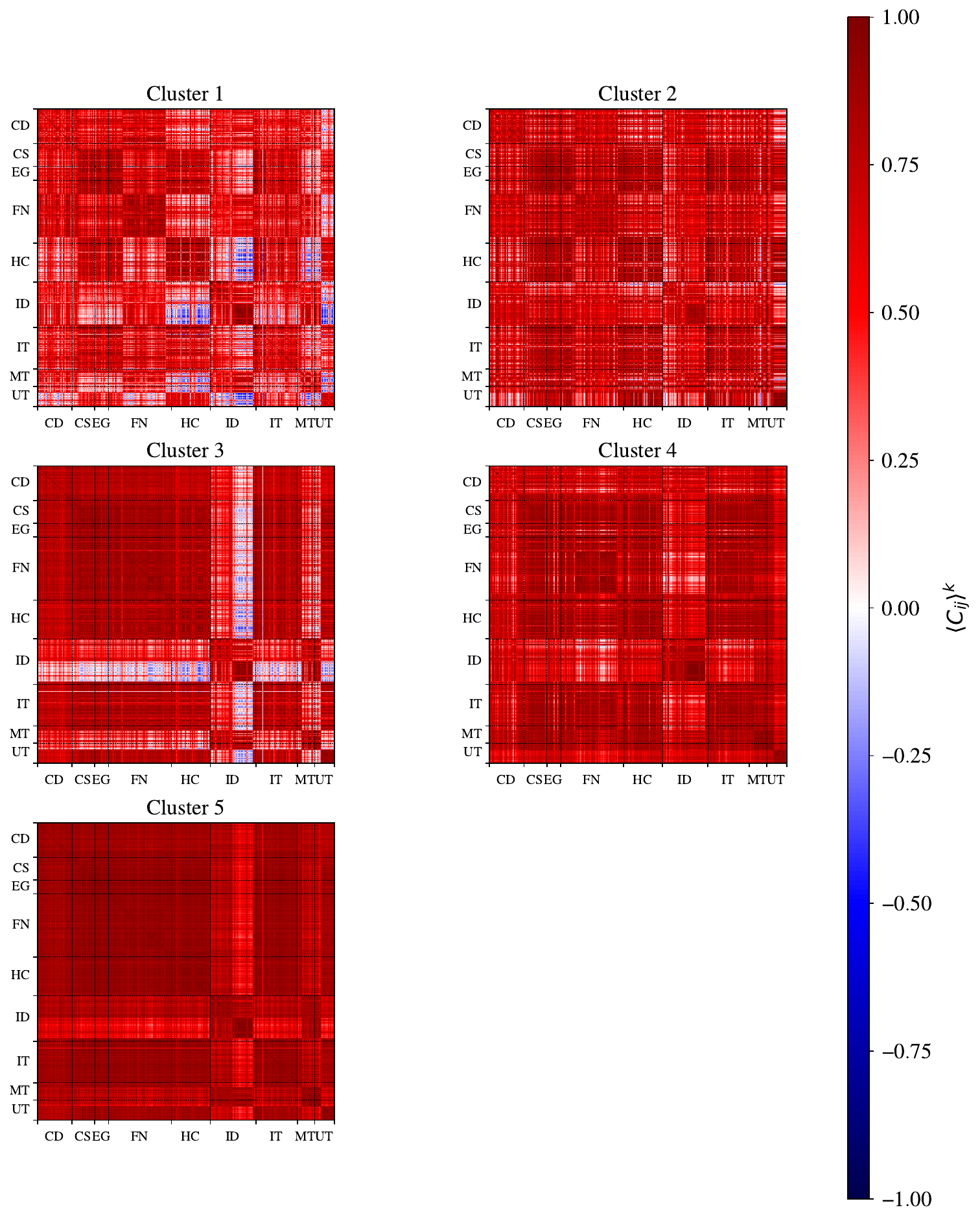}
  \caption[Heatmaps of the average \texorpdfstring{$C^2$}{C2} matrices by \emph{cluster}, \texorpdfstring{$q=40\;\text{y}\; k=5$}{q=40, k=5}]%
  {Heatmaps of the average $\mathbf{C}^{2}$ matrices obtained by \textit{k}-Means \emph{clustering} with $q=40$ and $k=5$. 
  In \emph{cluster} 1, delimited blocks with correlations close to $0$ or negative are observed. In \emph{cluster} 2 the blue anticorrelation bands disappear; 
  these reappear in \emph{cluster} 3. 
  \emph{Cluster} 5 concentrates the highest correlation values.}
  \label{fig:MeanMatrices_C2_q40_k5}
\end{figure}

\clearpage

\subsection{Matrices \texorpdfstring{$C^3$}{C³}}
\label{subsection: Matrices C3}
Next, the result corresponding to the matrix is presented  
\begin{equation}
\mathbf{C}^{3} \;=\; \mathcal{N}\!\left(\lambda_{N}\,\mathbf{v}_{N}\mathbf{v}_{N}^{\top} \;+\; \lambda_{N-1}\,\mathbf{v}_{N-1}\mathbf{v}_{N-1}^\top \;+\; \lambda_{N-2}\,\mathbf{v}_{N-2}\mathbf{v}_{N-2}^\top\right),
\end{equation}
\subsubsection{Temporal Evolution of \emph{MS}}
\index{Market!States}
Fig.~\ref{fig:C3_modaIEM_q20q40_k4k5} presents the temporal evolution of the
\emph{MS} induced by the matrices $\mathbf{C}^{3}$ in the windows $q=20$ and $q=40$. 
For $q=20$, the \textit{COVID} State appears in \emph{cluster} 2; nevertheless, the 2017-2018 period is not distinguished as an independent state, since the higher \emph{clusters} also concentrate a relevant population of observations. 
By contrast, for $q=40$ the \textit{COVID} State manifests itself intermittently at $k=4$ and separates more clearly when $k=5$ is considered. In none of the cases is \emph{cluster} 1, corresponding to the period between 2017 and 2018, isolated as a differentiated state, as it appears in Fig. \ref{fig:evolucion_ms_corr}.
In terms of stability, for $q=20$ and $k=4$, 98.7\% of the configurations attained $IEM=1.0$, whereas for $k=5$ the distribution was split between 56.5\% with $IEM=1.0$ and 39\% with $IEM=0.9333$.

When employing $q=40$, the behavior was analogous: with $k=4$, 96.5\% of the configurations recorded $IEM=1.0$, and with $k=5$ a distribution of 61.3\% at $IEM=1.0$ and 26.4\% at $IEM=0.7333$ was observed.
\subsubsection{Heatmaps}
\index{Heatmaps}
Fig.~\ref{fig:MeanMatrices_C3_q40_k5} presents the heatmaps of average matrices derived from the $\mathbf{C}^{3}$ construction with $q=40$ and $k=5$. 
Although a pattern similar to that of $\mathbf{C}^{2}$ is observed (Fig.~\ref{fig:MeanMatrices_C2_q40_k5}), particularly in the presence of white and blue bands within \emph{cluster} 1 that reappear in \emph{cluster} 3, there are notable differences. 
In $\mathbf{C}^{3}$, a greater grouping of null or near-zero correlations is evidenced, which suggests that the inclusion of an additional spectral component produces a more dispersed representation of the weak correlation relationships.
The heatmap for $q=40$ and $k=4$ is found in Appendix~\ref{fig:MeanMatrices_C3_q40_k4}.
\subsubsection{Transition Matrices}
 \index{Matrix!transition}
As observed in Fig.~\ref{fig: Matrices trans C3}, the matrices obtained with $q=20$ present greater variability in their configurations for $k=4$ and $k=5$, showing more noticeable changes between \emph{clusters}. By contrast, for $q=40$ the matrices prove more stable and consistent.

\begin{figure}[p]
\centering
\vspace*{\fill}

\begin{subfigure}{0.8\linewidth}
  \centering
  \includegraphics[width=\linewidth,height=0.20\textheight,keepaspectratio]{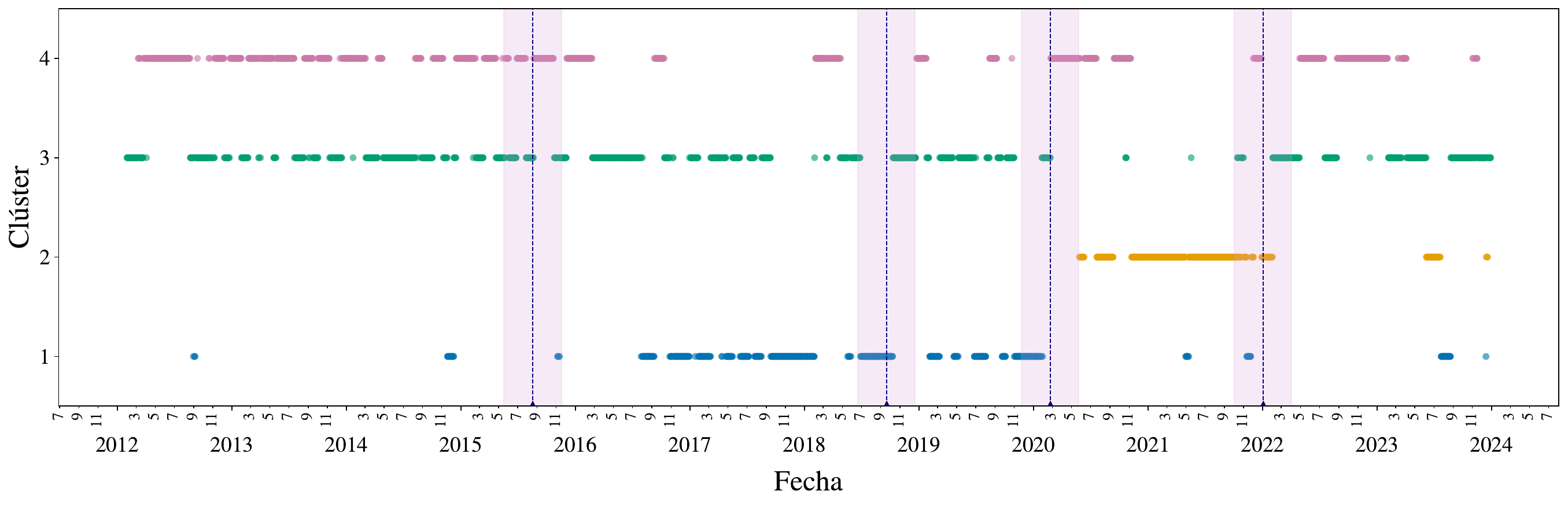}
  \caption*{\footnotesize a) $q=20$, $k=4$}
\end{subfigure}

\vfill

\begin{subfigure}{0.8\linewidth}
  \centering
  \includegraphics[width=\linewidth,height=0.20\textheight,keepaspectratio]{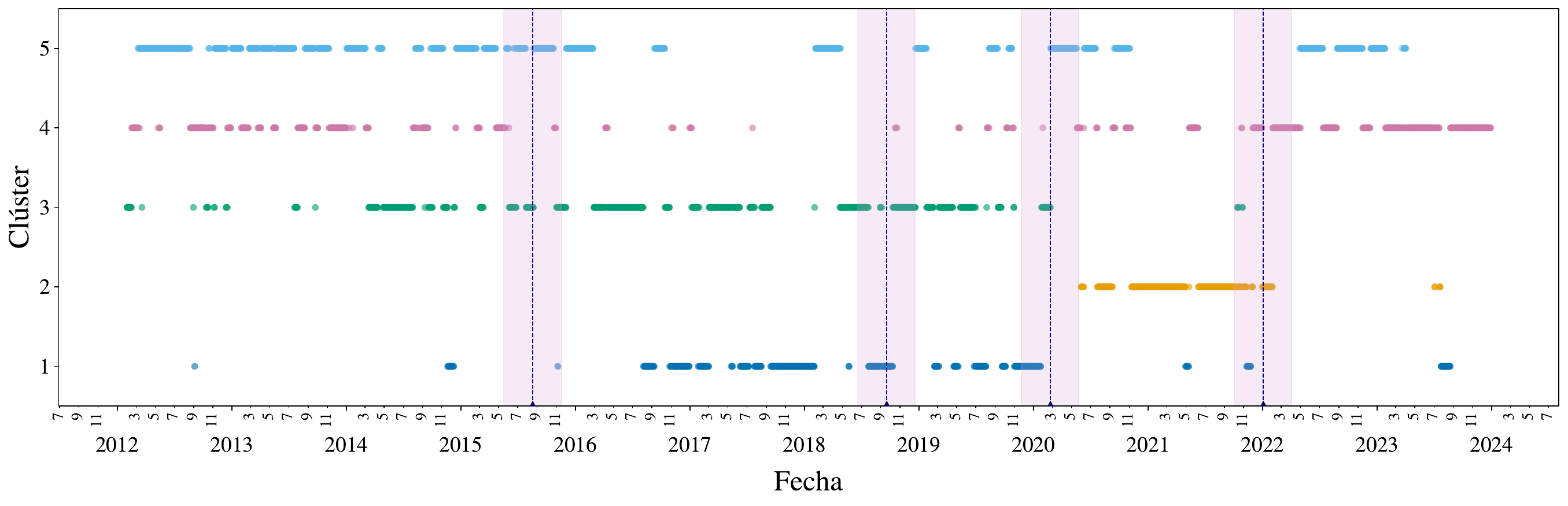}
  \caption*{\footnotesize b) $q=20$, $k=5$}
\end{subfigure}

\vfill

\begin{subfigure}{0.8\linewidth}
  \centering
  \includegraphics[width=\linewidth,height=0.20\textheight,keepaspectratio]{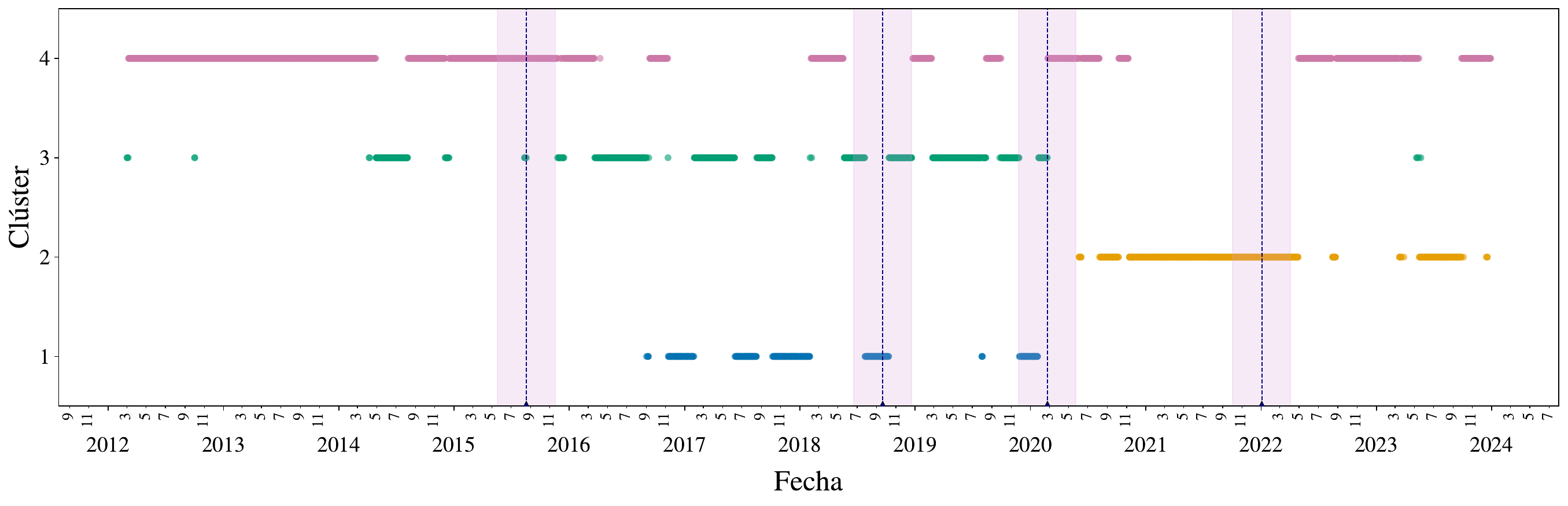}
  \caption*{\footnotesize c) $q=40$, $k=4$}
\end{subfigure}

\vfill

\begin{subfigure}{0.8\linewidth}
  \centering
  \includegraphics[width=\linewidth,height=0.20\textheight,keepaspectratio]{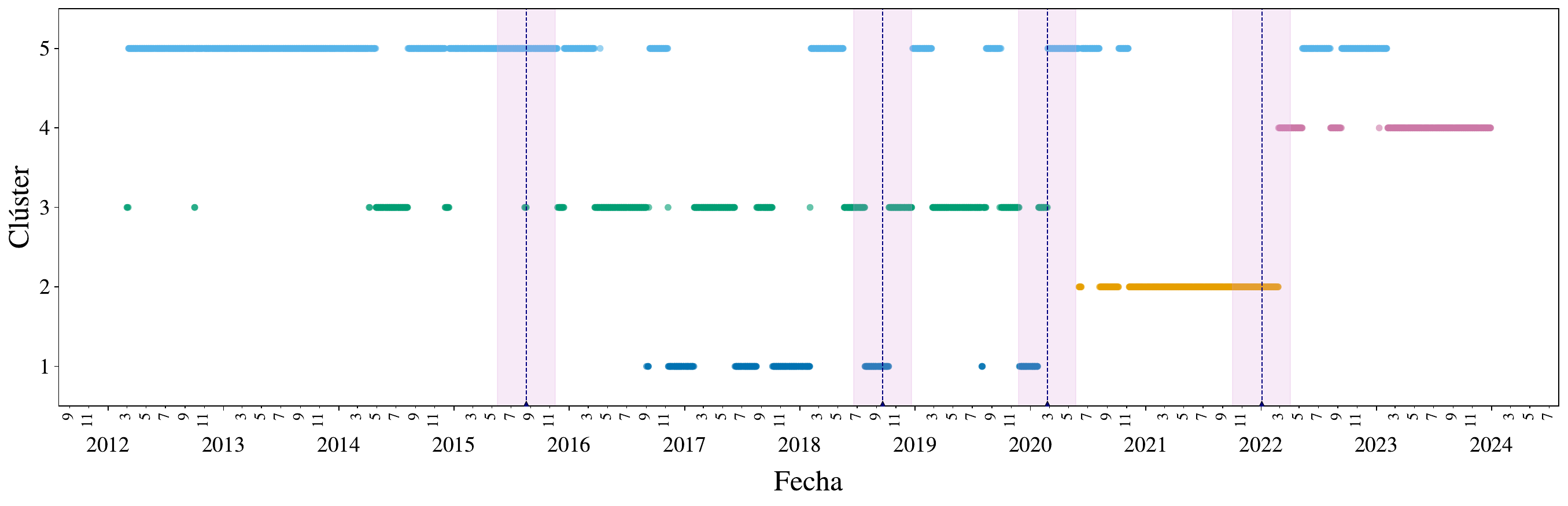}
  \caption*{\footnotesize d) $q=40$, $k=5$}
\end{subfigure}

\vspace*{\fill}

\caption[Temporal evolution of \emph{MS} from the \texorpdfstring{$C^3$}{C³} matrices]%
{Temporal evolution of the \emph{MS} obtained from the $\mathbf{C}^{3}$ matrices with $q=20$ and $q=40$. 
For $q=20$ the \textit{COVID} State begins to appear in \emph{cluster} 2; however, the 2017-2018 period is not isolated, since the higher \emph{clusters} also present a population. 
By contrast, with $q=40$ the \textit{COVID} State appears intermittently at $k=4$ and is isolated more clearly at $k=5$.}
\label{fig:C3_modaIEM_q20q40_k4k5}
\end{figure}
\begin{figure}[t]
  \centering
  
\includegraphics[width=\linewidth]{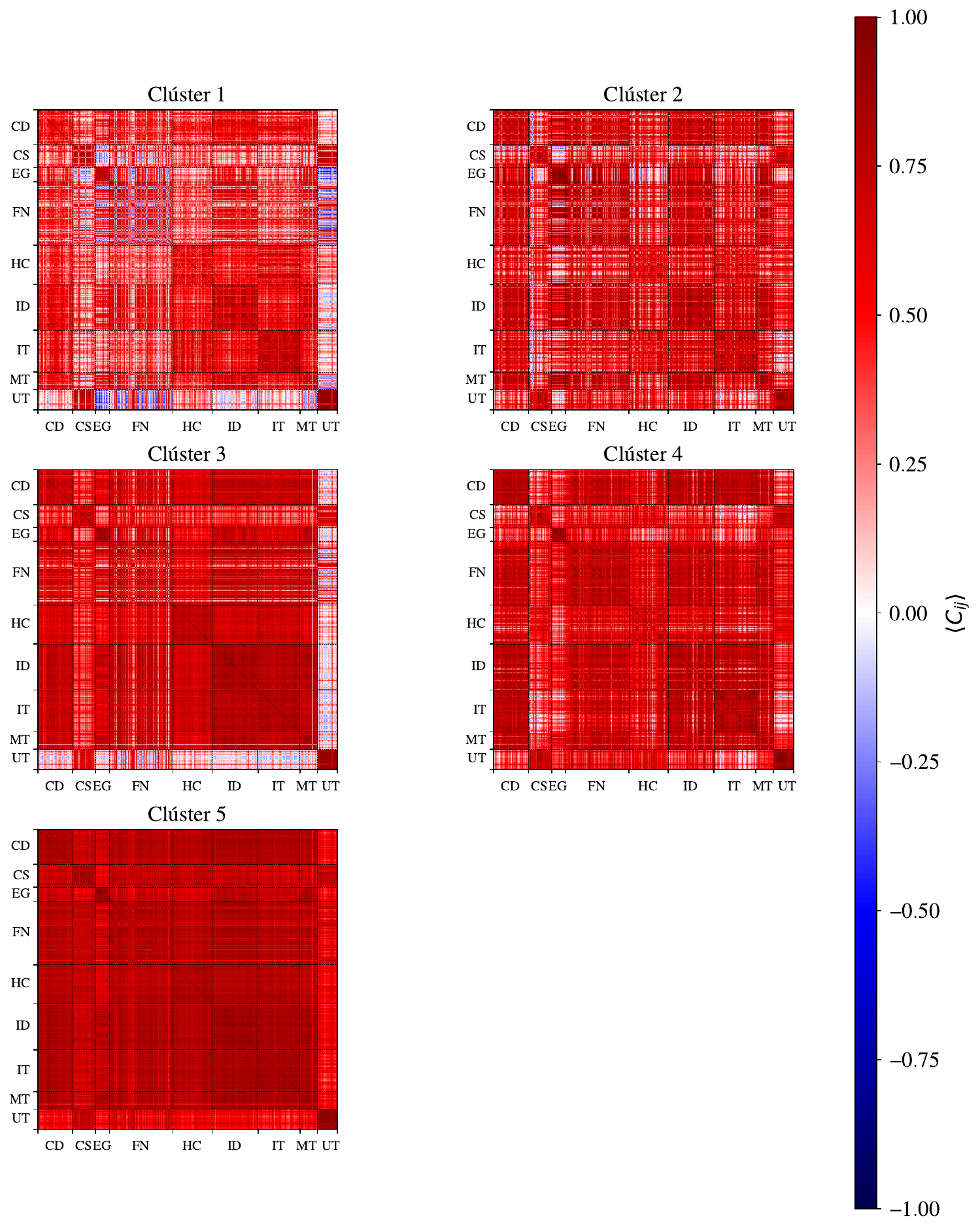}
\caption[Heatmaps of the average \texorpdfstring{$C^3$}{C³} matrices by \emph{cluster}, \texorpdfstring{$q=40\;\text{y}\; k=5$}{q=40, k=5}]%
{Heatmaps of the average $\mathbf{C}^{3}$ matrices obtained by means of \textit{k}-Means \emph{clustering} for $q=40$ and $k=5$. 
These plots exhibit a pattern similar to that of the $\mathbf{C}^{2}$ matrices (Fig.~\ref{fig:MeanMatrices_C2_q40_k5}). In \emph{cluster} 1, blue and white bands are appreciated, which reappear in \emph{cluster} 3.}
\label{fig:MeanMatrices_C3_q40_k5}
\end{figure}
\begin{figure}[ht]
\centering
\begin{minipage}{0.49\linewidth}
  \centering
  \includegraphics[width=\linewidth]{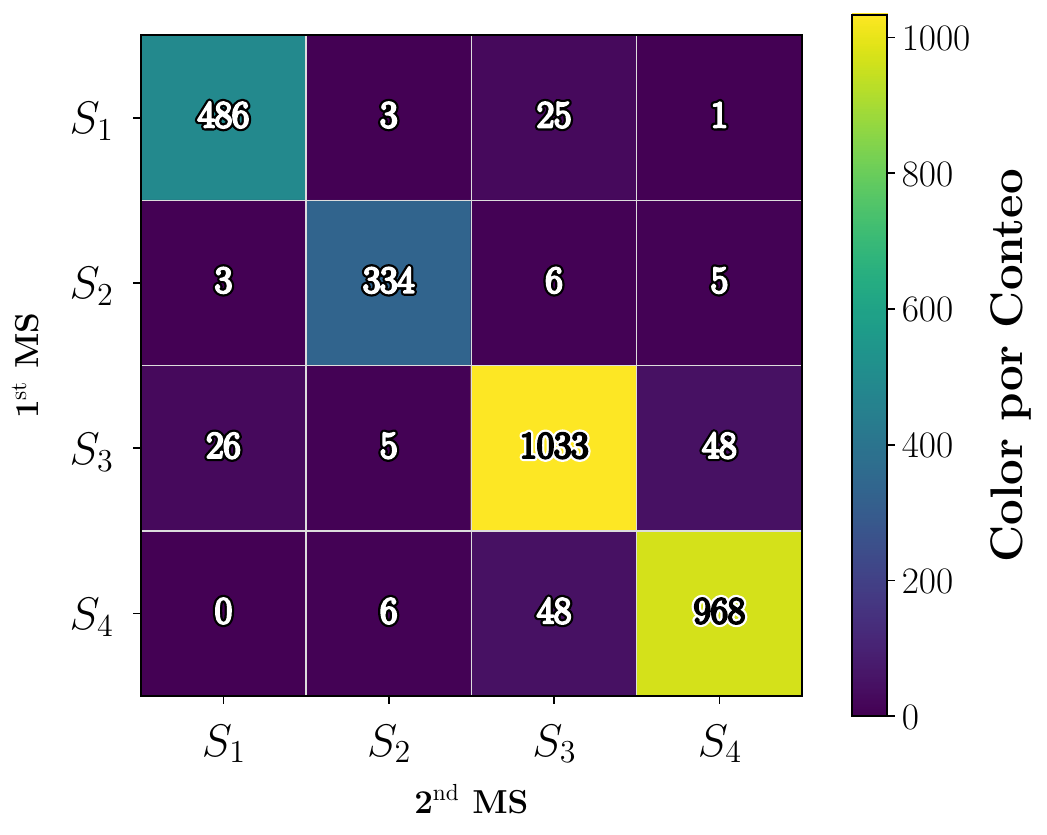}
  {\footnotesize a) $q=20$, $k=4$}
\end{minipage}\hfill
\begin{minipage}{0.49\linewidth}
  \centering
  \includegraphics[width=\linewidth]{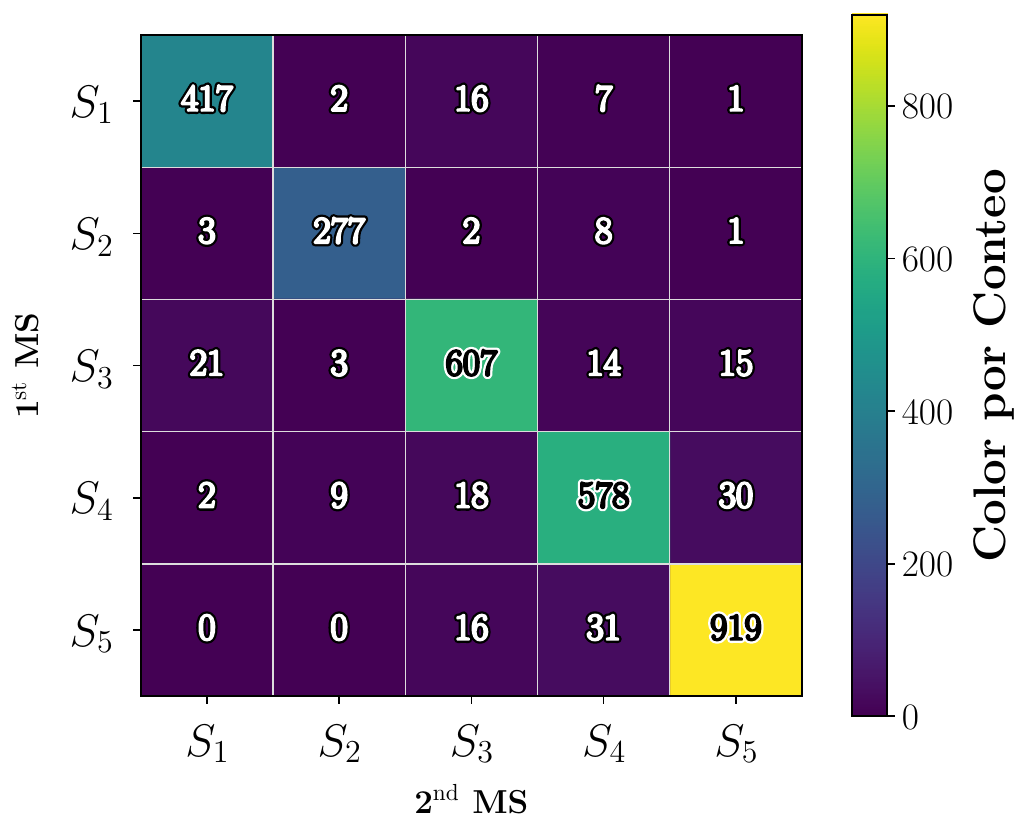}
  {\footnotesize b) $q=20$, $k=5$}
\end{minipage}
\vspace{0.8em}
\begin{minipage}{0.49\linewidth}
  \centering
  \includegraphics[width=\linewidth]{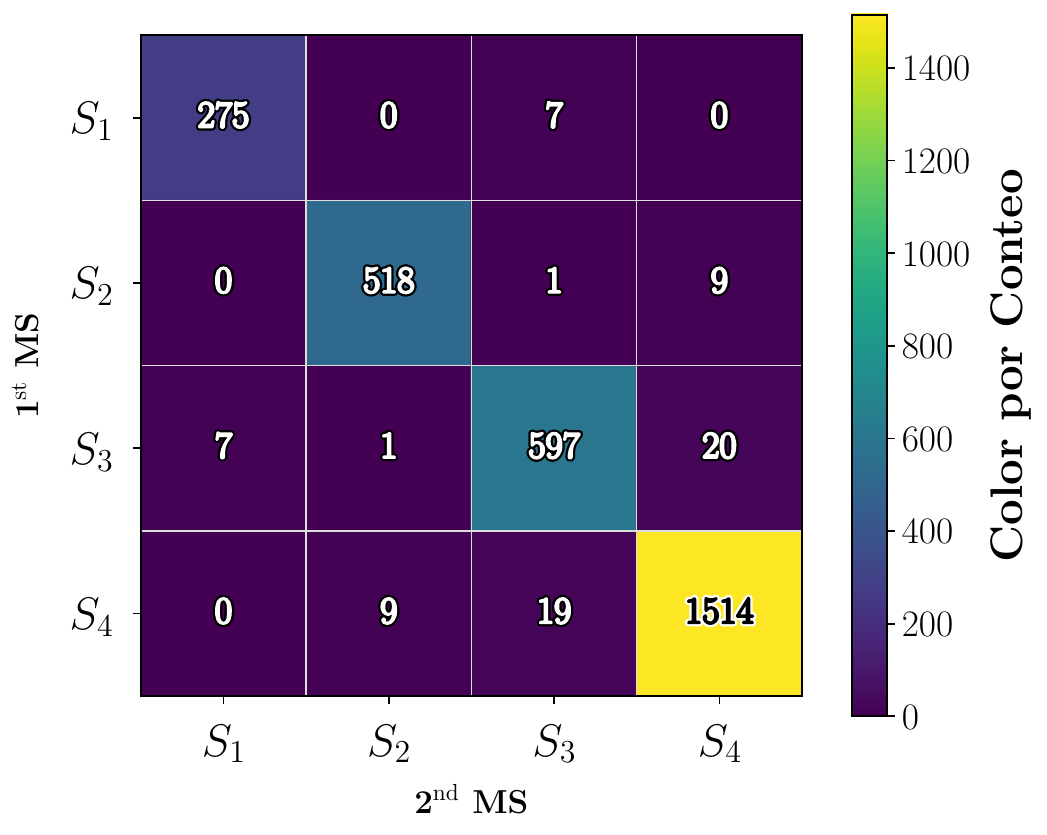}
  {\footnotesize c) $q=40$, $k=4$}
\end{minipage}\hfill
\begin{minipage}{0.49\linewidth}
  \centering
  \includegraphics[width=\linewidth]{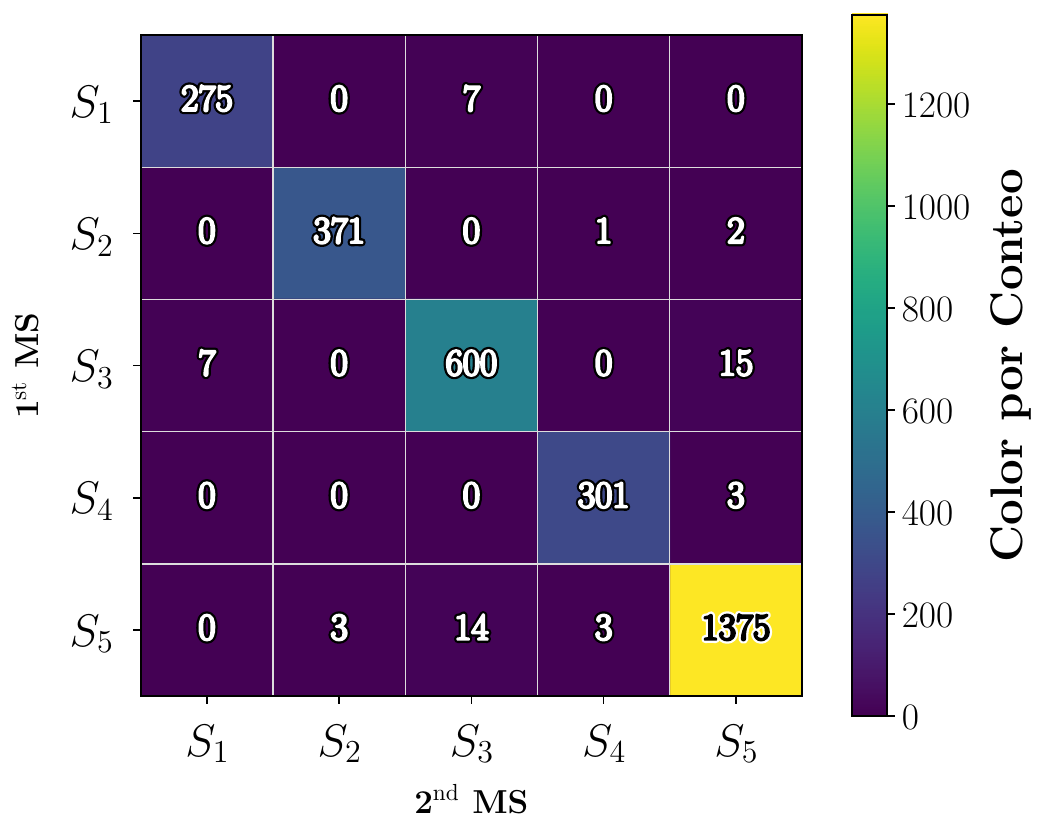}
  {\footnotesize d) $q=40$, $k=5$}
\end{minipage}
\caption[Transition matrices indicated by the \texorpdfstring{$C^3$}{C³} matrices]%
{Transition matrices derived from the $\mathbf{C}^{3}$ constructions. 
Consistent with Fig.~\ref{fig:C3_modaIEM_q20q40_k4k5}, the \emph{clusters} with the greatest occupation on the diagonal are \emph{cluster} 4 and \emph{cluster} 5. The configuration of the matrices varies considerably between $k=4$ and $k=5$ for $q=20$.} 
\label{fig: Matrices trans C3}
\end{figure}
\clearpage

\subsection{Stationary Vectors of \emph{MS}}
\label{subsection:vectoresEstacionariosMS}
\index{Vector!stationary}

To conclude this results section, a visual representation is presented of the stationary vectors
$\boldsymbol{\pi}$ associated with the correlation matrices $\mathbf{C}[\mathcal{E}(s,q)]$, defined in expression~\ref{eq:vector-estacionario}
and satisfying the asymptotic condition of~\ref{eq:lim-estacionario}. Fig.~\ref{fig:vectoresEstacionariosMS_q40_k5} shows the stationary
vectors corresponding to the transition matrices constructed from the \emph{clustering} applied to the different quantities analyzed in the present work.
The entries of the vector $\boldsymbol{\pi}$ are interpreted as the stationary probability of being in a specific state, that is, the expected fraction of time in that state within the time horizon considered.

In the figure, the elements are presented in the following order, from bottom to top:

\begin{itemize}
    \item Per-window correlation matrices, $\mathbf{C}[\mathcal{E}(s,q)]$ (Section~\ref{sec:matrices_corr}).
    \item Dominant eigenvalues $\lambda_{N}$ (Section~\ref{section: Dinamica Eigenvalores}).
    \item Dominant eigenvectors with squared entries, $[v_i^2]_{i=1}^{N}$ (Section~\ref{section:Eigenvectores al cuadrado}).
    \item Matrices $\mathbf{C}^{2}$ (Subsection~\ref{subsection: Matrices C2}).
    \item Matrices $\mathbf{C}^{3}$ (Subsection~\ref{subsection: Matrices C3}).
\end{itemize}

When comparing the stationary vectors obtained from each of the quantities shown in Fig.~\ref{fig:vectoresEstacionariosMS_q40_k5},
differences are appreciated in how the stationary probability is distributed among states. In particular, when considering the vector associated with the correlation matrices
$\mathbf{C}[\mathcal{E}(s,q)]$, the first entry (corresponding to state $S_1$) decreases when passing to $\mathbf{C}^{2}$ and $\mathbf{C}^{3}$.
By contrast, the component associated with state $S_5$ increases considerably: in $\mathbf{C}^{2}$ and $\mathbf{C}^{3}$ it reaches a value more than four times
that observed in the case of the correlation matrices.

In the stationary vectors associated with $\mathbf{C}^{2}$ and $\mathbf{C}^{3}$, the \emph{COVID} State is identified in \emph{cluster} 2 and maintains similar values between both constructions. By contrast, when considering the stationary vector associated with the dominant eigenvectors with squared entries, $[v_i^2]_{i=1}^{N}$, the \emph{COVID} episode appears in \emph{cluster} 3 for $q=40$ and $k=5$, in agreement with what was previously shown in Fig.~\ref{fig:evolucionMS_squared_q20q40}. On the other hand, the vector associated with $\lambda_{N}$ does not show a comparable concentration of the \emph{COVID} State. The complementary plots for other time windows $q$ and other values of the number of \emph{clusters} $k$ are presented in Appendix~\ref{sec: Vectores estacionarios complementarios de MS}.
\begin{figure}[ht]
  \centering
  \includegraphics[width=0.92\textwidth]{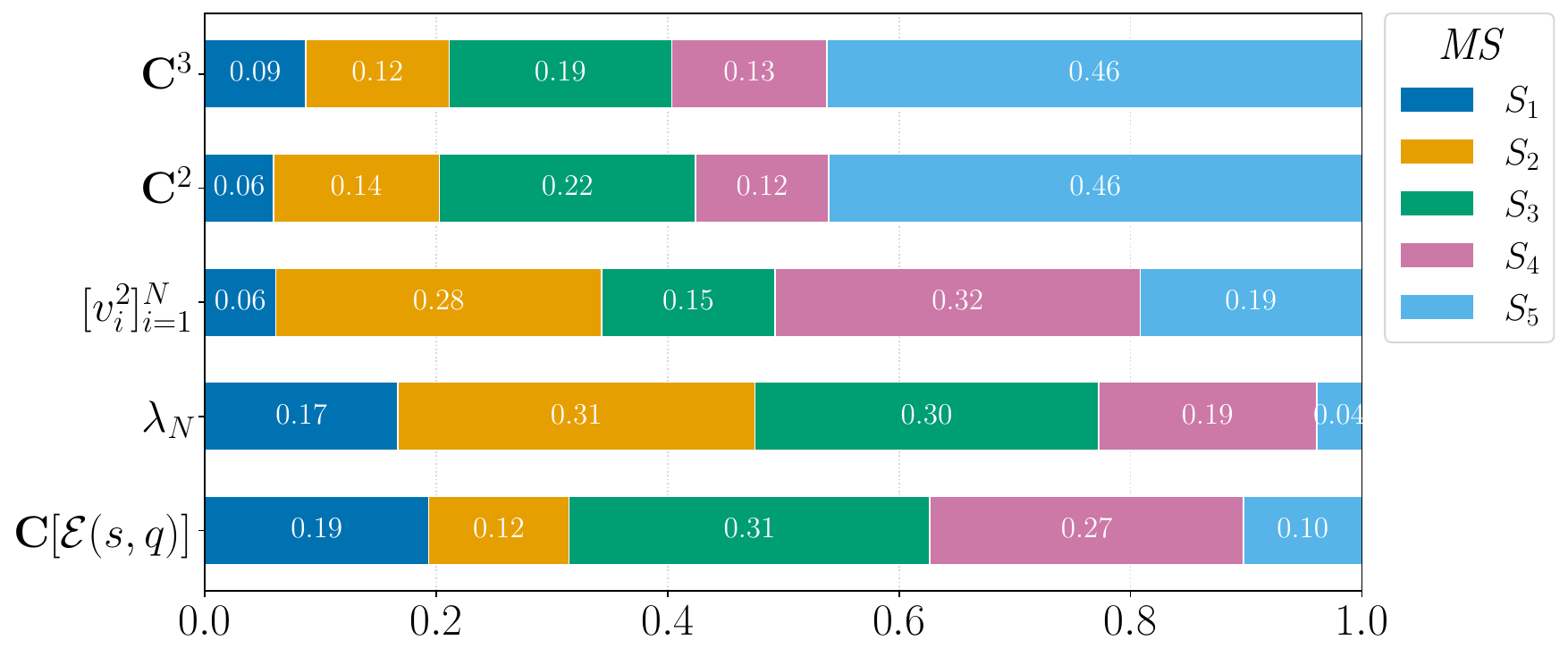}
  \caption[Stationary vectors of \emph{MS}, $q=40$, $k=5$]{Stationary vectors $\boldsymbol{\pi}$ of the transition matrices associated with the \emph{MS} for $q=40$ and $k=5$, obtained by means of \emph{clustering} from the quantities analyzed in the present work. Each entry of $\boldsymbol{\pi}$ is interpreted as the stationary probability of being in a market state. In the stationary vectors associated with $\mathbf{C}^{2}$ and $\mathbf{C}^{3}$, the \emph{COVID} episode is identified in \emph{cluster} 2, whereas for $[v_i^2]_{i=1}^{N}$ it appears in \emph{cluster} 3.}
  \label{fig:vectoresEstacionariosMS_q40_k5}
\end{figure}

\chapter*{Conclusions}
\phantomsection
\addcontentsline{toc}{chapter}{Conclusions}
This thesis focused on the analysis of the dynamics of the financial market from price fluctuations across different time windows. 
Throughout the research, dominant eigenvalues and eigenvectors were studied, as well as matrix constructions that preserved only the most relevant collective components. 
In addition, the square of the eigenvector components was analyzed with the purpose of interpreting them as a relative participation over the set of companies. This made it possible to characterize the collective behavior in each \emph{cluster}, since, during periods of crisis, the relative participation was distributed among a larger number of assets.

The $\mathbf{C}^{1}$ construction was insufficient to characterize the structure of the \emph{MS}, and was therefore discarded in the subsequent analyses. The results confirmed that the \emph{COVID} State\index{State!COVID} emerged as an \textbf{atypical state} in the correlation matrices; however, reconstructing it from spectral components required more than one principal component, since a single one is not sufficient to characterize it.

By contrast, the $\mathbf{C}^{2}$ and $\mathbf{C}^{3}$ constructions made it possible to reproduce the \emph{COVID} State, particularly for $q=40$ and $k=5$. On the other hand, when evaluating the stability by means of \emph{IEM}, the configuration $q=20$ and $k=4$ presented the greatest consensus within these same constructions. Within this framework, it was also observed that \emph{cluster} 1, corresponding to the 2017--2018 period, contributes to the description of the dynamics, although it does not manifest itself as an isolated state.

The exploration of unsupervised clustering techniques other than \textit{k}-Means\index{K-Means} remains open, since the stationary vectors derived from the transition matrices suggest limitations of the \emph{clustering} procedure when the representation of the system is based solely on a subset of components associated with the correlation matrices. In particular, a tendency was observed to decrease the population of the lower-index \emph{clusters} and to concentrate observations in those of higher index. In addition, the method is sensitive to the initial conditions and to the composition of the set of companies considered, which varies throughout the time horizon.

One line of methodological development consists in incorporating measures of informational complexity from Information Theory, such as the \textbf{Shannon entropy}\index{Entropy!Shannon}. 
In particular, the base-2 entropy provides a measure that makes it possible to compare configurations ranging from nearly uniform distributions to scenarios dominated by a reduced subset of assets.

The analysis developed in this thesis is not predictive; however, it has the potential to be applied in the design of investment portfolios and in \emph{trading} strategies, as discussed by Pharasi et al.~\cite{pharasiDynamicsMarketStates2024}. The objectives of this work focused on characterizing the collective dynamics of a financial market from the framework of complex systems and on identifying, from the correlation structure, a representative regime associated with the \emph{COVID} State, providing conceptual and methodological tools to understand its evolution more clearly.

Finally, once the model is improved and the methods proposed here are consolidated, a natural step will be to extend the comparative analysis to other international exchanges, such as the \textbf{Tokyo Stock Exchange}, \textbf{Euronext Paris}, \textbf{Frankfurter Wertpapierb\"orse}, the \textbf{Shanghai Stock Exchange}, and the \textbf{London Stock Exchange}, in order to assess the universality of the results and explore possible differences between markets.

\appendix
\formatoapendices

\chapter{Complementary Concepts}

\label{section:appendix1}

\section{Norms}
\label{appendix: Normas}


\subsection*{Euclidean Norm (\texorpdfstring{$\ell_2$}{l2})}\label{subsection:norma-euclidiana}

\noindent
For a vector \(\mathbf{v}\in\mathbb{R}^{n}\), the \textbf{Euclidean norm} is defined as:

\begin{equation}
  \|\mathbf{v}\|_{2}
  \;=\;
  \sqrt{\sum_{i=1}^{n} v_i^{2}}.
\label{eq:def-norma-euclidiana}
\end{equation}

\subsubsection*{Spectral Norm (Operator Norm)}\label{subsection:norma-espectral}
\index{Norm!spectral}
\noindent
Let \(A\) be a matrix with singular value decomposition (SVD)
\begin{equation}
  A \;=\; U\,\Sigma\,V^{\mathsf T},
\label{eq:matriz-espectral-svd}
\end{equation}
where:
\begin{equation}
  \Sigma \;=\;
  \begin{bmatrix}
    \sigma_{1} & 0          & \cdots & 0 \\
    0          & \sigma_{2} & \cdots & 0 \\
    \vdots     & \vdots     & \ddots & \vdots \\
    0          & 0          & \cdots & \sigma_{r}
  \end{bmatrix},
  \qquad
  \sigma_{1}\ge\sigma_{2}\ge\dots\ge\sigma_{r}>0.
  \label{eq:matriz-sigma}
\end{equation}
The \textbf{spectral norm} (or \emph{operator norm}) is denoted by \(\|A\|_{\mathrm{op}}\) and is defined as:
\begin{equation}
  \|A\|_{\mathrm{op}}
  \;=\;
  \sup_{\mathbf{v}\neq \mathbf{0}}
  \frac{\|A\mathbf{v}\|_{2}}{\|\mathbf{v}\|_{2}}
  \;=\;
  \sigma_{1}
  \;=\;
  \sqrt{\lambda_{\max}\!\big(A^{\mathsf T}A\big)}.
\label{eq:def-norma-espectral}
\end{equation}
In the case of symmetric matrices, \(\|A\|_{\mathrm{op}}=|\lambda_{\max}(A)|\).

\subsubsection*{Frobenius Norm}\label{subsection:norma-frobenius}
\index{Norm!Frobenius}
Taking into account the same matrix $A$ from~\ref{eq:matriz-espectral-svd}, the \textbf{Frobenius norm} is defined as:
\begin{equation}
  \|A\|_{F}
  \;=\;
  \sqrt{\sum_{i=1}^{r}\sigma_{i}^{2}}
  \;=\;
  \sqrt{\sum_{i=1}^{m}\sum_{j=1}^{n} a_{ij}^{2}}.
\label{eq:def-norma-frobenius}
\end{equation}
It is the Euclidean norm of the vector formed by all the entries of \(A\); it is invariant under orthogonal transformations.

\section{Eckart--Young--Mirsky Theorem}
\label{sec:cov_correl_espectral}
\index{Theorem!Eckart-Young-Mirsky}

\subsection{Singular Values}
Let \(A \in \mathbb{R}^{m \times n}\). \(A\) admits the \textbf{singular value decomposition} (SVD)\footnote{\emph{Singular Value Decomposition} (SVD): there exist orthogonal matrices \(U \in \mathbb{R}^{m \times m}\), \(V \in \mathbb{R}^{n \times n}\) and a (block) diagonal matrix \(\Sigma \in \mathbb{R}^{m \times n}\) such that \(A = U\,\Sigma\,V^{\mathsf T}\).}:
\begin{equation}
A \;=\; U\,\Sigma\,V^{\mathsf T},
\end{equation}
where
\begin{equation}
\Sigma \;=\;
\begin{bmatrix}
\sigma_{1} & 0          & \cdots & 0 \\
0          & \sigma_{2} & \cdots & 0 \\
\vdots     & \vdots     & \ddots & \vdots \\
0          & 0          & \cdots & \sigma_{r} \\
& & & & \mathbf{0}
\end{bmatrix},
\qquad
\sigma_{1}\ge\sigma_{2}\ge\dots\ge\sigma_{r}>0,
\qquad
r=\operatorname{rank}(A).
\end{equation}

The columns of \(U=(\mathbf{u}_1,\dots,\mathbf{u}_m)\) are called \textbf{left singular vectors}, while the columns of \(V=(\mathbf{v}_1,\dots,\mathbf{v}_n)\) are the \textbf{right singular vectors}. The scalars \(\sigma_i\) are the \textbf{singular values} of \(A\).

Unlike singular values, the \textbf{eigenvalues} \(\lambda\) are only defined for \emph{square matrices}. The fundamental relationship is:
\begin{equation}
A^{\mathsf T}A\,\mathbf{v}_i \;=\; \lambda_i\!\big(A^{\mathsf T}A\big)\,\mathbf{v}_i,
\qquad
\sigma_i \;=\; \sqrt{\,\lambda_i\!\big(A^{\mathsf T}A\big)\,}\;\ge 0,
\qquad i=1,\dots,r.
\end{equation}

Analogously,
\begin{equation}
AA^{\mathsf T}\,\mathbf{u}_i \;=\; \sigma_i^{2}\,\mathbf{u}_i,
\end{equation}
that is, the columns of \(U\) are eigenvectors of \(AA^{\mathsf T}\) with eigenvalues \(\sigma_i^{2}\).

\paragraph{Real and symmetric case.}
If \(M\in\mathbb{R}^{n\times n}\) is \textbf{symmetric} (\(M=M^{\mathsf T}\)), then \(M^{\mathsf T}M=M^{2}\) and
\(\lambda_i\!\big(M^{\mathsf T}M\big)=\lambda_i\!\big(M^{2}\big)=\lambda_i(M)^{2}\).
Therefore,
\[
\sigma_i(M)=\sqrt{\lambda_i\!\big(M^{\mathsf T}M\big)}=\sqrt{\lambda_i(M)^{2}}=|\lambda_i(M)|.
\]
If, in addition, \(M\succeq 0\) (as the correlation matrices used in this thesis), then \(\lambda_i(M)\ge 0\) and
\(\sigma_i(M)=\lambda_i(M)\), with \(\mathbf{u}_i=\mathbf{v}_i\) (the singular vectors coincide with the orthonormal eigenvectors of \(M\)).

\begin{theorem}[Eckart--Young--Mirsky]
Let \(A \in \mathbb{R}^{m \times n}\) be of rank \(N\), with singular values \(\sigma_1 \ge \cdots \ge \sigma_N > 0\).
For an integer \(r\) such that \(1 \le r < N\), define the rank-\(r\) SVD truncation
\[
A_r = \sum_{i=1}^{r} \sigma_i \mathbf{u}_i \mathbf{v}_i^{\mathsf{T}}.
\]
Then, for the spectral norm \(\|\cdot\|_2\) and the Frobenius norm \(\|\cdot\|_F\),
\[
\|A - A_r\| \le \|A - B\| \quad \text{for every matrix } B \text{ with } \operatorname{rank}(B) \le r,
\]
and, in particular,
\[
\|A - A_r\| = 
\begin{cases}
\sigma_{r+1}, & \text{if } \|\cdot\| = \|\cdot\|_2, \\
\left( \sum_{i=r+1}^{N} \sigma_i^2 \right)^{\frac{1}{2}}, & \text{if } \|\cdot\| = \|\cdot\|_F.
\end{cases}
\]
\end{theorem}

\begin{proof}
A proof adapted from \cite{wilkinson35LowrankApproximation} is presented below.

\vspace{0.5em}
\noindent
\textbf{i) Case of the spectral norm \(\|\cdot\|_2\):}

Let \(B \in \mathbb{R}^{n \times p}\) be of rank \(r\). By the rank--nullity theorem, the dimension of the null space \(N(B) \subseteq \mathbb{R}^p\) is \(\dim[N(B)] = p - r\). Consider the matrix
\[
V_{r+1} = [\mathbf{v}_1 \ldots \mathbf{v}_{r+1}] \in \mathbb{R}^{p \times (r+1)},
\]
which has rank \(r + 1\) and column space \(\mathcal{C}(V_{r+1}) \subseteq \mathbb{R}^p\). Since
\[
\dim[N(B)] + \dim[\mathcal{C}(V_{r+1})] = (p - r) + (r + 1) = p + 1,
\]
these subspaces cannot be disjoint. We then select
\[
\mathbf{w} \in N(B) \cap \mathcal{C}(V_{r+1}), \quad \|\mathbf{w}\|_2 = 1.
\]

From the definition of the spectral norm we obtain
\[
\|A - B\|_2^2 \geq \|(A - B)\mathbf{w}\|_2^2 = \|A\mathbf{w}\|_2^2,
\]
where the equality follows from the fact that \(\mathbf{w} \in N(B)\), implying \(B\mathbf{w} = 0\) and therefore \((A - B)\mathbf{w} = A\mathbf{w}\).

Applying the singular value decomposition \(A = U\Sigma V^\top\) and using the orthogonality of \(V\), we conclude that
\[
\begin{split}
\|A - B\|_2^2 &= \|A\mathbf{w}\|_2^2 = \mathbf{w}^\top V \Sigma^2 V^\top \mathbf{w} \\
&= \sum_{i=1}^{r+1} \sigma_i^2 w_i^2 \geq \sigma_{r+1}^2 \sum_{i=1}^{r+1} w_i^2 \\
&= \sigma_{r+1}^2 = \|A - A_r\|_2^2.
\end{split}
\]
This proves the desired inequality for the spectral norm.

\vspace{0.5em}

\noindent
\textbf{ii) Case of the Frobenius norm \(\|\cdot\|_F\):}

Consider the SVD of \(A = U\Sigma V^\top\), where \(\Sigma\) is diagonal with singular values \(\sigma_1 \geq \sigma_2 \geq \cdots \geq \sigma_N \geq 0\). Let \(B\) be a matrix of rank \(r\), and define
\[
B = U C V^\top, \quad C = U^\top B V,
\]
so that \(\operatorname{rank}(C) = r\). Owing to the invariance of the Frobenius norm under orthogonal transformations,
\[
\|A - B\|_F = \|\Sigma - C\|_F.
\]

We define the block-diagonal matrices:
\[
\Sigma_r = 
\begin{pmatrix}
\sigma_1 & & & \\
& \ddots & & \\
& & \sigma_r & \\
& & & 0 \\
& & & & \ddots
\end{pmatrix},
\quad
\Sigma_\perp = 
\begin{pmatrix}
0 & & & \\
& \ddots & & \\
& & 0 & \\
& & & \sigma_{r+1} \\
& & & & \ddots
\end{pmatrix},
\]
so that \(\Sigma = \Sigma_r + \Sigma_\perp\) and \(\operatorname{rank}(\Sigma_r) = r\).

Expanding the squared norm:
\[
\begin{split}
\|\Sigma - C\|_F^2 &= \langle \Sigma_r + \Sigma_\perp - C, \Sigma_r + \Sigma_\perp - C \rangle_F \\
&= \underbrace{\|\Sigma_r - C\|_F^2}_{\geq 0} + \underbrace{\|\Sigma_\perp\|_F^2}_{\text{independent of } C},
\end{split}
\]
where we have used the orthogonality \(\langle \Sigma_r, \Sigma_\perp \rangle_F = 0\).

Therefore,
\[
\|\Sigma_r - C\|_F^2 \geq \|\Sigma_\perp\|_F^2,
\]
with equality if and only if \(C = \Sigma_r\). Taking square roots and returning to the original variable \(A\), we conclude that
\[
\|A - B\|_F \geq \|A - A_r\|_F = \left( \sum_{i=r+1}^{N} \sigma_i^2 \right)^{1/2},
\]
which proves the best rank-\(r\) approximation property for the Frobenius norm.

\end{proof}

\subsection{Effective Rank of the Correlation Matrices}
\index{Rank!correlation matrices}
\label{subsec:Rango efectivo corr}
The temporal centering of the returns within a window $\mathcal{E}(s,q)$ determines fundamental algebraic properties of the resulting matrices, particularly in their \textbf{trace} and \textbf{rank}.

Let $\mathbf{C}\big[\mathcal{E}(s,q)\big]$
be the correlation matrix for the window $\mathcal{E}(s,q)$ defined in \eqref{eq:submatriz-q}; the centering operator is defined as
\begin{equation}
    H_q = I_q - \frac{1}{q}\,\mathbf{1}_q \mathbf{1}_q^{\top},
    \label{eq:def-Hq}
\end{equation}
where $I_q$ is the identity matrix of size $q\times q$ and 
\[
\mathbf{1}_q \mathbf{1}_q^{\top} \;=\;
\begin{pmatrix}
1 & \cdots & 1 \\
\vdots & \ddots & \vdots \\
1 & \cdots & 1
\end{pmatrix}\in \mathbb{R}^{q\times q}.
\]

\paragraph{Centered matrix.}
The centered version of the data in the window $\mathcal{E}(s,q)$ is:
\begin{equation}
    \mathbf{Z}_{\mathcal{E}(s,q)} = H_q \,\mathbf{C}\big[\mathcal{E}(s,q)\big] \;\in\; \mathbb{R}^{q\times N}.
    \label{eq:Z-centrada}
\end{equation}
The term $\tfrac{1}{q}\,\mathbf{1}_q \mathbf{1}_q^{\top}$ computes the temporal average in each column, and the action of $H_q$ corresponds to subtracting that average.

\paragraph{Idempotency of $H_q$.}
To verify idempotency, note that
\[
(\mathbf{1}_q \mathbf{1}_q^{\top})(\mathbf{1}_q \mathbf{1}_q^{\top}) = q\,\mathbf{1}_q \mathbf{1}_q^{\top}.
\]
In this way,
\[
H_q^2 = \Bigl(I_q - \tfrac{1}{q}\mathbf{1}_q \mathbf{1}_q^{\top}\Bigr)^2
= I_q - \tfrac{2}{q}\mathbf{1}_q \mathbf{1}_q^{\top} + \tfrac{1}{q^2}\,q\,\mathbf{1}_q \mathbf{1}_q^{\top}
= I_q - \tfrac{1}{q}\mathbf{1}_q \mathbf{1}_q^{\top} = H_q.
\]
Thus, $H_q$ is idempotent and acts as an orthogonal projector.

\paragraph{Trace and rank of $H_q$.}
The trace is computed directly as
\begin{equation}
\mathrm{tr}(H_q) = \mathrm{tr}(I_q) - \tfrac{1}{q}\,\mathrm{tr}(\mathbf{1}_q \mathbf{1}_q^{\top})
= q - \tfrac{1}{q}\,q = q-1.
\end{equation}
Therefore, $H_q$ has $q-1$ eigenvalues equal to $1$ and one equal to $0$, which implies
\begin{equation}
\mathrm{rank}(H_q) = q-1, 
\qquad 
\mathrm{null}(H_q) = 1.
\end{equation}

Since $\mathbf{Z}_{\mathcal{E}(s,q)}$ is the product of $H_q$ and $\mathbf{C}\big[\mathcal{E}(s,q)$, the bound holds
\begin{equation}
\mathrm{rank}\bigl(\mathbf{Z}_{\mathcal{E}(s,q)}\bigr) \;\le\; \min\{\,q-1,\,N\,\}.
\label{eq:cota-rango-Z}
\end{equation}

\paragraph{Invariant trace and eigenvalues.}
For a symmetric matrix $A \in \mathbb{R}^{N\times N}$ diagonalizable as $A=V\Lambda V^\top$, it holds that
\begin{equation}
\mathrm{tr}(A) 
= \mathrm{tr}(V\Lambda V^{\top}) 
= \mathrm{tr}(V^{\top}V\Lambda) 
= \mathrm{tr}(\Lambda) 
= \sum_{j=1}^{N}\lambda_j,
\label{eq:traza-espectro}
\end{equation}

Thus, the trace is invariant and is the sum of the eigenvalues. In particular, for the correlation matrices 
$
\mathbf{C}\big[\mathcal{E}(s,q)\big],
$
the bound \eqref{eq:cota-rango-Z} establishes that, in each window, at most $q-1$ eigenvalues can be strictly positive.

\section{Mar\v{c}enko-Pastur Distribution}
\index{Mar\v{c}enko-Pastur!distribution}
\label{sec: Marcenko-Pastur}
In the analysis of financial series, the correlation matrices are constructed from time windows of length $q$ with $N$ series of standardized returns. When the series present no real correlations, that is, when $\mathbb{E}[r_i r_j]=0$ for $i\neq j$ and $\mathbb{E}[r_i^2]=1$, the behavior of the eigenvalue spectrum is described by the Wishart model \cite{wishartGeneralisedProductMoment1928}. This result establishes that, in the high-dimensional regime where the ratio $Q=q/N$ is kept constant, the distribution of the eigenvalues converges to what was later known as the \emph{Mar\v{c}enko--Pastur Law} \cite{marcenkoDistributionEigenvaluesSets1967}.

The theoretical eigenvalue distribution is concentrated between the values
\begin{equation}
\lambda_{\pm}\;=\;\bigl(1\pm Q^{-1/2}\bigr)^{2},
\qquad Q>1,
\label{eq:mp-bordes}
\end{equation}
and within this interval $[\lambda_{-},\lambda_{+}]$, the density follows the expression
\begin{equation}
\rho_{\mathrm{MP}}(\lambda)\;=\;\frac{Q}{2\pi}\,
\frac{\sqrt{(\lambda_{+}-\lambda)(\lambda-\lambda_{-})}}{\lambda},
\qquad \lambda\in[\lambda_{-},\lambda_{+}],
\label{eq:mp-densidad}
\end{equation}
whereas for $Q\le 1$ a point peak appears at zero due to the singularity of the sample matrix \cite{marcenkoDistributionEigenvaluesSets1967}.

This theoretical delimitation makes it possible to separate the noise from the signal: the majority of the eigenvalues must remain within $[\lambda_{-},\lambda_{+}]$, whereas those exceeding these limits are interpreted as non-random structure. However, when working with short time series where $N \gg T$, we face the problem of singular correlation matrices that possess only $T-1$ nonzero eigenvalues, which limits the statistical analysis \cite{vyasMultivariateAnalysisShort2018}. Vyas et al. propose constructing ensembles of non-singular matrices by means of random selections of subsets of the $N$ available time series. This technique makes it possible to obtain smooth eigenvalue distributions and to adequately characterize both the central region of the spectrum and the outliers, even in highly singular regimes where conventional analysis would be statistically insufficient \cite{vyasMultivariateAnalysisShort2018}.
\chapter{Complementary Tables}

\label{section:appendix2}

\newcommand{\cabecera}[1]{\rule[-7pt]{0pt}{22pt}\color{white}\textbf{#1}}

\section{Sectors}
\label{appendix:sectores}
The financial sectors of the S\&P 500 used are shown below:
\index{Sectors!GICS}
\begin{table}[ht]
  \centering
  \small
  \begin{tabular}{|c|c|>{\centering\arraybackslash}m{8cm}|}
    \hline
    \rowcolor{teal!80}
    \cabecera{GICS} & 
    \cabecera{Sector} & 
    \cabecera{Brief Description} \\
    \hline
    CD & Consumer Discretionary & Non-essential consumer goods and services (retail, leisure, automotive). \\
    \hline
    CS & Consumer Staples       & Basic consumer goods and essential products (food, beverages, hygiene). \\
    \hline
    EG & Energy                 & Exploration, production, and services of oil and gas; traditional and transition energy. \\
    \hline
    FN & Financials             & Banks, insurers, exchanges, and diversified financial services. \\
    \hline
    HC & Health Care            & Health services, pharmaceuticals, biotechnology, and medical equipment. \\
    \hline
    ID & Industrials            & Manufacturing, transportation, aerospace, and industrial machinery. \\
    \hline
    IT & Information Technology & Software, hardware, semiconductors, and IT services. \\
    \hline
    MT & Materials              & Chemicals, metals, mining, paper, and construction materials. \\
    \hline
    TC & Telecommunication Services & Fixed and mobile telecommunication services, networks, and operators. \\
    \hline
    UT & Utilities              & Public electricity, gas, and water services for the population and industries. \\
    \hline
  \end{tabular}
  \caption[GICS sectors of the S\&P 500]{Sectors, GICS codes, and brief description used in the analysis of the S\&P 500.}
  \label{tab:sectores-gics-sp500}
\end{table}
\newpage
\section{S\&P 500 Companies Used}

\index{Companies!S\&P 500}

The list of companies analyzed is presented below, with their abbreviation (ticker), industry, and the corresponding GICS sector, ordered according to the sectors defined in Table \ref{tab:sectores-gics-sp500}. No companies from the \emph{TC} sector are used in this thesis.

\begingroup
\scriptsize 

\setlength{\tabcolsep}{4pt}          
\renewcommand{\arraystretch}{0.9}    

\setlength{\LTleft}{0pt}
\setlength{\LTright}{0pt}
\begin{longtable}{|p{0.06\textwidth}|p{0.32\textwidth}|p{0.08\textwidth}|p{0.46\textwidth}|}
\hline
\rowcolor{teal!80}
\cabecera{GICS} & \cabecera{Industry} & \cabecera{Ticker} & \cabecera{Company} \\
\hline
\endfirsthead

\hline
\rowcolor{teal!80}
\cabecera{GICS} & \cabecera{Industry} & \cabecera{Ticker} & \cabecera{Company} \\
\hline
\endhead

\hline
\multicolumn{4}{r}{\small\textit{Continued on the next page}}\\
\hline
\endfoot

\hline
\noalign{\vskip 12pt}
\caption[List of S\&P 500 companies used]{List of companies used with their abbreviated GICS sector, Industry, \textit{ticker}, and Company.}
\label{tab:companies-sp500}
\endlastfoot

CD & Internet Retail & AMZN & Amazon.com, Inc. \\
CD & Automotive Parts \& Equipment & APTV & Aptiv PLC \\
CD & Specialty Retail & AZO & AutoZone, Inc. \\
CD & Specialty Retail & BBWI & Bath \& Body Works, Inc. \\
CD & Specialty Retail & BBY & Best Buy Co., Inc. \\
CD & Travel Services & BKNG & Booking Holdings Inc. \\
CD & Auto Parts \& Equipment & BWA & BorgWarner Inc. \\
CD & Travel Services & CCL & Carnival Corporation \& plc \\
CD & Restaurants & CMG & Chipotle Mexican Grill, Inc. \\
CD & Residential Construction & DHI & D.R. Horton, Inc. \\
CD & Discount Stores & DLTR & Dollar Tree, Inc. \\
CD & Restaurants & DPZ & Domino's Pizza, Inc. \\
CD & Restaurants & DRI & Darden Restaurants, Inc. \\
CD & Travel Services & EXPE & Expedia Group, Inc. \\
CD & Auto Manufacturers & F & Ford Motor Company \\
CD & Auto Manufacturers & GM & General Motors Company \\
CD & Specialty Retail & GPC & Genuine Parts Company \\
CD & Scientific \& Technical Instruments & GRMN & Garmin Ltd. \\
CD & Leisure & HAS & Hasbro, Inc. \\
CD & Home Improvement Retail & HD & The Home Depot, Inc. \\
CD & Auto \& Truck Dealerships & KMX & CarMax, Inc. \\
CD & Residential Construction & LEN & Lennar Corporation \\
CD & Auto Parts & LKQ & LKQ Corporation \\
CD & Home Improvement Retail & LOW & Lowe's Companies, Inc. \\
CD & Resorts \& Casinos & LVS & Las Vegas Sands Corp. \\
CD & Lodging & MAR & Marriott International, Inc. \\
CD & Restaurants & MCD & McDonald's Corporation \\
CD & Resorts \& Casinos & MGM & MGM Resorts International \\
CD & Furnishings, Fixtures \& Appliances & MHK & Mohawk Industries, Inc. \\
CD & Footwear \& Accessories & NKE & NIKE, Inc. \\
CD & Residential Construction & NVR & NVR, Inc. \\
CD & Household \& Personal Products & NWL & Newell Brands Inc. \\
CD & Specialty Retail & ORLY & O'Reilly Automotive, Inc. \\
CD & Residential Construction & PHM & PulteGroup, Inc. \\
CD & Distributors & POOL & Pool Corporation \\
CD & Travel Services & RCL & Royal Caribbean Group \\
CD & Apparel Manufacturing & RL & Ralph Lauren Corporation \\
CD & Apparel Retail & ROST & Ross Stores, Inc. \\
CD & Restaurants & SBUX & Starbucks Corporation \\
CD & Tools \& Accessories & SWK & Stanley Black \& Decker, Inc. \\
CD & Discount Stores & TGT & Target Corporation \\
CD & Apparel Retail & TJX & The TJX Companies, Inc. \\
CD & Luxury Goods & TPR & Tapestry, Inc. \\
CD & Specialty Retail & TSCO & Tractor Supply Company \\
CD & Auto Manufacturers & TSLA & Tesla, Inc. \\
CD & Specialty Retail & ULTA & Ulta Beauty, Inc. \\
CD & Apparel Manufacturing & VFC & V.F. Corporation \\
CD & Furnishings, Fixtures \& Appliances & WHR & Whirlpool Corporation \\
CD & Resorts \& Casinos & WYNN & Wynn Resorts, Limited \\
CD & Restaurants & YUM & Yum! Brands, Inc. \\
CS & Farm Products & ADM & Archer-Daniels-Midland Company \\
CS & Agricultural Products \& Services & BG & Bunge Global S.A. \\
CS & Packaged Foods & CAG & Conagra Brands, Inc. \\
CS & Household \& Personal Products & CHD & Church \& Dwight Co., Inc. \\
CS & Household \& Personal Products & CL & Colgate-Palmolive Company \\
CS & Household \& Personal Products & CLX & The Clorox Company \\
CS & Discount Stores & COST & Costco Wholesale Corporation \\
CS & Packaged Foods & CPB & Campbell Soup Company \\
CS & Healthcare Plans & CVS & CVS Health Corporation \\
CS & Consumer Staples Merchandise Retail & DG & Dollar General Corporation \\
CS & Household \& Personal Products & EL & The Est\'ee Lauder Companies Inc. \\
CS & Packaged Foods & GIS & General Mills, Inc. \\
CS & Packaged Foods & HRL & Hormel Foods Corporation \\
CS & Confectioners & HSY & The Hershey Company \\
CS & Packaged Foods & K & Kellogg Company \\
CS & Beverages—Non-Alcoholic & KDP & Keurig Dr Pepper Inc. \\
CS & Household \& Personal Products & KMB & Kimberly-Clark Corporation \\
CS & Beverages—Non-Alcoholic & KO & The Coca-Cola Company \\
CS & Grocery Stores & KR & The Kroger Co. \\
CS & Confectioners & MDLZ & Mondelez International, Inc. \\
CS & Packaged Foods & MKC & McCormick \& Company, Incorporated \\
CS & Beverages—Non-Alcoholic & MNST & Monster Beverage Corporation \\
CS & Tobacco & MO & Altria Group, Inc. \\
CS & Beverages—Non-Alcoholic & PEP & PepsiCo, Inc. \\
CS & Household \& Personal Products & PG & The Procter \& Gamble Company \\
CS & Tobacco & PM & Philip Morris International Inc. \\
CS & Packaged Foods & SJM & The J. M. Smucker Company \\
CS & Beverages—Wineries \& Distilleries & STZ & Constellation Brands, Inc. \\
CS & Food Distribution & SYY & Sysco Corporation \\
CS & Beverages—Brewers & TAP & Molson Coors Beverage Company \\
CS & Farm Products & TSN & Tyson Foods, Inc. \\
CS & Pharmaceutical Retailers & WBA & Walgreens Boots Alliance, Inc. \\
CS & Discount Stores & WMT & Walmart Inc. \\
EG & Oil \& Gas E\&P & APA & APA Corporation \\
EG & Oil \& Gas Equipment \& Services & BKR & Baker Hughes Company \\
EG & Oil \& Gas E\&P & COP & ConocoPhillips \\
EG & Oil \& Gas Exploration \& Production & CTRA & Coterra Energy Inc. \\
EG & Oil \& Gas Integrated & CVX & Chevron Corporation \\
EG & Oil \& Gas E\&P & DVN & Devon Energy Corporation \\
EG & Oil \& Gas E\&P & EOG & EOG Resources, Inc. \\
EG & Oil \& Gas E\&P & EQT & EQT Corporation \\
EG & Oil \& Gas Equipment \& Services & HAL & Halliburton Company \\
EG & Oil \& Gas E\&P & HES & Hess Corporation \\
EG & Oil \& Gas Midstream & KMI & Kinder Morgan, Inc. \\
EG & Oil \& Gas Refining \& Marketing & MPC & Marathon Petroleum Corporation \\
EG & Oil \& Gas E\&P & MRO & Marathon Oil Corporation \\
EG & Oil \& Gas Midstream & OKE & ONEOK, Inc. \\
EG & Oil \& Gas E\&P & OXY & Occidental Petroleum Corporation \\
EG & Oil \& Gas Equipment \& Services & SLB & Schlumberger Limited \\
EG & Oil \& Gas Midstream & TRGP & Targa Resources Corp. \\
EG & Oil \& Gas Refining \& Marketing & VLO & Valero Energy Corporation \\
EG & Oil \& Gas Midstream & WMB & The Williams Companies, Inc. \\
EG & Oil \& Gas Integrated & XOM & Exxon Mobil Corporation \\
FN & Property \& Casualty Insurance & ACGL & Arch Capital Group \\
FN & Insurance—Life & AFL & Aflac Incorporated \\
FN & Insurance—Diversified & AIG & American International Group, Inc. \\
FN & Insurance—Specialty & AIZ & Assurant, Inc. \\
FN & Insurance Brokers & AJG & Arthur J. Gallagher \& Co. \\
FN & Property \& Casualty Insurance & ALL & The Allstate Corporation \\
FN & Asset Management & AMP & Ameriprise Financial, Inc. \\
FN & Telecom Tower REITs & AMT & American Tower Corporation \\
FN & Insurance Brokers & AON & Aon plc \\
FN & Office REITs & ARE & Alexandria Real Estate Equities, Inc. \\
FN & Multi-Family Residential REITs & AVB & AvalonBay Communities, Inc. \\
FN & Credit Services & AXP & American Express Company \\
FN & Banks—Diversified & BAC & Bank of America Corporation \\
FN & Asset Management & BEN & Franklin Resources, Inc. \\
FN & Asset Management & BK & The Bank of New York Mellon Corporation \\
FN & Asset Management & BLK & BlackRock, Inc. \\
FN & Insurance Brokers & BRO & Brown \& Brown, Inc. \\
FN & Office REITs & BXP & BXP, Inc. \\
FN & Banks—Diversified & C & Citigroup Inc. \\
FN & Property \& Casualty Insurance & CB & Chubb Limited \\
FN & Financial Exchanges \& Data & CBOE & Cboe Global Markets, Inc. \\
FN & Real Estate Services & CBRE & CBRE Group, Inc. \\
FN & Telecom Tower REITs & CCI & Crown Castle Inc. \\
FN & Insurance—Property \& Casualty & CINF & Cincinnati Financial Corporation \\
FN & Banks—Regional & CMA & Comerica Incorporated \\
FN & Financial Data \& Stock Exchanges & CME & CME Group Inc. \\
FN & Consumer Finance & COF & Capital One Financial Corporation \\
FN & Multi-Family Residential REITs & CPT & Camden Property Trust \\
FN & Real Estate Services & CSGP & CoStar Group, Inc. \\
FN & Consumer Finance & DFS & Discover Financial Services \\
FN & Data Center REITs & DLR & Digital Realty Trust, Inc. \\
FN & Reinsurance & EG & Everest Group, Ltd. \\
FN & Data Center REITs & EQIX & Equinix, Inc. \\
FN & Multi-Family Residential REITs & EQR & Equity Residential \\
FN & Multi-Family Residential REITs & ESS & Essex Property Trust, Inc. \\
FN & Self-Storage REITs & EXR & Extra Space Storage Inc. \\
FN & Financial Exchanges \& Data & FDS & FactSet Research Systems Inc. \\
FN & Banks—Regional & FITB & Fifth Third Bancorp \\
FN & REIT—Retail & FRT & Federal Realty Investment Trust \\
FN & Insurance—Life & GL & Globe Life Inc. \\
FN & Capital Markets & GS & The Goldman Sachs Group, Inc. \\
FN & Banks—Regional & HBAN & Huntington Bancshares Incorporated \\
FN & Insurance—Diversified & HIG & The Hartford Financial Services Group, Inc. \\
FN & Hotel \& Resort REITs & HST & Host Hotels \& Resorts, Inc. \\
FN & Financial Data \& Stock Exchanges & ICE & Intercontinental Exchange, Inc. \\
FN & Specialized REITs & IRM & Iron Mountain Incorporated \\
FN & Asset Management & IVZ & Invesco Ltd. \\
FN & Banks—Diversified & JPM & JPMorgan Chase \& Co. \\
FN & Banks—Regional & KEY & KeyCorp \\
FN & REIT—Retail & KIM & Kimco Realty Corporation \\
FN & Insurance—Property \& Casualty & L & Loews Corporation \\
FN & Insurance—Life & LNC & Lincoln National Corporation \\
FN & Credit Services & MA & Mastercard Incorporated \\
FN & Multi-Family Residential REITs & MAA & Mid-America Apartment Communities, Inc. \\
FN & Financial Data \& Stock Exchanges & MCO & Moody's Corporation \\
FN & Insurance—Life & MET & MetLife, Inc. \\
FN & Financial Exchanges \& Data & MKTX & MarketAxess Holdings Inc. \\
FN & Insurance Brokers & MMC & Marsh \& McLennan Companies, Inc. \\
FN & Capital Markets & MS & Morgan Stanley \\
FN & Financial Exchanges \& Data & MSCI & MSCI Inc. \\
FN & Banks—Regional & MTB & M\&T Bank Corporation \\
FN & Financial Data \& Stock Exchanges & NDAQ & Nasdaq, Inc. \\
FN & Asset Management & NTRS & Northern Trust Corporation \\
FN & REIT—Retail & O & Realty Income Corporation \\
FN & Asset Management & PFG & Principal Financial Group, Inc. \\
FN & Insurance—Property \& Casualty & PGR & The Progressive Corporation \\
FN & Industrial REITs & PLD & Prologis, Inc. \\
FN & Banks—Regional & PNC & The PNC Financial Services Group, Inc. \\
FN & Insurance—Life & PRU & Prudential Financial, Inc. \\
FN & Self-Storage REITs & PSA & Public Storage \\
FN & REIT—Retail & REG & Regency Centers Corporation \\
FN & Banks—Regional & RF & Regions Financial Corporation \\
FN & Capital Markets & RJF & Raymond James Financial, Inc. \\
FN & Telecom Tower REITs & SBAC & SBA Communications Corporation \\
FN & Capital Markets & SCHW & The Charles Schwab Corporation \\
FN & REIT—Retail & SPG & Simon Property Group, Inc. \\
FN & Financial Data \& Stock Exchanges & SPGI & S\&P Global Inc. \\
FN & Asset Management & STT & State Street Corporation \\
FN & Banks—Regional & TFC & Truist Financial Corporation \\
FN & Asset Management & TROW & T. Rowe Price Group, Inc. \\
FN & Insurance—Property \& Casualty & TRV & The Travelers Companies, Inc. \\
FN & Multi-Family Residential REITs & UDR & UDR, Inc. \\
FN & Banks—Regional & USB & U.S. Bancorp \\
FN & Credit Services & V & Visa Inc. \\
FN & Health Care REITs & VTR & Ventas, Inc. \\
FN & Health Care REITs & WELL & Welltower Inc. \\
FN & Banks—Diversified & WFC & Wells Fargo \& Company \\
FN & Insurance—Property \& Casualty & WRB & W. R. Berkley Corporation \\
FN & Insurance Brokers & WTW & Willis Towers Watson Public Limited Company \\
FN & Specialized REITs & WY & Weyerhaeuser Company \\
FN & Banks—Regional & ZION & Zions Bancorporation, National Association \\
HC & Diagnostics \& Research & A & Agilent Technologies, Inc. \\
HC & Medical Devices & ABT & Abbott Laboratories \\
HC & Medical Devices & ALGN & Align Technology, Inc. \\
HC & Drug Manufacturers—General & AMGN & Amgen Inc. \\
HC & Medical Instruments \& Supplies & BAX & Baxter International Inc. \\
HC & Medical Instruments \& Supplies & BDX & Becton, Dickinson and Company \\
HC & Drug Manufacturers—General & BIIB & Biogen Inc. \\
HC & Life Sciences Tools \& Services & BIO & Bio-Rad Laboratories, Inc. \\
HC & Drug Manufacturers—General & BMY & Bristol-Myers Squibb Company \\
HC & Medical Devices & BSX & Boston Scientific Corporation \\
HC & Health Care Distributors & CAH & Cardinal Health, Inc. \\
HC & Healthcare Plans & CI & The Cigna Group \\
HC & Healthcare Plans & CNC & Centene Corporation \\
HC & Medical Instruments \& Supplies & COO & The Cooper Companies, Inc. \\
HC & Life Sciences Tools \& Services & CRL & Charles River Laboratories International, Inc. \\
HC & Diagnostics \& Research & DGX & Quest Diagnostics Incorporated \\
HC & Life Sciences Tools \& Services & DHR & Danaher Corporation \\
HC & Medical Care Facilities & DVA & DaVita Inc. \\
HC & Medical Devices & DXCM & DexCom, Inc. \\
HC & Healthcare Plans & ELV & Elevance Health, Inc. \\
HC & Medical Devices & EW & Edwards Lifesciences Corporation \\
HC & Drug Manufacturers—General & GILD & Gilead Sciences, Inc. \\
HC & Health Care Facilities & HCA & HCA Healthcare, Inc. \\
HC & Medical Instruments \& Supplies & HOLX & Hologic, Inc. \\
HC & Medical Distribution & HSIC & Henry Schein, Inc. \\
HC & Healthcare Plans & HUM & Humana Inc. \\
HC & Diagnostics \& Research & IDXX & IDEXX Laboratories, Inc. \\
HC & Diagnostics \& Research & ILMN & Illumina, Inc. \\
HC & Biotechnology & INCY & Incyte Corporation \\
HC & Medical Instruments \& Supplies & ISRG & Intuitive Surgical, Inc. \\
HC & Drug Manufacturers—General & JNJ & Johnson \& Johnson \\
HC & Diagnostics \& Research & LH & Laboratory Corporation of America Holdings \\
HC & Drug Manufacturers—General & LLY & Eli Lilly and Company \\
HC & Medical Distribution & MCK & McKesson Corporation \\
HC & Medical Devices & MDT & Medtronic plc \\
HC & Healthcare Plans & MOH & Molina Healthcare, Inc. \\
HC & Drug Manufacturers—General & MRK & Merck \& Co., Inc. \\
HC & Diagnostics \& Research & MTD & Mettler-Toledo International Inc. \\
HC & Drug Manufacturers—General & PFE & Pfizer Inc. \\
HC & Medical Devices & PODD & Insulet Corporation \\
HC & Biotechnology & REGN & Regeneron Pharmaceuticals, Inc. \\
HC & Medical Instruments \& Supplies & RMD & ResMed Inc. \\
HC & Life Sciences Tools \& Services & RVTY & Revvity, Inc. \\
HC & Medical Instruments \& Supplies & STE & STERIS plc \\
HC & Medical Devices & SYK & Stryker Corporation \\
HC & Life Sciences Tools \& Services & TECH & Bio-Techne Corporation \\
HC & Medical Instruments \& Supplies & TFX & Teleflex Incorporated \\
HC & Diagnostics \& Research & TMO & Thermo Fisher Scientific Inc. \\
HC & Medical Care Facilities & UHS & Universal Health Services, Inc. \\
HC & Healthcare Plans & UNH & UnitedHealth Group Incorporated \\
HC & Biotechnology & VRTX & Vertex Pharmaceuticals Incorporated \\
HC & Drug Manufacturers—Specialty \& Generic & VTRS & Viatris Inc. \\
HC & Diagnostics \& Research & WAT & Waters Corporation \\
HC & Medical Instruments \& Supplies & WST & West Pharmaceutical Services, Inc. \\
HC & Medical Instruments \& Supplies & XRAY & DENTSPLY SIRONA Inc. \\
HC & Medical Devices & ZBH & Zimmer Biomet Holdings, Inc. \\
ID & Airlines & AAL & American Airlines Group Inc. \\
ID & Airlines & ALK & Alaska Air Group, Inc. \\
ID & Specialty Industrial Machinery & AME & AMETEK, Inc. \\
ID & Specialty Industrial Machinery & AOS & A. O. Smith Corporation \\
ID & Aerospace \& Defense & AXON & Axon Enterprise, Inc. \\
ID & Aerospace \& Defense & BA & The Boeing Company \\
ID & Data Processing \& Outsourced Services & BR & Broadridge Financial Solutions, Inc. \\
ID & Farm \& Heavy Construction Machinery & CAT & Caterpillar Inc. \\
ID & Integrated Freight \& Logistics & CHRW & C.H. Robinson Worldwide, Inc. \\
ID & Specialty Industrial Machinery & CMI & Cummins Inc. \\
ID & Diversified Support Services & CPRT & Copart, Inc. \\
ID & Railroads & CSX & CSX Corporation \\
ID & Specialty Business Services & CTAS & Cintas Corporation \\
ID & Passenger Airlines & DAL & Delta Air Lines, Inc. \\
ID & Farm \& Heavy Construction Machinery & DE & Deere \& Company \\
ID & Specialty Industrial Machinery & DOV & Dover Corporation \\
ID & Consulting Services & EFX & Equifax Inc. \\
ID & Specialty Industrial Machinery & EMR & Emerson Electric Co. \\
ID & Specialty Industrial Machinery & ETN & Eaton Corporation plc \\
ID & Integrated Freight \& Logistics & EXPD & Expeditors International of Washington, Inc. \\
ID & Industrial Distribution & FAST & Fastenal Company \\
ID & Integrated Freight \& Logistics & FDX & FedEx Corporation \\
ID & Aerospace \& Defense & GD & General Dynamics Corporation \\
ID & Specialty Industrial Machinery & GE & General Electric Company \\
ID & Electrical Components \& Equipment & GNRC & Generac Holdings Inc. \\
ID & Industrial Distribution & GWW & W.W. Grainger, Inc. \\
ID & Aerospace \& Defense & HII & Huntington Ingalls Industries, Inc. \\
ID & Industrial Conglomerates & HON & Honeywell International Inc. \\
ID & Specialty Industrial Machinery & IEX & IDEX Corporation \\
ID & Specialty Industrial Machinery & ITW & Illinois Tool Works Inc. \\
ID & Engineering \& Construction & J & Jacobs Solutions Inc. \\
ID & Integrated Freight \& Logistics & JBHT & J.B. Hunt Transport Services, Inc. \\
ID & Building Products \& Equipment & JCI & Johnson Controls International plc \\
ID & Research \& Consulting Services & LDOS & Leidos Holdings, Inc. \\
ID & Aerospace \& Defense & LHX & L3Harris Technologies, Inc. \\
ID & Aerospace \& Defense & LMT & Lockheed Martin Corporation \\
ID & Airlines & LUV & Southwest Airlines Co. \\
ID & Building Products \& Equipment & MAS & Masco Corporation \\
ID & Industrial Conglomerates & MMM & 3M Company \\
ID & Specialty Industrial Machinery & NDSN & Nordson Corporation \\
ID & Aerospace \& Defense & NOC & Northrop Grumman Corporation \\
ID & Railroads & NSC & Norfolk Southern Corporation \\
ID & Trucking & ODFL & Old Dominion Freight Line, Inc. \\
ID & Farm \& Heavy Construction Machinery & PCAR & PACCAR Inc \\
ID & Specialty Industrial Machinery & PH & Parker-Hannifin Corporation \\
ID & Specialty Industrial Machinery & PNR & Pentair plc \\
ID & Engineering \& Construction & PWR & Quanta Services, Inc. \\
ID & Staffing \& Employment Services & RHI & Robert Half Inc. \\
ID & Specialty Industrial Machinery & ROK & Rockwell Automation, Inc. \\
ID & Environmental \& Facilities Services & ROL & Rollins, Inc. \\
ID & Software—Application & ROP & Roper Technologies, Inc. \\
ID & Waste Management & RSG & Republic Services, Inc. \\
ID & Aerospace \& Defense & RTX & RTX Corporation \\
ID & Specialty Industrial Machinery & SNA & Snap-on Incorporated \\
ID & Aerospace \& Defense & TDG & TransDigm Group Incorporated \\
ID & Aerospace \& Defense & TDY & Teledyne Technologies Incorporated \\
ID & Building Products \& Equipment & TT & Trane Technologies plc \\
ID & Aerospace \& Defense & TXT & Textron Inc. \\
ID & Passenger Airlines & UAL & United Airlines Holdings, Inc. \\
ID & Railroads & UNP & Union Pacific Corporation \\
ID & Integrated Freight \& Logistics & UPS & United Parcel Service, Inc. \\
ID & Rental \& Leasing Services & URI & United Rentals, Inc. \\
ID & Research \& Consulting Services & VRSK & Verisk Analytics, Inc. \\
ID & Specialty Industrial Machinery & WAB & Westinghouse Air Brake Technologies Corporation \\
ID & Waste Management & WM & Waste Management, Inc. \\
ID & Specialty Industrial Machinery & XYL & Xylem Inc. \\
IT & Consumer Electronics & AAPL & Apple Inc. \\
IT & Information Technology Services & ACN & Accenture plc \\
IT & Software—Infrastructure & ADBE & Adobe Inc. \\
IT & Semiconductors & ADI & Analog Devices, Inc. \\
IT & Staffing \& Employment Services & ADP & Automatic Data Processing, Inc. \\
IT & Software—Application & ADSK & Autodesk, Inc. \\
IT & Software—Infrastructure & AKAM & Akamai Technologies, Inc. \\
IT & Semiconductor Equipment \& Materials & AMAT & Applied Materials, Inc. \\
IT & Semiconductors & AMD & Advanced Micro Devices, Inc. \\
IT & Software—Application & ANSS & ANSYS, Inc. \\
IT & Electronic Components & APH & Amphenol Corporation \\
IT & Semiconductors & AVGO & Broadcom Inc. \\
IT & Software—Application & CDNS & Cadence Design Systems, Inc. \\
IT & Software—Application & CRM & Salesforce, Inc. \\
IT & Communication Equipment & CSCO & Cisco Systems, Inc. \\
IT & Information Technology Services & CTSH & Cognizant Technology Solutions Corporation \\
IT & Information Technology Services & DXC & DXC Technology Company \\
IT & Internet Retail & EBAY & eBay Inc. \\
IT & Software—Infrastructure & FFIV & F5, Inc. \\
IT & Data Processing \& Outsourced Services & FI & Fiserv, Inc. \\
IT & Software—Application & FICO & Fair Isaac Corporation \\
IT & Information Technology Services & FIS & Fidelity National Information Services, Inc. \\
IT & Semiconductors & FSLR & First Solar, Inc. \\
IT & Software—Infrastructure & FTNT & Fortinet, Inc. \\
IT & Software—Infrastructure & GEN & Gen Digital Inc. \\
IT & Electronic Components & GLW & Corning Incorporated \\
IT & Specialty Business Services & GPN & Global Payments Inc. \\
IT & Computer Hardware & HPQ & HP Inc. \\
IT & Information Technology Services & IBM & International Business Machines Corporation \\
IT & Semiconductors & INTC & Intel Corporation \\
IT & Software—Application & INTU & Intuit Inc. \\
IT & Information Technology Services & IT & Gartner, Inc. \\
IT & Data Processing \& Outsourced Services & JKHY & Jack Henry \& Associates, Inc. \\
IT & Communication Equipment & JNPR & Juniper Networks, Inc. \\
IT & Semiconductor Equipment \& Materials & KLAC & KLA Corporation \\
IT & Semiconductor Equipment \& Materials & LRCX & Lam Research Corporation \\
IT & Semiconductors & MCHP & Microchip Technology Incorporated \\
IT & Semiconductors & MPWR & Monolithic Power Systems, Inc. \\
IT & Software—Infrastructure & MSFT & Microsoft Corporation \\
IT & Communication Equipment & MSI & Motorola Solutions, Inc. \\
IT & Semiconductors & MU & Micron Technology, Inc. \\
IT & Computer Hardware & NTAP & NetApp, Inc. \\
IT & Semiconductors & NVDA & NVIDIA Corporation \\
IT & Semiconductors & NXPI & NXP Semiconductors N.V. \\
IT & Semiconductors & ON & ON Semiconductor Corporation \\
IT & Software—Infrastructure & ORCL & Oracle Corporation \\
IT & Staffing \& Employment Services & PAYX & Paychex, Inc. \\
IT & Software—Application & PTC & PTC Inc. \\
IT & Semiconductors & QCOM & QUALCOMM Incorporated \\
IT & Software—Infrastructure & SNPS & Synopsys, Inc. \\
IT & Computer Hardware & STX & Seagate Technology Holdings plc \\
IT & Semiconductors & SWKS & Skyworks Solutions, Inc. \\
IT & Electronic Components & TEL & TE Connectivity Ltd. \\
IT & Semiconductor Equipment \& Materials & TER & Teradyne, Inc. \\
IT & Electronic Equipment \& Instruments & TRMB & Trimble Inc. \\
IT & Semiconductors & TXN & Texas Instruments Incorporated \\
IT & Software—Application & TYL & Tyler Technologies, Inc. \\
IT & Software—Infrastructure & VRSN & VeriSign, Inc. \\
IT & Computer Hardware & WDC & Western Digital Corporation \\
IT & Electronic Equipment \& Instruments & ZBRA & Zebra Technologies Corporation \\
MT & Specialty Chemicals & ALB & Albemarle Corporation \\
MT & Specialty Chemicals & APD & Air Products and Chemicals, Inc. \\
MT & Packaging \& Containers & AVY & Avery Dennison Corporation \\
MT & Metal, Glass \& Plastic Containers & BALL & Ball Corporation \\
MT & Specialty Chemicals & CE & Celanese Corporation \\
MT & Agricultural Inputs & CF & CF Industries Holdings, Inc. \\
MT & Specialty Chemicals & DD & DuPont de Nemours, Inc. \\
MT & Specialty Chemicals & ECL & Ecolab Inc. \\
MT & Specialty Chemicals & EMN & Eastman Chemical Company \\
MT & Copper & FCX & Freeport-McMoRan Inc. \\
MT & Agricultural Inputs & FMC & FMC Corporation \\
MT & Specialty Chemicals & IFF & International Flavors \& Fragrances Inc. \\
MT & Packaging \& Containers & IP & International Paper Company \\
MT & Industrial Gases & LIN & Linde plc \\
MT & Specialty Chemicals & LYB & LyondellBasell Industries N.V. \\
MT & Building Materials & MLM & Martin Marietta Materials, Inc. \\
MT & Agricultural Inputs & MOS & The Mosaic Company \\
MT & Gold & NEM & Newmont Corporation \\
MT & Steel & NUE & Nucor Corporation \\
MT & Packaging \& Containers & PKG & Packaging Corporation of America \\
MT & Specialty Chemicals & PPG & PPG Industries, Inc. \\
MT & Packaging \& Containers & SEE & Sealed Air Corporation \\
MT & Specialty Chemicals & SHW & The Sherwin-Williams Company \\
MT & Steel & STLD & Steel Dynamics, Inc. \\
MT & Building Materials & VMC & Vulcan Materials Company \\
UT & Utilities—Regulated Electric & AEE & Ameren Corporation \\
UT & Utilities—Regulated Electric & AEP & American Electric Power Company, Inc. \\
UT & Utilities—Diversified & AES & The AES Corporation \\
UT & Gas Utilities & ATO & Atmos Energy Corporation \\
UT & Water Utilities & AWK & American Water Works Company, Inc. \\
UT & Utilities—Regulated Electric & CMS & CMS Energy Corporation \\
UT & Utilities—Regulated Electric & CNP & CenterPoint Energy, Inc. \\
UT & Utilities—Regulated Electric & D & Dominion Energy, Inc. \\
UT & Utilities—Regulated Electric & DTE & DTE Energy Company \\
UT & Utilities—Regulated Electric & DUK & Duke Energy Corporation \\
UT & Utilities—Regulated Electric & ED & Consolidated Edison, Inc. \\
UT & Utilities—Regulated Electric & EIX & Edison International \\
UT & Utilities—Regulated Electric & ES & Eversource Energy \\
UT & Utilities—Regulated Electric & ETR & Entergy Corporation \\
UT & Utilities—Regulated Electric & EVRG & Evergy, Inc. \\
UT & Utilities—Regulated Electric & EXC & Exelon Corporation \\
UT & Utilities—Regulated Electric & FE & FirstEnergy Corp. \\
UT & Utilities—Regulated Electric & LNT & Alliant Energy Corporation \\
UT & Utilities—Regulated Electric & NEE & NextEra Energy, Inc. \\
UT & Utilities—Regulated Gas & NI & NiSource Inc. \\
UT & Utilities—Independent Power Producers & NRG & NRG Energy, Inc. \\
UT & Utilities—Regulated Electric & PCG & PG\&E Corporation \\
UT & Utilities—Regulated Electric & PEG & Public Service Enterprise Group Incorporated \\
UT & Utilities—Regulated Electric & PNW & Pinnacle West Capital Corporation \\
UT & Utilities—Regulated Electric & PPL & PPL Corporation \\
UT & Utilities—Regulated Electric & SO & The Southern Company \\
UT & Utilities—Diversified & SRE & Sempra \\
UT & Utilities—Regulated Electric & WEC & WEC Energy Group, Inc. \\
UT & Utilities—Regulated Electric & XEL & Xcel Energy Inc. \\

\end{longtable}
\endgroup

\section{Market \textit{Crash} Dates}
\index{Crash!dates}
Dates associated with episodes of \textit{crash} or stress in financial markets are listed. These references were used to contextualize the evolution of the correlation matrices and their associated indicators in the empirical analysis.

\begin{table}[ht]
  \centering
  \small
  \begin{tabular}{|c|>{\centering\arraybackslash}m{10cm}|}
    \hline
    \rowcolor{teal!80}
    \cabecera{Date} & 
    \cabecera{Event} \\
    \hline
    2015-08-18 & 2015--16 Chinese stock market crash \\
    \hline
    2018-09-20 & 2018 cryptocurrency crash \\
    \hline
    2020-02-24 & 2020 stock market crash \\
    \hline
    2022-01-03 & 2022 stock market decline \\
    \hline
  \end{tabular}
  \caption[Market \textit{crash} dates]{Dates and names of \textit{crash} or stress episodes used as a temporal reference.}
  \label{tab:fechas_episodios}
\end{table}

\section{Leading Companies of the Dominant Eigenvectors \texorpdfstring{$[v_i^2]_{i=1}^N$}{[vi²]}}

The $10$ companies with the maximum values from Fig.\ref{fig: EigVect cuadrado q20} are presented below.
\begin{table}[ht]
  \centering
  \small
  \begin{tabular}{|c|c|c|r|}
    \hline
    \rowcolor{teal!80}
    \cabecera{Position} &
    \cabecera{Ticker} &
    \cabecera{Sector} &
    \cabecera{Value} \\
    \hline
    1  & WTW  & FN & 0.004294 \\
    2  & AMGN & HC & 0.004002 \\
    3  & MOH  & HC & 0.003996 \\
    4  & CTRA & EG & 0.003955 \\
    5  & PFE  & HC & 0.003954 \\
    6  & DVN  & EG & 0.003950 \\
    7  & GIS  & CS & 0.003915 \\
    8  & TECH & HC & 0.003879 \\
    9  & GILD & HC & 0.003874 \\
    10 & ABT  & HC & 0.003845 \\
    \hline
  \end{tabular}
  \caption[Leading Companies in the $v_i^2$, $q=20$]{Leading companies in the entries $[v_i^2]_{i=1}^N$ for a window of $q=20$ days. The companies coincide with those of Table \ref{tab:top10-vn2-q40}, although in a different order.}
  \label{tab:top10-vn2-q20}
\end{table}

\section{Complementary Analysis of IPR and PR}
\label{sec:ipr-pr-q40-k4}

This section presents the values of $\boldsymbol{\mathrm{IPR}_{N}}$ and $\boldsymbol{\mathrm{PR}_{N}}$ associated with the dominant eigenvector for the case $q=40$ and $k=4$. In agreement with Fig.~\ref{fig:vn2-q40-k4}, \emph{cluster 1} exhibits a greater concentration of participation (higher $\boldsymbol{\mathrm{IPR}_{N}}$ and, therefore, lower $\boldsymbol{\mathrm{PR}_{N}}$) compared with the other clusters.

\vspace{0.35cm}

\begin{table}[htbp]
\centering
\begin{tabular}{|c|c|c|}
\hline
\rowcolor{teal!80}
\cabecera{\emph{Cluster}} &
\cabecera{$\boldsymbol{\mathrm{IPR}_{N}}$} &
\cabecera{$\boldsymbol{\mathrm{PR}_{N}}$} \\
\hline
\textbf{1} & 0.00345 & 289.9818 \\
\hline
\textbf{2} & 0.00291 & 343.8932 \\
\hline
\textbf{3} & 0.00276 & 362.9460 \\
\hline
\textbf{4} & 0.00249 & 402.1020 \\
\hline
\end{tabular}
\caption{Values of $\boldsymbol{\mathrm{IPR}_{N}}$ and $\boldsymbol{\mathrm{PR}_{N}}$ by \emph{cluster}, $q=40$ and $k=4$.}
\label{tab:ipr-pr-q40-k4}
\end{table}

\section{Complementary Analysis of the Leading Companies Indicated by the \emph{Clustering} of the Dominant Eigenvectors \texorpdfstring{$[v_i^2]_{i=1}^N$}{[vi²]}}

\label{sec:analisis-complementario-principales}
Table~\ref{tab:Top10_q40_k4} shows the 10 leading companies identified by the temporal averages of $[v_i^2]_{i=1}^N$ for $q=40$ and $k=4$, complementing Figure~\ref{fig:vn2-q40-k4}. 
\emph{Cluster} 1 is dominated exclusively by the \emph{HC} sector. 
The intersection between the 10 leading companies of each cluster is empty, although the company \emph{WTW} appears consistently in all the \emph{clusters} except the first, with high values.

It is worth highlighting that the first 7 companies of \emph{cluster} 3 (Table~\ref{tab:Top10_q40_k4}) coincide exactly with those of \emph{cluster} 2 of Table~\ref{tab:Top10_q40_k5}, which indicates that this group remains consistent when varying the number of \emph{clusters}. 
Likewise, \emph{cluster} 2 of Table~\ref{tab:Top10_q40_k4} reproduces the same order and values of \emph{cluster} 3 of Table~\ref{tab:Top10_q40_k5}.

These results confirm two consistent findings:
\begin{enumerate}[label=\roman*)]
    \item The \emph{Health Care} sector exclusively dominates the first cluster in both configurations.
    \item No company appears simultaneously in all the \emph{clusters}, although \emph{WTW} maintains a consistent presence in multiple groups.
\end{enumerate}

\begin{table}[ht]
\centering

\begin{subtable}[t]{0.48\linewidth}
\centering
\scriptsize
\begin{tabular}{|c|c|c|r|}
    \hline
    \rowcolor{teal!80}
    \cabecera{Position} &
    \cabecera{Ticker} &
    \cabecera{Sector} &
    \cabecera{Value} \\
    \hline
    1  & LH   & HC & 0.007271 \\
    2  & HCA  & HC & 0.007138 \\
    3  & PFE  & HC & 0.007048 \\
    4  & A    & HC & 0.006934 \\
    5  & DXCM & HC & 0.006840 \\
    6  & HUM  & HC & 0.006803 \\
    7  & BSX  & HC & 0.006794 \\
    8  & BDX  & HC & 0.006582 \\
    9  & DGX  & HC & 0.006531 \\
    10 & REGN & HC & 0.006527 \\
    \hline
\end{tabular}
\caption*{a) \emph{cluster} 1}
\end{subtable}\hfill
\begin{subtable}[t]{0.48\linewidth}
\centering
\scriptsize
\begin{tabular}{|c|c|c|r|}
    \hline
    \rowcolor{teal!80}
    \cabecera{Position} &
    \cabecera{Ticker} &
    \cabecera{Sector} &
    \cabecera{Value} \\
    \hline
    1  & WTW  & FN & 0.005238 \\
    2  & MOH  & HC & 0.005056 \\
    3  & PFE  & HC & 0.004993 \\
    4  & DHR  & HC & 0.004847 \\
    5  & V    & FN & 0.004793 \\
    6  & INCY & HC & 0.004737 \\
    7  & PM   & CS & 0.004733 \\
    8  & PEP  & CS & 0.004713 \\
    9  & HSIC & HC & 0.004681 \\
    10 & CTRA & EG & 0.004620 \\
    \hline
\end{tabular}
\caption*{b)\emph{cluster} 2}
\end{subtable}

\vspace{0.6em}

\begin{subtable}[t]{0.48\linewidth}
\centering
\scriptsize
\begin{tabular}{|c|c|c|r|}
    \hline
    \rowcolor{teal!80}
    \cabecera{Position} &
    \cabecera{Ticker} &
    \cabecera{Sector} &
    \cabecera{Value} \\
    \hline
    1  & AMGN & HC & 0.004623 \\
    2  & CTRA & EG & 0.004588 \\
    3  & TECH & HC & 0.004580 \\
    4  & WTW  & FN & 0.004441 \\
    5  & ACGL & FN & 0.004389 \\
    6  & GIS  & CS & 0.004380 \\
    7  & DLR  & FN & 0.004357 \\
    8  & DVN  & EG & 0.004293 \\
    9  & CVX  & EG & 0.004240 \\
    10 & PM   & CS & 0.004149 \\
    \hline
\end{tabular}
\caption*{c)\emph{cluster} 3}
\end{subtable}\hfill
\begin{subtable}[t]{0.48\linewidth}
\centering
\scriptsize
\begin{tabular}{|c|c|c|r|}
    \hline
    \rowcolor{teal!80}
    \cabecera{Position} &
    \cabecera{Ticker} &
    \cabecera{Sector} &
    \cabecera{Value} \\
    \hline
    1  & WTW  & FN & 0.004167 \\
    2  & AMGN & HC & 0.003912 \\
    3  & ABT  & HC & 0.003863 \\
    4  & DVN  & EG & 0.003813 \\
    5  & GILD & HC & 0.003809 \\
    6  & TECH & HC & 0.003809 \\
    7  & CTRA & EG & 0.003709 \\
    8  & JNJ  & HC & 0.003707 \\
    9  & MOH  & HC & 0.003693 \\
    10 & GIS  & CS & 0.003669 \\
    \hline
\end{tabular}
\caption*{d)\emph{cluster} 4}
\end{subtable}

\caption[Leading companies by \emph{cluster}, \texorpdfstring{$q=40\;\text{y}\; k=4$}{q=40, k=4}]{Leading companies indicated by the entries $[v_i^2]_{i=1}^{N}$ for $q=40$ and $k=4$. It stands out that, in \emph{cluster} 1, the $10$ leading contributions belong to the \emph{HC} sector.}
\label{tab:Top10_q40_k4}
\end{table}
\chapter{Complementary Figures}

\label{section:appendix3}

\section{Heatmaps}
\index{Heatmaps}
\subsection{Correlation Matrix}
Fig. \ref{fig: Mapas de Calor apendices} shows the heatmaps obtained for $q=20,\;q=40$, $k=4$. In both cases, the same pattern as in Fig. \ref{fig:avg_corr_mats} is observed: \emph{cluster} 1 concentrates the lowest values (including anticorrelations) and, as the \emph{cluster} index increases, the intensity of the correlation increases; the highest-index \emph{cluster} presents the highest average correlations.

\subsection{Matrices \texorpdfstring{$\mathbf{C}^2$}{C²}}

Fig.~\ref{fig:MeanMatrices_C2_q40_k4} presents the average $\mathbf{C}^{2}$ matrices obtained by means of \emph{clustering} with the \textit{k}-Means algorithm for $q=40$ and $k=4$. In \emph{cluster} 1, delimited bands with correlations close to $0$ or even negative are observed. In \emph{cluster} 2, the anticorrelation regions disappear, whereas in \emph{cluster} 3 those blue bands associated with negative correlation values reappear. Finally, \emph{cluster} 4 concentrates the highest correlations, showing a much more homogeneous and positive structure. 

\subsection{Matrices \texorpdfstring{$\mathbf{C}^3$}{C³}}
As shown in Fig.~\ref{fig:MeanMatrices_C3_q40_k4}, the heatmaps obtained from the $\mathbf{C}^3$ matrices with $q=40$ and $k=4$ exhibit a behavior consistent with what was observed in Fig.~\ref{fig:MeanMatrices_C3_q40_k5}. 
In \emph{cluster} 1, blue and white bands are distinguished that reflect negative or weak correlations, which reappear with greater intensity in \emph{cluster} 3. 
The sectors are not clearly delimited and, although the maps of the $\mathbf{C}^3$ matrices bear similarity to those of $\mathbf{C}^2$, they present a relevant difference: in \emph{clusters} 2 and 3 white zones emerge that indicate an absence of correlation.

\begin{figure}[!htbp]
  \centering
  \begin{subfigure}{\textwidth}
    \centering
    \includegraphics[width=0.75\textwidth, height=0.46\textheight]{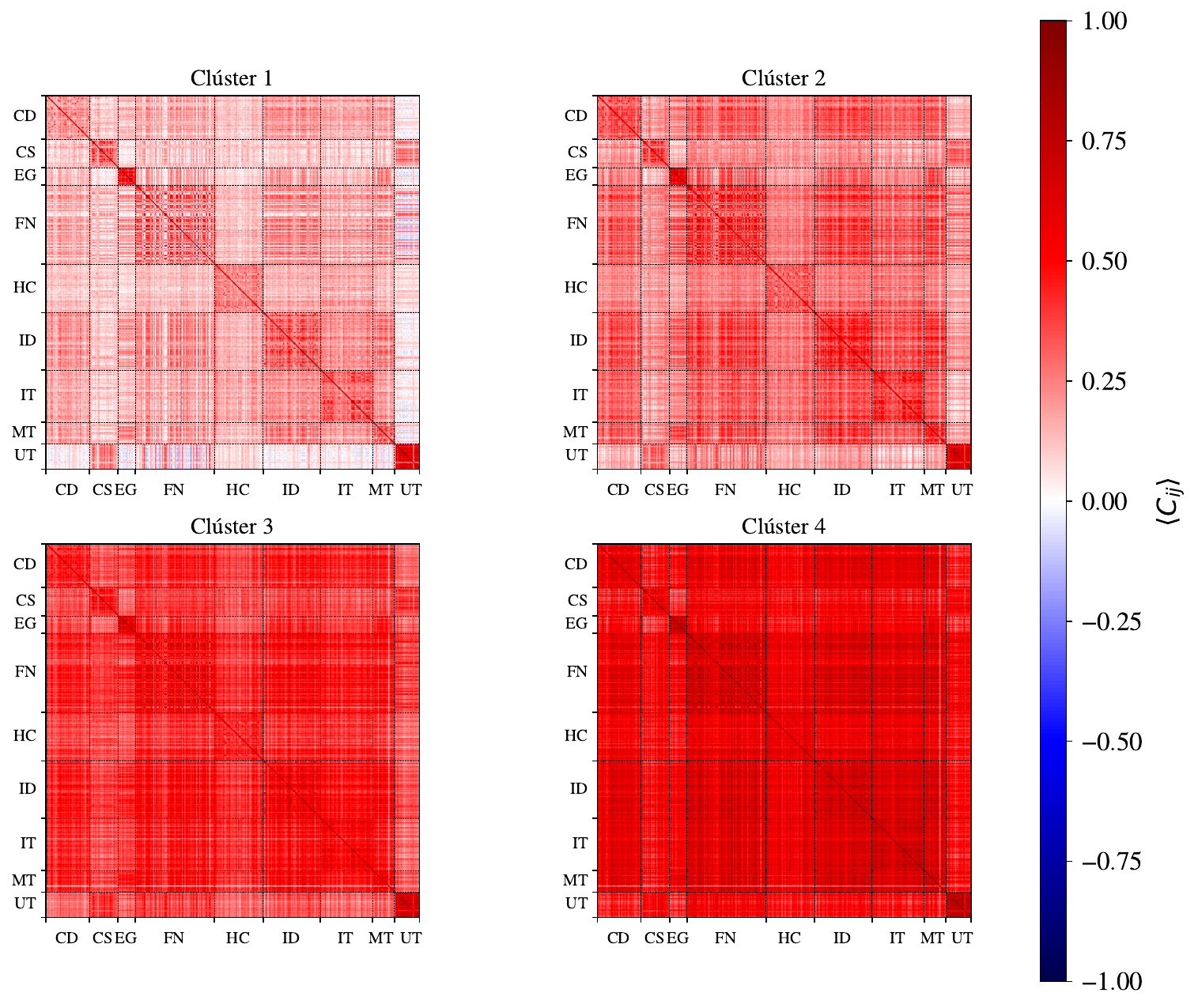}
    \caption{$q=20$, $k=4$.}
    \label{fig:heat-q20-k5}
  \end{subfigure}

  \medskip

  \begin{subfigure}{\textwidth}
    \centering
    \includegraphics[width=0.75\textwidth,height=0.64\textheight,keepaspectratio]{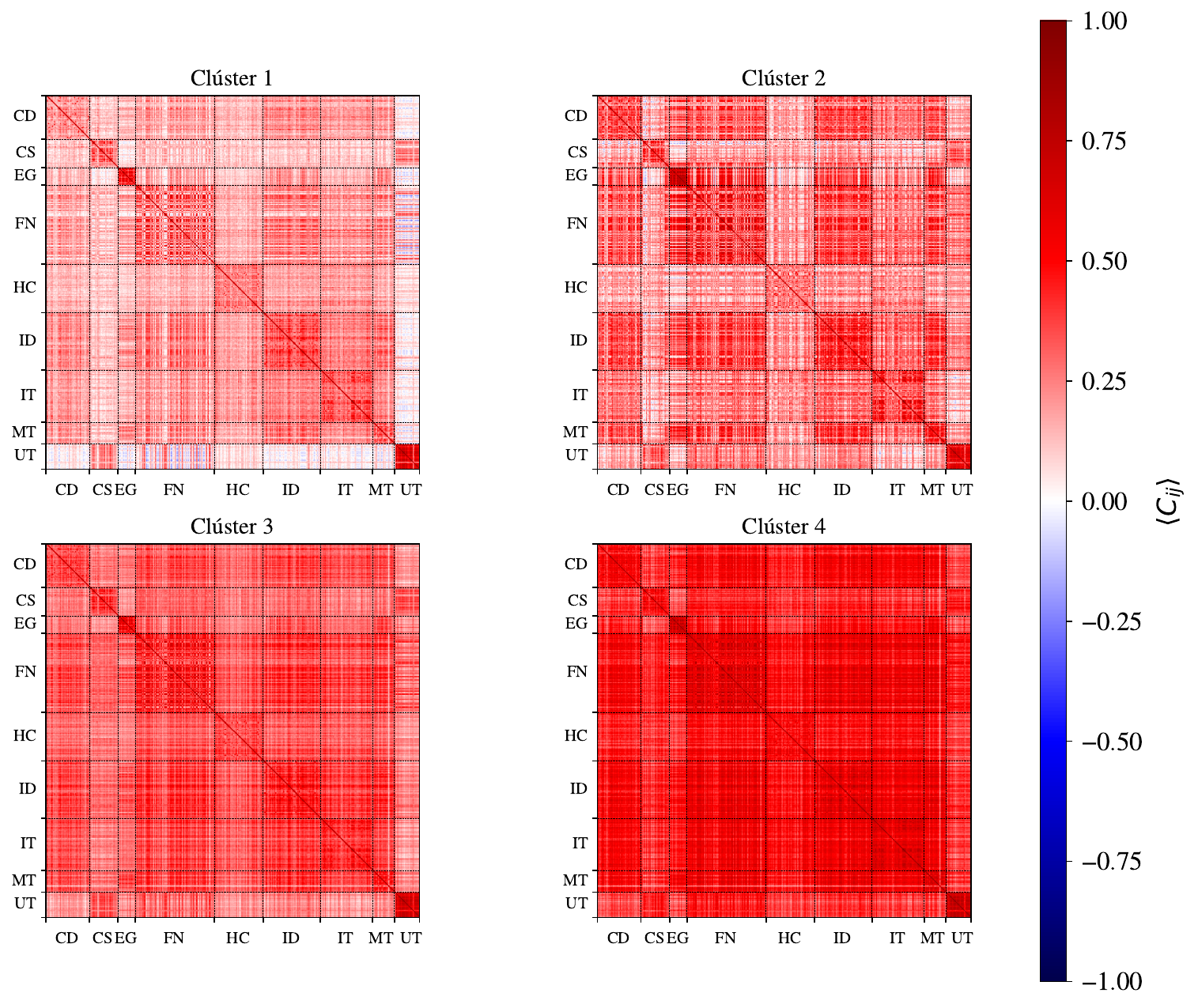}
    \caption{$q=40$, $k=4$.}
    \label{fig:heat-q40-k5}
  \end{subfigure}

  \caption[Complementary heatmaps of Correlation Matrices]{Heatmaps of Correlation Matrices. The highest-index \emph{cluster} concentrates the highest correlations.}
  \label{fig: Mapas de Calor apendices}
\end{figure}


\begin{figure}[t]
  \centering
  \includegraphics[width=\linewidth]{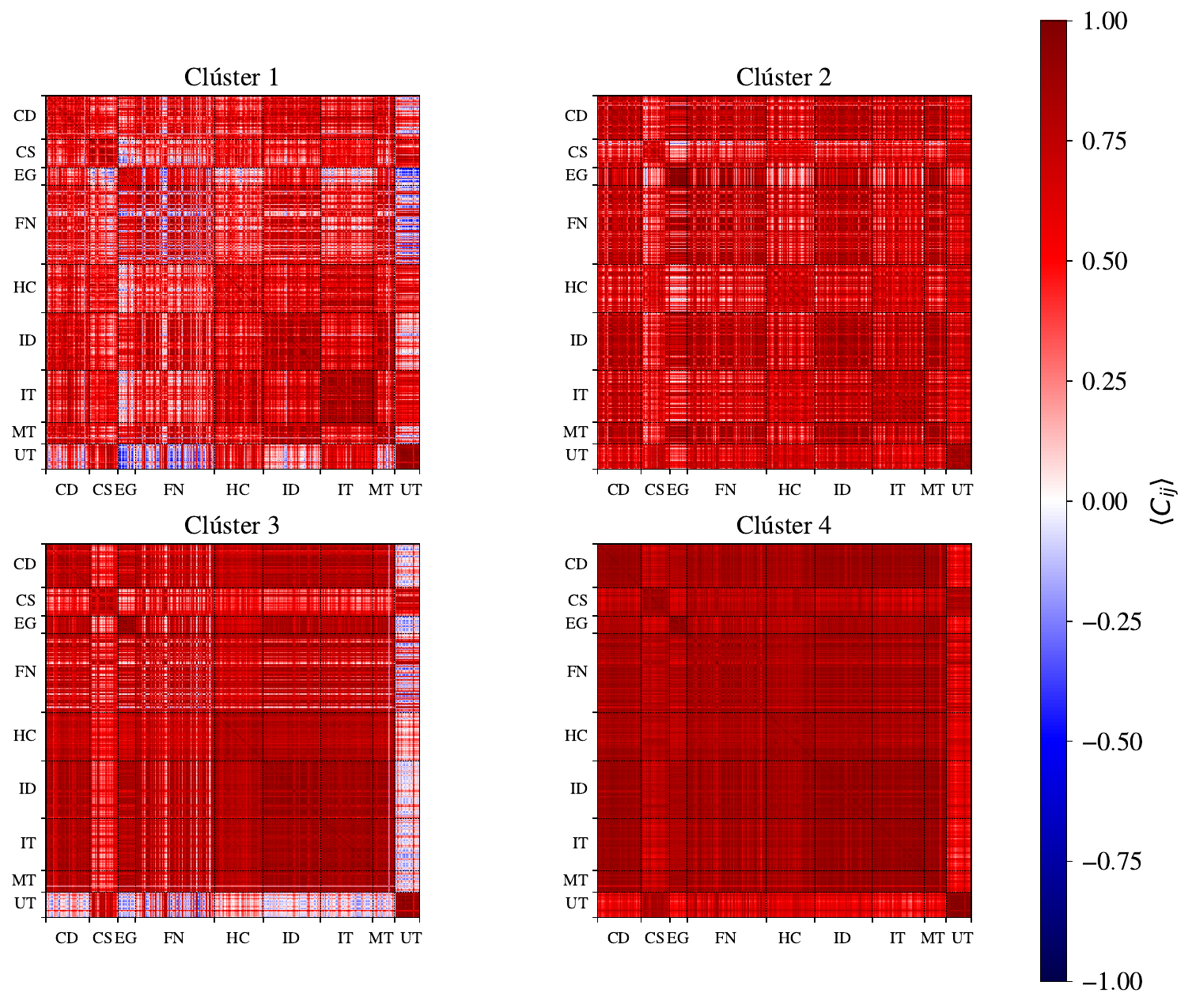}
  \caption[Heatmaps of the average \texorpdfstring{$C^2$}{C²} matrices by \emph{cluster}, \texorpdfstring{$q=40\;\text{y}\;k=4$}{q=40, k=4}]%
  {Heatmaps of the average $\mathbf{C^2}$ matrices obtained by \textit{k}-Means \emph{clustering} with $q=40$ and $k=4$. 
  In \emph{cluster} 1, delimited blocks with correlations close to $0$ or negative are observed. In \emph{cluster} 2, the blue anticorrelation bands disappear; 
  these reappear in \emph{cluster} 3.}
  \label{fig:MeanMatrices_C2_q40_k4}
\end{figure}

\begin{figure}[ht]
\centering
\includegraphics[width=\linewidth]{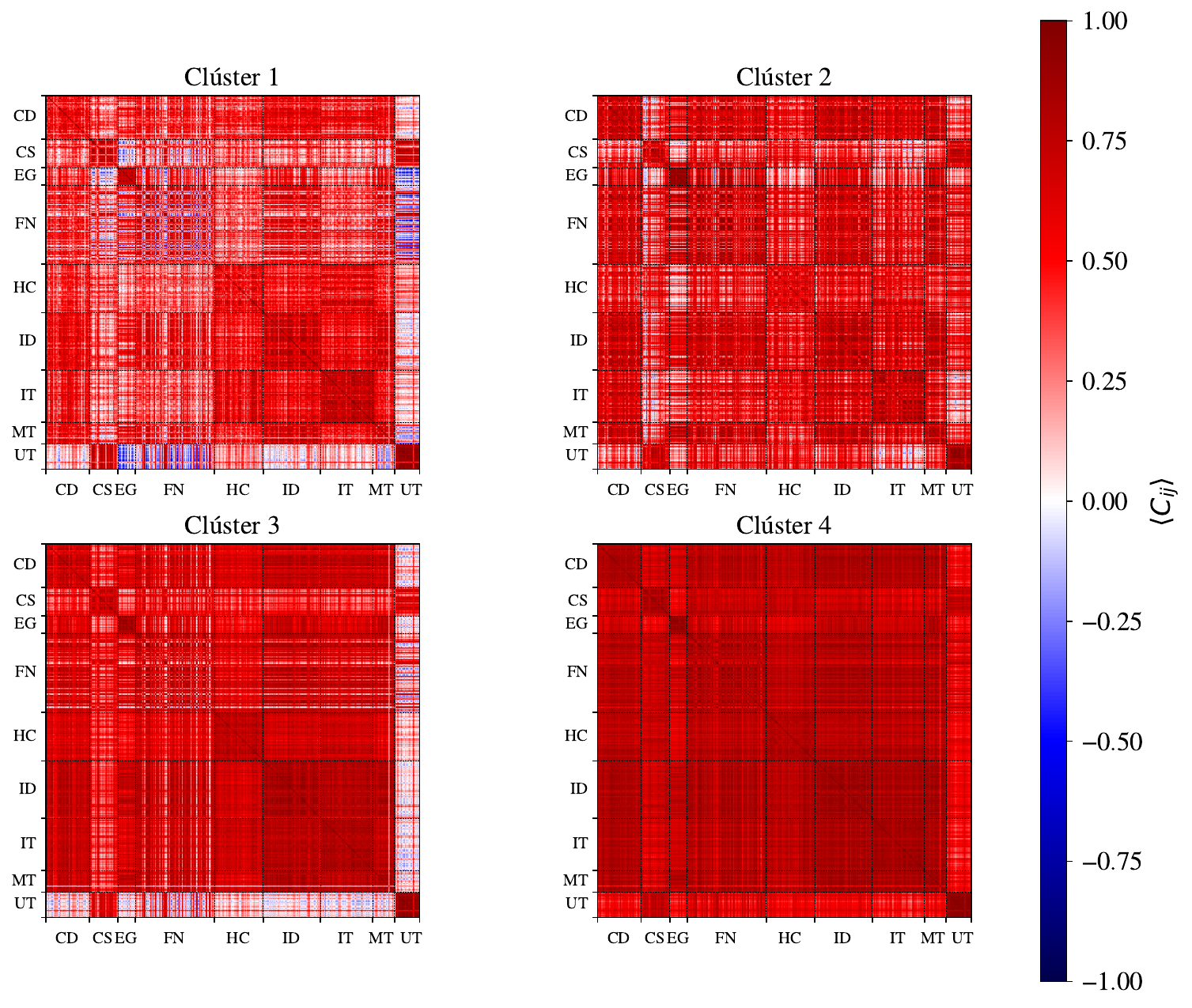}
\caption[Heatmaps of the average \texorpdfstring{$C^3$}{C³} matrices, \texorpdfstring{$q=40\;\text{y}\;k=4$}{q=40, k=4}]%
{Heatmaps of the average $\mathbf{C^3}$ matrices for $q=40$ days and $k=4$. 
The behavior is consistent with what was observed in Fig.~\ref{fig:MeanMatrices_C3_q40_k5}, since in \emph{cluster} 1 blue and white bands associated with weak or negative correlations are appreciated, whereas in \emph{cluster} 3 these reappear in a more marked manner. 
The sectors are not fully delimited and, unlike $\mathbf{C}^2$, the intensity of the red color is reduced in several \emph{clusters}, which indicates correlation values that are not necessarily at the extreme.}
\label{fig:MeanMatrices_C3_q40_k4}
\end{figure}

\section{Second and Third Largest Eigenvalues}
\label{sec: segundos y terceros eigen}
Next, in Fig.~\ref{fig:lambdaNmenos1_grid} and Fig.~\ref{fig:lambdaNmenos2_grid}, the results corresponding to the second ($\lambda_{N-1}$) and third ($\lambda_{N-2}$) leading eigenvalues are presented, organized according to the window size and the number of \emph{clusters}. Unlike what was observed for $\lambda_N$ in Fig.~\ref{fig:lambdaN_scatter_png}, after a crisis period, no significant increases in value are identified in these eigenvalues.
\index{Eigenvalue!penultimate}
\begin{figure}[ht]
\centering

\begin{subfigure}{0.49\linewidth}
  \centering
  \includegraphics[width=\linewidth]{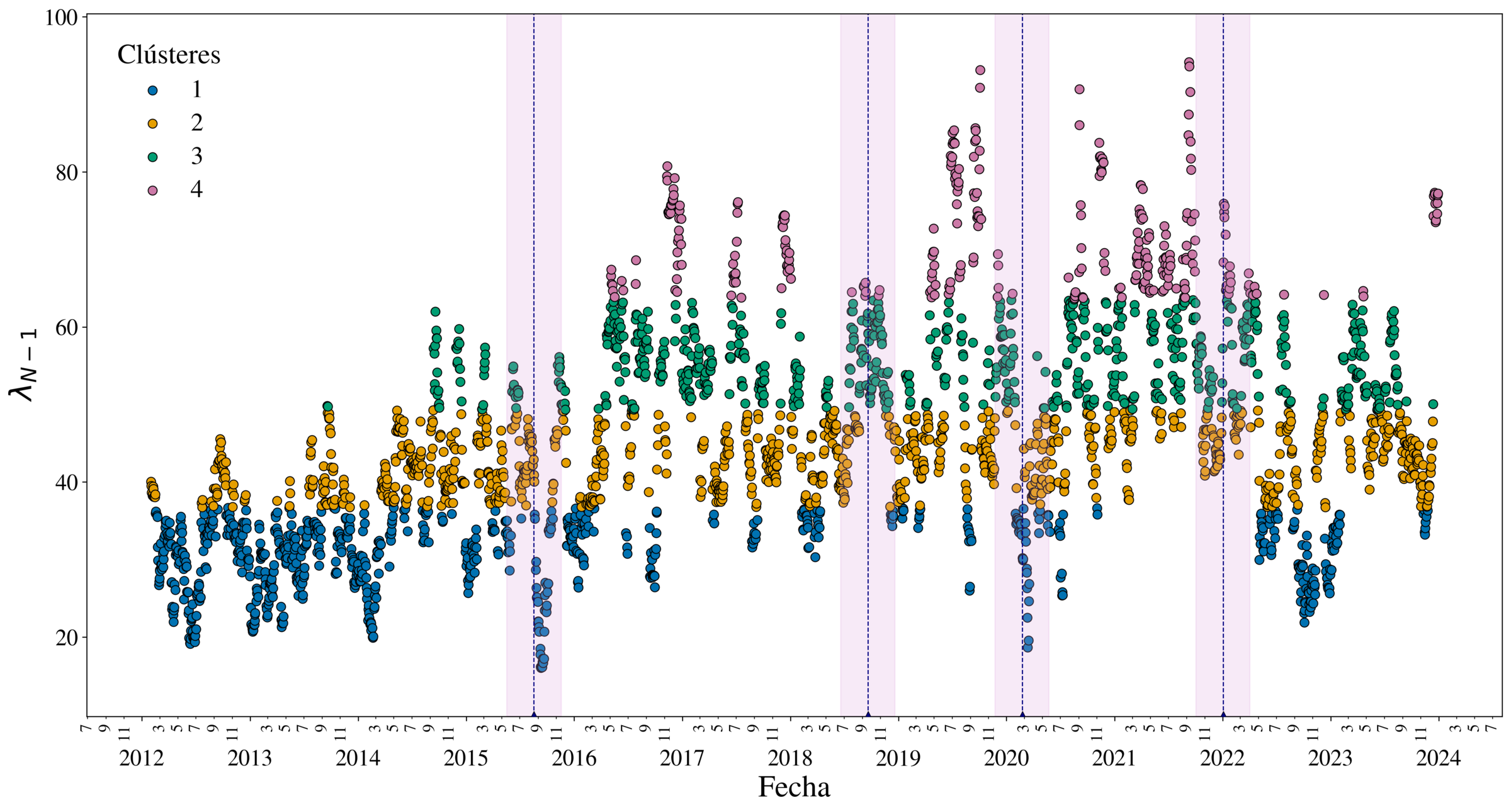}
  \caption*{a) $q=20$, $k=4$}
\end{subfigure}\hfill
\begin{subfigure}{0.49\linewidth}
  \centering
  \includegraphics[width=\linewidth]{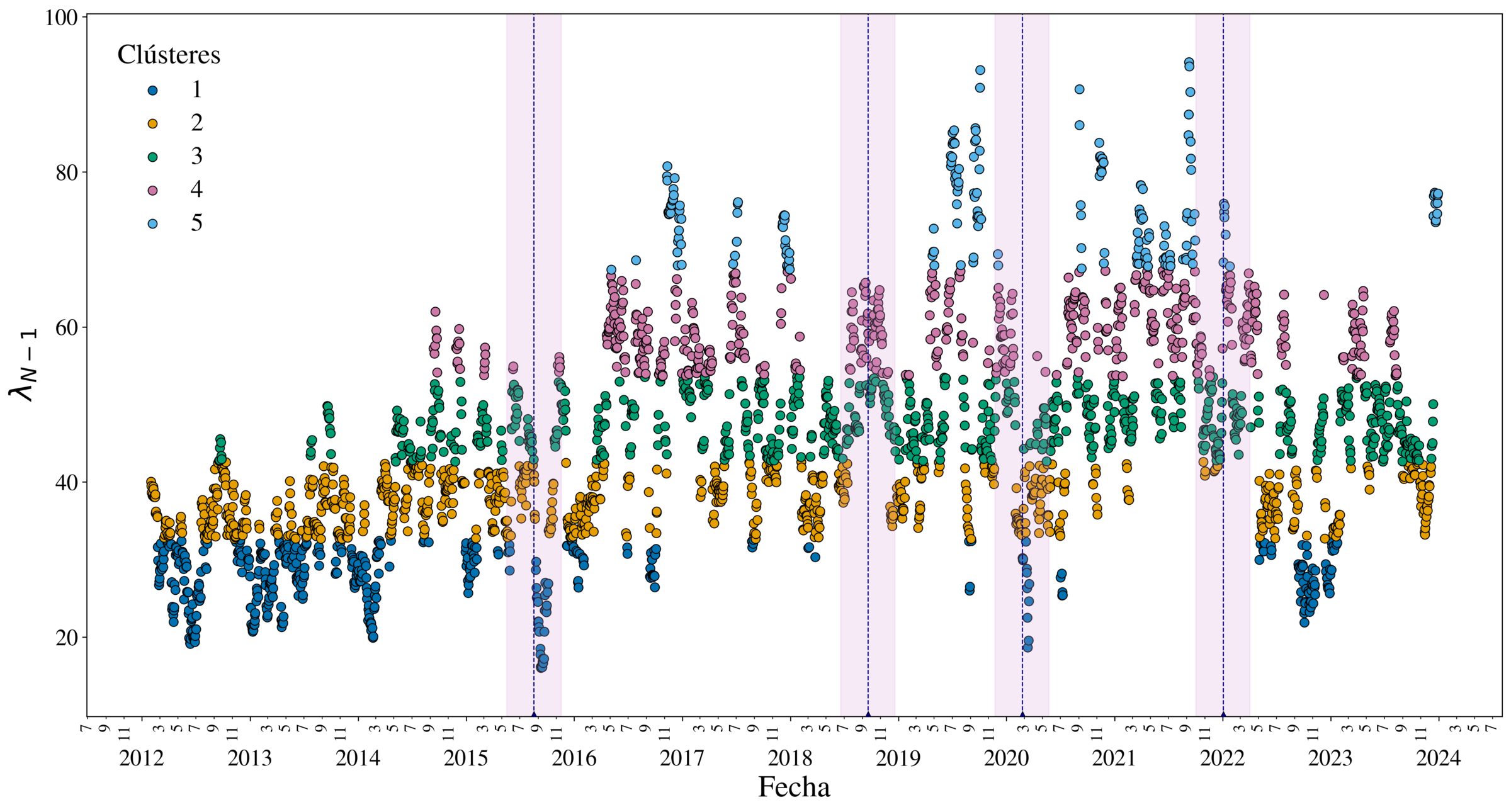}
  \caption*{b) $q=20$, $k=5$}
\end{subfigure}

\vspace{0.3em}

\begin{subfigure}{0.49\linewidth}
  \centering
  \includegraphics[width=\linewidth]{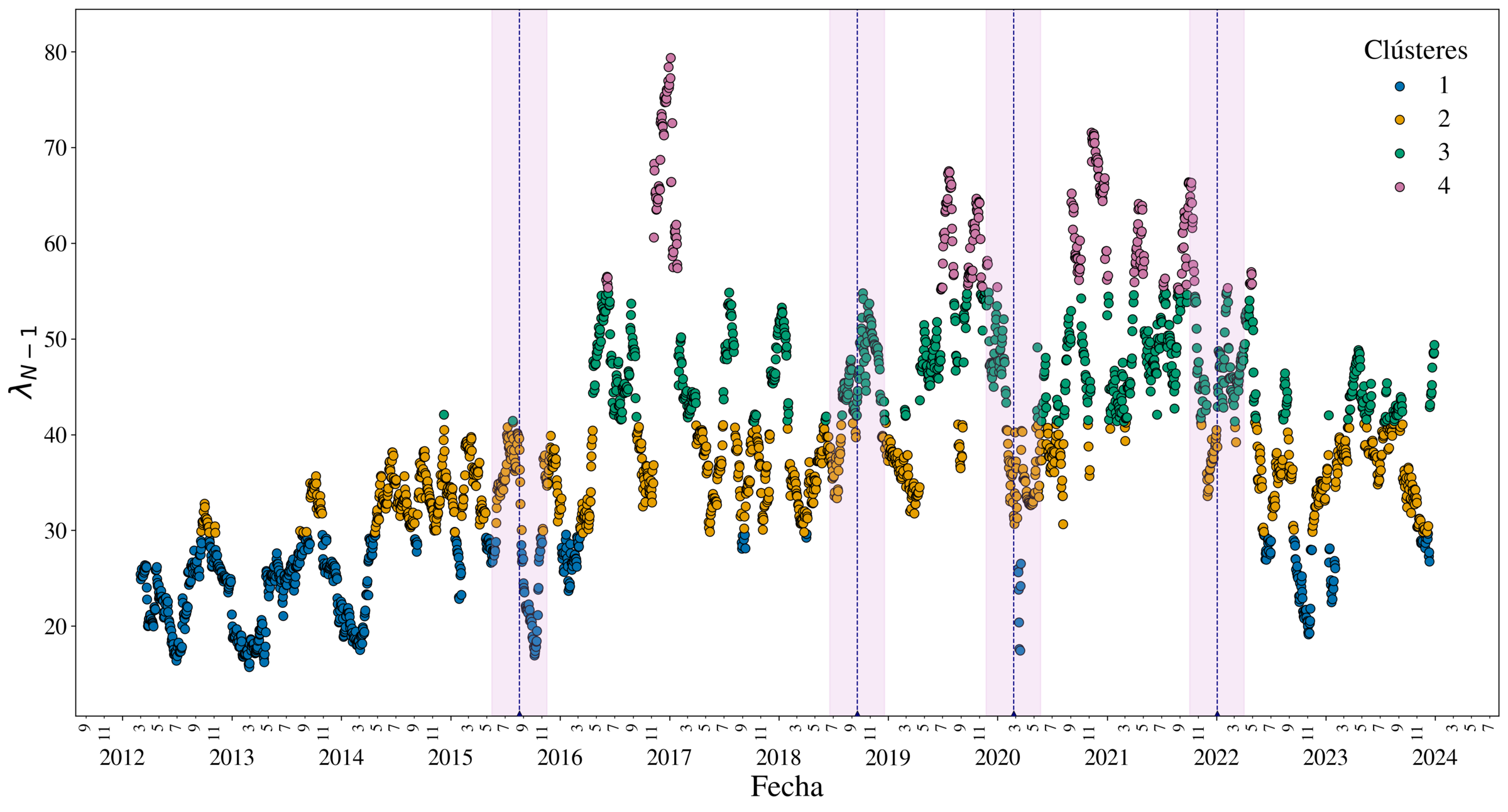}
  \caption*{c) $q=40$, $k=4$}
\end{subfigure}\hfill
\begin{subfigure}{0.49\linewidth}
  \centering
  \includegraphics[width=\linewidth]{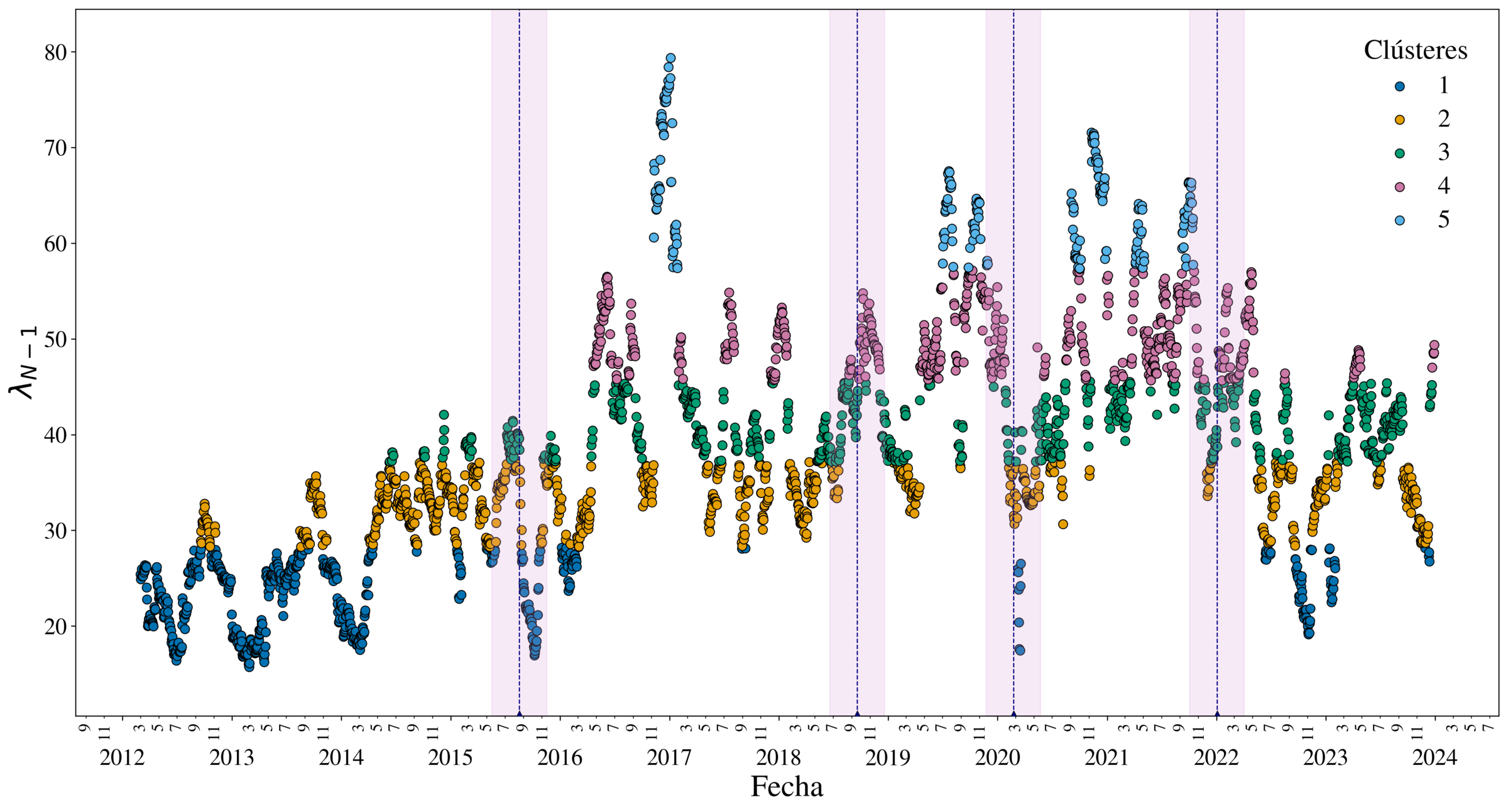}
  \caption*{d) $q=40$, $k=5$}
\end{subfigure}

\caption[Comparison of the second leading eigenvalues]{Comparison of $\lambda_{N-1}$ by \emph{cluster}. 
Top row: $q=20$; bottom row: $q=40$; columns: $k=4$ and $k=5$.}
\label{fig:lambdaNmenos1_grid}
\end{figure}

\index{Eigenvalue!antepenultimate}
\begin{figure}[ht]
\centering

\begin{subfigure}{0.49\linewidth}
  \centering
  \includegraphics[width=\linewidth]{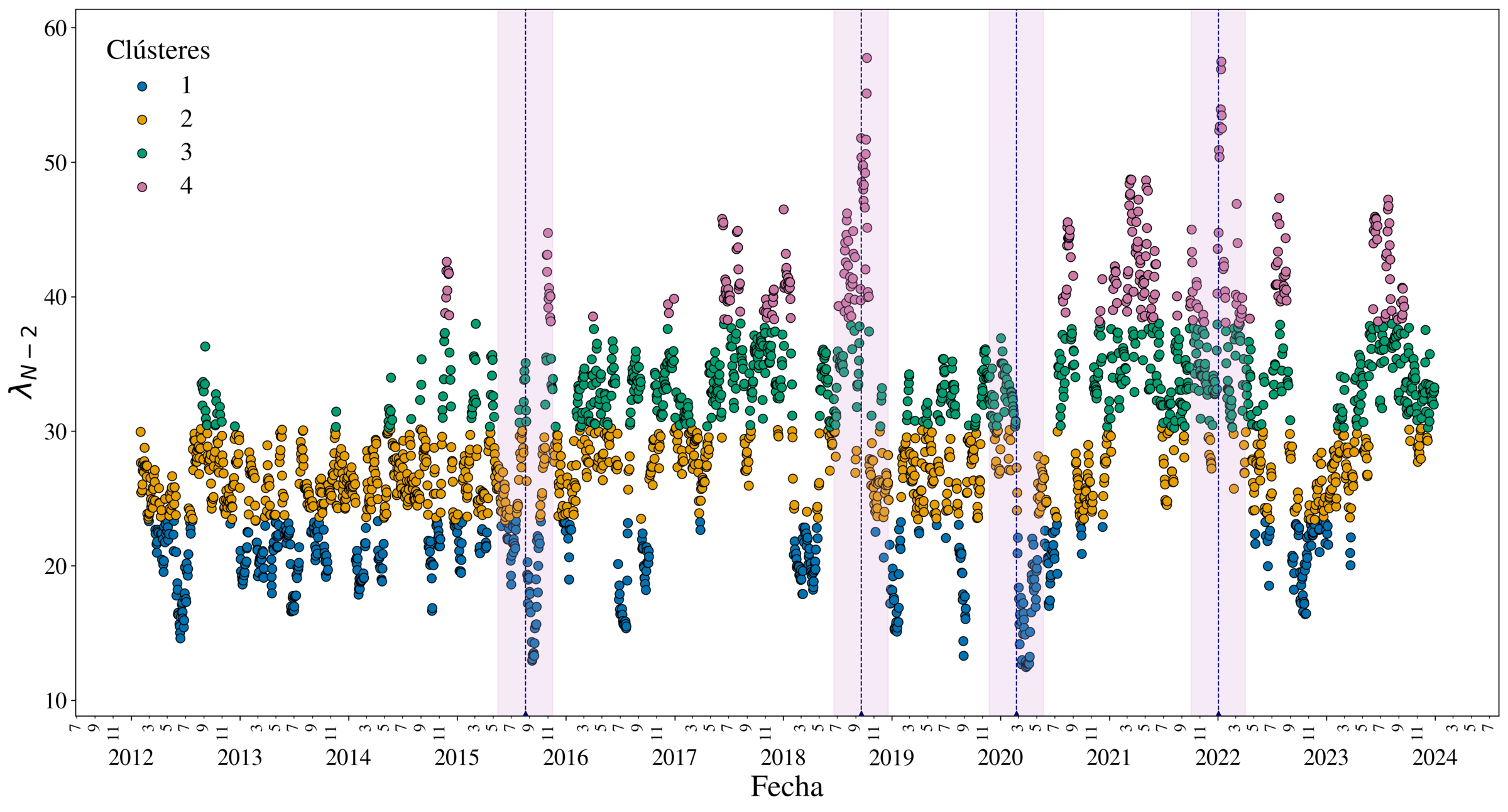}
  \caption*{a) $q=20$, $k=4$}
\end{subfigure}\hfill
\begin{subfigure}{0.49\linewidth}
  \centering
  \includegraphics[width=\linewidth]{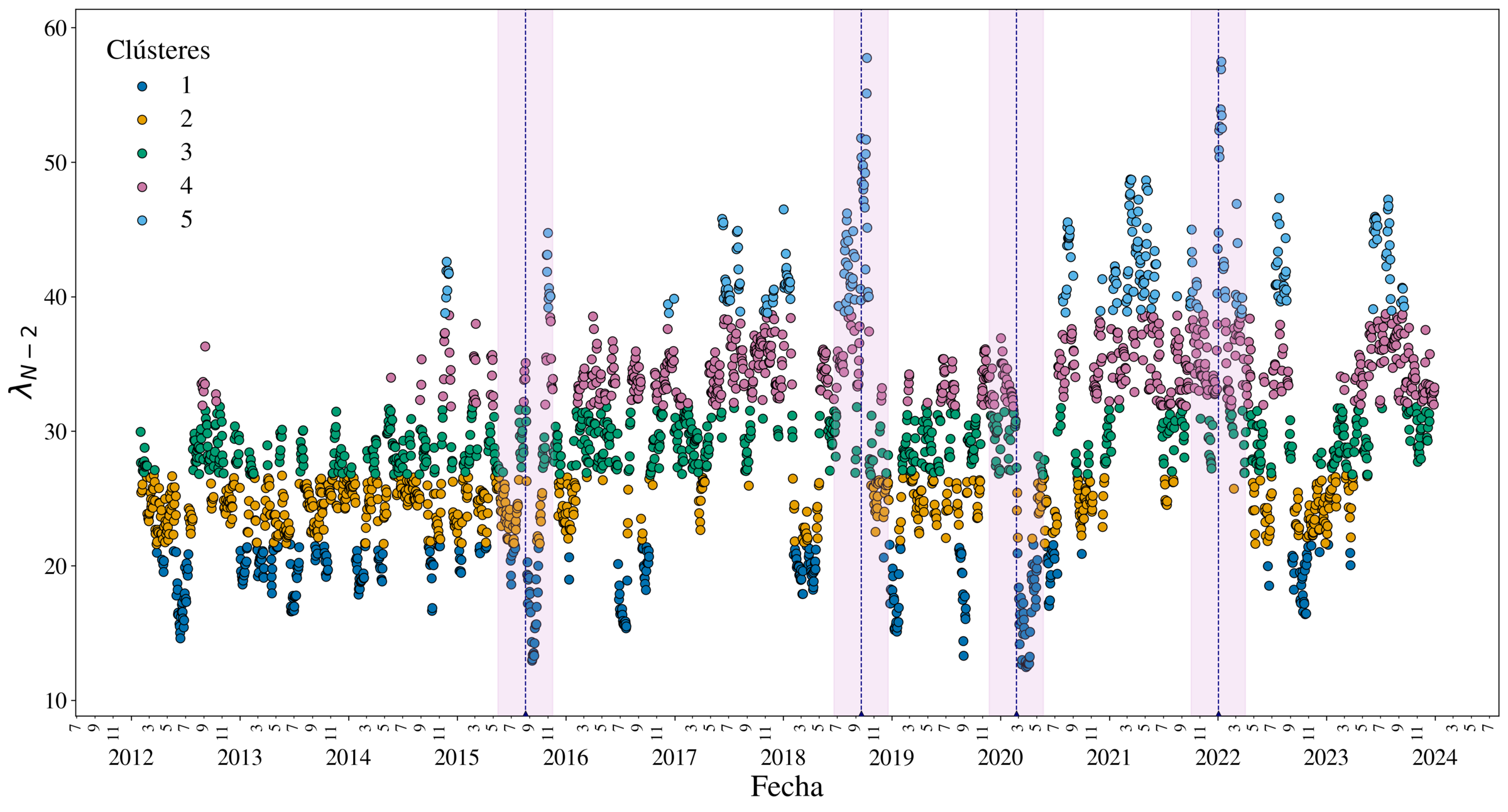}
  \caption*{b) $q=20$, $k=5$}
\end{subfigure}

\vspace{0.3em}

\begin{subfigure}{0.49\linewidth}
  \centering
  \includegraphics[width=\linewidth]{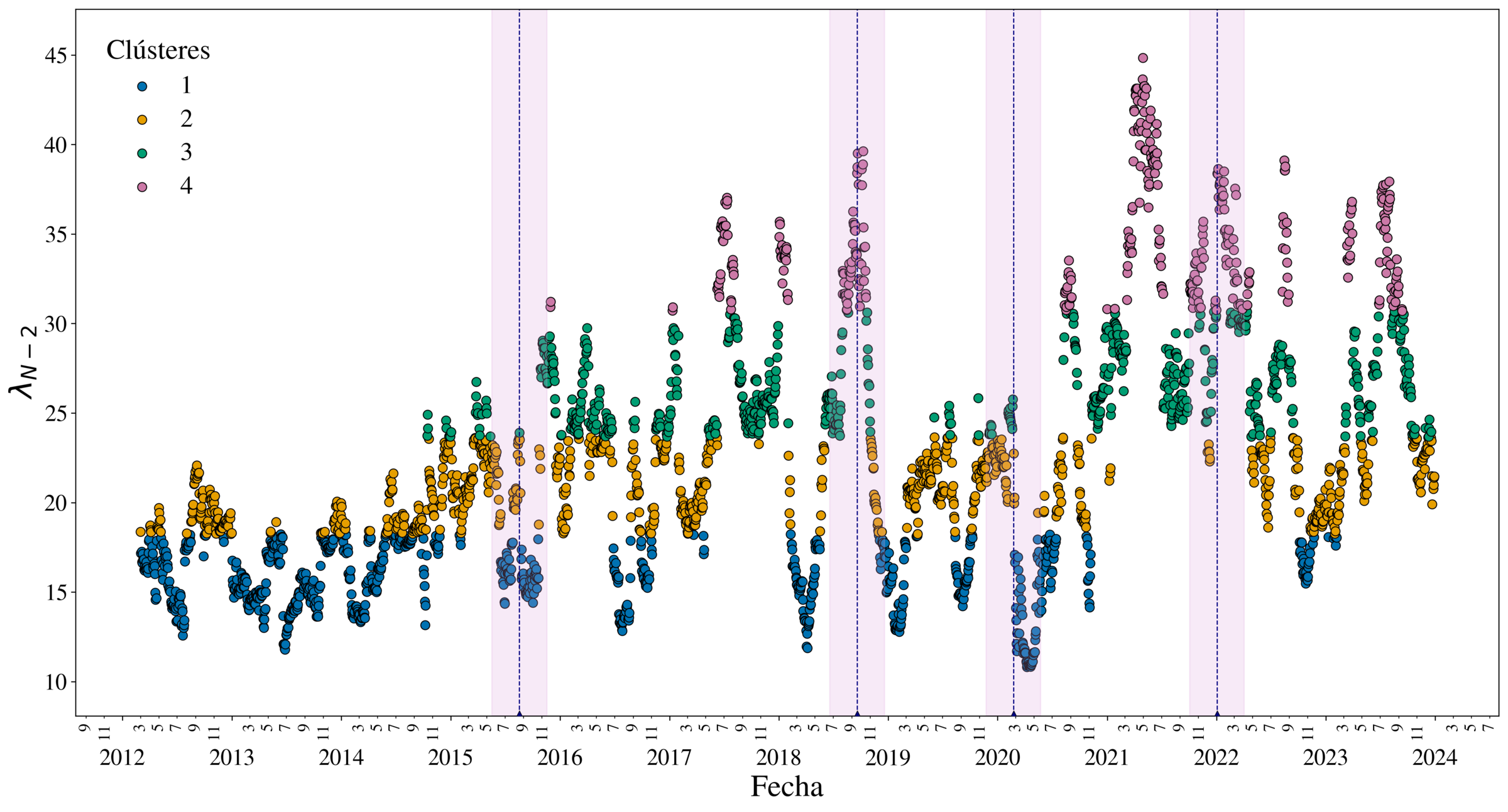}
  \caption*{c) $q=40$, $k=4$}
\end{subfigure}\hfill
\begin{subfigure}{0.49\linewidth}
  \centering
  \includegraphics[width=\linewidth]{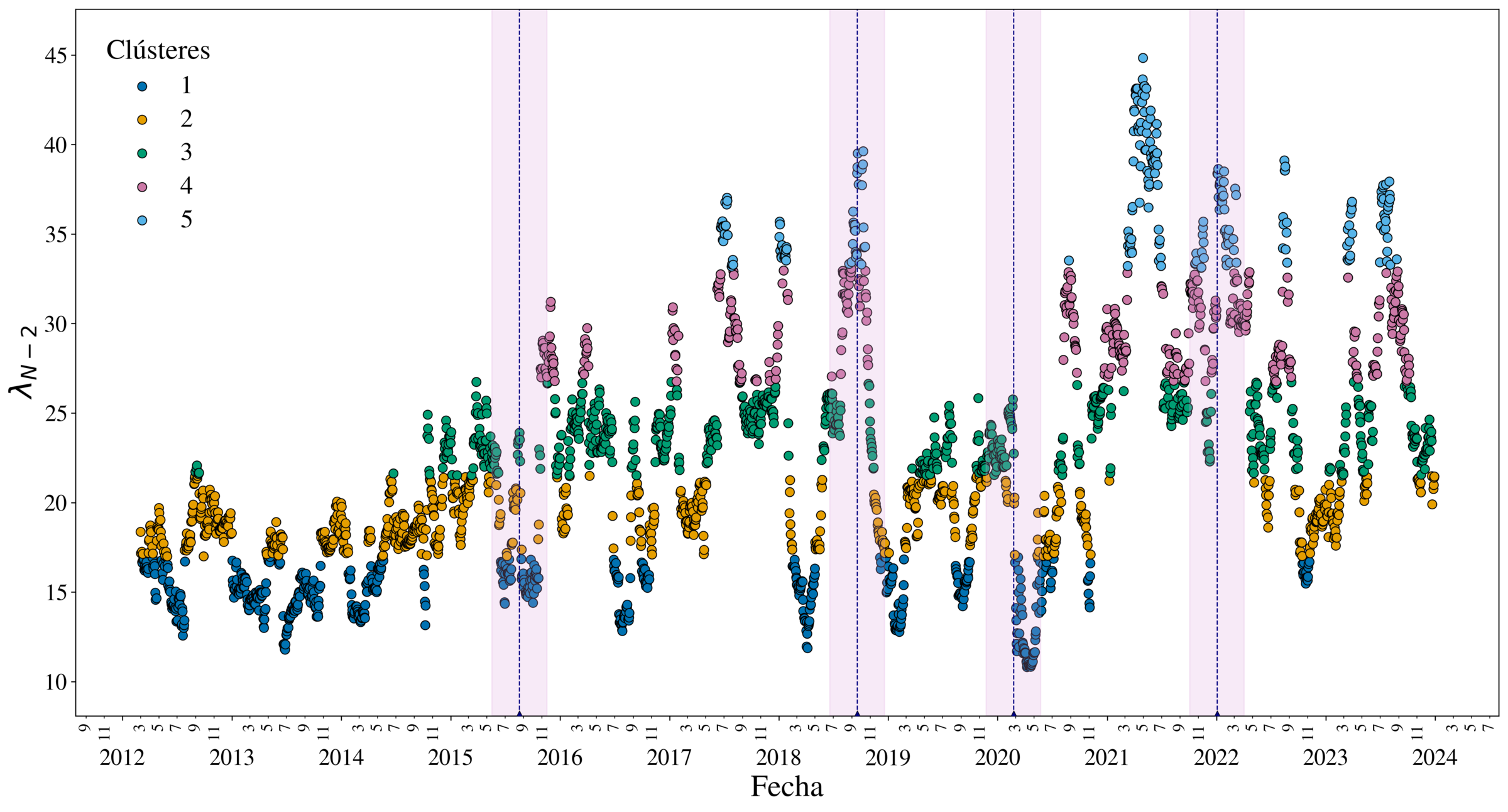}
  \caption*{d) $q=40$, $k=5$}
\end{subfigure}

\caption[Comparison of the third leading eigenvalues]{Comparison of $\lambda_{N-2}$ by \emph{cluster}. 
Top row: $q=20$; bottom row: $q=40$; columns: $k=4$ and $k=5$.}
\label{fig:lambdaNmenos2_grid}
\end{figure}

\clearpage

\section{Dominant Eigenvectors \texorpdfstring{$[v_i^2]_{i=1}^N$}{[vi²]}}
\index{Eigenvector!squared entries}
\index{Sectors!GICS}
Fig.~\ref{fig: EigVect cuadrado q20} corresponds to the case $q=20$ and preserves the same visual structure as Fig.~\ref{fig: EigVect cuadrado q40}.  
Each bar represents the temporal average of the \emph{relative participation weight} $w_i \equiv v_i^{2}$, where $v$ is the eigenvector associated with $\lambda_{N}$, computed over all the time windows. Due to the normalization of the eigenvector, it holds that $\sum_{i=1}^{N} w_i = 1$. For this case a value of $\boldsymbol{\mathrm{IPR}_{N}}=0.00254$ and $\boldsymbol{\mathrm{PR}_{N}}=394.1433$ was obtained.

\begin{figure}[htbp]
    \centering
    \includegraphics[width=\textwidth]{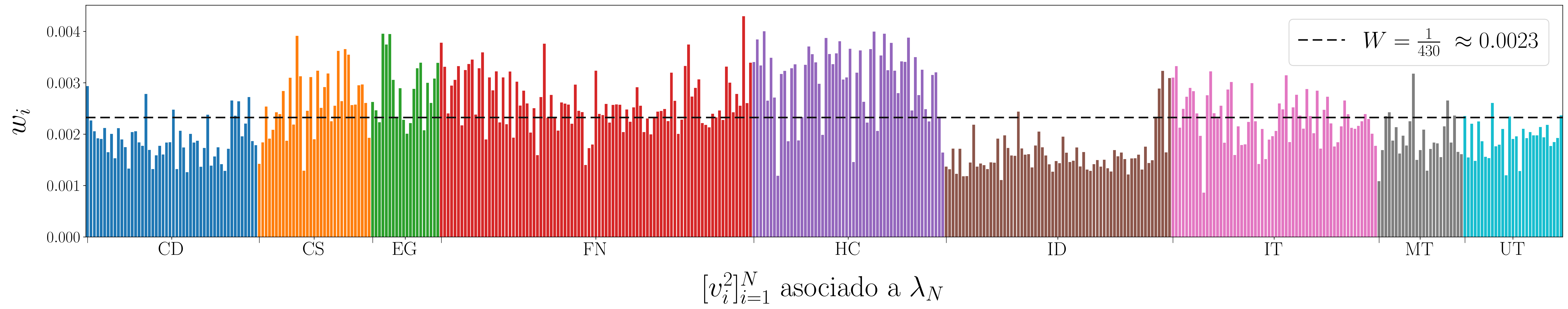}
    \caption[Averages of dominant eigenvectors $v_i^2$, \texorpdfstring{$q=20$}{q=20}]%
    {Temporal averages of the \emph{relative participation weights} $w_i \equiv v_i^{2}$, where $v$ is the eigenvector associated with $\lambda_{N}$, for $q=20$. 
    Due to the normalization of the eigenvector, it holds that $\sum_{i=1}^{N} w_i = 1$. 
    The dotted line indicates the uniform level $W=\frac{1}{430}\approx 0.0023$.
    For this case a value of $\boldsymbol{\mathrm{IPR}_{N}}=0.00254$ and $\boldsymbol{\mathrm{PR}_{N}}=394.1433$ was obtained.}
    \label{fig: EigVect cuadrado q20}
\end{figure}

\section{\textit{Clustering} of the Dominant Eigenvectors \texorpdfstring{$[v_i^2]_{i=1}^N$}{[vi²]}}
\label{sec: Clustering de los eigenvectores cuadrados}
\index{Clusters}

Fig.~\ref{fig:vn2-q40-k4} shows the \emph{temporal averages} of the \emph{relative participation weights} $w_i\equiv v_i^2$, where $v$ is the eigenvector associated with $\lambda_N$, for $q=40$ and $k=4$, grouped by sector (Table \ref{tab:sectores-gics-sp500}). 
By the normalization of the eigenvector, it holds that $\sum_{i=1}^{N}w_i=1$. In this case $N=430$, and the dotted line marks the uniform level $w=1/N=1/430\approx 0.0023$. 
In \emph{cluster} 1, \emph{HC} dominates, whereas in \emph{cluster} 2, \emph{CS}, \emph{EG}, \emph{FN}, and \emph{HC} stand out. 
In addition, the $10$ leading companies for each \emph{cluster} are presented in Table \ref{tab:Top10_q40_k4}.

\begin{figure}[ht]
  \centering
  \captionsetup[subfigure]{justification=centering,singlelinecheck=false}

  \begin{subfigure}[t]{0.98\textwidth}
    \centering
    \includegraphics[width=\linewidth]{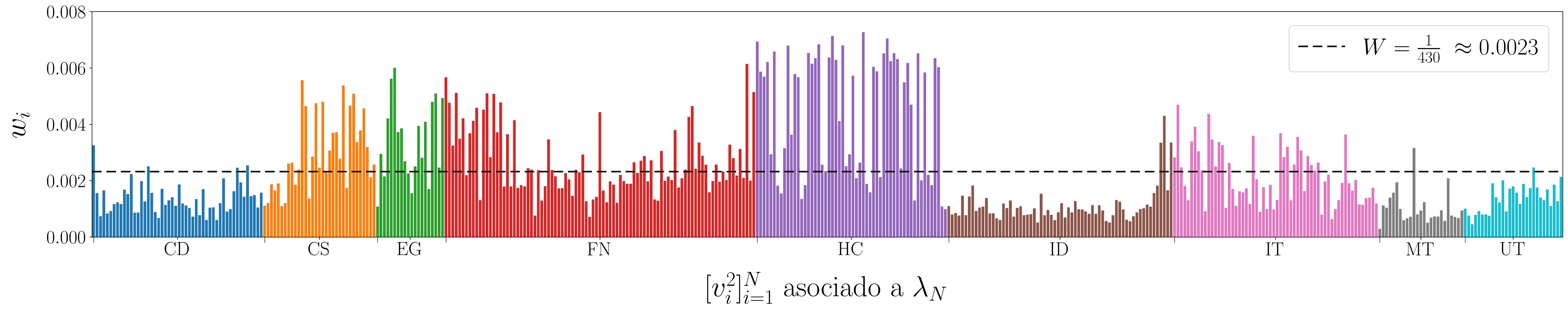}
    \caption{\emph{cluster} 1}
  \end{subfigure}

  \begin{subfigure}[t]{0.98\textwidth}
    \centering
    \includegraphics[width=\linewidth]{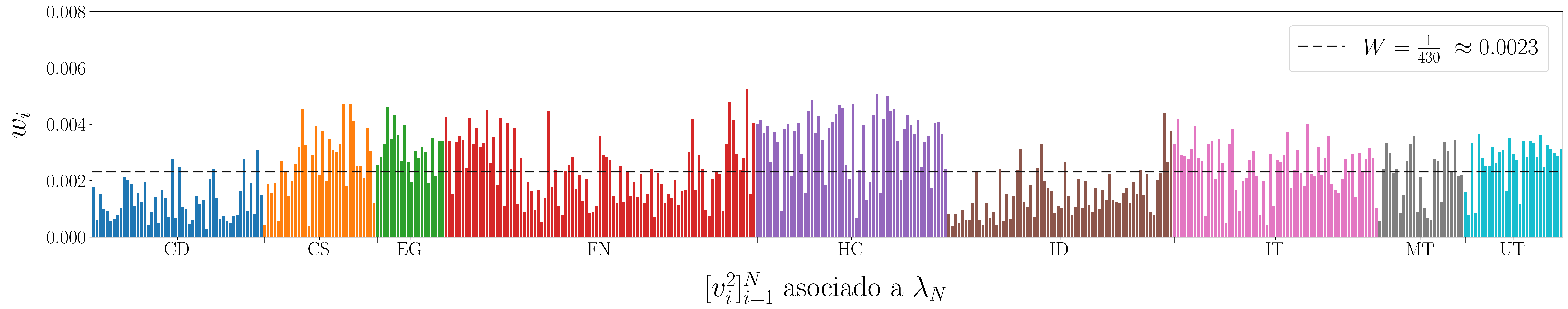}
    \caption{\emph{cluster} 2}
  \end{subfigure}

  \begin{subfigure}[t]{0.98\textwidth}
    \centering
    \includegraphics[width=\linewidth]{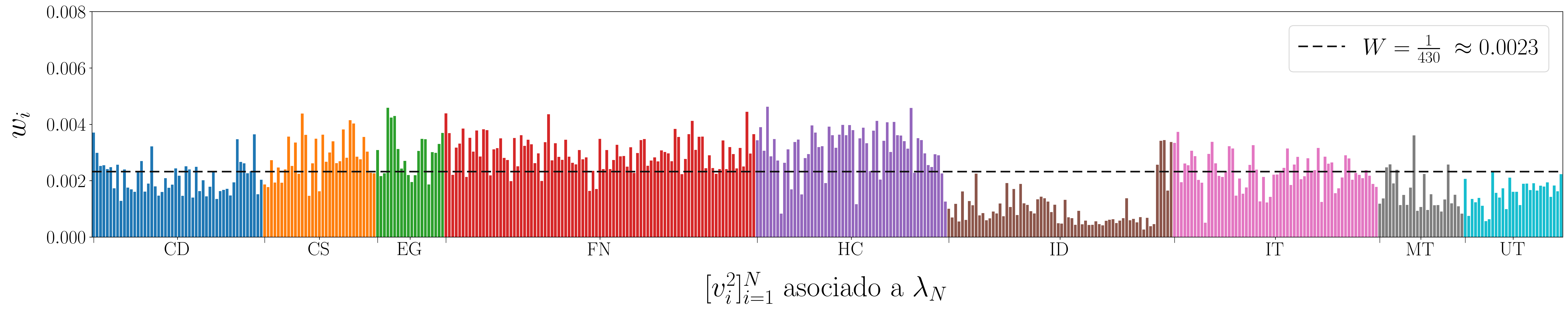}
    \caption{\emph{cluster} 3}
  \end{subfigure}

  \begin{subfigure}[t]{0.98\textwidth}
    \centering
    \includegraphics[width=\linewidth]{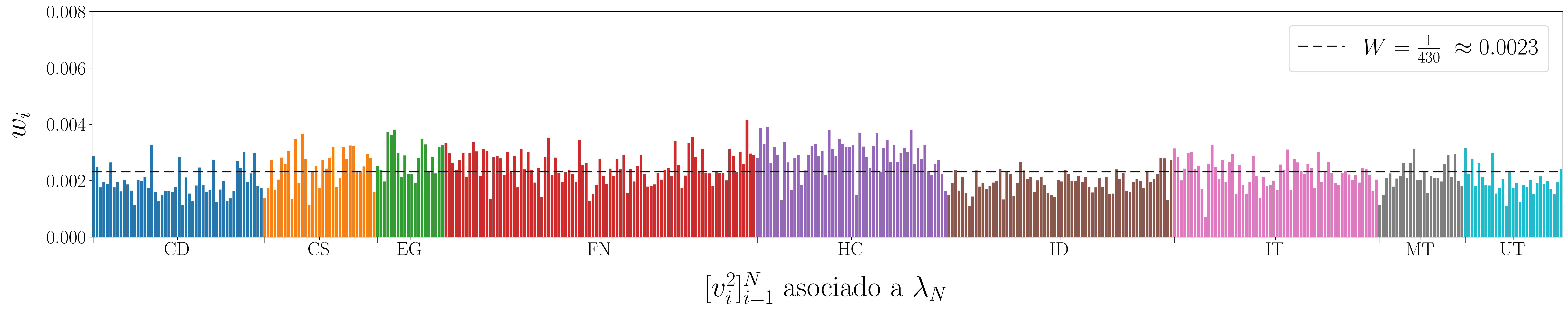}
    \caption{\emph{cluster} 4}
  \end{subfigure}

 \caption[Averages of dominant eigenvectors $v_i^2$, \texorpdfstring{$q=40\;\text{y}\; k=4$}{q=40, k=4}]%
{Temporal averages of the \emph{relative participation weights} $w_i\equiv v_i^2$, where $v$ is the eigenvector associated with $\lambda_N$, for $q=40$ and $k=4$. By the normalization of the eigenvector, it holds that $\sum_{i=1}^{N}w_i=1$, so that $w_i$ is interpreted as a \textbf{discrete participation measure} of company $i$. In this case $N=430$ and the dotted line marks the uniform level $W=1/N=1/430\approx 0.0023$. Both \emph{cluster} 1 and \emph{cluster} 2 follow a trend similar to Fig.~\ref{fig:vn2-q40-k5}.}
\label{fig:vn2-q40-k4}

\end{figure}

\section{Stationary Vectors of \emph{MS}}
\label{sec: Vectores estacionarios complementarios de MS}
\index{Vector!stationary}
Fig.~\ref{fig:vectoresEstacionariosMS_ab} presents complementary visualizations of the stationary vectors associated with the configurations $q=20$, $k=5$ and $q=40$, $k=4$, obtained from different accumulations carried out in the present work. In the case $q=20$ and $k=5$, the same behavior described in Fig.~\ref{fig:vectoresEstacionariosMS_q40_k5} is observed: for $\mathbf{C}^{2}$ and $\mathbf{C}^{3}$, the \emph{COVID} State is identified in \emph{cluster} 2 and maintains similar values between both constructions, whereas for $[v_i^2]_{i=1}^{N}$ it appears in \emph{cluster} 3. For its part, for $q=40$ and $k=4$, the \emph{COVID} State is located in \emph{cluster} 2 in $\mathbf{C}^{2}$, $\mathbf{C}^{3}$, and $[v_i^2]_{i=1}^{N}$, in agreement with Fig.~\ref{fig:evolucionMS_squared_q20q40}. In these constructions, a tendency is also appreciated to reduce the stationary probability associated with the first \emph{cluster} and to increase it in the highest-index \emph{cluster}; finally, the vector associated with $\lambda_{N}$ does not clearly reproduce the \emph{COVID} State.
\begin{figure}[htbp]
  \centering
  \begin{minipage}{0.98\textwidth}
    \centering
    \includegraphics[width=\textwidth]{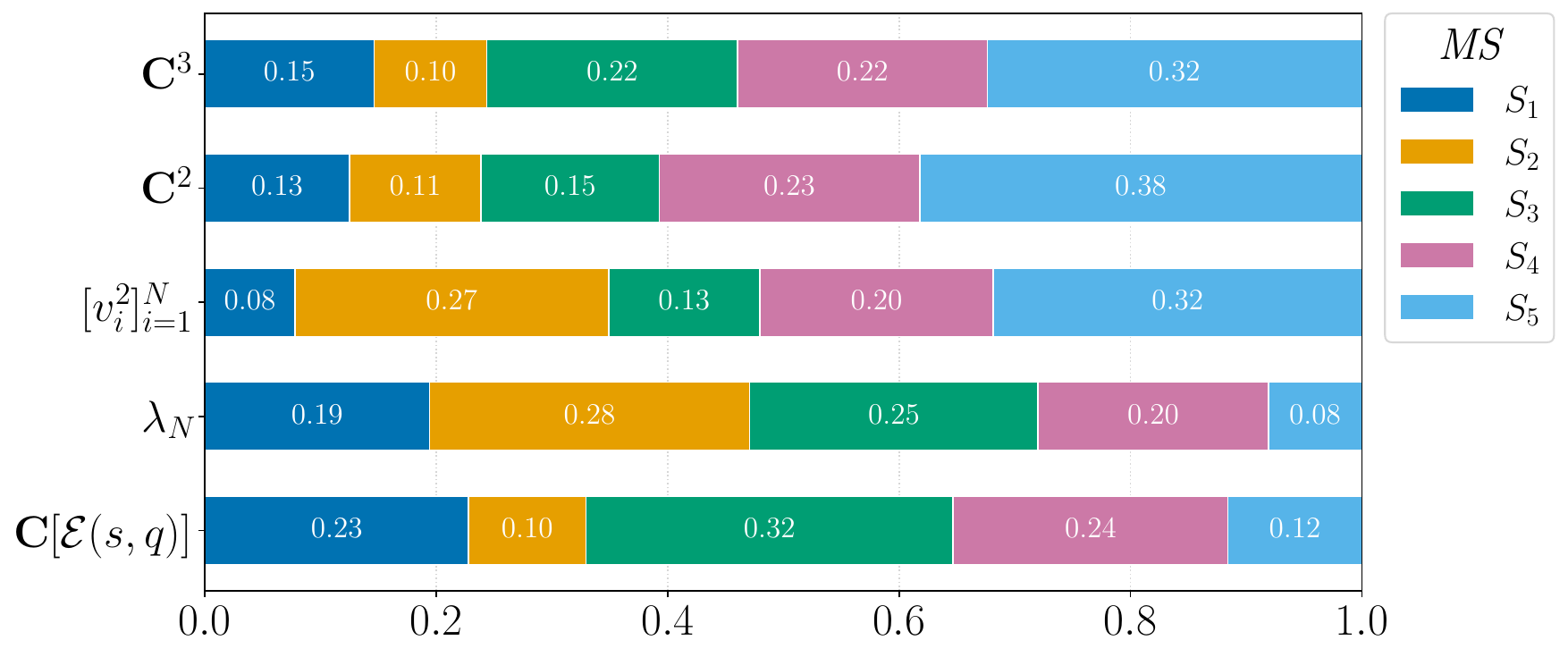}
    \caption*{\textbf{a)} $q=20$, $k=5$.}
  \end{minipage}
  \vspace{0.8cm}
  \begin{minipage}{0.98\textwidth}
    \centering
    \includegraphics[width=\textwidth]{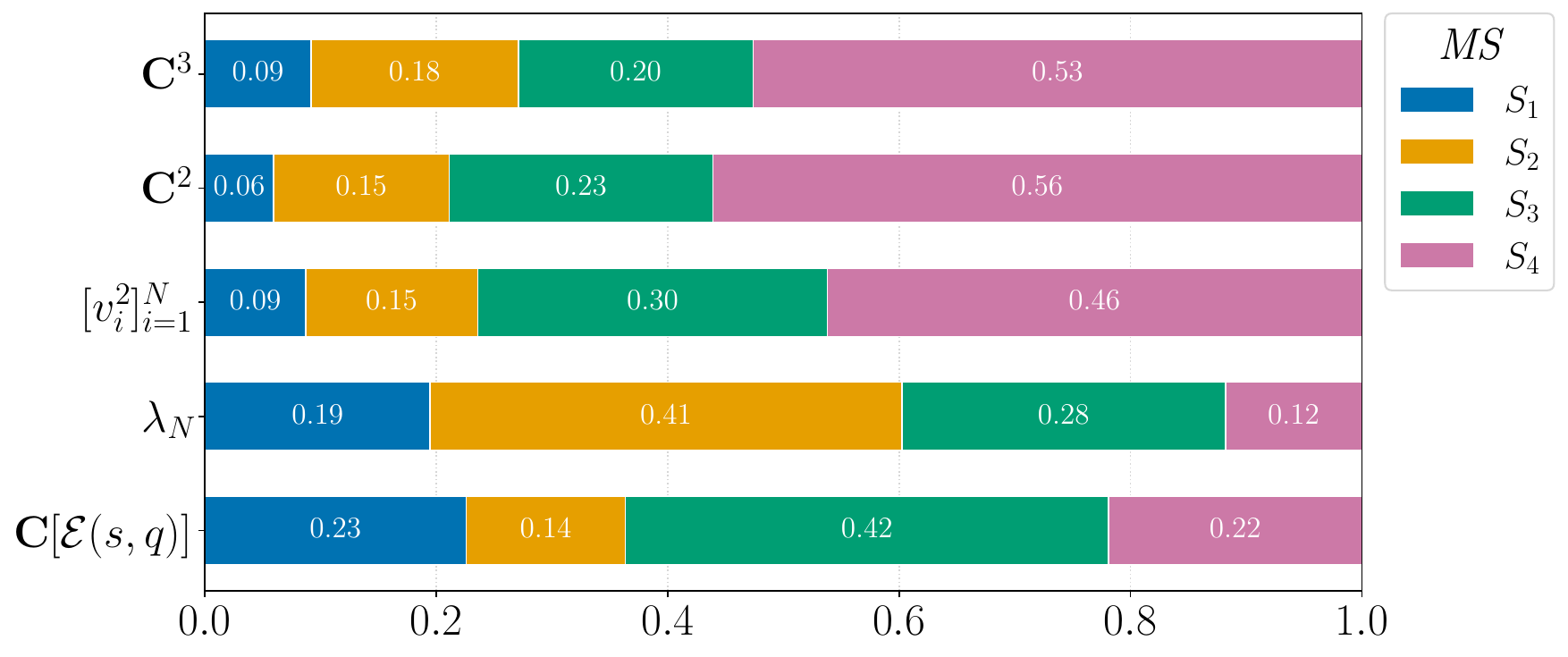}
    \caption*{\textbf{b)} $q=40$, $k=4$.}
  \end{minipage}
  \caption[Stationary vectors associated with \emph{MS}, $q=20$, $k=5$ and $q=40$, $k=4$.]{Fig.~\ref{fig:vectoresEstacionariosMS_ab} presents complementary visualizations of the stationary vectors associated with the configurations $q=20$, $k=5$ and $q=40$, $k=4$, obtained from different accumulations carried out in the present work. In the case $q=20$ and $k=5$, for $\mathbf{C}^{2}$ and $\mathbf{C}^{3}$ the \emph{COVID} State is identified in \emph{cluster} 2 and maintains similar values between both constructions, whereas for $[v_i^2]_{i=1}^{N}$ it appears in \emph{cluster} 3. For its part, for $q=40$ and $k=4$, the \emph{COVID} State is located in \emph{cluster} 2 in $\mathbf{C}^{2}$, $\mathbf{C}^{3}$, and $[v_i^2]_{i=1}^{N}$, in agreement with Fig.~\ref{fig:evolucionMS_squared_q20q40}. In these constructions, a tendency is also appreciated to reduce the stationary probability associated with the first \emph{cluster} and to increase it in the highest-index \emph{cluster}. Finally, the vector associated with $\lambda_{N}$ does not clearly reproduce the \emph{COVID} State.}
  \label{fig:vectoresEstacionariosMS_ab}
\end{figure}

\clearpage
\printbibliography[heading=bibintoc]

\clearpage
\printindex

\end{document}